\documentclass[11pt,lettersize]{article}
\usepackage[utf8]{inputenc}
\usepackage[T1]{fontenc}
\usepackage{microtype}
\usepackage{lmodern}

\usepackage[margin=1in]{geometry}

\usepackage{subcaption}
\usepackage{nicefrac}
\usepackage{amsmath}
\usepackage{amsfonts}
\usepackage{amssymb}
\usepackage{amsthm}
\usepackage{amstext}
\usepackage{thmtools}

\usepackage{relsize}
\usepackage{complexity}

\usepackage{csquotes}
\usepackage{xcolor}
\usepackage{graphicx}

\usepackage{mathtools}
\usepackage{xspace}
\usepackage{enumitem}
\usepackage[linesnumbered,ruled,vlined]{algorithm2e}
\usepackage{multicol}

\usepackage{soul}
\usepackage{bm}

\usepackage{url}
\usepackage[colorlinks,breaklinks]{hyperref} %
\hypersetup{citecolor=green!60!black,linkcolor=blue!90!black,urlcolor=blue!90!black}

\usepackage[skip=5pt plus 3pt minus 1pt,parfill=0pt]{parskip}
\makeatletter
\def\thm@space@setup{%
  \thm@preskip=\parskip \thm@postskip=0pt
}
\makeatother

\allowdisplaybreaks

\usepackage[capitalize,nameinlink]{cleveref}
\Crefname{algocf}{Algorithm}{Algorithms}
\crefname{algocfline}{line}{lines}

\usepackage[
 backend=biber,
 style=alphabetic,
 citetracker,
 hyperref=auto,
 maxcitenames=5,
 sortcites,
 sorting=nyt,
 maxbibnames=12,
 date=year,
 isbn=false,
 url=false,
 doi=false,
 eprint=false,
]{biblatex}

\newbibmacro{string+doiurlisbn}[1]{
  \iffieldundef{doi}{
    \iffieldundef{url}{
      #1
    }{
      \href{\thefield{url}}{#1}
    }
  }{
    \href{http://dx.doi.org/\thefield{doi}}{#1}
  }
}

\DeclareFieldFormat
[article,inbook,incollection,inproceedings,patent,thesis,unpublished]
  {title}{\usebibmacro{string+doiurlisbn}{#1}}

\newtheorem{theorem}{Theorem}
\newtheorem{corollary}[theorem]{Corollary}
\newtheorem{lemma}[theorem]{Lemma}
\newtheorem{claim}[theorem]{Claim}
\newtheorem{fact}[theorem]{Fact}
\newtheorem{observation}[theorem]{Observation}
\newtheorem{proposition}[theorem]{Proposition}
\theoremstyle{definition}
\newtheorem{definition}{Definition}

\crefname{invariant}{Invariant}{Invariants}

\newcommand{\opt}{\mathsf{opt}}
\newcommand{\Reduce}{\mathsf{RED}}
\newcommand{\reduce}{\mathsf{red}}
\newcommand{\OPT}{\mathsf{OPT}}
\newcommand{\ALG}{\mathsf{ALG}}

\newcommand{\FEW}{\textsf{FEW}\xspace}
\newcommand{\MANY}{\textsf{MANY}\xspace}

\newcommand{\hG}{\ensuremath{\bm\hat{G}}\xspace}

\newcommand{\cF}{\mathcal{F}}
\newcommand{\cT}{\ensuremath{\mathcal{T}}\xspace}

\newcommand{\cFour}{\ensuremath{\mathcal{C}_4}\xspace}
\newcommand{\cFive}{\ensuremath{\mathcal{C}_5}\xspace}
\newcommand{\cSix}{\ensuremath{\mathcal{C}_6}\xspace}
\newcommand{\cSeven}{\ensuremath{\mathcal{C}_7}\xspace}
\newcommand{\cEight}{\ensuremath{\mathcal{C}_8}\xspace}

\newcommand{\typA}{\ensuremath{\mathsf{A}}\xspace}
\newcommand{\typB}{\ensuremath{\mathsf{B}}\xspace}
\newcommand{\typBi}{\ensuremath{\mathsf{B1}}\xspace}
\newcommand{\typBii}{\ensuremath{\mathsf{B2}}\xspace}
\newcommand{\typC}{\ensuremath{\mathsf{C}}\xspace}
\newcommand{\typCi}{\ensuremath{\mathsf{C1}}\xspace}
\newcommand{\typCii}{\ensuremath{\mathsf{C2}}\xspace}
\newcommand{\typCiisubba}{\ensuremath{\mathsf{C2(i)}}\xspace}
\newcommand{\typCiisubb}{\ensuremath{\mathsf{C2(ii)}}\xspace}
\newcommand{\typCiisubc}{\ensuremath{\mathsf{C2(iii)}}\xspace}
\newcommand{\typCiii}{\ensuremath{\mathsf{C3}}\xspace}
\newcommand{\typesThreeVC}{\cT_3}
\newcommand{\typesTwoVC}{\cT_2}

\DeclarePairedDelimiter\ceil{\lceil}{\rceil}

\DeclarePairedDelimiter{\abs}{\lvert}{\rvert}
\DeclareMathOperator*{\argmin}{arg\,min}
\DeclareMathOperator*{\argmax}{arg\,max}

\newcommand{\cost}{\mathrm{cost}\xspace}
\newcommand{\credit}{\mathrm{cr}\xspace}
\newcommand{\col}{\mathrm{col}\xspace}
\newcommand{\loan}{\mathrm{loan}\xspace}
\newcommand{\eps}{\varepsilon}
\newcommand{\cre}{(\frac{1}{4}-\delta)}

\title{A Better-Than-$5/4$-Approximation for Two-Edge Connectivity}

\date{}

\author{Felix Hommelsheim\thanks{Universidad de Chile, Santiago, Chile.}
\and Alexander Lindermayr\thanks{Simons Institute for the Theory of Computing, UC Berkeley, USA. This work was done while at University of Bremen and supported by the ``Humans on Mars Initiative'', funded by the Federal State of Bremen and the University of Bremen.}
\and Zhenwei Liu\thanks{Zhejiang University, Hangzhou, China.}}

\begin{document}

\pagenumbering{gobble}
\maketitle

\begin{abstract}
The 2-Edge-Connected Spanning Subgraph Problem (2ECSS) is a fundamental problem in
survivable network design. Given an undirected $2$-edge-connected graph, the goal is
to find a $2$-edge-connected spanning subgraph with the minimum number of edges; a
graph is 2-edge-connected if it is connected after the removal of any single edge.
2ECSS is \APX-hard and has been extensively studied in the context of approximation
algorithms. Very recently, Bosch-Calvo, Garg, Grandoni, Hommelsheim, Jabal Ameli,
and Lindermayr showed the currently best-known approximation ratio of
$\nicefrac{5}{4}$ [STOC 2025]. This factor is tight for many of their techniques and
arguments, and it was not clear whether $\nicefrac{5}{4}$ can be improved.

We break this natural barrier and present a $(\nicefrac{5}{4} - \eta)$-approximation
algorithm, for some constant $\eta \geq 10^{-6}$. On a high level, we follow the
approach of previous works: take a triangle-free $2$-edge cover and transform it
into a 2-edge-connected spanning subgraph by adding only a few additional edges. For
$\geq \nicefrac{5}{4}$-approximations, one can heavily exploit that a $4$-cycle in
the 2-edge cover can ``buy'' one additional edge. This enables simple and nice techniques,
but immediately fails for our improved approximation ratio. To overcome this, we design two
complementary algorithms that perform well for different scenarios: one for few $4$-cycles
and one for many $4$-cycles. Besides this, there appear more obstructions when
breaching $\nicefrac54$, which we surpass via new techniques such as colorful bridge
covering, rich vertices, and branching gluing paths.
\end{abstract}

\newpage

\tableofcontents

\newpage

\pagenumbering{arabic}

\section{Introduction}\label{sec:Introductin}

Survivable network design studies how to compute cost-effective networks with connectivity properties
that are resilient to edge or node failures. One of the most important and fundamental problems
in this area is the 2-Edge-Connected Spanning Subgraph Problem (2ECSS). Here, we are given an
undirected 2-edge-connected (2EC) graph $G$, and the goal is to compute a 2-edge-connected spanning
subgraph (2ECSS) $H$ of $G$ with a minimum number of edges.

2ECSS is \APX-hard~\cite{F98,CL99}, and admits a simple $2$-approximation by augmenting a DFS tree. Substantial research has been devoted to prove smaller approximation ratios using non-trivial techniques.
After two early better-than-2 approximations~\cite{KV94,CSS01}, until recently, the best known approximation ratio remained $\nicefrac{4}{3}$, which was achieved independently with different techniques by \textcite{SV14} and \textcite{HVV19}.
In a breakthrough result, \textcite{GargGA23improved} improved the
ratio to $\nicefrac{118}{89} + \eps < 1.326$, which was quickly tuned to $1.3+\eps$ by
\textcite{kobayashi2023approximation} by replacing a 2-edge cover, which is used as a lower
bound in~\cite{GargGA23improved}, with a triangle-free 2-edge cover.
While this already reduced the technical complexity of \cite{GargGA23improved}, very recently,
\textcite{BGGHJL25} made a substantial leap forward   %
to a $\nicefrac{5}{4}$-approximation via an even simpler proof. %
Their key idea is to always glue smaller 2EC components to a large 2EC component, which eventually becomes a 2ECSS. %
\cite{BGGHJL25} is the result of merging the best parts of two independent works, by \textcite{BGJ24ecss} and \textcite{GHL24}, who achieved this approximation ratio with slightly different techniques.\footnote{\cite{GHL24} showed an approximation ratio of $\nicefrac{5}{4} + \eps$.}

This state of affairs strongly indicates that $\nicefrac{5}{4}$ is a natural barrier for \emph{many} known techniques used for 2ECSS.
Despite this, we breach the factor of $\nicefrac{5}{4}$ and show the following result.
We did not optimize the constant $\eta$ too much and aimed for a simple approach whenever possible.

\begin{theorem}\label{thm:main}
	There is a deterministic $(\nicefrac{5}{4} - \eta)$-approximation algorithm for 2ECSS that runs in polynomial time for some $\eta \geq 10^{-6}$.
\end{theorem}

While the quantitative improvement over \cite{BGGHJL25}
is only small, our result gives qualitative insights to approximating 2ECSS.
Many techniques and arguments in \cite{BGGHJL25} are tight for the approximation ratio $\nicefrac{5}{4}$,
and it was not clear whether this factor can be improved by
building up on the framework used
in \cite{GargGA23improved,BGJ24ecss,GHL24,kobayashi2023approximation,BGGHJL25}. This approach broadly consists of the following four steps:
\begin{enumerate}[nosep]
	\item \textbf{Preprocessing:} Reduce the input graph to a structured graph with additional properties.
	\item \textbf{Canonical 2-edge cover:} Compute a subgraph $H$ with $|H| \leq \opt$ (a minimum triangle-free 2-edge cover) and impose additional structural properties on it.
	\item \textbf{Bridge covering:} Cover the bridges of $H$ so that $H$ is composed of 2EC components.
	\item \textbf{Gluing:} Glue the components of $H$ together to a 2ECSS.
\end{enumerate}
Bridge covering and gluing may also happen simultaneously.
As our main technical contributions, we introduce new techniques to overcome the following challenges that appear in this framework for approximation ratios $\alpha < \nicefrac{5}{4}$:
\begin{itemize}
	\item For bridge covering and gluing, we keep track of the cost of our solution via a credit scheme where each edge of the initial 2-edge cover receives a credit of $\alpha - 1$. For $\alpha \geq \nicefrac{5}{4}$,  even $4$-cycles, which are the smallest 2EC components that might appear, have a total credit of at least $1$ and thus can ``buy'' one edge to glue themselves to another component---a crucial argument that was extensively used in previous work. This immediately fails for $\alpha < \nicefrac{5}{4}$. We address this by a careful case distinction on the number of $4$-cycles in the initial 2-edge cover, and then combine two very different approximation algorithms.
	\item In the preprocessing, we contract constant-sized subgraphs for which any solution must contain many edges. Intuitively, these \emph{contractible} subgraphs can be trivially included in an $\alpha$-approximation. For $\alpha \geq \nicefrac{5}{4}$, $5$-cycles are contractible if any solution must include at least $4$ edges, which streamlines many arguments. If $\alpha < \nicefrac{5}{4}$, $5$-cycles are only contractible if every 2ECSS includes all $5$ edges. This is much less likely, and makes this technique much weaker.
	\item As a consequence of the previous point, we cannot enforce as much structure on complex components (connected component with bridges) that appear in intermediate steps compared to \cite{BGGHJL25,GargGA23improved}. This leads to several complications when covering bridges. To fix this, we allow complex components initially to charge loans to satisfy invariants, which we carefully repay during bridge covering using a \emph{color encoding}. Moreover, we allow uncovered (isolated) vertices in our 2-edge cover during bridge covering, which can glue themselves ``on their own'' to a 2EC component later (\emph{rich vertices}).
\end{itemize}

We give a walk-through of our approach in \cref{sec:overview}, and describe how we address and solve these challenges in more detail.

As a byproduct of our approach, we also give the following conditional result.

\begin{theorem}
\label{thm:main2}
	If there exists a polynomial-time $\mu$-approximation algorithm for computing a minimum $\{3,4\}$-cycle-free 2-edge cover of a 2-edge-connected graph, then there exists a polynomial-time $(1.24 \mu + \varepsilon)$-approximation %
	algorithm for 2ECSS for every $\varepsilon > 0$.
\end{theorem}

It is known that a minimum triangle-free 2-edge cover can be computed in polynomial time~\cite{kobayashi2023approximation}, as it is equivalent to computing a triangle-free 2-matching~\cite{hartvigsen2024finding,P23}. For both, a PTAS is also available~\cite{BGJ24ptas,KN25ptas}. On the other side, the problem of computing a minimum $\{3,4,5\}$-cycle-free 2-matching is known to be NP-hard~\cite{CornuejolsP80}. The complexity status of a minimum $\{3,4\}$-cycle-free 2-matching and its cover variant are still open.

\subsection{Related Work}

In network augmentation problems, we are given a graph with a specific structure
(e.g., spanning tree, matching, forest, cactus, or disjoint paths) along with a set
of additional edges (called links). The task is to add a minimum-sized set of links to
the structure so that the resulting graph becomes spanning and 2-edge connected;
in 2ECSS we start with an empty graph.
For the (spanning) Tree Augmentation Problem, several better-than-$2$-approximations are
known~\cite{A19,CTZ21,CG18,CG18a,CN13,EFKN09,FGKS18,GKZ18,KN16simple,KN18lp,N03,N21,TZ21,TZ22}, where the currently best-known approximation ratio is 1.393~\cite{CTZ21}.
For the Forest Augmentation Problem, the first non-trivial approximation was obtained by \textcite{GJT22} with a ratio of 1.9973, which recently has been slightly improved~\cite{Hommelsheim25}.
For the Matching Augmentation Problem, \textcite{GHM23} gave a $\nicefrac{13}8$-approximation after several other non-trivial results~\cite{BDS22,CDGKN20,CCDZ23}.
In the $k$-Connectivity Augmentation Problem (kCAP), we are given a $k$-edge-connected graph and a set of links with the goal of adding as few links as possible to make the graph $(k+1)$-edge connected.
The k-Edge-Connected Spanning Subgraph Problem for $k\geq 2$ is another natural generalization of 2ECSS (see, e.g.,~\cite{CHNSS22, CT00,GG12, GGTW09,HershkowitzKZ24}).

All the above problems have a natural weighted version.
A general result by Jain \cite{J01} yields a $2$-approximation. %
For the Weighted Tree Augmentation Problem, the current best approximation ratio is $1.5+\varepsilon$ by \textcite{TZ22}.
Recently, \textcite{TZ23} also obtained a $(1.5+\eps)$-approximation algorithm for the weighted version of kCAP.
For Weighted 2ECSS, obtaining a better-than-2-approximation is a major open problem.

The $2$-Vertex-Connected Spanning Subgraph Problem (2VCSS) is the vertex-connectivity version of 2ECSS.
Here, we are given a 2EC graph $G$, and the goal is to find the 2-vertex-connected\footnote{A graph is $k$-vertex-connected (kVC) if it remains connected after the removal of an arbitrary subset of $k-1$ vertices.}
spanning subgraph $H$ with minimum number of edges. There are various trivial 2-approximations. \Textcite{KV94} gave the first non-trivial $\nicefrac{5}{3}$-approximation, which was improved to $\nicefrac{3}{2}$~\cite{GVS93} and $\nicefrac{10}{7}$~\cite{HV17}. Recently, \textcite{BGJ23} gave a $\nicefrac{4}{3}$-approximation.

Closely related is also the Traveling Salesperson Problem (TSP).
The best-known approximations for metric TSP~\cite{christofides2022worst,karlin2021slightly} make use of the Help-Karp LP-relaxation of the problem, which has an integrality gap of $\nicefrac{4}{3}$. %
The example of this integrality gap is actually an instance of graphic TSP:
given a graph $G$, the task is to find an Eulerian multigraph with a minimum the number of edges. 2ECSS is a relaxation of graphic TSP.
Seb\"{o} and Vygen~\cite{SV14} proved the currently best approximation guarantee of $\nicefrac{7}{5}$ for graphic TSP. Interestingly, they use very similar techniques %
to also obtain a $\nicefrac{4}{3}$-approximation for 2ECSS, which highlights the close connection between both problems.
Thus, advancements for 2ECSS may lead to better bounds for graphic TSP.

\subsection{Preliminaries}\label{sec:preliminaries}
We use standard graph theory notation. For a graph $G$, $V(G)$ and $E(G)$ refers to its vertex and edge sets, respectively. %
For a graph $G$, and a vertex set $S\subseteq V(G)$, $G[S]$ refers to the subgraph of $G$ induced on the vertex set $S$. We use components of a graph to refer to its maximal connected subgraphs.
A $k$-vertex cut in a graph for $k\geq 1$ refers to a set of $k$ vertices such that deleting them results in increasing the number of components of the graph.
We call the single vertex of a 1-vertex cut a cut vertex.
A graph with no cut vertices is called 2-vertex connected (2VC).
A bridge in a graph refers to an edge in the graph whose removal results in increasing the number of components in the graph. Note that a connected graph without bridges is 2-edge connected (2EC). %
A cycle refers to a simple cycle.
We use $\mathcal{C}_i$ to represent a cycle on $i$ vertices.
Given a graph $G$ and a vertex set $S\subseteq V(G)$, $G|S$ denotes the graph obtained from contracting the vertex set $S$ into a single vertex. Graph contraction may give rise to parallel edges and self-loops. The edges of $G$ and $G|S$ are in one-to-one correspondence. For a graph $G$ and a subgraph $H$, we use $G|H$ to denote $G|V(H)$.
For a graph $G$, we use $|G|$ to denote $|E(G)|$.
For a 2EC graph $G$, we use $\OPT(G)$ to denote a 2ECSS of $G$ with the minimum number of edges. We define $\opt(G)\coloneq|\OPT(G)|$ (when $G$ is clear from the context we may just say $\opt$ and $\OPT$ instead).

\section{Overview of Our Approach}\label{sec:overview}

In this section, we present our algorithm and give an outline of our analysis. As an input to the problem, we are given a 2EC graph $G$.

\subsection{Reduction to Structured Graphs}
\label{sec:overview:preprocessing}
The first step of our approach is to reduce the input graph
to a graph with more structural properties while only sacrificing little in the approximation ratio. We will later use these structural properties extensively to prove a good approximation for 2ECSS on structured graphs.

In \cite{GargGA23improved,GHL24,BGGHJL25},
$(\alpha,\eps)$-structured graphs have sufficiently many vertices and
do not contain $\alpha$-contractible subgraphs, irrelevant edges, non-isolating 2-vertex cuts,  or large 3-vertex cuts.
We defer definitions of these previously considered structures to \Cref{app:prep}.
In this work, we impose additional connectivity properties of structured graphs by preprocessing $\mathcal C_k$ cuts and large $4$-vertex cuts. A \emph{$\mathcal C_k$ cut} is a set of vertices $\{v_1,\ldots,v_k\}$ that forms a cycle in $G$ whose removal increases the number of connected components in $G$.
We define structured graphs as follows.

\begin{definition}
	[($\alpha, \varepsilon)$-structured graph]\label{def:structured}
	Given $\alpha \geq 1$ and $\varepsilon > 0$, a graph $G$ is $(\alpha, \varepsilon)$-\emph{structured} if it is simple, 2VC, it contains at least $\nicefrac 8 \varepsilon$ vertices, and it does not contain irrelevant edges,
	\begin{enumerate}[nosep,]
		\item $\alpha$-contractible subgraphs of size at most $\nicefrac 8 \varepsilon$,
		\item non-isolating 2-vertex cuts,
		\item large 3-vertex cuts (each side has $\geq \ceil{\frac{2}{\alpha - 1}} - 1$ vertices; in this paper, this is equal to $8$),
		\item large $\mathcal C_k$ cuts for $k \in \{4,5,6,7,8\}$ (each side has $\geq \ceil{\frac{k}{\alpha - 1}} - 1$ vertices; in this paper, this is equal to $4k$), and
		\item large $4$-vertex cuts (each side has $\geq \ceil{\frac{6}{\alpha - 1}} - 1$ vertices; in this paper, this is equal to $24$).
	\end{enumerate}
\end{definition}

\cite{GargGA23improved} gave an approximation-preserving reduction for approximation ratios $\alpha \geq \nicefrac 6 5$ to instances without parallel edges, self loops, irrelevant edges, cut-vertices, small $\alpha$-contractible subgraphs, and non-isolating 2-vertex cuts. This reduction was also used in \cite{BGGHJL25}.
Afterwards, \cite{GHL24} significantly extended this reduction to instances without
large 3-vertex cuts for approximation ratios $\alpha \geq \nicefrac 5 4$, which loses a factor of $1+\eps$ in the final
approximation ratio.

Our contribution for this step is twofold. First, we show that large 3-vertex cuts can also be avoided for approximation ratios %
$\alpha \geq \nicefrac 6 5$
by adapting the requirement on the number of vertices on each side of the cut. Second, we show that we additionally reduce to instances without large $\mathcal C_k$ cuts for $k \in \{4,5,6,7,8\}$ and large $4$-vertex cuts for approximation ratios $\alpha \geq \nicefrac 6 5$.
The proof of the following reduction is given in \cref{app:prep}.

\begin{restatable}{lemma}{lemmaReduction}\label{lem:reduction-to-structured}
	For all $\alpha \in [\nicefrac{6}{5},\nicefrac{5}{4}]$ and $\varepsilon \in (0, \nicefrac{1}{100}]$, if there exists a deterministic polynomial-time $\alpha$-approximation algorithm for 2ECSS on ($\alpha, \varepsilon)$-structured graphs, then there exists a deterministic polynomial-time $(\alpha + 4 \varepsilon)$-approximation algorithm for 2ECSS.
\end{restatable}

For the proof of \Cref{thm:main}, we will use $\alpha \coloneq \nicefrac{5}{4} - 10^{-5}$ and $\eps \coloneq 10^{-6}$.
We will show in the following the existence of a deterministic polynomial-time $\alpha$-approximation algorithm for 2ECSS on structured graphs. %
Thus, from now on we simply write contractible and structured instead of $\alpha$-contractible and $(\alpha, \varepsilon)$-structured, respectively.

A crucial benefit of structured graphs is the existence of large matchings between vertex bipartitions. In the analysis of our algorithm for structured graphs, we will make generous use of the following matching lemmas, similar to \cite{GargGA23improved,GHL24,BGGHJL25}.

\begin{lemma}[$3$-Matching Lemma~\cite{GargGA23improved}]
	\label{lem:3-matching}
	Let $G$ be a structured graph.
	Consider any partition $(V_1, V_2)$ of $V$ such that for each $i \in \{1, 2\}$, $|V_i| \geq 3$, and if $|V_i| = 3$, then $G[V_i]$ is a triangle.
	Then, there exists a matching of size $3$ between $V_1$ and $V_2$.
\end{lemma}

\begin{lemma}[4-Matching Lemma~\cite{GHL24}]
	\label{lem:4-matching}
	Let $G$ be a structured graph.
	Consider any partition $(V_1, V_2)$ of $V(G)$ such that for each $i \in \{1, 2\}$, $\abs{V_i} \geq \ceil{\frac{2}{\alpha - 1}} + 2 = 11$.
	Then, there exists a matching of size $4$ between $V_1$ and $V_2$ in $G$.
\end{lemma}

\subsection{Canonical 2-Edge Cover}

From now on, we assume that we are given a structured input graph $G$.
Our next step is to compute a lower bound on the size of a minimum 2ECSS of $G$.
For this, we follow the approach of \cite{kobayashi2023approximation,BGGHJL25}, and compute a minimum triangle-free 2-edge cover $H$ of $G$, which can be done in polynomial time~\cite{kobayashi2023approximation,hartvigsen2024finding}.
Formally, a 2-edge cover  is a set of edges of a graph $G$ such that every vertex of $G$ is incident to at least two edges of the cover.
A 2-edge cover is triangle-free if it does not contain a triangle as a component. We have $|H| \leq \opt(G)$, because every feasible solution to 2ECSS is a 2-edge cover and triangle-free, because it is connected and contains at least $4$ vertices.

Given an arbitrary minimum triangle-free 2-edge cover $H$, we next transform it in polynomial time into a \emph{canonical} minimum triangle-free 2-edge cover $H_0$ of $G$ to impose the following properties.

\begin{definition}[complex]
	\label{def:complex}
	Let $H$ be a $2$-edge cover of some graph.
	We call $H$ \emph{bridgeless} if it contains no bridges, i.e., all the components of $H$ are $2$-edge connected.
	A component of $H$ is \emph{complex} if it contains a bridge.
	Inclusion-wise maximal $2$-edge-connected subgraphs of $H$ are \emph{blocks}.
	A block $B$ of some complex component $C$ of $H$ is called \emph{pendant} if $C \setminus V(B)$ is connected.
	Otherwise it is non-pendant.
	Note that by contracting each block of a complex component to a single vertex, we obtain a tree whose leaves correspond to pendant blocks.
\end{definition}

\begin{definition}[canonical]
	\label{def:canonicalD2}
	A $2$-edge cover $H$ is \emph{canonical} if
	\begin{itemize}[nosep]
		\item each non-complex component of $H$ is either a $\mathcal C_i$, for $4 \leq i \leq 8$, or contains at least $9$ edges,
		\item for every complex component, each of its pendant blocks $B$ either
		      \begin{itemize}[nosep]
			      \item consists of a $\cFive$, labeled clockwise as ${b_1,b_2,b_3,b_4,b_5}$, with $b_1$ incident to the bridge, and $b_2,b_5$ have degree $2$ in $G$. We call such a $\cFive$ a \emph{nice} $\cFive$,
			      \item or has at least $6$ vertices,
		      \end{itemize}
		\item for every complex component, each of its non-pendant blocks $B$ has at least $5$ edges,
		\item there are no two non-complex components of $H$ that are both $\cFour$'s and can be merged into a single $\cEight$ by removing two edges of $H$ and adding two edges, and
		\item $H$ contains a component that has at least $32$ vertices.
	\end{itemize}
\end{definition}

\begin{figure}
	\centering
	\begin{subfigure}[t]{0.49\textwidth}
		\centering
		\includegraphics[width=\textwidth]{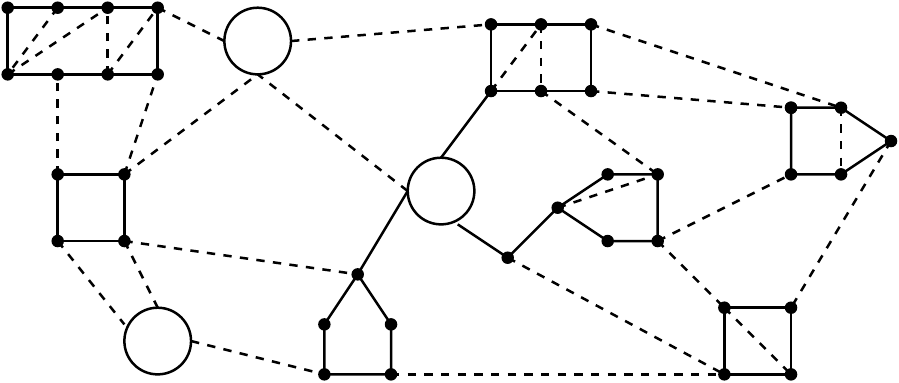}
		\caption{A canonical 2-edge cover $H$ (solid edges) of a graph $G$ (solid and dashed edges). Vertices of $G$ are black-filled small circles, 2EC components of $H$ with at least $9$ vertices and blocks of complex components are white-filled big circles.}
		\label{fig:canonical}
	\end{subfigure}%
	\hfill
	\begin{subfigure}[t]{0.49\textwidth}
		\centering
		\includegraphics[width=0.7\textwidth]{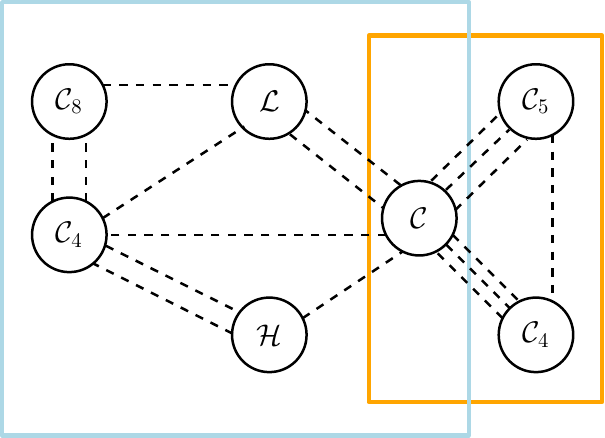}
	\caption{The component graph $\hG_H$ of the canonical 2-edge cover $H$ shown in \Cref{fig:canonical}. Here, $\mathcal C$ denotes the node of a complex component, $\mathcal L$ of a large component, and $\mathcal H$ of a huge component. The left $5$ nodes and the $3$ right nodes each form a segment.}
	\label{fig:component-graph}
	\end{subfigure}%
	\caption{An example of a 2-edge cover and its component graph.}
\end{figure}

Our definition of canonical is similar to~\cite{GargGA23improved,BGGHJL25}, except that in this paper components with $8$ edges have to be a $\cEight$, each pendant block of a complex component that is a $\cFive$ has to be nice, and we cannot further merge two $\cFour$'s into a single $\cEight$.
Merging two $\cFour$'s whenever possible is new and crucial to later steps of our pipeline.
In \Cref{sec:canonical}, we show that we can compute a canonical 2-edge cover from a minimum triangle-free 2-edge cover in polynomial time.

\begin{restatable}{lemma}{lemmaCanonicalMain}
	\label{lem:canonical-cover:main}
	Given a structured graph $G$, we can compute a canonical $2$-edge cover $H_0$ of $G$ in polynomial time, where $H_0$ is also a minimum triangle-free $2$-edge cover of $G$.
\end{restatable}

\subsection{Credit Schemes, Bridge Covering, and Gluing}

There are two reasons why a canonical 2-edge cover $H_0$ of a structured graph $G$ is not already a 2ECSS of $G$: First, $H_0$ might have multiple components. Second, some of these components might have bridges.
In the following two subsections, we will present two algorithms that transform $H_0$ into a 2ECSS of $G$ by covering bridges and gluing components together.
Throughout these transformations, we will keep the property that the intermediate solution is a canonical 2-edge cover.

To keep track of the size of our solution throughout these modifications, we use a \emph{credit-based argument}.
Initially, we assign each edge in $H_0$ a certain credit, and then redistribute the credits to the
(sub)-components according to a \emph{credit scheme}, which we maintain throughout the algorithms. We use different credit schemes for each algorithm. For a 2-edge cover $H$, we define $\credit(H)$ as the sum of all credits assigned to $H$, and the \emph{cost} of $H$ as $\cost(H) = |H| + \credit(H)$.
If we can modify $H_0$ to a 2ECSS $H$ in polynomial time such that (i) $\cost(H_0) \leq \alpha \cdot |H_0|$ and (ii) $\cost(H) \leq \cost(H_0)$,
we have that $H$ is a $\alpha$-approximation on structured graphs, since the credits are non-negative.

To modify $H_0$ to a 2ECSS, we mainly work on the component graph $\hG_H$ of $G$.

\begin{definition}[component graph and segments]
	The \emph{component graph} $\hG_H$
	w.r.t.\ $G$ and $H$ is the multigraph obtained from  $G$ by contracting the vertices of each 2EC component of $H$ into a single node.
	We call the vertices of $\hG_H$ \emph{nodes} to distinguish them from the vertices in $G$. There is a one-to-one correspondence between the edges of $\hG_H$ and $G$. Furthermore, there is a one-to-one correspondence between the nodes of $\hG_H$ and the components of $H$. We will denote the component corresponding to a node $A$ by $C_A$. We say a \emph{segment} of $\hG_H$ is an inclusion-wise maximal 2-vertex-connected subgraph of $\hG_H$.
\end{definition}

For ease of exposition, we remove all self-loops from the component graph $\hG_H$.
Observe that the component graph $\hG_H$ is 2EC, because $G$ is 2EC and contracting a set of vertices or deleting self-loops in a 2EC graph results in a 2EC graph. %
We can break $\hG_H$ into segments in polynomial time such that two segments only overlap at a cut node of $\hG_H$ \cite{West2001,KorteV2002}.

\begin{definition}[large and huge components]\label{def:large-huge}
	We say that a component $C$ of a canonical 2-edge cover is \emph{huge} if $|V(C)| \geq 32$, and \emph{large} if it is 2EC and $|V(C)| \geq 9$.
\end{definition}

\subsection{Case Distinction on the Number of $\cFour$s in $H_0$ (Proof of \Cref{thm:main})}

Next, we will present two different algorithms that convert a canonical 2-edge cover $H_0$ into a 2ECSS of $G$. Each algorithm handles a different case depending on the number of $\cFour$'s in $H_0$. To this end, let $\beta \in [0,1]$ such that the number of edges of $\cFour$'s in $H_0$ is $\beta |H_0|$.
We distinguish two cases:
\begin{itemize}[nosep]
    \item \textbf{Case} \textsf{FEW} (\Cref{lem:main:FEW}, see \Cref{sec:few}): If $\beta \leq 1 - \frac{1}{1000}$, there is an algorithm that computes in polynomial time a 2ECSS $H$ of $G$ such that $|H| \leq  \left( \frac{5}{4} - (1-\beta) \frac{1}{100} \right) \cdot |H_0| \leq 1.24999 \cdot \opt$.
    \item \textbf{Case} \textsf{MANY} (\Cref{lem:main:MANY}, see \Cref{sec:many}): If $\beta > 1 - \frac{1}{1000}$, there is an algorithm that computes in polynomial time a 2ECSS $H$ of $G$ such that $|H| \leq \left( \frac{5}{4} - \frac{1}{28} \right) \cdot \beta |H_0| + 5 (1 - \beta) \cdot |H_0| + \frac{1}{35} \cdot \opt \leq  1.24999 \cdot \opt$.
\end{itemize}
Here we use that $|H_0| \leq \opt$. Given these two algorithm, we can easily prove \Cref{thm:main,thm:main2}.

\begin{proof}[Proof of \Cref{thm:main}]
Combining \Cref{lem:main:FEW,lem:main:MANY} yields a polynomial-time $(\nicefrac{5}{4} - 10^{-5})$-approximation algorithm for 2ECSS on structured graphs.
We now apply \Cref{lem:reduction-to-structured} with $\alpha = \nicefrac{5}{4} - 10^{-5}$ and $\varepsilon = 10^{-6}$ to conclude \Cref{thm:main}.
\end{proof}

\begin{proof}[Proof of \Cref{thm:main2}]
	We first preprocess the input graph $G$ using \Cref{lem:reduction-to-structured} using $\alpha = 1.24$. In particular, we  do not exclude large $4$-vertex cuts and large $\mathcal{C}_k$ cuts for $4 \leq k \leq 8$, because we do not need them for \FEW. Then, we use the given algorithm to compute a $\mu$-approximation $H_0$ of a minimum $\{3,4\}$-cycle-free 2-edge cover of $G$.
Note that a minimum $\{3,4\}$-cycle-free 2-edge cover of a 2-edge-connected graph $G$ is a lower bound on a feasible 2ECSS of $G$.
Hence, $|H_0| \leq \mu \cdot \opt$.
Now, we make $H_0$ canonical (\Cref{lem:canonical-cover:main}) and apply the algorithm for \FEW with the starting solution $H_0$ and $\beta = 0$. Hence, we obtain a 2ECSS $H$ with $|H| \leq \left( \frac{5}{4} - \frac{1}{100} \right) |H_0| = 1.24 \mu \cdot \opt$ (cf.~\Cref{lem:main:FEW}).
\end{proof}

The next two subsections give an outline on the proofs of \Cref{lem:main:FEW} and \Cref{lem:main:MANY}, respectively.

\subsection{Algorithm for \FEW}\label{sec:few}

First, we present an approach that gives a good approximation ratio for \FEW, i.e., $\beta \leq 1 - \frac{1}{1000}$, and prove the following lemma. We defer proofs and further details to \Cref{sec:few:main}.

\begin{restatable}{lemma}{lemmamainFEW}
\label{lem:main:FEW}
Given a structured graph $G$ and a canonical $2$-edge cover $H_0$, in polynomial time we can compute a 2EC spanning subgraph $H$ of $G$ such that
$$
|H| \leq \frac{5}{4} \beta |H_0| + \left( \frac{5}{4} - \frac{1}{100} \right) (1- \beta) |H_0| = \left( \frac{5}{4} - (1-\beta) \frac{1}{100} \right) |H_0|                     \ ,
$$
where  $\beta \in [0,1]$ such that the number of edges of $\cFour$'s in $H_0$ is $\beta |H_0|$.
\end{restatable}

On a high-level, we use a similar framework as in \cite{BGGHJL25}:
We start from a canonical $2$-edge cover $H_0$ and first modify it to a bridgeless one.
This is particularly challenging due to insufficient credits on some complex components and $\cFour$'s.
Finally, we glue the 2EC components together.

A crucial innovation over \cite{BGGHJL25} is that
we do not distribute credits uniformly among the edges of $H_0$; otherwise we could not give a credit of $1$ to each $\cFour$, which is important for later arguments.
Instead, in \FEW we can afford to give each $\cFour$ a credit of 1
initially: we give each edge $e$ in $H_0$ a credit of $\nicefrac{1}{4} - \delta$ if $e$ is not in a $\cFour$ of $H_0$, and a credit of $\nicefrac{1}{4}$ if $e$ is in a $\cFour$ of $H_0$, where $\delta = \nicefrac{1}{100}$. Hence, the total credit that we initially distribute to $H_0$ is $
	\left( \nicefrac 1 4 - \delta \right) |H_0| + \delta \cdot \beta |H_0| = \left(\nicefrac{1}{4} - (1-\beta)\delta\right) |H_0|$, as desired.

Moreover, we slightly relax the condition that the current solution $H$ is a canonical $2$-edge cover by allowing some isolated vertices, which we will later give enough credits to buy edges. We call them \emph{rich vertices}.
We call such a $2$-edge cover a \emph{weakly canonical $2$-edge cover.}

Our goal is to redistribute the initial credits on the edges according to the following credit scheme:

\begin{restatable}[credit scheme for \FEW]{definition}{definitionFEWcredit}
	\label{def:initial_credit}
	Let $H$ be a weakly canonical $2$-edge cover of a structured graph $G$. We keep the following credit invariants:
	\begin{itemize}[nosep]
		\item Let $C$ be a 2EC component of $H$ that is a $\mathcal C_4$. Then $ \credit(C) = 1$.
		\item Let $C$ be a 2EC component of $H$ that is a $\mathcal C_i$, for $5 \leq i \leq 8$. Then $\credit(C) = (\nicefrac{1}{4} - \delta) i$.
		\item Let $C$ be a 2EC component of $H$ with at least $9$ edges. Then $\credit(C) = 2$.
		\item Let $C$ be a complex component of $H$. Then $\credit(C)=1$.
		\item Let $B$ be a block of a complex component $C$ of $H$. Then $\credit(B) =1$.
		\item Let $e$ be a bridge of a complex component $C$ of $H$. Then $\credit(e) = \nicefrac{1}{4} - \delta$.
		\item Let $v$ be a rich vertex (component) of $H$. Then $\credit(v) = 2(\nicefrac{5}{4}-\delta)$.
	\end{itemize}
\end{restatable}

Ideally, we want each bridge to keep its credits, and each block and each component to get a credit of $1$.
However, for complex components, this may not be possible by only using the credits of the edges.
Hence, we take some \emph{loans} that we will pay back later.
The key for this is that we may not need all the credits of the bridges to perform the bridge covering: we will buy edges using the credits of some bridges, while the unused credits of the remaining bridges will be used to pay back the loans. We next give more details on the bridge covering and how we repay the loans.

\subsubsection{Bridge Covering}

To make the loan repayments transparent, we maintain information about which bridges need to pay back which amount of loans via a color encoding:
We initially create a distinct color $\col(C)$ for each complex component $C$, and give each bridge $e$ of $C$ the same color $\col(e) = \col(C)$, which is fixed and does not change. Let $\cF$ be the set of all colors.
For each color $f \in \cF$, and for each 2-edge cover $H$ that we consider in the process of bridge covering, we will maintain $\loan_f(H)$. We write $\loan(H) = \sum_{f \in \cF} \loan_f(H)$.
Now $\credit(H)$ has $2$ parts: $\loan(H)$ credits that are borrowed and $\credit(H)-\loan(H)$ credits that are self-owned.
We call a color $f$ \emph{deficient} for $H$ if $\loan_f(H) > 0$.

Throughout the bridge covering process, we maintain the following invariants.

\begin{restatable}[color and loan invariant]{invariant}{definitionFEWcolorloan}
\label{def:loan_invariant}
	Let $H$ be a weakly canonical 2-edge cover of a structured graph $G$.
	We always maintain the following invariants:
	\begin{itemize}[nosep]
		\item[(i)] For each color $f \in \cF$, the bridges of color $f$ are in the same complex component $C$ of $H$. \label[Definition(i)]{def:loan_invariant:i}
		\item[(ii)] Let $f$ be a deficient color and $F$ be the set of bridges in color $f$. \label{def:loan_invariant:ii} %
		      \begin{itemize}[nosep]
			      \item $F$ induces a connected subtree $T_C(F)$ of $T_C$, where $T_C$ is obtained from $C$ by contracting its blocks.
			      \item $T_C(F)$ has at most $4$ leaves and its leaves are block nodes.
				\item If $T_C(F)$ has $4$ leaves, then $\loan_f(H)\leq 20 \delta$.
				\item If $T_C(F)$ has $3$ leaves, then
				$\loan_f(H) \leq \nicefrac{1}{4}+11 \delta$.
				\item If $T_C(F)$ has $2$ leaves, then $\loan_f(H) \leq \nicefrac{1}{2}+10\delta$. Additionally, if $\loan_f(H) > 2(\nicefrac{1}{4} - \delta)$, i.e., $\loan_f(H)$ is larger than the credits of $2$ bridges, then the 2 leaf blocks are nice $\cFive$'s.
		      \end{itemize}
		\item[(iii)] For each color $f$, $\loan_f(H)$ is at most the total credit of bridges of $H$ of color $f$. \label{def:loan_invariant:iii}
	\end{itemize}
\end{restatable}

One can check that $H_0$ does not satisfy \cref{def:loan_invariant}(iii) if there is a complex component $C$ consisting of exactly two blocks (which are nice $\cFive$'s) and exactly two bridges.
For this case, to satisfy \cref{def:loan_invariant}(iii), we slightly adapt $H_0$ as follows.
Whenever there is such a complex component, we remove the two bridges and treat their intersection vertex as a component; we will call it a \emph{rich} vertex.
Since we remove two edges from $H_0$ and each bridge had a credit of $\nicefrac14 - \delta$, we give each rich vertex a credit of $2(1 + \nicefrac 14 - \delta)$.
For ease of presentation, we still call the resulting solution $H_0$, which is now a weakly canonical $2$-edge cover without the specific type of components as mentioned above.
All isolated vertices of $H_0$ are obtained by this operation.

\begin{restatable}{lemma}{fewBridgeCoveringInitial}\label{lem:initial_loan}
	There are valid colors $\cF$ and loans for $H_0$ such that \cref{def:loan_invariant} is satisfied and
	$\cost(H_0)-\loan(H_0) \leq (\nicefrac{5}{4}-(1-\beta)\delta) |H_0|$.
\end{restatable}
The above invariants imply that for each color $f$, the credits of at most $3$ bridges are enough to pay back the loans.
Hence, we can perform the bridge-covering process similarly as in~\cite{BGGHJL25} when there are enough bridges in color $f$.
When the number of bridges in color $f$ is small, we carefully cover the remaining bridges without too much cost so that we can pay back the loans using the credits of some bridges before all bridges in color $f$ are covered.

We will apply the following lemma to iteratively adapt $H_0$ into a canonical 2-edge cover without bridges while maintaining the above invariants.
In each iteration, we add several paths to reduce the number of bridges, while not increasing the real cost of the current solution, i.e., $\cost - \loan$.

\begin{restatable}{lemma}{lemmaFEWbridge}
 \label{lem: bridgeCover-main}
	Given a weakly canonical 2-edge cover $H$ of a structured graph $G$ and a complex component, there is a polynomial-time algorithm that outputs a weakly canonical 2-edge cover $H'$ and a valid loan $\loan(H')$ of $G$ with fewer bridges such that
	$\cost(H')-\loan(H') \leq \cost(H)-\loan(H)$.
\end{restatable}

By exhaustively applying the above lemma, we obtain a bridgeless weakly canonical 2-edge cover $H'$ of $G$ %
such that
$\cost(H')-\loan(H') \leq \cost(H_0)-\loan(H_0) \leq  (\nicefrac{5}{4}-(1-\beta)\delta) |H_0|$.
Since $H'$ is bridgeless, we can conclude that $\loan(H') = 0$ by \Cref{def:loan_invariant}(iii).
Thus, $|H'| \leq \cost(H') = \cost(H')-\loan(H') \leq  (\nicefrac{5}{4}-(1-\beta)\delta) |H_0|$ as desired.

To get an intuition for \Cref{lem: bridgeCover-main}, consider a complex %
component $C$ of $H$. %
We denote by $G_C$ the multi-graph obtained from $G$ by contracting each block of $C$ and each connected component $C' \neq C$ of $H$ into a single node.
Let $T_C$ be the tree in $G_C$ induced by the bridges of $C$. We call the vertices of $T_C$ corresponding
to blocks \emph{block nodes}, and the remaining vertices of $T_C$ \emph{lonely nodes}.
Observe that the leaves of $T_C$ are block nodes. %
As for the component graph, we call vertices of $G_C$ that arise from components as \emph{nodes} and refer to vertices of $G_C$ that correspond to vertices in $G$ as vertices.

At a high level, we will transform $H$ into a new solution $H'$ containing a component $C'$ that contains the vertices of $C$ %
such that no new bridge is created and at least one bridge $e$ of $C$ is not a bridge of $C'$; we say $e$ is \emph{covered}. %
During this process, %
we maintain that $H'$ is weakly canonical.

A crucial tool for proving \Cref{lem: bridgeCover-main} are \emph{bridge-covering paths}.
A bridge-covering path $P_C$ is a path in $G_C \setminus E(T_C)$ with its distinct endpoints $u$ and $v$ in $T_C$, and the remaining internal vertices outside $T_C$.
Notice that $P_C$ might consist of a single edge, possibly parallel to some edge in $E(T_C)$.
Augmenting $H$ along $P_C$ means adding the edges of $P_C$ to $H$ to obtain~$H'$.
Notice that $H'$ is again a weakly canonical $2$-edge cover and has fewer bridges than $H$: all the bridges of $H$ along the $u$-$v$ path in $T_C$ are no longer bridges in $H'$ as they were covered, and the bridges of $H'$ are a subset of the bridges of $H$.
Note that a bridge-covering path, if one exists, can be computed in polynomial time; specifically, one can use breadth-first search after truncating $T_C$.

If we are able to show that we can find %
a bridge-covering path such that the resulting solution satisfies $\cost(H')-\loan(H') \leq \cost(H)-\loan(H)$, we are done (cf.\ \cref{lem: bridgeCover-main}).
In particular, the inequality is satisfied if $P_C$ involves at least $2$ block nodes or $1$ block node and at least $5$
bridges.
During bridge-covering, we ensure that bridges only pay back the loan of their own color, even if we are currently covering bridges of a different complex component.
In this way, we can carefully argue that we satisfy %
\Cref{def:loan_invariant}.

\subsubsection{Gluing}\label{sec:overview:few:gluing}

After bridge covering, we have a bridgeless weakly canonical 2-edge cover $H_1$ of a structured graph $G$ with $\cost(H_1) \leq \left( \nicefrac{5}{4} - (1-\beta)\delta \right) |H_0|$ such that $H_1$ is composed of
    \begin{itemize}[nosep]
        \item large components, each with a credit of $2$,
        \item rich vertices, each with a credit of $2(\nicefrac{5}{4}-\delta)$,
        \item $\cFour$'s with a credit of $1$, and
        \item $\mathcal C_i$'s with a credit of $(\nicefrac14 - \delta) i$ for $5\leq i\leq 8$.
    \end{itemize}
Note that every component receives a credit of at least 1.
Our main contribution in this section is to prove the following lemma.

\begin{restatable}{lemma}{lemmaMANYgluing}
	\label{lem:gluing:main}
	Given a structured graph $G$ and a bridgeless weakly canonical 2-edge cover $H_1$ of $G$, we can transform $H_1$ in polynomial time into a 2ECSS $H_2$ of $G$ such that $\cost(H_2) \leq \cost(H_1)$.
\end{restatable}

With the above lemma, we can prove \Cref{lem:main:FEW}.
We first compute a canonical $2$-edge cover $H_0$ using \Cref{lem:canonical-cover:main}.
We then transform this into a bridgeless weakly canonical 2-edge cover $H_1$ such that $\cost(H_1) \leq \cost(H_0)-\loan(H_0)$ by exhaustively using \Cref{lem: bridgeCover-main}.
Then, by exhaustively applying \Cref{lem:gluing:main}, we can transform $H_1$ into a 2ECSS $H_2$ such that $\cost(H_2) \leq \cost(H_1)$.
Since by \Cref{lem:initial_loan}, $\cost(H_1) \leq  \left(\nicefrac{5}{4} - (1-\beta)\delta\right) |H_0|$,
we can conclude~\Cref{lem:main:FEW}.

Hence, the remaining section is dedicated to proving \Cref{lem:gluing:main}.
Recall that an inclusion-wise maximal 2-vertex-connected subgraph of $\hG_H$ is a segment, and we can break
$\hG_H$ into segments such that two segments can potentially overlap only at a cut node. %
Hence, in a segment $S$ with $|S| = 2$, each node is either a cut node or a pendant node of $\hG_H$.
Moreover, a cut node of $\hG_H$ cannot correspond to a rich vertex, because otherwise there is a cut vertex in $G$, contradicting that it is structured. %

We roughly follow the approach of~\cite{BGGHJL25}.
Consider a segment $S$ that contains a node $L$ corresponding to a huge component. By the definition of a canonical 2-edge cover (\Cref{def:canonicalD2}), such a segment must exist. The benefit of having a huge component is that there exists a matching of size $4$ between $L$ and the other vertices if those are at least 11 (\Cref{lem:4-matching}).

Our approach is the following:
If we can compute a cycle $F$ in $S$ through $L$ and another node $A$, such that $F$ is incident to two distinct vertices $u$ and $v$ in $C_A$, and there is a Hamiltonian path $P$ between $u$ and $v$ in $G[V(C_A)]$,
then we set $H' \coloneq (H \setminus E(C_A)) \cup F \cup P$. Note that $H'$ is a canonical bridgeless 2-edge cover of $G$ with fewer components than $H$. %
Furthermore, $\cost(H') - \cost(H) = (|H'| - |H|) + (\credit(H') - \credit(H))  = (|F| - 1) + (2-|F|-1) \leq 0$, since we added $|F|$ edges, effectively removed $1$ edge as $|C_A| - |P|=1$, $\credit(C_L)=2$, and $\credit(C_A)\geq 1$ in $H$, whereas in $H'$ those form a single component with credit 2.
We do not need the existence of such a Hamiltonian path in  $C_A$ if $A$ is large, as this component already has a credit of $2$. Hence, such a cycle also meets our goal.
We call such cycles \emph{good} cycles.

In most cases, we can simply find such a good cycle as above. However, this is not always possible.
We sometimes need to find two particular cycles in $\hG_H$ that intersect at a single node.
In the following, we distinguish whether $|S| = 2 $ or $|S| \geq 3$.

\begin{restatable}{lemma}{lemmaMANYgluingtrivial}
	\label{lem:gluing:trivial-segment}
	If $|V(S)|=2$, then in polynomial time we can compute a weakly canonical bridgeless 2-edge cover $H'$ of $G$ such that $H'$ has fewer components than $H$, $H'$ contains a huge component, and $\cost(H') \leq \cost(H)$.
\end{restatable}

In this case, $S$ contains two components: a huge component $L$ and another component $C$.
We can find a good cycle between $L$ and $C$ if $C$ is either large or if $C$ is a $\mathcal{C}_i$ for $4 \leq i \leq 7$ and there is a $3$-matching or $4$-matching between $L$ and $C$, which already deals with most cases.
In the remaining cases and in particular if $C$ is a $\cEight$, we show that we can find two cycles in $\hG_H$ that intersect at $C$ such that adding this cycle to $H$ and removing a specific set of edges from $H$ results in a canonical $2$-edge cover $H'$ such that $\cost(H') \leq \cost(H)$ and $H'$ has fewer components than $H$.

In the remaining case, we have $|S| \geq 3$.

\begin{restatable}{lemma}{lemmaMANYgluingnontrivial}
	\label{lem:gluing:non-trivial-segment}
	If $|V(S)| \geq 3$, then in polynomial time we can compute a weakly canonical bridgeless 2-edge cover $H'$ of $H$ such that $H'$ has fewer components than $H$, $H'$ contains a huge component, and $\cost(H') \leq \cost(H)$.
\end{restatable}

Here, we first use a lemma from~\cite{BGGHJL25}, which essentially says that if $S$ contains the huge component $L$ and another component $C$ that is a $\cFour$, then we can find the desired $H'$.
If there is no $\cFour$, any component in $S$ has credit of at least $5 (\nicefrac{1}{4} - \delta) \geq \nicefrac{6}{5}$.
Now, if there is a cycle $K$ in~$S$ through~$L$ of length at least 6, we %
have that $H' = H \cup K$ is a canonical $2$-edge cover $H'$, $H'$ has fewer components than $H$, and
$|H'| \leq |H| + |K|$ and $\credit(H') \leq \credit(H) + 2 -  2 - (|K|-1) \cdot \nicefrac{6}{5}  \leq \credit(H) + |K|$ for $|K| \geq 6$.
Otherwise, we make a case distinction on the longest length of a cycle in~$S$ through~$L$, which can only be of length 3, 4, or 5. We show that in either case we can find a cycle through $L$ such that we can find the desired $H'$.
This finishes the gluing part and the algorithm for case \FEW.

\subsection{Algorithm for \MANY}\label{sec:many}

Next, we present our algorithm for \MANY, i.e., $\beta > 1 - \frac{1}{1000}$.
We show the following lemma and defer proofs and further details to \Cref{sec:manyC4}.

\begin{restatable}{lemma}{lemmamainMANY}
\label{lem:main:MANY}
Given a structured graph $G$ and a canonical $2$-edge cover $H_0$, in polynomial time we can compute a 2EC spanning subgraph $H$ of $G$ such that
$$|H| \leq \left( \frac{5}{4} - \frac{1}{28} \right) \cdot \beta |H_0| + 5 (1 - \beta)|H_0| + \frac{1}{35} \opt       \ ,
$$
where  $\beta \in [0,1]$ such that the number of edges of $\cFour$'s in $H_0$ is $\beta |H_0|$.
\end{restatable}

To prove this lemma, we use the following credit scheme.

\begin{restatable}[credit assignment for the case of many $\cFour$'s]{definition}{definitionMANYcredit}
Let $H$ be a $2$-edge cover of $G$ and let $\delta \coloneq \nicefrac{1}{28}$. We keep the following credit invariant for blocks, components, and bridges:
	\label{def:initial_credit_2}
	\begin{itemize}[nosep]
		\item Each $\mathcal C_4$  receives a credit of $ \credit(C) = 4 \cdot (\nicefrac 1 4 - \delta)$.
		\item Each 2EC component $C$ of $H$ that is a $\mathcal C_i$, for $5 \leq i \leq 8$, receives credit $ \credit(C) = 4i$.
		\item Each 2EC component $C$ of $H$ that contains 9 or more edges receives credit $ \credit(C) = 2$.
		\item Each block $B$ of a complex component $C$ of $H$ receives a credit $\credit(B) =1$.
		\item Each bridge $e$ of a complex component $C$ of $H$ receives credit $\credit(e) = 4$.
		\item Each complex component $C$ of $H$ receives a component credit $\credit(C)=1$.
	\end{itemize}
\end{restatable}

In a nutshell, the credit for edges not contained in $\cFour$'s is very large so that any \emph{short} cycle covering a bridge or involving a $\mathcal{C}_i$ for some $5 \leq i \leq 8$ can be used to transform $H$ into a canonical $2$-edge cover $H'$ with fewer components or bridges such that $\cost(H') \leq \cost(H)$. This enables us to mainly focus on components that are $\cFour$, large, or complex.

Using this credit scheme, it is easy to show that the initial canonical $2$-edge cover satisfies the invariant of the cost of the partial solution that we maintain throughout our algorithm.

\begin{restatable}{lemma}{lemmaMANYinitial}
	\label{lem:manyC4_starting-solution-cost}
	Given a canonical $2$-edge cover $H$, we have $\cost(H) \leq (\nicefrac 54 - \delta) \cdot m_4 + 5 \cdot m_r$, where $m_4$ and $m_r$ are the number of edges of $H$ contained in a $\cFour$ or not contained in a $\cFour$, respectively.
\end{restatable}

Our approach for \MANY is different from the previous approach for \FEW. Instead of first covering all bridges and afterwards gluing the different 2EC components together, we do both at the same time.
Therefore, throughout most of the algorithm, we also have complex components.

Unfortunately, we cannot turn the initial 2-edge cover $H_0$ into a 2ECSS $H$ such that  $|H| \leq \left( \nicefrac{5}{4} - \nicefrac{1}{28} \right) \cdot \beta |H_0| + 5 (1 - \beta)|H_0|$, i.e., we never increase the cost of the solution and only bound $H$ in terms of $|H_0|$.
The reason for this is the following:
Consider the case that $H_0$ has one large (but constant size) component $L$ and a large number $k$ of $\cFour$'s, such that each $\cFour$ is only adjacent to $L$.
Then any feasible solution must include at least $5k \approx \nicefrac 5 4 |H_0|$ edges to be connected; hence $\opt \geq 5k$ and therefore we cannot find a solution $H$ such that $|H| \leq \nicefrac{5}{4} \cdot c  \cdot |H_0|$, for some $c < 1$.
However, we can find the following \emph{core-square configuration}.
Then, we use a generalized version of the above improved bound on $\opt$ to obtain our final approximation guarantee of $|H| \leq \left( \nicefrac{5}{4} - \nicefrac{1}{28} \right) \cdot \beta |H_0| + 5 (1 - \beta)|H_0| + \nicefrac{1}{35} \cdot \opt      $.

\begin{restatable}{definition}{definitionMANYcore}
	\label{def:core-square-configuration}
	A canonical $2$-edge cover $H$ is in a \emph{core-square configuration} if
	\begin{itemize}[nosep]
		\item[(i)] there is exactly one component $L$ in $H$ with $|V(L)| \geq 32$ that is either complex or $2$EC,
		\item[(ii)] every component of $H$ except $L$ is a $\cFour$,
		\item[(iii)] $L$ is only contained in segments of $\hG_H$ of size at most $2$, and
		\item[(iv)] if $C \neq L$ is a cut component of $H$, then it is in exactly $2$ segments of $\hG_H$ and one of the segments contains $L$.
				Let $S$ and $S'$ be the $2$ segments of $\hG_H$ containing $C$
				and $L \in S$.
				Then $S'$ contains at most $4$ components and $C$ is the only cut component in $S'$.
	\end{itemize}
	Finally, we say that $H$ is a \emph{bridgeless} core-square configuration if $L$ is $2$-edge connected.
\end{restatable}

\begin{figure}
	\centering
	\begin{subfigure}[t]{0.49\textwidth}
		\centering
        \includegraphics[width=0.9\textwidth]{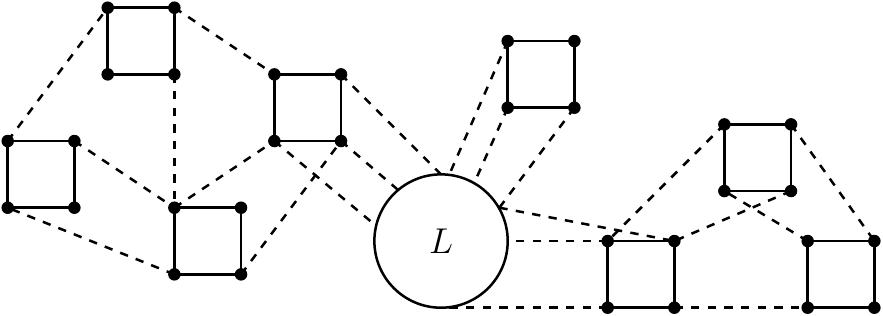}
        \caption{}
        \label{fig:core-square-configuration}
    \end{subfigure}
    \hfill
    \begin{subfigure}[t]{0.49\textwidth}
        \centering
        \includegraphics[width=0.9\textwidth]{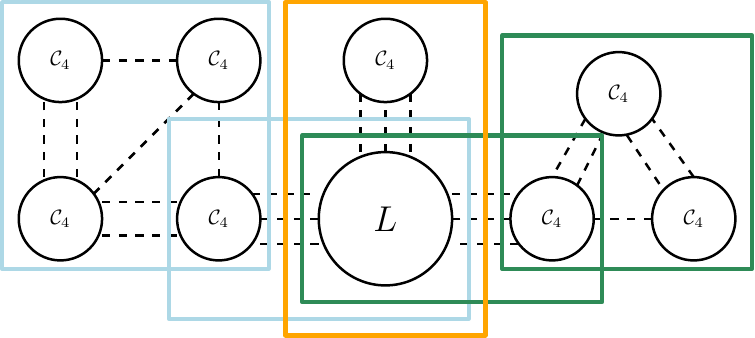}
        \caption{}
        \label{fig:core-square-configuration-component-graph}
    \end{subfigure}
    \caption{An example of a core-square configuration (\Cref{def:core-square-configuration}) in \Cref{fig:core-square-configuration}, and its component graph and segments in \Cref{fig:core-square-configuration-component-graph}.}
\end{figure}

A similar notion of core configuration together with the use of the different lower bound on $\opt$ from above has been used in~\cite{GargGA23improved}.
The main goal of this section is to compute a bridgeless core-square configuration, summarized in the following lemma.

\begin{restatable}{lemma}{lemmaMANYcore}
	\label{lem:manyC4_core-configuration_main}
	Given a canonical $2$-edge cover $H$ of $G$, in polynomial time we can compute a bridgeless core-square configuration $H'$ of $G$ such that $\cost(H') \leq \cost(H) +2$.
\end{restatable}

In the remainder of this section, we outline the proof of \Cref{lem:manyC4_core-configuration_main}.
Similar to before, we know that our canonical $2$-edge cover $H$ contains a huge component $L$.
We consider different cases on the sizes of segments $S$ of the component graph $\hG_H$ and the position of $L$ within $\hG_H$.
In a nutshell, in most cases we can argue that the size of the current segment $S$ is either at most $4$, there is a segment containing either a short cycle involving a bridge of a complex component or a $\mathcal{C}_i$ for some $5 \leq i \leq 8$, or we can conclude that there is a large $\mathcal{C}_i$ cut, which is a contradiction to $G$ being structured.
If we only have such cases, it is not too difficult to show that we can find $H'$ as required in \Cref{lem:manyC4_core-configuration_main}.

The difficult case is if $L$ is contained in a segment of size at least $5$, which is our main technical contribution and summarized in the following lemma.

\begin{restatable}{lemma}{lemmaMANYgluingpath}
	\label{lem:core:many-components}
	Given a canonical $2$-edge cover $H$ of a structured graph $G$ such that the component graph $\hG_H$ contains a segment $S$ with at least $5$ components where at least one is huge, then we can compute in polynomial time a canonical $2$-edge cover $H'$ with fewer components than $H$ and $\cost(H') \leq \cost(H)$.
\end{restatable}

To prove \Cref{lem:core:many-components}, we use the definition of a \emph{gluing path}. A similar concept in a weaker form has been used in~\cite{GargGA23improved}. In particular, unlike in \cite{GargGA23improved}, we will later consider so-called \emph{branching} gluing paths.

\begin{restatable}{definition}{definitionMANYgluing}
	\label{def:nice-path}
	Given a canonical $2$-edge cover $H$ of a structured graph $G$ and a segment $S$, a \emph{gluing path} is a simple path $P$ in the component graph $\hG_H$ satisfying the following conditions:
	\begin{itemize}[nosep]
	\item[(i)] Every node of $P$ is a node of $S$.
		\item[(ii)] If $e_1$ and $e_2$ are two edges of $P$ incident to a component $C$ corresponding to a $\cFour$, then the endpoints of these edges in $V(C)$ are adjacent in $C$ (w.r.t.\ the edges in $H$).
		\item[(iii)] If $e_1$ and $e_2$ are two edges of $P$ incident to a complex component $C$, and both endpoints $u_1$ of $e_1$ and $u_2$ of $e_2$ in $V(C)$ are not contained in some block of $C$, then $u_1 \neq u_2$.
	\end{itemize}
\end{restatable}

To illustrate the benefit of a gluing path, consider one of length $k \geq 4$.
If there is some edge $e$ from a node $v$ of $P$ to some node $u$ of $P$ such that their distance on $P$ is at least $4$, then consider the cycle $C$ formed by the subpath of $P$ from $u$ to $v$ and the edge $e$.
Clearly, adding the %
edges of $C$ to $H$ turns the $|C|$ components of $C$ into a single 2EC component $C'$. Let $H'$ be the new $2$-edge cover.
Furthermore, due to the definition of a gluing path, each component of $C$ that is a $\cFour$ except $C_u$ and $C_v$ can be shortcut.  That is, we can remove one edge for each such component such that $C'$ remains $2$-edge connected.
These edges pay for most of the cost of adding $|C|$.
Therefore, it can be easily checked that $H'$ is also a canonical $2$-edge cover such that $\cost(H') \leq \cost(H)$.

Of course, such edges may not always exist.
Hence, our goal is the following:
Given some canonical $2$-edge cover and some gluing path $P$ of length $k$, either we want to make progress by turning $H$ into some canonical $2$-edge cover $H'$ such that $H'$ contains fewer components than $H$ and $\cost(H') \leq \cost(H)$, or we can extend $P$ to a longer gluing path $P'$. %
Since the length of a gluing path can be easily bounded by $|V(G)|$, after polynomially many extension we must find the desired $H'$.
We then show that we can always make progress, which proves \Cref{lem:manyC4_core-configuration_main}. %

An important innovation over \cite{GargGA23improved} are branches. Informally, a gluing path $P$ has a branch $P'$ if $P$ and $P'$ are identical for a large prefix and then ``branch'' out. By distinguishing structural properties of branches, we can prove the existence of more involved cycles $C$ with which we can make progress; sometimes we even add two intersecting cycles. An example is illustrated in \Cref{fig:gluing-path-branch-intro}.

\begin{figure}
	\centering
    \includegraphics[width=0.5\textwidth]{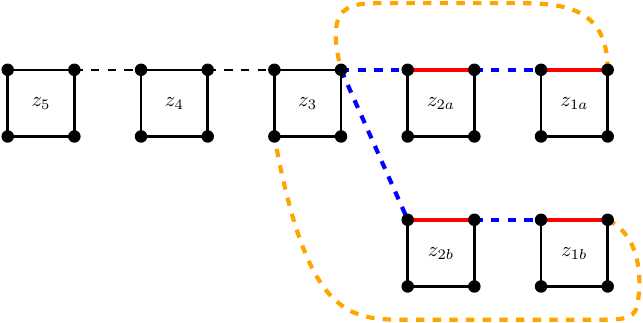}
    \caption{An example of a gluing path $P = z_5 - \ldots - z_{1a}$ and a branch $P' = z_5 - \ldots - z_{1b}$. We show the existence of the orange edges. Then we can add the two orange edges and the four blue path edges, and remove the four red edges to make progress.}
    \label{fig:gluing-path-branch-intro}
\end{figure}

\subsection{Conclusion}
\label{sec:conclusion}

We give a polynomial-time $(\nicefrac54 - \eta)$-approximation algorithm for 2ECSS
for some $\eta \geq 10^{-6}$, breaching the natural barrier of $\nicefrac54$ that has been achieved very recently~\cite{BGGHJL25,BGJ24ecss,GHL24}.
Our main contribution is the introduction and novel combination of techniques that allow
approximation factors below $\nicefrac54$.
We focused on a clean and qualitative exposition, and did not optimize the approximation ratio much; we
believe that $\eta \geq 10^{-4}$ should be easily possible by using more fine-grained constants.
Furthermore, we show a more substantial improvement to $1.24+\eps$ for every $\eps > 0$ conditioned on the existence of a polynomial-time algorithm
that computes a minimum $\{3,4\}$-cycle-free $2$-edge cover on structured graphs.
The complexity status of this problem is open. In the light of our result, finding an optimal polynomial-time algorithm (or a good approximation algorithm) for this problem seems to be an interesting open question with strong implications.
Other future directions that build up on our work are better approximation algorithms for (special cases of) Forest Augmentation and graphic TSP.
Since current lower bounds on the approximation ratio are close to 1, one can only achieve tight or almost tight approximation ratios for 2ECSS if we improve upon these lower bounds. This is another promising direction for future work.

\section{Algorithm for \FEW}\label{sec:few:main}

This section is dedicated to the proof of \Cref{lem:main:FEW}, which we restate for convenience.

\lemmamainFEW*

To prove the lemma, we first transform the initial canonical $2$-edge cover $H_0$ into a bridgeless canonical $2$-edge cover. This step is done in \Cref{sec:bridge-covering}.
Afterwards, we glue the different 2EC components together into a single 2EC component. This step is done in \Cref{sec:gluing}.

\subsection{Bridge Covering}
\label{sec:bridge-covering}

The key ingredient for covering the bridges is finding good bridge-covering paths, which we have already introduced in \Cref{sec:few}. We restate the credit scheme and the loan invariant and give some more details here.

\definitionFEWcredit*

\definitionFEWcolorloan*

One can check that $H_0$ does not satisfy \cref{def:loan_invariant}(iii) if there is a complex component $C$ consisting of exactly two blocks (which are nice $\cFive$'s) and exactly two bridges.
For this case, to satisfy \cref{def:loan_invariant}(iii), we slightly adapt $H_0$ as follows.
Whenever there is such a complex component, we remove the two bridges and treat their intersection vertex as a component; we will call it a \emph{rich} vertex.
Since we remove two edges from $H_0$ and each bridge had a credit of $\nicefrac14 - \delta$, we give each rich vertex a credit of $2(1 + \nicefrac 14 - \delta)$.
For ease of presentation, we still call the resulting solution $H_0$, which is now a weakly canonical $2$-edge cover without the specific type of components as mentioned above.

We remark that each component of the initial (canonical) $2$-edge cover $H_0$ that is not 2EC contains at least $12$ vertices, and we only possibly merge together components in further stages of the construction.
As a consequence, the new solution is also canonical.

Recall the definition of a \emph{bridge-covering path}.
A bridge-covering path $P_C$ is any path in $G_C \setminus E(T_C)$ with its distinct endpoints $u$ and $v$ in $T_C$, and the remaining internal vertices outside $T_C$.
Notice that $P_C$ might consist of a single edge, possibly parallel to some edge in $E(T_C)$.
Augmenting $H$ along $P_C$ means adding the edges of $P_C$ to $H$ to obtain a $2$-edge cover
$H'$.
Notice that $H'$ has fewer bridges than $H$: all the bridges of $H$ along the $u$-$v$ path in $T_C$ are no bridges in $H'$ (they are covered), and the bridges of $H'$ are a subset of the bridges of $H$.
If we are able to show that we can find %
a bridge-covering path such that the resulting solution satisfies $\cost(H')-\loan(H') \leq \cost(H)-\loan(H)$, we are done (cf.\ \cref{lem: bridgeCover-main}).

Hence, it remains to analyze $\cost(H')-\loan(H')$.
Let $P_C$ be a bridge-covering path of
$T_C$ from $u$ to $v$. We consider the following quantities.
\begin{itemize}
	\item $br$ is the distance between $u$ and $v$ in $T_C$ (i.e., the path contains $br$ many bridges).
	\item $bl$ is the number of blocks on a path in $T_C$ between $u$ and $v$ (including $u$ and $v$)
\end{itemize}

Then the number of edges w.r.t.\ $H$ grows by $|E(P_C )|$.
The number of credits w.r.t.\ $H$ decreases by at least $\cre br + bl + |E(P_C)| - 1$ since we remove $br$ bridges, $bl$ blocks, and $|E(P_C)| - 1$ components (each one having at least one credit).
However, the number of credits also grows by 1 since we create a new block $B'$,  which needs $1$ credit, or a new 2EC component $C'$, which needs $1$ additional credit w.r.t.\ the credit of $C$.
Altogether $\cost(H) - \cost(H') \geq \cre br + bl - 2$.
That is, we can afford to buy the extra edges on $P_C$ using the credits if $\cre br + bl - 2 \geq 0$.

\begin{definition}[Cheap and Expensive Bridge-Covering Paths]
	We say that a bridge-covering path $P_C$ is \emph{cheap} if $\cre br + bl - 2 \geq 0$, and \emph{expensive} otherwise.
\end{definition}

In particular, $P_C$ is cheap if it involves at least $2$ block nodes or $1$ block node and at least $5$
bridges.
Notice that a bridge-covering path, if at least one such path exists, can be computed in polynomial time; specifically, one can use breadth-first search after truncating $T_C$.
Observe that in the resulting solution $H'$, \cref{def:loan_invariant}(i) is maintained.
However, to maintain \cref{def:loan_invariant}(iii), we have to carefully use a bridge-covering path that involves only $1$ block.
For example, if there are only $6$ bridges in color $f$ and $\loan_f(H)$ is larger than the credits of one bridge, then adding such a bridge-covering path would violate \cref{def:loan_invariant}(iii).

Now we show if such a path exists that the resulting solution $H'$ satisfies $\cost(H') \leq \cost(H)$,
we can adapt the loans so that $\cost(H')-\loan(H') \leq \cost(H)-\loan(H)$.
A bride-covering path may cover some bridge $e$ in another color.
If $\col(e)$ has positive loans, we use the credits to pay (partially) for the loans since the credits are not used for buying extra edges on the path.
In this way, both $\credit(H')$ and $\loan(H')$ decrease by $\cre$.
So the value of $\cost(H')-\loan(H')$ is unchanged,
and \cref{def:loan_invariant}(iii) is maintained.
Another case is that  there exist $br'<br$ that $\cre br' + bl -2 \geq 0$, i.e., we do not need all the credits from the covered bridges to afford the extra bought edges.
In this case, we let $\loan(H') \coloneq \max\{\loan(H)-(br-br')\cre, 0\}$, and we have
\begin{align*}
	(\cost(H)-\cost(H'))-(\loan(H)-\loan(H')) &= \left( \frac14 - \delta \right) br + bl - 2  - \left( \frac14 - \delta \right) (br-br') \\
	&= \left( \frac14 - \delta \right) br' +bl -2 \geq 0 \ ,
\end{align*}

which implies $\cost(H')-\loan(H') \leq \cost(H)-\loan(H)$.
Further, if $bl=2$, then $br'=0$ and \cref{def:loan_invariant}(iii) is maintained by the following arguments since the credits and the loans decrease by the same amount.
If $bl=1$ and $br \geq 5$ ($br'=5$), our algorithm will ensure to add such a bridge-covering path only if $(k-5)\cre \geq \loan_f(H)$, where $k$ is the number of bridges in color $f$ in $H$.
In the resulting solution $H'$, we have $k-br$ bridges in color $f$.
So $\loan_f(H') = \loan_f(H)-(br-5)\cre \leq (k-5) \cre - (br-5)\cre \leq (k-br)\cre$, which satisfies \cref{def:loan_invariant}(iii).

To argue \cref{def:loan_invariant}(ii), we fix any deficient color $f$ in a component $C$ of $H$ and let $F$ (resp. $F'$) be the set of edges in color $f$ in $H$ (resp. $H'$).
Notice that $T_C'(F')$ can be obtained from $T_C(F)$ by contracting a connected subtree $T' \subseteq T_C(F)$.
Hence, $T_C'(F')$ has fewer or equal number of leaves than $T_C(F)$.
Since $\loan_f(H') \leq \loan_f(H)$, the conditions in \cref{def:loan_invariant}(ii) are satisfied.

\paragraph*{Main lemmas.}

Let $H_0$ be a canonical 2-edge cover of $G$ (\cref{def:canonicalD2}).
We show that we can assign credits to elements of $H_0$ as in \cref{def:initial_credit}.

\fewBridgeCoveringInitial*

\begin{proof}%
Consider the following process.
Initially, we give $\nicefrac14 - \delta$ credits to each edge that is not in a $\cFour$ of $H_0$, and give $\nicefrac{1}{4}$ credits to each edge in $\cFour$s.
Hence we initially give $\left(\frac{1}{4} - (1-\beta)\delta\right) |H_0|$ credits in total.
We then redistribute the credits over other elements (components, blocks, bridges) of $H_0$ to satisfy the credit scheme in \Cref{def:initial_credit}.
Each edge that is in a 2EC component gives its credits to the component it is in, which satisfies the credit scheme for non-complex components.
Consider a complex component $C$ of $H_0$.
Each bridge keeps its credits.
We perform a case-distinction on the number of blocks in $C$ and show how to redistribute the credits of non-bridge edges to components and blocks, and how to take loans.

\begin{description}
	\item[$C$ has at least $5$ blocks.]
	      Each block has at least $5$ edges and hence at least $5 \cre$ credits in total.
	      We use the credits of the edges in the blocks to pay for the block credits.
	      Then for each block, we still have $(\frac{1}{4}-5\delta)$ unused credits.
	      In total, we have at least $\frac{5}{4}-25\delta \geq 1$ unused credits (using $\delta \leq \frac{1}{100}$) over all the blocks.
	      We use the unused credits to pay for the component credit.
	\item[$C$ has exactly $4$ blocks.]
	      Let the total number of edges in these blocks be $k$.
	      Hence we have $k \cre$ credits from these edges.
	      Note that $k\geq 20$ since each block has at least $5$ edges.
	      We need $4$ credits for blocks and $1$ credit for the component.
	      If $k \geq 21$, then $k\cre \geq \frac{21}{4}-21\delta \geq 5$ (using $\delta \leq \frac{1}{84}$), we are done.
	      If $k = 20$, we have $20 \cre = 5-20\delta$ credits.
	      But we need $5$ credits to pay for the $4$ blocks and $1$ credit for the component.
	      Hence we have to take a loan of $20 \delta$ credits.
	      Note that the credits from one bridge are enough to pay back the loan.
		  We will create a new color $f$, mark all the bridges in this component with color $f$, and set $\loan_f(H_0)\coloneq 20\delta$.
	\item[$C$ has exactly $3$ blocks.]
	      Let the total number of edges in these blocks be $k$.
	      Hence we have $k \cre$ credits from these edges.
	      Note that $k\geq 15$ since each block has at least $5$ edges.
	      We need $3$ credits for blocks and $1$ credit for the component.
	      If $k \geq 17$, then $k\cre \geq \frac{17}{4}-17\delta \geq 4$ (using $\delta \leq \frac{1}{68}$), we are done.
	      If $k = 16$, we have $16 \cre = 4-16\delta$ credits.
	      But we need $4$ credits to pay for the $3$ blocks and $1$ credit for the component.
	      Hence we have to take a loan of $16 \delta$ credits.
	      If $k = 15$, we have $15 \cre = \frac{15}{4}-15\delta$ credits.
	      Hence we have to take a loan of $\frac{1}{4}+15 \delta$ credits.
	      Note that the credits from $2$ bridges are enough to pay back the loan.
		  We will create a new color $f$, mark all the bridges in this component with color $f$, and set $\loan_f(H_0) = 16 \delta$ if $k=16$ and $\loan_f(H_0) = \frac{1}{4}+15 \delta$ if $k=15$.
	\item[$C$ has exactly $2$ blocks.] Let the total number of edges in these blocks be $k$.
	      Hence we have $k \cre$ credits from these edges.
	      Note that $k\geq 10$ since each block has at least $5$ edges.
	      We need $2$ credits for blocks and $1$ credit for the component.
	      If $k \geq 13$, then $k\cre \geq \frac{13}{4}-13\delta \geq 3$ (using $\delta \leq \frac{1}{52}$), we are done.
	      If $10 \leq k \leq 12$,
	      we have to take a loan of $3- k\cre = \frac{12-k}{4}+k\delta \leq \frac{1}{2}+10\delta \leq 3 \cre$ credits.
	      Note that the credits from at most $3$ bridges are enough to pay back the loan.
		  We will create a new color $f$, mark all the bridges in this component with color $f$, and set $\loan_f(H_0)\coloneq \frac{1}{2}+10\delta $.
		  Note that we need to satisfy the loan invariant that the total credits of bridges in color $f$ must be at least $\loan_f(H_0)$.
		  This is violated if and only if if $C$ has exactly $2$ blocks with both of them being $\cFive$, and has at most $2$ bridges.
	      If this is the case, recall that we treat $C$ differently.
	      Note that $C$ has exactly $2$ edges as otherwise we can remove the only bridge to obtain a triangle-free $2$-edge cover with fewer edges.
	      We remove the two bridges and treat their intersection vertex as a component.
	      We call such a vertex a \emph{rich} vertex or component.
	      Further, we give the credits of the two edges to the rich vertex so that it has $2(\frac{5}{4}-\delta)$ credits.
\end{description}
By the previous discussion, for each color $f$, $\loan_f(H_0) \leq \frac{1}{2}+10\delta$ and the credits of at most $3$ bridges are enough to pay back the loans.
Further, if the credits on $2$ bridges are not enough, it must be the case where the corresponding complex component in $H_0$ has exactly two blocks and each of them are $\cFive$.
Otherwise, we have $\loan_f(H_0) \leq \frac{1}{4}+11 \delta$.
\end{proof}

We use the notations defined in \cref{sec:few}.
We say that a vertex $v \in V(T_C) \setminus \{u \}$ is \emph{reachable} from $u \in V(T_C)$ if there exists a bridge-covering path between $v$ and $u$.
Let $R(W)$ be the vertices in $V(T_C) \setminus W$ reachable from some vertex in $W \subseteq V(T_C)$, and let us use $R(u) \coloneq R(\{u\})$ for $u \in V(T_C)$.
Notice that $v \in R(u)$ if and only if $u \in R(v)$.
We first show the following useful lemma.

\begin{lemma}[Lemma 29 in the full version of~\cite{BGGHJL25}]
	\label{lem:bridge-covering-helper1}
	Let $e = xy \in E(T_C)$ and let $X_C$ and $Y_C$ be the two sets of vertices of the two trees obtained from $T_C$ after removing the edge $e$, where $x \in X_C$ and $y \in Y_C$.
	Then $R(X_C)$ contains a block node or $R(X_C) \setminus \{y\}$ contains at least 2 lonely nodes.
\end{lemma}

Let $T_C(u, v)$ denote the path in $T_C$ between vertices $u$ and $v$.
The following lemmas show that we can make progress with $2$ bridge-covering paths.

\begin{lemma}
	\label{lem:bridge-covering-helper2}
	Let $b$ and $b'$ be two pendant (block) nodes of $T_C$.
	Let $u \in R(b)$ and $u' \in R(b')$ be vertices of $V (T_C) \setminus \{b, b' \}$.
	Suppose that $T_C(b, u)$ and $T_C(b', u')$ both contain some vertex $w$ (possibly $w = u = u'$) and $|E(T_C(b, u)) \cup E(T_C(b', u'))| \geq 5$.
	Assume that the bridges in $E(T_C(b, u)) \cup E(T_C(b', u'))$ are in the same color $f$.
	Then in polynomial time one can find a 2-edge cover $H'$ satisfying the conditions of~\Cref{lem: bridgeCover-main}.
\end{lemma}
\begin{proof}
Let $P_{bu}$ (resp.\ $P_{b'u'}$) be the bridge-covering paths between $b$ and $u$ (resp.\ $b'$ and $u'$).
If the two paths intersect internally, then there exists a bridge-covering path $P$ between $b$ and $b'$, which is cheap.
Hence $H'\coloneq H \cup P$ satisfies the conditions of~\Cref{lem: bridgeCover-main}.
Otherwise, let $H'\coloneq H \cup P_{bu} \cup P_{b'u'}$.
Note that these two bridge-covering paths together cover $k \geq 5$ bridges and merge at least two blocks.
Hence, we have $|H'|= |H|+|P_{bu}|+|P_{b'u'}|$ and $\credit(H') = \credit(H)+1-(|P_{bu}|-1)-(|P_{b'u'}|-1) -2-k\cre  \leq \credit(H)-|P_{bu}|-|P_{b'u'}|-(k-5)\cre$.
So $\cost(H') \leq \cost(H)-(k-5)\cre$, and we can let $\loan_f(H') \coloneq \max \{\loan_f(H) - (k-5)\cre,0 \}$ since these $(k-5)\cre$ credits are not needed for buying the bridge-covering paths.
This proves the lemma.
\end{proof}

\begin{lemma}
	\label{lem:bridge-covering-helper3}
	Let $b$ and $b'$ be two pendant (block) nodes of $T_C$.
	Let $u \in R(b)$ and $u' \in R(b')$ be vertices of $V (T_C) \setminus \{b, b' \}$.
	Suppose that $T_C(b, u)$ and $T_C(b', u')$ intersect on only one edge $e$, that is, $E(T_C(b, u)) \cap E(T_C(b', u')) = \{e\}$.
	Then in polynomial time one can find a 2-edge cover $H'$ satisfying the conditions of \Cref{lem: bridgeCover-main}.
\end{lemma}
\begin{proof}
Let $P_{bu}$ (resp.\ $P_{b'u'}$) be the bridge-covering paths between $b$ and $u$ (resp.\ $b'$ and $u'$).
If the two paths intersect internally, then there exists a bridge-covering path $P$ between $b$ and $b'$, which is cheap.
Hence $H'\coloneq H \cup P$ satisfies the conditions of~\Cref{lem: bridgeCover-main}.
Otherwise, let $H'\coloneq (H\setminus \{e\}) \cup P_{bu} \cup P_{b'u'}$.
Note that these two bridge-covering paths merge at least two blocks, and $k \geq 1$ bridges are covered or removed.
We have $|H'|= |H|+|P_{bu}|+|P_{b'u'}|-1$ and $\credit(H') = \credit(H)+1-(|P_{bu}|-1)-(|P_{b'u'}|-1) -2-k\cre  = \credit(H)-|P_{bu}|-|P_{b'u'}|+1-k\cre$.
So $\cost(H') \leq \cost(H)-k\cre$, and we can let $\loan_f(H') \coloneq \max\{\loan_f(H) - k\cre, 0 \}$ since these $k\cre$ credits are not needed for buying the bridge-covering paths.
This proves the lemma.
\end{proof}

\begin{lemma}
	\label{lem:bridge-covering-helper4}
	Let $b$ be a pendant (block) node of $T_C$ and $b'$ be any other node (not necessarily a block node) of $T_C$.
	Let $u \in R(b)$ and $u' \in R(b')$ be vertices of $V (T_C) \setminus \{b, b' \}$.
	Suppose that $T_C(b, u)$ and $T_C(b', u')$ intersect on only one edge $e$ and $|E(T_C(b, u)) \cup E(T_C(b', u'))| \geq 5$.
	Then in polynomial time one can find a 2-edge cover $H'$ satisfying the conditions of~\Cref{lem: bridgeCover-main}.
\end{lemma}
\begin{proof}
Let $P_{bu}$ (resp.\ $P_{b'u'}$) be the bridge-covering paths between $b$ and $u$ (resp.\ $b'$ and $u'$).
If the two paths intersect internally, then there exists a bridge-covering path $P$ between $b$ and $b'$, which is cheap.
Hence $H'\coloneq H \cup P$ satisfies the conditions of~\Cref{lem: bridgeCover-main}.
Otherwise, let $H'\coloneq (H\setminus \{e\}) \cup P_{bu} \cup P_{b'u'}$.
Note that these two bridge-covering paths merge at least two blocks, and $k \geq 5$ bridges are covered or removed.
we have $|H'|= |H|+|P_{bu}|+|P_{b'u'}|-1$ and $\credit(H') = \credit(H)+1-(|P_{bu}|-1)-(|P_{b'u'}|-1) -1-k\cre  = \credit(H)-|P_{bu}|-|P_{b'u'}|+1-(k-5)\cre$.
So $\cost(H') \leq \cost(H)-(k-5)\cre$, and we can let $\loan_f(H') \coloneq \max \{ \loan_f(H) - (k-5)\cre , 0 \}$ since these $(k-5)\cre$ credits are not needed for buying the bridge-covering paths.
This proves the lemma.
\end{proof}

We will frequently use the above lemmas by passing the parameters in order of $(b,b',u,u')$.
The following lemma shows how to cover bridges of non-deficient colors. %
Note that we always maintain~\cref{def:loan_invariant}.
\begin{lemma}\label{bridgeCover: noloan}
	Given a canonical 2-edge cover $H$ of a structured graph $G$ and a complex component $C$ \emph{without edges in deficient colors}, there is a polynomial-time algorithm that outputs a canonical 2-edge cover $H'$ of $G$ satisfying the conditions of~\Cref{lem: bridgeCover-main}.
\end{lemma}
\begin{proof}
	If there exists a cheap bridge-covering path $P_C$ (condition that we can check in
	polynomial time), we simply augment $H$ along $P_C$ hence obtaining the desired $H'$.

	Thus, we next assume that no such path exists.
	Let $P' = b, u_1, \ldots, u_\ell$ in $T_C$ be a longest path in $T_C$ (interpreted as a sequence of vertices).
	Notice that $b$ must be a leaf of $T_C$, hence a block node.
	Let us consider $R(b)$. Since by assumption there is no cheap bridge-covering path, $R(b)$ does not contain any block node.
	Hence $|R(b) \setminus \{u_1\} | \geq 2$
	by \Cref{lem:bridge-covering-helper1} (applied to $xy = bu_1$).
	This implies that $\ell \geq 3$.

	We next distinguish a few subcases depending on $R(b)$.
	Let $V_i$, $i \geq 1$, be the vertices in $V(T_C) \setminus V(P')$ such that their path to $b$ in $T_C$ passes through $u_i$ and not through $u_{i+1}$.
	Notice that $\{ V(P'), V_1, \ldots, V_\ell \}$ is a partition of $V(T_C)$.
	We observe that any vertex in $V_i$ is at distance at most $i$ from $u_i$ in $T_C$ as otherwise $P'$ would not be a longest path in $T_C$.
	We also observe that the leaves of $T_C$ in $V_i$ are block nodes.
	This implies the following observations.
	\begin{claim}
		All the vertices in $V_1$ are pendant block nodes;
		all the vertices in $V_2$ are either block nodes or are non-pendant vertices at distance $1$ from $u_2$ in $T_C$;
		all the vertices in $V_3$ are either block nodes or are non-pendant vertices at distance at most $2$ from $u_3$ in $T_C$.
	\end{claim}
	\textbf{Case 0:} $\ell=4$.
	Note that both $b$ and $u_4$ are pendant block nodes.
	Further, $V_3$ can only contain pendant block nodes as otherwise $P'$ was not a longest path.
	We apply \Cref{lem:bridge-covering-helper1} twice by setting $xy=bu_1$ and $xy=u_{4} u_{3}$ respectively.
	Since $R(b)$ and $R(u_4)$ contain no block node, $R(b)\setminus \{u_1\}$ contains at least two nodes in $\{u_2,u_3\}\cup V_2$ and $R(u_4)\setminus \{u_3\}$ contains at least two nodes in $\{u_1,u_2\}\cup V_2$.
	If $R(b)\setminus \{u_1\}$ contains some node $u$ in $V_2$,
	let $u' \in R(u_4)\setminus \{u_3\}$.
	Then the tuple $(b,u_4,u,u')$ satisfies the conditions of \Cref{lem:bridge-covering-helper2}.
	The case where $R(u_4)\setminus \{u_3\}$ contains some node in $V_2$ is symmetric.
	Hence we can assume $R(b)\setminus \{u_1\} = \{u_2,u_3\}$ and $R(u_4)\setminus \{u_3\} \{u_1,u_2\}$.
	Then the tuple $(b,u_4,u_3,u_2)$ satisfies the conditions of \Cref{lem:bridge-covering-helper3}.

	For the following cases, we can assume that $\ell \geq 5$.

	\textbf{Case 1:} There exists $u \in R(b)$ with $u \notin \{ u_1, u_2, u_3 , u_4 \} \cup V_1 \cup V_2   \cup V_3$. By definition there exists a bridge-covering path between $b$ and $u$ covering at least $5$ bridges, hence cheap.
	This is excluded by the assumption that there is no cheap bridge-covering path.

	\textbf{Case 2:} There exists $u \in R(b)$ with $u \in V_1 \cup V_2  \cup V_3$.

	First consider the case $u \in V_3$.
	There must be some pendant block $b' \in V_3$ such that $u \in T_C(b',u_3)$.
	Consider $R(b')$ and let $w$ be the unique neighbor of $b'$ in $T_C$ (possibly $w=u$).
	By the assumption that there are no cheap bridge-covering paths and \Cref{lem:bridge-covering-helper1} (applied to $xy = b'w$),  $R(b')$ contains at least one node $u' \notin \{b', b, w \}$. The tuple $(b, b', u, u')$ satisfies the conditions of \Cref{lem:bridge-covering-helper2} (specifically, both $T_C(b, u)$ and $T_C(b', u')$ contain $u$), hence we can obtain the desired $H'$.

	Now consider the case where $u \in V_1 \cup V_2$.
	Since $u$ is not a block node, $u$ must be a non-pendant vertex in $V_2$ at distance $1$ from $u_2$.
	Furthermore, $V_2$ must contain at least one pendant block node $b'$ adjacent to $u$.
	Consider $R(b')$.
	By the assumption that there are no cheap bridge-covering paths and \Cref{lem:bridge-covering-helper1} (applied to $xy = b'u$), $|R(b') \setminus \{u \}| \geq 2$.
	Let $u' \in R(b')$.
	The tuple $(b, b', u, u')$ satisfies the conditions of \Cref{lem:bridge-covering-helper2} (specifically, both $T_C(b, u)$ and $T_C(b', u')$ contain $u_2$) unless $u'=u_1$ or $u'=u_2$ ($|E(T_C(b, u)) \cup E(T_C(b', u'))| = 4 < 5$).
	Since $|R(b') \setminus \{u \}| \geq 2$, we can let $u'=u_2$.
	Then the tuple $(b,b',u,u')$ satisfies the conditions of~\Cref{lem:bridge-covering-helper3}.

	\textbf{Case 3:} $R(b) \setminus \{u_1 \} \subseteq \{u_2, u_3,u_4 \}$; so $|R(b) \setminus \{u_1 \} \cap \{u_2, u_3,u_4 \}| \geq 2$.
	Let $j$ ($3\leq j \leq 4$) be the maximum index such that $u_j \in R(b)$.

	We first consider the case that $V_1 \cup V_2 \neq \emptyset$.
	Take any pendant (block) node $b' \in V_1 \cup V_2$, say $b' \in V_i$.
	Let $\ell'$ be the vertex adjacent to $b'$.
	By the assumption that there are no cheap bridge-covering paths and \Cref{lem:bridge-covering-helper1} (applied to $xy = b'\ell'$), $R(b') \setminus \{ \ell' \}$ has cardinality at least 2 and contains only lonely nodes.
	Choose any $u' \in R(b') \setminus \{ \ell' \}$. Notice that $u' \notin V_1$ (but it could be a lonely node in $V_2$ other than $\ell'$).
	Observe that the tuple $(b, b', u_j, u')$ satisfies the conditions of \Cref{lem:bridge-covering-helper2}
	if $j=4$ or $|T_C(b',u')\setminus P'| \geq 2$.
	It remains to consider the case that $j=3$ and $b'$ is a pendant block node adjacent to $u_1$ or $u_2$.
	If $b' \in V_1$, we can let $b'\notin \{u_1,u_2\}$ since $|R(b)\setminus \{u_1\}|\geq 2$. Then the tuple $(b, b', u_2, u')$ satisfies the conditions of \Cref{lem:bridge-covering-helper3}.
	If $b' \in V_2$, then the tuple $(b, b', u_3, u')$ satisfies the conditions of \Cref{lem:bridge-covering-helper2} unless $u'\in \{u_1,u_3\}$.
	If $u'=u_1$ or $u'=u_3$, the tuple $(b, b', u_3, u')$ satisfies the conditions of \Cref{lem:bridge-covering-helper3}.
	In any case, we obtain the desired $H'$.
	In the following, we assume that $V_1 \cup V_2 = \emptyset$.

	\textbf{Case 3a:} $R(b) \setminus \{u_1 \} = \{u_2, u_3 \}$.
	This implies $u_2$ and $u_3$ are lonely nodes as otherwise there is a cheap bridge-covering path.
	By \Cref{lem:bridge-covering-helper1} (applied to $xy = u_2 u_3$) the set $R( \{b, u_1,u_2 \})$ contains a block node or $R( \{b, u_1,u_2 \}) \setminus \{ u_3 \}$ contains at least $2$ lonely nodes.
	Since $R(b) \setminus \{u_1 \}=\{u_2,u_3\}$ and $u_3$ is a lonely node, $R(\{u_1,u_2\})\setminus \{u_3\}$ either contains a block node or at least $2$ lonely nodes.
	If $R(\{u_1,u_2\})\setminus \{u_3\}$ contains some node $w$ at distance at least $2$ from $u_3$, then the tuple $(b,u_1,u_2,w)$ or $(b,u_2,u_3,w)$ satisfies the conditions of \Cref{lem:bridge-covering-helper4} depending on $w \in R(u_1)$ or $w \in R(u_2)$.
	Hence we assume that $R(\{u_1,u_2\})\setminus \{u_3\}$ only contains neighbors of $u_3$.

	If $R(\{u_1,u_2\})\setminus \{u_3\}$ contains some block node $b'$, then there is a bridge-covering path from $b'$ to $u_i$ for some $i \in \{1,2\}$.
	Hence we can apply \Cref{lem:bridge-covering-helper3} using the tuple $(b,b',u_{i+1},u_i)$ and make progress.
	The remaining case is that $R(\{u_1,u_2\})\setminus \{u_3\}$ contains at least two lonely nodes among the neighbors of $u_3$, which implies that $V_3 \neq \emptyset$ and there is a bridge-covering path from some $u_i$ with $i \in \{1,2\}$ to some $u\in V_3$.
	If $u$ is a block node, then we can apply \Cref{lem:bridge-covering-helper3} with $(b,u,u_{i+1},u_i)$.

	If $u$ is not a block node, there must be some block node $b' \in V_3$ such that $u$ lies on the path in $T_C$ from $b'$ to $u_3$.
	Then $R(b')$ must contain some lonely node $u' \notin \{b,u\}$.
	If the distance between $b'$ and $u_3$ is $2$, then $u'$ (which is not a block node) is not in $V_3$.
	Then the tuple $(b,b',u,u')$ satisfies the conditions of~\cref{lem:bridge-covering-helper3}.
	If the distance between $b'$ and $u_3$ is $3$, then either the tuple $(b,b',u,u')$ satisfies the conditions of~\cref{lem:bridge-covering-helper3} or $u \in T_C(u',u_3)$.
	In the latter case, we can instead consider the longest path $T_C(b',u_l)$ and we are in Case~2.

	\textbf{Case 3b:} $\{u_2,u_4\}\subseteq R(b)$ or $\{u_3,u_4\}\subseteq R(b)$.
	We apply \Cref{lem:bridge-covering-helper1} with $xy=u_3u_4$.
	Let $X_C$ be the subtree obtained from $T_C$ by removing $u_3u_4$ with $u_3 \in X_C$ and $u_4 \in Y_C$.
	Then either $R(X_C)$ contains a block node in $Y_C$ or two lonely nodes in $Y_C \setminus \{u_4\}$.
	If $R(X_C)$ contains a block node $b'$ in $Y_C$, it must be in $Y_C \setminus \{u_4\}$ since $u_4 \in R(b)$ and $u_4$ cannot be a block node.
	Then we can apply \Cref{lem:bridge-covering-helper2} using the tuple $(b,b',u_4,u)$, where $u\in X_C$ and $b' \in R(u)$.
	In the following, we assume that $R(X_C)$ contains at least $2$ lonely nodes in $Y_C \setminus \{u_4\}$.

	If there is a bridge-covering path from $u_3\cup V_3$ to $Y_C$, then we can apply \Cref{lem:bridge-covering-helper4} to obtain the desired $H'$.
	Hence we can assume that the bridge-covering paths from $X_C$ to $R(X_C) \setminus \{u_4\}$
	start from either $u_1$ or $u_2$.
	Further, if there is no bridge-covering path from $u_1$ (resp. $u_2$) to $R(X_C) \setminus \{u_4\}$, then $\{u_2,u_4\}$ (resp. $\{u_1,u_4\}$) would yield a non-isolating $2$-vertex-cut, which is impossible since our graph is structured.
	Hence, there must be some bridge-covering path from $u_1$ to $v_1 \in Y_C \setminus \{u_4\}$ and some bridge-covering path from $u_2$ to $v_2 \in Y_C \setminus \{u_4\}$.
	Note that we have either $u_2 \in R(b)$ or $u_3 \in R(b)$.
	If $u_2 \in R(b)$, then the tuple $(b,u_1,u_2,v_1)$ satisfies the conditions of \Cref{lem:bridge-covering-helper4}.
	If $u_3 \in R(b)$, then the tuple $(b,u_2,u_3,v_2)$ satisfies the conditions of \Cref{lem:bridge-covering-helper4}.
\end{proof}

Now consider some deficient color $f$.
Let $F$ be the set of bridges in color $f$.
By \cref{def:loan_invariant}, $F$ is in the same complex component $C$ of $H$ and $F$ induces a connected subtree of $T_C$, which we denote as $T_C(F)$.
Since we care about covering bridges in $T_C(F)$, if a bridge-covering path covers some bridge in $T_C(F)$, we can assume that both its endpoints are in $T_C(F)$.
If not, it covers some bridge $e'$ in color $f'\neq f$ in $T_C \setminus T_C(F)$, we use the credits on $e'$ to pay back $\loan_{f'}(H)$.
This does not change the value of $\cost(H)-\loan(H)$ and even reduces $\loan(H)$.
So in the worst case, we can assume this never happens.
By simulating the proof of \cref{bridgeCover: noloan}, we have the following corollary.

\begin{corollary}\label{corr:five_plus_loan_edges}
	Let $\ell$ be the smallest integer such that $\ell \cre \geq \loan_f(H)$.
	If $T_F(C)$ has at least $5+\ell$ bridges, then in polynomial time one can find a 2-edge cover $H'$ satisfying the conditions of~\Cref{lem: bridgeCover-main}.
\end{corollary}
\begin{proof}
	In the proof of \cref{bridgeCover: noloan}, the key is to find one or two bridge-covering paths that cover at least $5$ bridges (to gain $5\cre \geq 1$ credits).
	We simulate the same process, and whenever we cover $k \geq 5$ bridges,
	we use the credits of $5$ bridges to pay for the bought edges and the credits of the remaining $(k-5)$ bridges to pay back the loans $\loan_f(H)$.
	In this way, we cover bridges as well as maintain \cref{def:loan_invariant}.
\end{proof}

The following lemma implies that we can always make progress if $T_C(F)$ has two block nodes $b$ and $b'$ between which there are only lonely nodes of degree $2$.

\begin{lemma}\label{lem: bridgeCoverPath}
	If there are two block nodes in $T_C(F)$, $b$ and $b'$, such that the path between them contains only lonely nodes of degree $2$, then there is a polynomial-time algorithm that outputs a canonical 2-edge cover $H'$ of $G$ satisfying the conditions of \Cref{lem: bridgeCover-main}.
\end{lemma}
\begin{proof}
	We label the nodes on $T_C(b,b')$ as $u_1=b,u_2,\dots,u_q=b'$.
	Note that $u_2, \cdots, u_{q-1}$ are nodes of degree $2$.
	We show that we can find bridge-covering paths to cover the bridges on the path $T_C(u_1,u_q)$.
	Let $T_{u_1}$ (resp.\ $T_{u_q}$) be the subtree of $T_C\setminus \{u_1u_2\}$ (resp.\ $T_{u_{q-1}u_q}$) containing $u_1$ (resp. \ $u_q$).
	Since we aim to cover the bridges on the path $T_C(u_1,u_q)$, we can assume w.l.o.g.\ that any bridge-covering path starting from $T_{u_1}$ (resp.\ $T_{u_q}$) starts from $u_1$ (resp.\ $u_q$).
	Further, we can consider $T_C'$ obtained from $T_C$ by contracting $T_{u_1}$ and $T_{u_q}$
	That is, $T_C'$ is a path whose two endpoints are block nodes and all internal nodes are lonely nodes.
	We slightly abuse the notations by denoting the path $T_C'$ as $u_1,u_2,\dots,u_q$.

	If $u_q \in R(u_1)$, we can add any bridge-covering path between $u_1$ and $u_q$ since it covers $2$ blocks and is cheap.
	In the following, we assume that $u_q \notin R(u_1)$ and symmetrically $u_1 \notin R(u_q)$.
	If $q \geq 9$, i.e., there are at least $8$ bridges.
	Then we can immediately apply \cref{corr:five_plus_loan_edges} with $l=3$.

	Now we assume that $q \leq 8$.
	We apply~\Cref{lem:bridge-covering-helper1} with $xy = u_1u_2$, and we can conclude that $|R(u_1) \cap \{u_3, u_4, \dots, u_{q-1} \}| \geq 2$.
	Symmetrically, we have $|R(u_q) \cap \{u_2, u_3, \dots, u_{q-2} \}| \geq 2$.
	If $q \leq 6$, i.e., there are at most $5$ bridges, there must be some index $i$ such that $u_i \in R(u_q)$ and $u_{i+1} \in R(u_1)$.
	Then the tuple $(u_1,u_q,u_{i+1},u_i)$ satisfies the conditions of~\Cref{lem:bridge-covering-helper3} and we can make progress.
	Note that we use the credits on the covered bridges to pay for $\loan_f(H)$.
	In the following, we assume $7 \leq q \leq 8$, i.e., there are $6$ or $7$ bridges.

	We first consider the case where $q=8$, i.e., $T_C'$ is a path of $7$ edges.
	We will cover all these edges using up to $3$ bridge-covering paths.
	These $7$ bridges have $7\cre$ credits in total and we use them to fully pay for $\loan_f(H) \leq \frac{1}{2}+10\delta$.
	After that, we have at least $7\cre - \loan_f(H) \geq \frac{5}{4}-17\delta \geq 1$ unused credits, $\loan(H')=0$ and $\cost(H')-\loan(H') = \cost(H)-\loan_f(H)$ since we decrease the credits and the loans by the same amount.
	We will make use of this $1$ unused credit to pay for the bridge-covering paths using the following operations.

	\begin{itemize}
		\item[(i)] (Only for $q=8$) Add two bridge-covering paths, $P$ from $u_1$ to some $u_i$ and $Q$ from $u_j$ to $u_q$ where $1 < j \leq i < q$.
		In this way, we obtain $H'$ with $|H'|=|H|+|P|+|Q|$, and $\credit(H') \leq \credit(H)+1 - (|P|-1+|Q|-1 +2 + 1) = \credit(H)-|P|-|Q|$.
		The last $+1$ is the unused credit.
		Hence $\cost(H')\leq \cost(H)$ and we are done.
		\item[(ii)] (Only for $q=8$) Add three pairwise disjoint bridge-covering paths $P,Q,R$ so that they cover all the bridges, and remove $1$ edge so that no bridge is created.
		In this way, we obtain $H'$ with $|H'|=|H|+|P|+|Q|+|R|-1$, and $\credit(H') \leq \credit(H)+1 - (|P|-1+|Q|-1+|R|-1 +2 + 1) = \credit(H)-|P|-|Q|-|R|+1$.
		The last $+1$ is the unused credit.
		Hence $\cost(H')\leq \cost(H)$ and we are done.
		\item[(iii)] (Do not need the unused credits from the bridges)
		Add three pairwise disjoint bridge-covering paths $P,Q,R$ so that they cover all the bridges, and remove $2$ edges so that no bridge is created.
		In this way, we obtain $H'$ with $|H'|=|H|+|P|+|Q|+|R|-2$, and $\credit(H') \leq \credit(H)+1 - (|P|-1+|Q|-1+|R|-1 +2) = \credit(H)-|P|-|Q|-|R|+2$.
		Hence $\cost(H')\leq \cost(H)$ and we are done.
	\end{itemize}

	If there are indices $i$ and $j$ with $i \leq j $ such that $u_{j} \in R(u_1)$ and $u_i \in R(u_7)$,
	let the corresponding bridge-covering paths be $P$ and $Q$.
	If $P$ and $Q$ intersect internally, then there is another bridge-covering path between $u_1$ and $u_8$, contracting $u_8 \notin R(u_1)$.
	If $P$ and $Q$ are internally disjoint, we can perform Operation (i).

	If there are not such $i$ and $j$, the only remaining case is that $R(u_1)\setminus \{u_2\} = \{u_3, u_4\}$ and $R(u_8) \setminus \{u_7\} = \{u_5,u_6\}$.
	Then we apply~\Cref{lem:bridge-covering-helper1} with $xy=u_5u_4$ and obtain that either $R(\{u_5,u_6,u_7,u_8\}) \setminus \{u_4\}$ contains $2$ lonely nodes or $R(\{u_5,u_6,u_7,u_8\})$ contains a block node (hence $u_1$).
	However, the latter case cannot happen since $R(u_1) \subseteq \{u_2,u_3,u_4\}$.
	Consider the former case where  $|R(\{u_5,u_6,u_7,u_8\}) \setminus \{u_4\}| \geq 2$, which implies $R(\{u_5,u_6,u_7,u_8\}) \setminus \{u_4\} = \{u_2,u_3\}$.
	Let $P$ be any bridge-covering path from $u_3$ to some $u_j$ with $j \geq 5$.
	We add to $H$ a bridge-covering path $Q$ from $u_1$ to $u_4$, $P$, and a bridge-covering path $R$ from $u_5$ to $u_8$.
	Further, we remove the edge $u_3u_4$.
	The resulting graph $H'$ has fewer bridges than $H$, and we argue that $\cost(H') \leq \cost (H)$.
	Note that $P$ and $Q$ are internally disjoint as otherwise there would be another bridge-covering path from $u_1$ to $u_j$.
	Similarly, $Q$ and $R$ are internally disjoint.
	Therefore, we can apply Operation~(ii).

	We finally consider the case where $q=7$, i.e., $T_C'$ is a path of $6$ edges.
	Unless explicitly written, we will immediately pay back the loans using all the credits of the bridges, and will not use the credits on the bridges for buying edges.
	By the above discussion, we have $R(u_1) \cap \{u_3, u_4, \dots, u_{6} \}| \geq 2$ and
	$|R(u_7) \cap \{u_2, u_3, \dots, u_{5} \}| \geq 2$.
	If there is some index $i$ such that $u_{i+1} \in R(u_1)$ and $u_i \in R(u_7)$, then the tuple $(u_1,u_7,u_{i+1},u_i)$ satisfies the conditions of~\Cref{lem:bridge-covering-helper3}, and we can make progress.
	Hence we can assume this does not happen.
	We make a case distinction on the maximum index $i$ such that $u_i \in R(u_1)$.

	\textbf{Case 1:} $\argmax_i \{ u_i \mid u_i \in R(u_1) \} = 6$.
	We have $R(u_7) \subseteq \{u_2,u_3,u_4\}$.
	There are two subcases.

	\textbf{Case 1.1:} $R(u_1)\setminus \{u_2\}=\{u_5,u_6\}$ and $R(u_7)\setminus \{u_6\}=\{u_2,u_3\}$.
	Consider the $3$-vertex cut $\{u_2,u_3,u_6\}$.
	Let $A$ (resp.\ $B$) be the set of vertices in $G\setminus \{u_2,u_3,u_6\}$ reachable from the vertices in the block $u_1$ (resp.\ $u_7$).
	Let $R$ be $V\setminus (A \cup B \cup C)$.

	If $B$ contains some other component, then $|B| \geq 9$ since $u_7$ contains at least $5$ vertices and any other component contains at least $4$ vertices.
	Let $S:=A \cup R$, $T:=B$.
	We claim that $\{u_2,u_3,u_6\}$ is a large $3$-vertex cut.
	Suppose it is not a large $3$-vertex cut.
	Then $|A| \leq 7$, which implies that any bridge-covering path between $u_1$ and $u_5$ (or $u_6$) must be a single edge, and the block node $u_1$ consists of a nice $\cFive$.
	We number the $5$ vertices $b_1,\dots,b_5$ of $u_1$ as in~\Cref{def:canonicalD2}.
	By~\cref{lem:3-matching}, there is a 3-matching between $V(u_7)$ and $V\setminus V(u_7)$.
	Since $b_2$ and $b_5$ are vertices of degree $2$, the $3$-matching must be incident to $b_1,b_3$ and $b_4$, i.e., the $3$-matching matches $b_1,b_3,b_4$ respectively to $u_2,u_5,u_6$ (but the correspondence between the endpoints is not clear).
	If there is an edge between $b_3$ and $u_2$,
	then there is another edge between $b_4$ and some $w \in \{u_5,u_6\}$.
	We can add $b_3u_2$ and $b_4w$, and remove $b_3b_4$ and $b_1u_1$ to obtain
	the desired $H'$ with fewer bridges.
	Symmetrically, we can assume that there is no edge between $b_4$ and $u_2$.
	Hence, we can assume $b_3u_5 \in E$ and $b_4u_6\in E$.
	Then we add $b_3u_5$ and $b_4u_6$, and remove $b_3b_4$ and $u_5u_6$ to obtain the desired $H'$ with fewer bridges.

	If $B$ does not contain any other component,
	then any bridge-covering path starting from $u_7$ is a single edge.
	If $u_7$ is not an original block in $H_0$, i.e., it is obtained after some steps of bridge-covering, then $|B| \geq 8$, which implies $\{u_2,u_3,u_6\}$ is a large $3$-vertex cut.
	Now assume that $u_7$ is an original block of $H_0$.
	If it contains at least $6$ vertices, by the loan invariants, we have $\loan_f(H) \leq \frac{1}{4}+11\delta$.
	We have $6$ bridges and $6\cre$ credits in total, and we use them to fully pay for $\loan_f(H) \leq \frac{1}{4}+11\delta$.
	After that, we have at least $6\cre - \loan_f(H) \geq \frac{5}{4}-17\delta \geq 1$ credits.
	Then we can perform the same operations as we do for the case with $q=8$.
	Finally, consider the case $u_7$ is a (nice) $\cFive$.
	Then we can perform the same operations as the previous case where $|A|=7$.

	\textbf{Case 1.2:}
	$R(u_1)\setminus \{u_2\} = \{u_4,u_6\}$ and $R(u_7)\setminus \{u_6\} = \{u_2, u_4\}$.
	We claim that $\{u_2,u_4,u_6\}$ is a large $3$-vertex cut, contradicting that $G$ is structured.
	Let $A$ (resp.\ $B$) be the set of vertices in $G\setminus \{u_2,u_4,u_6\}$ reachable from the vertices in the block $u_1$ (resp.\ $u_7$).
	Let $R$ be $V\setminus (A \cup B \cup C)$.
	We can w.l.o.g.\ $|A| \geq |B|$ by symmetry.
	We follow the notations in~\Cref{def:large_three_cut}.
	Let $V_1:=A\cup R$ and $V_2:=B \cup \{u_3,u_5\}$.
	If both $|A| \geq 8$ and $|B| \geq 8$, then $\{u_2,u_4,u_6\}$ is a large $3$-vertex cut, a contraction.
	Now we assume $|B|<8$, which implies $u_7$ must be a nice $\cFive$ and $V(u_7)$ can only be adjacent to $u_2$ and $u_4$ by the definition of $B$.
	We number the $5$ vertices $b_1,\dots,b_5$ in $u_7$ as in~\Cref{def:canonicalD2}.
	By~\cref{lem:3-matching}, there is a 3-matching between $V(u_7)$ and $V\setminus V(u_7)$.
	Since $b_2$ and $b_5$ are vertices of degree $2$, the $3$-matching must be incident to $b_1,b_3$ and $b_4$.
	As in Case 1.1, we can assume that the $3$-matching contains $b_3u_2$ and $b_4u_4$.
	Then we can perform Operation~(iii) by adding the bridge-covering path between $u_1$ and $u_6$, $b_3u_2$, and $b_4u_4$,
	and removing $b_1u_6$ and $b_3b_4$.

	\textbf{Case 2:} $\argmax_i \{ u_i \mid u_i \in R(u_1) \} = 5$.
	We can also assume $\argmin_j \{ u_j \mid u_j \in R(u_7) \} \geq 3$.
	Otherwise, we can swap the roles of $u_1, u_7$ and end up with the previous case.
	Recall that there is no such index $i$ that $u_{i+1} \in R(u_1)$ and $u_i \in R(u_7)$,
	Hence, we have $R(u_7)\setminus \{u_6\} = \{u_3,u_5\}$, which implies $R(u_1)\setminus \{u_2\} =\{u_3,u_5\}$.
	We apply~\Cref{lem:bridge-covering-helper1} with $xy=u_2u_3$ and obtain that either $R(\{u_1,u_2\}) \setminus \{u_3\}$ contains $2$ lonely nodes or $R(\{u_1,u_2\})$ contains a block node (hence $u_7$).
	The latter is impossible since $u_1,u_2 \notin R(u_7)$.
	In the former case,
	there must be some bridge-covering path $P$ from $u_2$ to $u_j$ for some $j\in \{4,6\}$.

	\textbf{Case 2.1:}  $j=4$.
	Let $Q$ be any bridge-covering path from $u_1$ to $u_3$ and $R$ be any bridge-covering path from $u_7$ to $u_3$.
	By arguments similar to previous cases, $P,Q,R$ are pairwise disjoint.
	So we can perform Operation~(iii): we add $P,Q,R$ to $H$ and remove the edges $u_2u_3,u_3u_4$.

	\textbf{Case 2.2:}  $j=6$.
	Let $Q$ be any bridge-covering path from $u_1$ to $u_3$ and $R$ be any bridge-covering path from $u_7$ to $u_5$.
	By arguments similar to previous cases, $P,Q,R$ are pairwise disjoint.
	So we can perform Operation~(iii): we add $P,Q,R$ to $H$ and remove the edges $u_2u_3,u_5u_6$.

	\textbf{Case 3:} $\argmax_i \{ u_i \mid u_i \in R(u_1) \} = 4$. We can also assume $\argmin_j \{ u_j \mid u_j \in R(u_7) \} = 4$. Otherwise, we can swap the roles of $u_1, u_7$ and end up with the previous cases.
	We apply~\Cref{lem:bridge-covering-helper1} with $xy=u_5u_4$ and obtain that either $R(\{u_5,u_6,u_7\}) \setminus \{u_4\}$ contains $2$ lonely nodes or $R(\{u_5,u_6,u_7\})$ contains a block node (hence $u_1$).
	However, the latter case cannot happen since $R(u_1) \subseteq \{u_2,u_3,u_4\}$.
	Hence $R(\{u_5,u_6,u_7\}) \setminus \{u_4\}$ contains $2$ lonely nodes, which are $u_2$ and $u_3$.
	Let $P$ be a bridge-covering path from $u_2$ to some $u_j$ with $5\leq j \leq 6$.
	Let $Q$ be any bridge-covering path from $u_1$ to $u_3$ and $R$ be any bridge-covering path from $u_7$ to $u_{j-1}$.
	By arguments similar to previous cases, $P,Q,R$ are pairwise disjoint.
	So we can perform Operation~(iii): we add $P,Q,R$ to $H$ and remove the edges $u_2u_3,u_{j-1}u_j$.
\end{proof}

By \cref{def:loan_invariant}, $T_C(F)$ has at most $4$ leaves.
By exhaustively applying~\cref{lem: bridgeCoverPath}, we are done with the case where $T_C(F)$ has $2$ leaves.
Now we consider the case where $T_C(F)$ has $3$ leaves, i.e., it consists of $3$ paths intersecting at one common node $r$.
By the arguments on the loans of each color, we have $\loan_f(H) \leq \frac{1}{4}+15\delta$, which implies the credits of two bridges are enough to pay for the loan.
Let the $3$ leaves be $u_1,u_2,u_3$ respectively.
We can assume that we cannot use~\Cref{lem: bridgeCoverPath} to make progress.
Then $u_1,u_2,u_3$ are the only block nodes in $T_C(F)$.
For $1 \leq i \leq 3$, let $T_{u_i}$ be the subtree of $T_C\setminus F$ containing $u_i$.
Since we aim to cover the bridges in $F$, we can assume w.l.o.g.\ that any bridge-covering that has one endpoint in $T_{u_i}$ starts at $u_i$.
Hence, we can consider $T_C'$ obtained from $T_C$ by contracting $T_{u_i}$.
We slightly abuse the notations by denoting the leaves of $T_C'$ as $u_1,u_2,u_3$.

\begin{lemma}
	If $T_C(F)$ is a tree with $3$ leaves and every non-leaf node is a lonely node, then there is a polynomial-time algorithm that outputs a canonical 2-edge cover $H'$ of $G$ satisfying the conditions of \Cref{lem: bridgeCover-main}.
\end{lemma}
\begin{proof}
	We can assume that w.l.o.g.\ $u_1,u_2,u_3$ are sorted in non-increasing distances to $r$.
	Hence $T_C'(u_1,u_2)$ is a longest path in $T_C'$.
	Suppose $T_C'$ has at least $7$ bridges.
	Then we can immediately apply \cref{corr:five_plus_loan_edges} with $l=2$.
	Hence, we can assume that $T_C'$ has at most $6$ edges in the following.
	Let $d_i$ be the distance between $u_i$ and $r$ on $T_C'$.
	Then the only possible configurations of $(d_1,d_2,d_3)$ are
	$(3,2,1), (3,1,1), (2,2,2), (2,2,1),(2,1,1)$ and $(1,1,1)$.
	We assume that there is no bridge-covering path between $u_i$ and $u_j$ for any $1 \leq i, j\leq 3$, since otherwise we can add these cheap bridge-covering path and make progress.

	\textbf{Case 1:} $(d_1,d_2,d_3) = (3,2,1)$ or $(d_1,d_2,d_3) = (3,1,1)$.
	We only show that case where $(d_1,d_2,d_3) = (3,2,1)$.
	The other case is similar.
	We label the nodes on the path between $u_1$ and $u_2$ as $u_1,v_1,v_2,v_3,v_4,u_2$.
	By applying~\Cref{lem:bridge-covering-helper1} with $xy=u_1v_1$, we have $|R(u_1)\setminus \{v_1\} \cap \{v_2,v_3,v_4\}| \geq 2$.
	Similarly, we have $|R(u_2)\setminus \{v_4\} \cap \{v_1,v_2,v_3\}| \geq 2$.
	Then there must be some index $i$ such that $v_{i+1} \in R(u_1)$ and $v_i \in R(u_2)$.
	Hence we can apply~\Cref{lem:bridge-covering-helper3} with $(u_1,u_2,v_{i+1},v_i)$ to make progress.

	\textbf{Case 2:} $(d_1,d_2,d_3) = (2,2,2)$ or $(d_1,d_2,d_3) = (2,2,1)$.
	We only show that case where $(d_1,d_2,d_3) = (2,2,2)$.
	The other case is similar.
	Let $v_i$ be the node between $u_i$ and $r$ for $1 \leq i \leq 3$.
	By applying~\Cref{lem:bridge-covering-helper1} with $xy=u_1v_1$, we have $|R(u_1)\setminus \{v_1\} \cap \{r,v_2,v_3\}| \geq 2$.
	We can assume w.l.o.g.\ that $v_2 \in R(u_1)$.
	By applying~\Cref{lem:bridge-covering-helper1} with $xy=u_2v_2$, we have $|R(u_1)\setminus \{v_2\} \cap \{r,v_1,v_3\}| \geq 2$.
	If $r \in R(u_2)$, we can apply~\Cref{lem:bridge-covering-helper3} with $(u_1,u_2,v_2,v_1)$.
	If $r \notin  R(u_2)$, then $v_3 \in R(u_2)$ and we can apply~\cref{lem:bridge-covering-helper2} with $(u_1,u_2,v_2,v_3)$.

	\textbf{Case 3:} $(d_1,d_2,d_3) = (2,1,1)$ or $(d_1,d_2,d_3) = (1,1,1)$.
	By applying~\Cref{lem:bridge-covering-helper1} with $xy=u_3r$, $R(u_3)\setminus \{r\}$ contains $2$ lonely nodes, which is impossible.
\end{proof}

Finally, we consider the case where $T_C(F)$ has $4$ leaves.
\begin{lemma}
	If $T_C(F)$ is a tree with $4$ leaves and every non-leaf node is a lonely node, then there is a polynomial-time algorithm that outputs a canonical 2-edge cover $H'$ of $G$ satisfying the conditions of \Cref{lem: bridgeCover-main}.
\end{lemma}
\begin{proof}
	By \cref{def:loan_invariant}, we have $\loan_f(H) \leq 20 \delta$, which implies the credits of one bridge are enough to pay for the loan.
	Let the $4$ leaf block nodes be $b_1,b_2,b_3,b_4$.
	If $T_C(F)$ has at least $6$ edges, then we can apply \cref{corr:five_plus_loan_edges} with $\ell=1$.
	In the following, we can assume that $T_C(F)$ has at most $5$ edges and the longest path in $T_C(F)$ has at most $3$ edges.

	\textbf{Case 1:} the longest path has $2$ edges.
	In this case, $T_C(F)$ must be isomorphic to a star, i.e., the $4$ blocks are incident to a common node $r$.
	We apply~\Cref{lem:bridge-covering-helper1} with $xy=b_1r$.
	Since $r$ is the only node that could be a lonely node, there must be a bridge-covering path between $b_1$ and some other $b_i$ for $2 \leq i \leq 4$.
	Hence, we can add the cheap bridge-covering path to make progress.

	\textbf{Case 2:} the longest path has $3$ edges.
	We assume that the longest path consists of $b_1,v_1,v_2,b_2$.
	Since it is a longest path, $b_3$ and $b_4$ must be incident to either $v_1$ or $v_2$.
	We apply~\Cref{lem:bridge-covering-helper1} with $xy=b_1v_1$.
	Since $R(b_1)\setminus \{v_1\}$ cannot contain $2$ lonely nodes (only $v_2$ can be a lonely node), there must be a bridge-covering path between $b_1$ and some other block node.
	Hence, we can add the bridge-covering path to make progress.
\end{proof}

\subsection{Gluing}
\label{sec:gluing}
We have a structured graph $G$ as input and a bridgeless canonical 2-edge cover $H$ of $G$ with $\cost(H)\leq \cre \opt$ such that $H$ satisfies the credit invariant for each 2EC component:
\begin{itemize}[nosep]
	\item Each large component receives a credit of $2$,
	\item each rich vertex receives a credit of $2(\frac{5}{4}-\delta)$,
	\item each $\cFour$ receives a credit of $1$, and
	\item each $\mathcal C_i$ receives a credit of $\cre i$ for $5\leq i\leq 8$.
\end{itemize}

The main goal of this section is to prove \Cref{lem:gluing:main}, which we restate for convenience.

\lemmaMANYgluing*

We first introduce some definitions. In what follows, unless otherwise mentioned, $G$ is a structured graph and $H$ is a bridgeless weakly canonical 2-edge cover of~$G$.

\begin{observation} \label{obs:rich-node}
	Let $A$ be a node of $\hG_H$ such that $C_A$ is a rich vertex.
	Then $A$ cannot be a cut node of $\hG_H$.
\end{observation}
Otherwise, $C_A$ is a cut vertex of $G$, which contradicts the fact that $G$ is structured.

\begin{observation}\label{obs:cut-pendant}
	Let $S$ be a segment of $\hG_H$ with $|S|=2$ and let $A$ be a node in $S$. Then $A$ is either a cut node of $\hG_H$ or a pendant node of $\hG_H$.
\end{observation}
\begin{proposition}\label{pros:3-matching-in-small-segment}
	Let $A$ and $B$ be adjacent nodes in $\hG_H$ such that $A\cup B$ is a segment. If neither $C_A$ nor $C_B$ is a rich vertex, then there exists a matching of size $3$ going between $V(C_A)$ and $V(C_B)$ in $G$.
\end{proposition}

There is a node $L$ of $\hG_H$ such that $C_L$ is a huge component.
Let $S$ be a segment of $\hG_H$ that contains~$L$.

\lemmaMANYgluingtrivial*

\lemmaMANYgluingnontrivial*

As explained before,
the above two lemmas immediately imply~\cref{lem:gluing:main}. All that remains is to prove the above two lemmas, which we do in the next two subsections.
The following lemma is a useful tool to prove the above two lemmas.

\begin{lemma}\label{lem:cycle-hamiltonian}
	Let $L$ be a node of $\hG_H$ such that $C_L$ is a huge component.
	Let $S$ be a segment of $\hG_H$ that contains $L$ and let $A$ be another node in $S$.
	If there is a cycle $F$ in $S$ through $L$ and $A$, such that $F$ is incident to two distinct vertices $u$ and $v$ in $C_A$, and there is a Hamiltonian path $P$ between $u$ and $v$ in $G[V(C_A)]$, then we can compute in polynomial time a canonical bridgeless 2-edge cover $H'$ of $G$ such that $\cost(H') \leq \cost(H)$.
\end{lemma}
\begin{proof}
	We set $H' \coloneq (H \setminus E(C_A)) \cup F \cup P$, which is clearly a canonical bridgeless 2-edge cover of $G$ with fewer components than $H$. Furthermore, $\cost(H') - \cost(H) = (|H'| - |H|) + (\credit(H') - \credit(H))  = (|F| - 1) + (2-|F|-1) \leq 0$, since we added $|F|$ edges, effectively removed $1$ edge as $|C_A| - |P|=1$, $\credit(C_L)=2$ and $\credit(C_A)\geq 1$ in $H$, whereas in $H'$ those combine to form a single component with credit 2. Thus, $\cost(H') \leq \cost(H)$, meeting our goal.
\end{proof}

\subsubsection{Proof of~\cref{lem:gluing:trivial-segment}}

\begin{lemma}[Lemma 10 in~\cite{BGGHJL25}]\label{lem:niceCycle}
	Let $S$ be a segment of $\hat{G}_H$ and let $C_1$ and $C_2$ be two
	distinct components of $S$. Given edges $u_1X_1, u_2 X_2$, where $u_1\in V(C_1), u_2\in V(C_2), X_1\in V(S)\setminus\{ C_1\}, X_2\in V(S)\setminus\{C_2\}$, one can compute in polynomial time a cycle $F$ in $S$ containing $C_1$ and $C_2$, such that $F$ is incident on $u_1$ and another vertex in $C_1$, and incident on $u_2$ and another vertex in $C_2$.
\end{lemma}

\begin{lemma}[Lemma 11 in~\cite{BGGHJL25}]\label{lem:cycle-through-cfour}
	Let $S$ be a segment of $\hG_H$ and let $C_1, C_2$ be two different components of $S$ such that
	$C_1$ is a $4$-cycle. Given an edge $u_2 X$, where $u_2 \in V(C_2)$ and $X \in V (S) \setminus \{C_2\}$,
	one can compute in polynomial time a cycle $F$ in $S$ incident to distinct vertices $u_i, v_i \in V (C_i)$ for
	$i \in \{1, 2\}$ such that there is a Hamiltonian $u_1, v_1$-path in $G[V(C_1)]$.
\end{lemma}

\begin{proof}[Proof of~\cref{lem:gluing:trivial-segment}]
	Our goal is to construct a canonical bridgeless 2-edge cover $H'$ of $G$ with fewer components compared to $H$ such that $H'$ contains a huge component and $\cost(H')\leq\cost(H)$.
	Let $L$ be the node corresponding to the huge component in $S$, and $A$ be the other node in $S$.

	If $C_A$ is a rich vertex, then by \cref{obs:rich-node} and \cref{obs:cut-pendant}, $A$ is a pendant node of $\hG_H$.
	By the definition of a rich vertex, $C_A$ is incident to at least $2$ edges in $C_L$.
	Then, we simply add any two edges between $C_L$ and $C_A$ to $H$ to obtain $H'$, and it can be easily checked that $H'$ meets our goal as $\cost(H') - \cost(H) = (|H'| - |H|) + (\credit(H') - \credit(H))  \leq 2 + (2-2-2)\leq 0$ (we buy 2 edges, $\credit(C_L)=2, \credit(C_A)\geq 2$, and, in $H'$, $C_L$ and $C_A$ combine to form a single component with credit at least $2$).

	From now on, we assume that $C_A$ is not a rich vertex.
	By \Cref{pros:3-matching-in-small-segment}, we know that there must be a 3-matching $M_3$ between $C_L$ and $C_A$.
	First, note that if $C_A$ is a large component, we simply add any two edges of $M_3$ to $H$ to obtain $H'$, and it can be easily checked that $H'$ meets our goal as $\cost(H') - \cost(H) = (|H'| - |H|) + (\credit(H') - \credit(H))  \leq 2 + (2-2-2)\leq 0$ (we add 2 edges, $\credit(C_L)=\credit(C_A)=2$, and, in $H'$, $C_L$ and $C_A$ combine to form a single component of credit $2$). Thus, we can assume that $C_A$ is $\mathcal C_i$ for some $4 \leq i \leq 8$.

	We now show that we can always apply~\cref{lem:cycle-hamiltonian}. If $C_A$ is a $\cFour$ or a $\cFive$, then in either case, there must be two distinct edges in the matching $M_3$, say $e_1$ and $e_2$, whose endpoints in $C_A$ are adjacent via an edge $e$ in $C_A$. Thus, the conditions of~\cref{lem:cycle-hamiltonian} hold by setting $P=E(C_A)\setminus \{e\}$. Otherwise, $C_A$ is $\mathcal C_i$ for some $6 \leq i \leq 8$.

	\textbf{Case 1: $A$ is not a pendant node in $\hG_H$.}
	Consider partitioning $V(G)$ in two parts to have $V(C_L)$ in one part $V_1$ and $V(C_A)$ in another part $V_2$ such that the only edges crossing the parts are between $V(C_L)$ and $V(C_A)$.
	Since $L$ is huge, $|V_1|\geq 11$.

	\textbf{Case 1.1: $|V_2|\geq 11$.}

	Then \Cref{lem:4-matching} is applicable, implying a 4-matching $M_4$ between $C_A$ and $C_L$.
	If $C_A$ is a $\cSix$ or $\cSeven$,
	as before, there must exist edges $e_1$ and $e_2$ whose endpoints in $C_A$ are adjacent via an edge $e$ in $C_A$. Again, we can apply~\cref{lem:cycle-hamiltonian} by setting $P= E(C_A)\setminus \{e\}$.
	If $C_A$ is a $\cEight$ and the conditions of~\cref{lem:cycle-hamiltonian} do not hold, then $L$ must be adjacent to every other vertex of $C_A$.
	Let the vertices of the cycle $C_A$ be labeled as  $v_1 - v_2 - v_3 - v_4 - v_5 - v_6 - v_7-v_8-v_1$ and without loss of generality assume that $L$ is adjacent to $v_i$ for $i \in \{1,3,5,7\}$.
	Let $S'$ be any segment of $\hG_H$ containing $A$ but not containing $L$.

	We say an edge in $G[V(C_A)]$ is a nice diagonal if both its endpoints are in $\{v_2,v_4,v_6,v_8\}$.
	If there exists a nice diagonal $e$ in $G[V(C_A)]$,
	then observe that there are two distinct edges in $M_4$ that are incident to say $v_i$ and $v_j$ in $C_A$, and there is a Hamiltonian Path $P\subseteq \{e\}\cup E(C_A)$ from $v_i$ to $v_j$ in $G[V(C_A)]$.
	So we can apply~\cref{lem:cycle-hamiltonian}.
	If there is no nice diagonal, we claim that there is some edge $v_jX$ for some $j \in \{2,4,6,8\}$ and some $X \in V(S') \setminus A$.
	Otherwise, the $4$ sets of edges of $G$, namely, edges incident to $v_2$, edges incident to $v_4$, edges incident to $v_6$, and edges incident to $v_8$ are all inside $G[V(C_A)]$ and mutually disjoint.
	As any 2ECSS of $G$ must contain at least $2$ edges incident on each vertex, it has to contain at least $8$ edges from $G[V(C_A)]$, and hence $C_A$ is contractible.
	We can assume without loss of generality that $v_j=v_2$.

	Suppose $S'$ does not contain any component that is a $\cFour$.
	Then each component (which can be a rich vertex) has at least $5\cre$ credits.
	We apply~\Cref{lem:niceCycle} to $A$ and any other node in $S'$ to find a cycle $F$ in $S'$ containing $A$.
	Further, $F$ is incident to $u_2$ and some other vertex $u_k$ ($k \neq 2$).
	Let $H':= (H\setminus\{v_1v_2\}) \cup F \cup \{v_1L, v_3L\}$.
	Notice that $|H'|=|H|+|F|+1$.
	In $H'$ there is a large component $C'$ spanning the nodes of all components of $H$ incident on $F$ and the component $C_L$.
	Every such component yields at least $5\cre$ credits, while $L$ yields $2$ credits and $A$ yields $8\cre$ credits.
	Then we have $\credit(H') \leq \credit(H)-(|F|-1)\cdot 5\cdot \cre-\credit(L)-\credit(A) +\credit(C')\leq \credit(H)-(|F|-2)\cdot 5\cre-5\cre-2-8\cre+2 = \credit(H)-(|F|-2)-13\cre \leq \credit(H)-|F|-1$.
	Thus $\cost(H)-\cost(H') = |H|-|H'|+\credit(H)-\credit(H') \geq -|F|-1 +|F|+1 \geq 0$.

	If $H'$ contains another node $B$ that is a $\cFour$,
	we apply~\Cref{lem:cycle-through-cfour} to find a cycle $F$ in $S$ incident to distinct nodes $b_1, b_2 \in V (B)$ such that there is a Hamiltonian $b_1, b_2$-path $P$ in $G[V(B)]$, and incident to $v_2$ and some $v_k$ for some $k \neq 2$.
	Set $H' := (H \setminus (E(C_B)\cup \{v_1v_2\})) \cup \{v_1L,v_3L\}  \cup F \cup P$.
	Notice that $|H'|=|H|+|F|$.
	In $H'$, there is a large component $C'$ spanning the nodes of all components of $H$ incident on $F$ and the component $C_L$. Every such component yields at least $1$ credit, while $L$ yields $2$ credits.
	Then we have $\credit(H') \leq \credit(H)- (|F|-2)-\credit(L)-\credit(A)-\credit(B)+\credit(C')=\credit(H) -(|F|-2)-2-1-1+2 \leq \credit(H) -|F|$.
	Hence $\cost(H)-\cost(H')=|H|-|H'|+\credit(H)-\credit(H') \geq -|F|+|F| \geq 0$.

	\textbf{Case 1.2: $|V_2|\leq 10$.}
	Since $|A| \geq 6$, the only possible case is that $A$ has a neighbor $B \neq L$ in $\hG_H$ where $B$ is a $\cFour$.
	If the conditions of~\cref{lem:cycle-hamiltonian} do not hold, then $L$ must be adjacent to every other vertex of $C_A$.
	Let the vertices of the cycle $C_A$ be labeled as  $v_1 - v_2 - v_3 - v_4 - v_5 - v_6-v_1$ and without loss of generality assume that $L$ is adjacent to $v_i$ for $i \in \{1,3,5\}$.
	Let $S$ be the segment of $\hG_H$ containing $A$ and $B$.
	There must be some edge $v_jX$ for some $j \in \{2,4,6\}$ and some $X \in V(S) \setminus A$, as otherwise $C_A$ is contractible; the arguments are essentially the same as in Case 1.1.
	We can assume without loss of generality that $v_j=v_2$.
	We apply~\Cref{lem:cycle-through-cfour} to find a cycle $F$ in $S$ incident to distinct nodes $b_1, b_2 \in V (B)$ such that there is a Hamiltonian $b_1, b_2$-path $P$ in $G[V(B)]$, and incident to $v_2$ and some $v_k$ for some $k \neq 2$.
	Set $H' := (H \setminus (E(C_B)\cup \{v_1v_2\})) \cup \{v_1L,v_3L\}  \cup F \cup P$.
	Notice that $|H'|=|H|+|F|$.
	In $H'$ there is a large component $C'$ spanning the nodes of all components of $H$ incident on $F$ and the component $C_L$. Every such component yields at least $1$ credit, while $L$ yields $2$ credits.
	Then we have $\credit(H') \leq \credit(H)- (|F|-2)-\credit(L)-\credit(A)-\credit(B)+\credit(C')=\credit(H) -(|F|-2)-2-1-1+2 \leq \credit(H) -|F|$.
	Hence, $\cost(H)-\cost(H')=|H|-|H'|+\credit(H)-\credit(H') \geq -|F|+|F| \geq 0$.

	\textbf{Case 2: $A$ is a pendant node in $\hG_H$.}

	\textbf{Case 2.1: $C_A$ is a $\cSix$ or $\cSeven$.}

	Assuming that the condition of \cref{lem:cycle-hamiltonian} does not hold, we prove that $C_A$ is contractible, which contradicts the fact that $G$ is structured. Let the vertices of the cycle $C_A$ be labeled as  $v_1 - v_2 - v_3 - v_4 - v_5 - v_6 - v_1$ if $C_A$ is a $\cSix$ or $v_1 - v_2 - v_3 - v_4 - v_5 - v_6 - v_7 - v_1$ if $C_A$ is a $C_7$.
	Since the condition of \cref{lem:cycle-hamiltonian} does not hold and there is the 3-matching $M_3$ between $C_L$ and $C_A$, $C_L$ is adjacent to exactly three mutually non-adjacent vertices of $C_A$ via $M_3$, say w.l.o.g.\ $\{v_1, v_3, v_5\}$ (up to symmetry this is the only possibility in both the $\cSix$ and $\cSeven$ cases).
	We say an edge in $G[V(C_A)]$ is a nice diagonal if both its endpoints are in $\{v_2,v_4,v_6\}$.
	Therefore, the set of nice diagonals are $\{ v_2 v_4, v_4 v_6, v_6 v_1 \}$.
	If there exists a nice diagonal $e$ in $G[V(C_A)]$,
	then observe that there are two distinct edges in $M_3$ that are incident to say $u$ and $v$ in $C_A$, and there is a Hamiltonian Path $P\subseteq \{e\}\cup E(C_A)$ from $u$ to $v$ in $G[V(C_A)]$.
	For example, if the diagonal $v_2 v_4$ exists in $G$ and $C_A$ is $\cSix$, then the desired Hamiltonian path is $v_3 -v_2-v_4 -v_5 -v_6- v_1$ (for $\cSix$, up to symmetry, this is the only case; the $\cSeven$ case is left to the reader to verify, a simple exercise).
	Hence, the condition of \cref{lem:cycle-hamiltonian} holds, a contradiction.
	Otherwise, if none of these nice diagonals are in $G$, we can conclude that $C_A$ is contractible: The three sets of edges of $G$, namely, edges incident to $v_2$, edges incident to $v_4$, and edges incident to $v_6$ are all inside $G[V(C_A)]$ and mutually disjoint.
	As any 2ECSS of $G$ must contain at least $2$ edges incident on each vertex, it has to contain at least $6$ edges from $G[V(C_A)]$, and hence $C_A$ is contractible.

	\textbf{Case 2.2: $C_A$ is a $\cEight$.}
	Let $C_A$ be $v_1 - v_2 - v_3 - v_4 - v_5 - v_6 - v_7 -v_8 - v_1$.
	Since there is a $3$-matching between $C_A$ and $C_L$, there are at least $3$ vertices of $C_A$ adjacent to $C_L$.
	Since the condition of \cref{lem:cycle-hamiltonian} does not hold, up to symmetry,
	$C_L$ is adjacent to exactly $\{v_1,v_3,v_5,v_7\}$, or $\{v_1,v_3,v_5\}$, or $\{v_1, v_3, v_6\}$.

	\textbf{Case 2.2.1: $C_L$ is adjacent to exactly $\{v_1,v_3,v_5,v_7\}$.}
	Similar to Case 2.1, we can show that there is no edge between $v_i$ and $v_j$ for even values of $i$ and $j$, which implies that $C_A$ is contractible.

	\textbf{Case 2.2.2: $C_L$ is adjacent to exactly $\{v_1,v_3,v_5\}$.}
	Similar to Case 2.1, we can show that none of the diagonals $\{v_2v_4,v_2v_6,v_4v_6,v_2v_8,v_4v_8\}$ exists in $G$.
	Hence, the $4$ sets of edges of $G$, namely, edges incident to $v_2$, edges incident to $v_4$, edges incident to $v_6$, and edges incident to $v_8$ are all inside $G[V(C_A)]$ and mutually disjoint except $v_6v_8$.
	As any 2ECSS of $G$ must contain at least $2$ edges incident to each of $v_2,v_4,v_6,v_8$, it has to contain at least $7$ edges (not $8$ due to the possible existence of $v_6v_8$) from $G[V(C_A)]$, and hence $C_A$ is contractible.

	\textbf{Case 2.2.3: $C_L$ is adjacent to exactly $\{v_1,v_3,v_6\}$.}
	Similar to Case 2.1, we can show that none of the diagonals $\{v_2v_4,v_2v_5,v_2v_6,v_2v_8\}$ exists in $G$.
	Hence any 2ECSS of $G$ has to contain contain at least $2$ edges whose other endpoints are not in $\{v_4,v_5,v_7,v_8\}$.
	Further, any 2ECSS of $G$ has to contain contain at least $5$ edges incident to vertices in $\{v_4,v_5,v_7,v_8\}$, which implies at least $7$ edges in $G[V(C_A)]$.
	We conclude that $C_A$ is contractible.
\end{proof}

\subsubsection{Proof of~\cref{lem:gluing:non-trivial-segment}}

\begin{lemma}[Local $3$-Matching Lemma]\label{lem:local3matching}
	Let $S$ be a segment of $\hat{G}_H$ and let $(\hat V_1, \hat V_2)$ be a partition of the nodes of $S$ such that $\hat V_1 \neq \emptyset \neq \hat V_2$. Let $V_i=\cup_{\hat{C}_i\in \hat{V}_i}V(C_i)$ be the set of vertices of $G$ that correspond to $\hat V_i$, for $i\in\{1, 2\}$. Then, there is a matching of size $3$ between $V_1$ and $V_2$ in $G$.
\end{lemma}

\begin{lemma}[Local $4$-Matching Lemma]\label{lem:local4matching}
	Let $S$ be a segment of $\hat{G}_H$ and let $(\hat V_1, \hat V_2)$ be a partition of the nodes of $S$ such that $\hat V_1 \neq \emptyset \neq \hat V_2$. Let $V_i=\cup_{\hat{C}_i\in \hat{V}_i}V(C_i)$ be the set of vertices of $G$ that correspond to $\hat V_i$, for $i\in\{1, 2\}$.
	If $|V_i| \geq 11$ for $i\in\{1, 2\}$,
	then there is a matching of size $4$ between $V_1$ and $V_2$ in $G$.
\end{lemma}

\begin{lemma}[2 edges between $L$ and a $\cFive$]\label{lem:L-C5}
	Let $S$ be a segment of $\hat{G}_S$ with $|K| \geq 3$ and each component in $S$ has at least $5$ vertices.
	Let $L$ be a node in $S$ such that $C_L$ is a huge component and let $A$ be a node in $S$ such that $C_A$ is a $\cFive$.
	If there is a 2-matching $M_2$ between  $C_L$ and $C_A$, then in polynomial time we can compute a canonical bridgeless 2-edge cover $H'$ of $H$ such that $H'$ has fewer components than $H$, $H'$ contains a huge component, and $\cost(H') \leq \cost(H)$.
\end{lemma}
\begin{proof}
	Let the vertices in $C_A$ be labeled as $v_1-v_2-v_3-v_4-v_5-v_1$.
	If the two matching edges are incident to adjacent vertices in $C_A$, say $v_1$ and $v_2$, then we can set $H' \coloneq (H \setminus \{v_1v_2\}) \cup M_2$.
	We have $\cost(H') - \cost(H) = (|H'| - |H|) + (\credit(H') - \credit(H))  = (2 - 1) + (2-2-\credit(C_A))\leq 0$, since we added two edges, effectively removed 1 edge, $\credit(C_L)=2$ and $\credit(C_A)\geq 1$ in $H$, whereas in $H'$ those combine to form a single component with credit 2. Thus, $\cost(H') \leq \cost(H)$, meeting our goal.

	In the following, we assume that the two matching edges are incident to non-adjacent vertices in $C_A$, say $v_1$ and $v_3$ without loss of generality.
	There must be a $3$-matching between $C_A$ and $V(S)\setminus C_A$
	by~\Cref{lem:local3matching}.
	So there exists some edge between $\{v_2,v_4,v_5\}$ and $V(S)\setminus C_A$.
	If there is some edge between $\{v_2,v_4,v_5\}$ and $L$, then it follows as in the previous case where $M_2$ is incident to adjacent vertices in $C_A$.
	Hence there is some edge between $v_i$ for $i \in \{2,4,5\}$ and $C_B$ for some $B \in S\setminus \{L,A\}$.
	Let $v_j \in \{v_1,v_3\}$ be the neighbor of $v_i$ with  $v_jL \in M_2$.
	Since $S$ is 2-node-connected, there must be a path $P$ in $S$ between $L$ and $B$ without passing through $A$.
		Let $F\coloneq P \cup \{v_jL, v_iB\}$.
	Let $H':= H \cup F \setminus \{v_iv_j\}$.
	In $H'$ there is a large component $C'$ spanning components incident to $F$. Every such component yields at least $1$ credit, while $L$ yields $2$ credits.
	Then we have $\credit(H') \leq \credit(H)- \credit(L)-(|F|-1)+\credit(C')=\credit(H) -2-(|F|-1)+2 \leq \credit(H) -(|F|-1)$.
	Hence $\cost(H)-\cost(H')=|H|-|H'|+\credit(H)-\credit(H') \geq -|F|+1 + (|F|-1) = 0$.
\end{proof}

\begin{lemma}[3 edges between $L$ and a $\cSix$ or a $\cSeven$]\label{lem:L-C6}
	Let $S$ be a segment of $\hat{G}_S$ with $|S| \geq 3$.
	Let $L$ be a node in $S$ such that $C_L$ is a huge component and let $A$ be a node in $S$ such that $C_A$ is a $\cSix$ or a $\cSeven$.
	If there is a 3-matching $M_3$ between  $C_L$ and $C_A$, then in polynomial time we can compute a canonical bridgeless 2-edge cover $H'$ of $H$ such that $H'$ has fewer components than $H$, $H'$ contains a huge component, and $\cost(H') \leq \cost(H)$.
\end{lemma}
\begin{proof}
	Let the vertices of $C_A$ be labeled as $v_1-v_2-v_3-v_4-v_5-v_6-v_1$ or  $v_1-v_2-v_3-v_4-v_5-v_6-v_7-v_1$.
	If the edges of $M_3$ are incident to adjacent vertices in $C_A$, say $v_1$ and $v_2$, then we can set $H' \coloneq (H \setminus \{v_1v_2\}) \cup M_2$.

	Without loss of generality, we can assume that the edges of $M_3$ are incident to $v_1,v_3$ and $v_5$.
	Consider $\{v_2,v_4,v_6\}$.
	If $E(G[\{v_2,v_4,v_6\}]) \neq \emptyset$,
	then there is a Hamiltonian path in $G[C_A]$ and we can apply~\cref{lem:cycle-hamiltonian}.
	In the following, we assume that $E(G[\{v_2,v_4,v_6\}]) = \emptyset$.
	Then there must be some edge in $G$ between $\{v_2,v_4,v_6\}$ and $C_B$ for some $B \notin \{L, A\}$.
	Otherwise, $C_A$ is contractible.
	Without loss of generality, we can assume that $v_2$ is adjacent to $C_B$ for some $B \notin \{L, A\}$.

	If $B \in S$, then we can find a cycle $F$ in $S$ through $L$ and $A$, containing the two edges $v_1L$ and $v_2B$ as follows.
	Since $S$ is 2-node-connected, there must be a path $P$ in $S$ between $L$ and $B$ without passing through $A$.
	Then $F\coloneq P \cup \{v_1L, v_2B\}$ is a desired cycle.
	Further, we can apply~\cref{lem:cycle-hamiltonian} to find $H'$ using the $v_1$,$v_2$-Hamiltonian path in $G[V(C_A)]$ from $v_1$ to $v_2$.

	In the remaining case, $B$ is in some other segment $S' \neq S$ ($A$ is a cut node).
	If the total number of vertices in $S'$ is at least $11$, then we can apply~\cref{lem:4-matching} to find a 4-matching between $C_B$ and $S \setminus C_B$.
	Hence we are in the previous case where there is some edge in $G$ between $\{v_2,v_4,v_6\}$ and $C_B$ for some $B \notin \{L, A\}$.
	If the total number of vertices in $S'$ is at most $10$, then either $|S'|=2$ and it contains another node $B$ that $C_B$ is a $\cFour$, or $S'$ contains a rich vertex.

	In the former case, by~\cref{lem:cycle-through-cfour} we can find a cycle $F$ in $S'$ incident to two vertices $a,b$ of $B$ such that there is a Hamiltonian $a$,$b$-path $P$ in $G[V(B)]$.
	Further, $F$ is incident to $v_2$ and some $v_k$ for some $k \neq 2$ in $C_A$.
	Set $H' := (H \setminus (E(C_B)\cup \{v_1v_2\})) \cup \{v_1L,v_3L\}  \cup F \cup P$.
	Notice that $|H'|=|H|+|F|$.
	In $H'$, there is a large component $C'$ spanning the nodes of all components of $H$ incident on $F$ and the component $C_L$. Every such component yields at least $1$ credit, while $L$ yields $2$ credits.
	Then we have $\credit(H') \leq \credit(H)- (|F|-2)-\credit(L)-\credit(A)-\credit(B)+\credit(C')=\credit(H) -(|F|-2)-2-1-1+2 \leq \credit(H) -|F|$.
	Hence $\cost(H)-\cost(H')=|H|-|H'|+\credit(H)-\credit(H') \geq -|F|+|F| \geq 0$.
	In the latter case with some rich vertex in $S'$,
	it follows similarly since a rich vertex contains at least $2$ credits (we do not save one edge inside $B$, but we collect one more credit).
\end{proof}

We can generalize the above lemma by replacing one of the edge in the $3$-matching by a path.
This follows from that adding the path is equivalently to adding a single edge since each internal node on the path has at least $1$ credit and can pay for one bought edge.
\begin{corollary}[2 edges and a path between $L$ and a $\cSix$]\label{corr:Lp-C6}
	Let $S$ be a segment of $\hat{G}_S$ with $|S| \geq 3$.
	Let $L$ be a node in $S$ such that $C_L$ is a huge component and let $A$ be a node in $S$ such that $C_A$ is a $\cSix$, labelled as $v_1-v_2-v_3-v_4-v_5-v_6-v_1$.
	Suppose that $v_1L, v_3L \in E(G)$ and there is a path $P_M$ in $S$ between $L$ and $A$ that is incident to $v_5$.
	Then in polynomial time we can compute a canonical bridgeless 2-edge cover $H'$ of $H$ such that $H'$ has fewer components than $H$, $H'$ contains a huge component, and $\cost(H') \leq \cost(H)$.
\end{corollary}
\begin{proof}
	By simulating the proof of \cref{lem:L-C6}, we can assume that  $v_j \in \{v_2,v_4,v_6\}$ is adjacent to $C_B$ for some $B \notin \{L, A\}$.

	If $B \in S$, then we can find a cycle $F$ in $S$ through $L$ and $A$, containing the two edges $v_jB$ and $v_iL$ where $v_i$ is adjacent to $v_j$ in $C_A$ as follows.
	Since $S$ is 2-node-connected, there must be a path $P$ in $S$ between $L$ and $B$ without passing through $A$.
	By the conditions of the corollary, there must be some $v_i \neq v_5$ such that $v_iL \in E(G)$.
	Hence, we can apply~\cref{lem:cycle-hamiltonian} to find $H'$ using the $v_j$,$v_i$-Hamiltonian path in $G[V(C_A)]$ from $v_j$ to $v_i$.

	If $B \notin S$, the arguments in the proof of~\cref{lem:L-C6} follow similarly as adding the path between $L$ and $v_5$ is equivalently to adding a single edge $v_5L$.
\end{proof}

\begin{proof}[Proof of~\cref{lem:gluing:non-trivial-segment}]
	Let $C_L$ be a huge component such that $L$ is part of a segment $S$ ($|S|\geq 3$) of $\hG_H$.
	If $S$ contains a rich vertex $r$,
	then we find a cycle $F$ in $S$ through $L$ and $r$.
	Let $S':= S \cup F$.
	Notice that $|S'|=|S|+|F|$.
	In $S'$ there is a large component $C'$ spanning the nodes of all components of $S$ incident on $F$ and the component $C_L$. Every such component yields at least $1$ credit, while $L$ and $r$ yield at least $2$ credits.
	Then we have $\credit(S') \leq \credit(S)- (|F|-2)-\credit(L)-\credit(A)+\credit(C')=\credit(S) -(|F|-2)-2-2+2 \leq \credit(S) - |F|$.
	Hence, $\cost(S)-\cost(S')=|S|-|S'|+\credit(S)-\credit(S') \geq -|F|+|F| \geq 0$.

	If $S$ contains another component $A$ that is a $\cFour$, then we apply~\Cref{lem:cycle-through-cfour} to find a cycle $F$ in $S$ incident to distinct nodes $b_1, b_2 \in V (A)$ such that there is a Hamiltonian $b_1, b_2$-path $P$ in $G[V(A)]$.
	Set $H' := (H\setminus E(C_A)) \cup F \cup P$.
	Notice that $|H'|=|H|+|F|-1$.
	In $H'$ there is a large component $C'$ spanning the nodes of all components of $H$ incident on $F$ and the component $C_L$. Every such component yields at least $1$ credit, while $L$ has $2$ credits.
	Then, we have $\credit(H') \leq \credit(H)- (|F|-2)-\credit(L)-\credit(A)+\credit(C')=\credit(H) -(|F|-2)-2-1+2 \leq \credit(S) - |F|+1$.
	Hence $\cost(H)-\cost(H')=|H|-|H'|+\credit(H)-\credit(H') \geq 1-|F|+|F|-1 \geq 0$.

	Further, if there is a cycle $F$ in $\hG_H$ containing $L$ with $|F| \geq 6$, then we can add the cycle, i.e., $H' = H \cup F$.
	In $H'$ there is a large component $C'$ spanning the nodes of all components of $S$ incident on $F$ and the component $C_L$.
	Every such component yields at least $5\cre$ credit, while $L$ has $2$ credits.
	Then we have $\credit(H') \leq \credit(H)- (|F|-1)\cdot 5 \cre-\credit(L)+\credit(C')=\credit(H) -(|F|-1)\cdot 5 \cre-2+2 \leq \credit(H) - (|F|-1)\cdot (\frac{5}{4}-5\delta)$.
	Hence $\cost(H)-\cost(H')=|H|-|H'|+\credit(S)-\credit(H') \geq -|F|+(|F|-1)\cdot (\frac{5}{4}-5\delta)= (|F|-1)\cdot (\frac{1}{4}-5\delta) -1 \geq \frac{5}{4}-25\delta -1  \geq 0$.

	In the following, we assume that
	\begin{itemize}
	    \item $H$ contains only components of size at least $5$ and
		\item any cycle containing $L$ in $S$ has length at most $5$.
    \end{itemize}
	Let $F^*$ be a \emph{longest} cycle in $S$ containing~$L$. For the remainder of this proof, we distinguish different cases based on $|F^*|$.

	\textbf{Case 1: $|F^*|=3$.}
	In this case, $S$ contains only $3$ components, where one is $L$.
	Let $A$ and $B$ be the other two components in $S$. Let $F^*=\{LA,LB,AB\}$.
	If $C_A$ and $C_B$ together have at least $3$ credits, then $S':= S \cup F^*$
	satisfies the lemma since we buy $3$ edges, which can be paid by the credits
	of $A$ and $B$. In the following, we assume that $C_A$ and $C_B$ together have
	less than $3$ credits. So neither of them can be a large component and the
	total number of edges in them is at most $12$. Without loss of generality, we
	assume that $|C_A|\leq |C_B|$.

	\textbf{Case 1.1: $A$ and $B$ are both a $\cFive$.}
	There must be a $3$-matching between $C_L$ and $C_A\cup C_B$
	by~\Cref{lem:local3matching}. Since $A$ and $B$ are both $\cFive$, either $A$
	or $B$ must be incident to at least $2$ matching edges. Thus, we can
	apply~\cref{lem:L-C5} to find $H'$, and we are done.

	\textbf{Case 1.2: $A$ is a $\cFive$ and $B$ is a $\cSix$ or a $\cSeven$.}
	We apply \cref{lem:local4matching} to $L$ and $A\cup B$ so that there is a $4$-matching between $L$ and $A\cup B$, because there are at least $11$ vertices in $A \cup B$.
	If there are at least $2$ matching edges incident to $A$, then we can apply~\cref{lem:L-C5} to find $H'$.
	Hence, we can assume that there is only one matching edge incident to $A$ and $3$ matching edges incident to $B$.
	Then we can apply~\cref{lem:L-C6} to find $H'$.

	\textbf{Case 1.3: $A$ and $B$ are both a $\cSix$.}
	We apply \cref{lem:local4matching} to $L$ and $A\cup B$ so that there is a
	$4$-matching between $L$ and $A\cup B$. If there are at least $3$ matching
	edges incident to one of $A$ or $B$, then we can apply~\cref{lem:L-C6} to
	find $H'$. In the remaining case, the $4$-matching has $2$ edges incident to
	$A$ and $2$ edges incident to $B$. Let the vertices of $C_A$ be labeled as
	$v_1-v_2-v_3-v_4-v_5-v_6-v_1$. Without loss of generality, we can assume that
	the $4$-matching is incident to $v_1$ and to one of $\{v_2, v_3,v_4\}$; the cases with $v_5$ or $v_6$ follow symmetrically. If
	$v_2L \in E(G)$, then there is a Hamiltonian $v_1,v_2$ path in $G[C_A]$, and
	we can apply~\cref{lem:cycle-hamiltonian}. If $v_4L \in E(G)$, we know there
	is a $3$-matching between $C_A$ and $C_L\cup C_B$ by~\cref{lem:local3matching}.
	Hence there exists some edge in $G$ between $\{v_2,v_3,v_5,v_6\}$ and $C_B$.
	Without loss of generality, we can assume that $v_2B \in E(G)$.
	Then we can apply~\cref{lem:cycle-hamiltonian} with $F=\{v_1L,v_2B,LB\}$, because there is a Hamiltonian $v_1,v_2$ path in $G[C_A]$.

	In the last case, the $4$-matching is incident to $v_1$ and $v_3$.
	By~\cref{lem:local3matching}, there must be a 3-matching between $C_A$ and $C_L\cup C_B$.
	Hence for some $i \in \{2,4,5,6\}$, $v_iB \in E(G)$.
	We can again apply~\cref{lem:cycle-hamiltonian} if $i \in \{2,4,6\}$.
	If $i=5$, we can apply~\cref{corr:Lp-C6}.

	\textbf{Case 2: $|F^*|=4$.}
	In this case, $S$ contains at least $4$ components, where one is $L$. Let  $A_1, A_2$, and $A_3$ be some other components. Let $\bar{F}=\{LA_1,A_1A_2,A_2A_3,A_3L\}$.
	If $A_1,A_2$, and $A_3$ together have at least $4$ credits, then $S':= S \cup F^*$ satisfies the lemma since we buy $4$ edges, which can be paid by the $4$ credits.
	In the following, we assume that $A_1,A_2$, and $A_3$ together have less than $4$ credits.
	So neither of them can be a large component and the total number of edges in them is at most $16$.
	We conclude that they can be either three $\cFive$'s or two $\cFive$'s and a $\cSix$.
	Since $S$ is 2VC and the longest cycle in $S$ through $L$ is of length $4$, we can observe that there are only $3$ cases on the structure of $S$.

	\textbf{Case 2.1: $S=\{L,A_1,A_2=D_1,A_3, D_2, \ldots, D_k \}$ for $k \geq 2$; $E[\hat{G}_S] = \{D_iA_1, D_iA_3 \mid i=1,\ldots,k\} \cup \{LA_1, LA_3\}$.}
	By~\cref{lem:local4matching}, there must be a $4$-matching between $C_L$
	and the remaining vertices in $V(S)$. Therefore, there is a 2-matching between $C_L$
	and $C_{A_1}$ or $C_{A_3}$. If both $C_{A_1}$ and $C_{A_3}$ are $\cFive$'s,
	then we can apply~\cref{lem:L-C5} to find $H'$. Therefore, $C_{A_1}$ and
	$C_{A_3}$ must be a $\cSix$ and a $\cFive$, and $C_{A_2}$ must be a $\cFive$.
	Without loss of generality, we can assume that $C_{A_1}$ is a $\cFive$ and
	$C_3$ is a $\cSix$. Further, we can assume that the $4$-matching has one edge
	incident to $C_{A_1}$ and $3$ edges incident to $C_3$.
	Then we can apply~\cref{lem:L-C6} to find $H'$.

	\textbf{Case 2.2: $S=\{L,A_1,A_2,A_3\}$.}
	By~\cref{lem:local4matching}, there must be a $4$-matching between $C_L$ and
	$C_{A_1}\cup C_{A_2}\cup C_{A_3}$. If all of $C_{A_i}$ are $\cFive$, then one
	of them is incident to $2$ edges in the $4$-matching, and we can
	apply~\cref{lem:L-C5} to find $H'$. Hence, one of $C_{A_i}$ must be a $\cSix$
	and the others are $\cFive$. Further, the $4$-matching has $2$ edges incident
	to the $\cSix$ and one edge incident to each of the $\cFive$. Without loss of
	generality, we assume that $C_{A_1}$ is a $\cSix$, and we label the vertices of
	$C_{A_1}$ as $u_1-u_2-u_3-u_4-u_5-u_6-u_1$. Also, we can assume that the
	$4$-matching has $2$ edges incident to non-adjacent vertices, i.e., $u_1$ and
	$u_4$, or to $u_1$ and $u_3$.

	Consider the case the $u_1L,u_4L \in E(G)$.
	By~\cref{lem:local3matching}, there must be some edge $e'$ between $u_i \in \{u_2,u_3,u_5,u_6\}$ and $L \cup C_{A_2} \cup C_{A_3}$.
	Let $u_i$ be the adjacent vertex of $u_j$ in $C_{A_1}$ such that $u_iL \in E(G)$.
	By the previous arguments, there exists a cycle $F$ in $\hG_H$ through $L$ and $A_1$ containing the edges $u_iL$ and $e'$.
	We set $H':= (H\setminus \{u_iu_j\}) \cup F $.
	Then we have $|H'| = |H| + |F| - 1$ and $\credit(H') \leq \credit(H) - (|F|-1)-2+2 = \credit(H) - |F|+1$.
	Hence, $\cost(H') \leq \cost(H)$ and $H'$ satisfies the lemma.

	In the remaining case, $u_1L,u_3L \in E(G)$.
	By~\cref{lem:local3matching}, there must be some edge $e'$ between $u_i \in \{u_2,u_4,u_5,u_6\}$ and $L \cup C_{A_2} \cup C_{A_3}$.
	If $u_i$ is adjacent to $u_1$ or $u_3$, then we can find $H'$ as in the previous case.
	Hence, we can assume that $u_i = u_5$.

	If $e'$ is incident to $L$, we can apply~\cref{lem:L-C6} to find $H'$.
	If $e'$ is incident to $A_j$ for $j \in \{2,3\}$, then
	there is a path $u_5-A_j-L$ from $u_5$ to $L$.
	So we can apply~\cref{corr:Lp-C6}.

	\textbf{Case 2.3:} $S=\{L,A_1=D_1,A_2,A_3=D_2, D_3, \ldots, D_k \}$ for $k \geq 3$; $E[\hat{G}_S] = \{D_iL, D_iA_2 \mid i=1,\ldots,k\}$,
	or $E[\hat{G}_S] = \{D_iL, D_iA_2 \mid i=1,\ldots,k\} \cup \{LA_2\}$.
	If there are $D_i$ and $D_j$ such that $D_i$ and $D_j$ are both $\cSix$'s.
	Then we can obtain $H'$ by adding to $H$ the edges in the cycle
	$L-D_i-A_2-D_j$. The total number of edges in $D_i, A_2, D_j$ is at least $17$, so they have at least $4$ credits, which can pay for the $4$ edges
	that we add. In the following, we assume that at most one of $D_i$ is a $\cSix$ and
	the others are $\cFive$'s. For each $i$ such that $D_i$ is a $\cFive$,
	by~\cref{lem:local3matching} we can find a $3$-matching between $C_{D_i}$ and
	$C_L \cup C_{A_2}$. Further, the $3$-matching has at least $2$ edges incident
	to two adjacent vertices $u_i,v_i$ in $C_{D_i}$. We denote the $2$ edges
	incident to two adjacent vertices $u_i,v_i$ as $u_iU_i,v_iV_i$ for $U_i,V_i \in \{L, A_2\}$.

	If there is a $D_i$ that is a $\cSix$, say, $D_1$, then we set $H:= H \cup \{D_1L, D_1A_2, D_2L, D_2A_2, u_3U_3,v_3V_3\} \setminus \{u_3v_3\}$.
	We have $|H'| = |H| + 6 -1 = |H|+5$ and $\credit(H') \leq \credit(H) - 2 - (6+5+5+5)\cre + 2 \leq   \credit(H) - 5$.
	Thus, $\cost(H') \leq \cost(H)$.

	Hence, assume that there is no $D_i$ that is a $\cSix$; that is, all $D_i$ are $\cFive$.
	If $k \geq 4$, then we set $H:= H \cup \{D_1L, D_1A_2, D_2L, D_2A_2, u_3U_3,v_3V_3,u_4U_4,v_4V_4\} \setminus (\{u_3v_3\}\cup \{u_4v_4\})$.
	We have $|H'| = |H| + 8 -2 = |H|+6$ and $\credit(H') \leq \credit(H) - 2 - (5+5+5+5+5)\cre + 2 \leq   \credit(H) - 6$.
	Thus, $\cost(H') \leq \cost(H)$.

	If $k=3$, we first assume that there is no edge in $G$ between $L$ and $A_2$. Then by~\cref{lem:local4matching} there is a $4$-matching between $C_L$ and $C_{D_1} \cup C_{D_2} \cup C_{D_3}$.
	Hence, there is a $2$-matching between $C_L$ and $C_{D_i}$ for some $i \in \{1,2,3\}$ and we can apply~\cref{lem:L-C5} to find $H'$.
	Now consider the case where $k=3$ and $LA_2 \in E(G)$.
	For each $D_i$, if $U_i = V_i = L$ or $\{U_i,V_i\} = \{L, A_2\}$, then
	we can apply~\cref{lem:cycle-hamiltonian}.
	So we can assume that $U_i = V_i = A_2$, i.e., $u_iA_2, v_iA_2 \in E(G)$.
	If $|E[C_{A_2}]| \geq 6$, then $D_1,D_2,D_3,A_2$ together have at least $21$ edges and $4$ credits.
	In this case, we can set $H:= H \cup \{LD_3,D_3A_2,LA_2, u_1A_2,v_1A_2,u_2A_2,v_2A_2\} \setminus (\{u_1v_1\}\cup \{u_2v_2\})$.
	We have $|H'| = |H| + 7 -2 = |H|+5$ and $\credit(H') \leq \credit(H) - 2 - (5+5+5+6)\cre + 2 \leq   \credit(H) - 5$.
	Thus, $\cost(H') \leq \cost(H)$.
	Hence, in the remaining case, we have that $A_2$ is a $\cFive$, labeled as $q_1-q_2-q_3-q_4-q_5-q_1$.
	Suppose that for some $D_i$, the two edges $u_iA_2, v_iA_2$ are incident to adjacent vertices of $C_{A_2}$, say, $q_1$ and $q_2$.
	Then we can let $H'\coloneq H \setminus \{u_iv_i, q_1q_2\} \cup \{u_iA_2,v_iA_2\}$.
	We have $|H'|=|H|$ and $\credit(H') \leq \credit(H) - (5+5)\cre + 2 \leq \credit(H)$ since we merge two $\cFive$s into a bigger 2EC component.
	Hence we can make progress as $\cost(H') \leq \cost(H)$.
	In the remaining case we know that $u_i A_2, v_i A_2$ are not incident to adjacent vertices of $A_2$. Hence, w.l.o.g.\ we assume that $u_1q_1, v_1q_3 \in E(G)$.
	Suppose that there are two adjacent vertices in $C_{A_2}$, say, $q_1$ and $q_2$, such that $q_1 D_i, q_2 D_j \in E(G)$ for some $i \neq j$.
	Then we can apply~\cref{lem:cycle-hamiltonian} using the cycle $L-D_i-A_2-D_j$ and a $q_1,q_2$-Hamiltonian path in $C_{A_2}$.
	Therefore, in the following we can assume that for each $D_i$, the edges between $D_i$ and $A_2$ are incident to either $q_1$ or $q_3$.
	By~\cref{lem:local3matching} there must be a 3-matching between $C_{A_2}$ and $V(S) \setminus C_{A_2}$.
	Hence the edge $LA_2$ must be incident to $q_2,q_4$, or $q_5$.
	Without loss of generality, we assume the edge is incident to $q_2$.
	Then we can apply~\cref{lem:cycle-hamiltonian} with the cycle $L-D_3-A_2-L$ and a $q_1,q_2$-Hamiltonian path in $C_{A_2}$.

	\textbf{Case 3: $|F^*|=5$.}
	In this case, $S$ contains at least $5$ components, where one is $L$. Let  $A_1, A_2$, $A_3$, and $A_4$ be some other components. Let $F^*= \{LA_1,A_1A_2,A_2A_3,A_3A_4,A_4L\}$. If $A_1,A_2,A_3,A_4$ together have at least $5$
	credits, or equivalently, have at least $21$ edges, then $H':= H \cup F^*$
	satisfies the lemma since we buy $4$ edges, which can be paid by the credits
	of $A_i$.

	Hence, we can assume in the following that all $A_i$ are $\cFive$'s.

	\textbf{Case 3.1: $L$ has a neighbor $B$ in $S$ where $ B \notin \{A_1,\dots,A_4\}$.}
	In this case, there must be a path in $S$ from $B$ to some $A_i$ without passing through $L$ since $S$ is 2-node-connected.
	Hence, there is a cycle $F$ in $S$ containing $L$, $B$, and $k$ ($k \geq 3$) nodes in $\{A_1,A_2,A_3,A_4\}$.
	If $B$ is not a $\cFive$ ($|C_B| \geq 6$), then we can set $H':= H \cup F$.
	We have $\credit(H') = \credit(H) - 2 - (6+5k)\cre +2 \leq \credit(H)- k-2$ and $|H'| = |H|+ k+2$, which satisfies the lemma.
	If $B$ is a $\cFive$, let the vertices of $C_B$ be labeled as $v_1-v_2-v_3-v_4-v_5$ and $v_1L \in E(G)$.
	Since $S$ is 2-node-connected, there must be a path $P$ in $\hat{G}[S]$ from $B$ to some $A_i$ without passing through $L$.
	If $i=1$ or $i=4$, then there is a longer cycle in $\hat{G}[S]$, which contradicts the definition of $F^*$.
	Without loss of generality, we can assume that $i=3$.
	Further, $P$ consists of a single edge as otherwise there is a longer cycle, i.e., there is an edge in $G$ between $A_3$ and $B$.
	By~\cref{lem:local3matching}, there must be a $3$-matching between $C_B$ and the rest of vertices in $V(S)$, which implies that there is an edge $e'$ in $\hG[S]$ incident to some $v_j \in \{v_2,v_4,v_5\}$.
	Let $P'$ be any path in $\hG[S]$ from $v_j$ to $\{L, A_1,A_2,A_3,A_4\}$ and let $v_k$ be the neighbor of $v_j$ in $C_B$.
	We set $H':= H\setminus \{v_jv_k\} \cup F^* \cup P'$.
	We have $|H'| = |H| + |F^*| + |P'|+1 - 1=|H|+|P'|+5$ and $\credit(H') \leq \credit(H) - 25\cre-(|P'|-1)-2+2 = \credit(H) -6-|P'|+1 =\credit(H) -|P'|-5 $.
	Hence $\cost(H') \leq \cost(H)$ and $H'$ satisfies the lemma.

	\textbf{Case 3.2: $L$ has no other neighbors in $S$ than $A_1,\dots,A_4$.}
	By~\cref{lem:local4matching}, there must be a $4$-matching between $C_L$ and $C_{A_1}\cup C_{A_2}\cup C_{A_3}\cup C_{A_4}$.
	If there is a $2$-matching between $C_L$ and $C_{A_i}$ for some $i \in \{1,2,3,4\}$, then we can apply~\cref{lem:L-C5} to find $H'$.
	Hence, we can assume that the $4$-matching has one edge incident to
	each of $C_{A_1}, C_{A_2}$, $C_{A_3}$, and $C_{A_4}$.

	We first show that $S$ only contains $L,A_1,A_2,A_3$, and $A_4$. Suppose $A_1$ has some neighbor $B$ that is not in $\{L, A_2,A_3,A_4\}$. Then there must be a path $P$ in $S$ from $B$ to $A_2$, $A_3$, or $A_4$ without passing through $A_1$ since $S$ is 2VC.
	Further, the path cannot pass through $L$ since $L$ has no other neighbors in $S$ other than $A_1,\dots,A_4$.
	Since there are edges in $G$ between $L$ and each of $A_1,A_2,A_3,A_4$, no matter where $P$ ends at $A_2$, $A_3$, or $A_4$, there must be a cycle in $S$ through $L$ of length at least $6$. This contradicts the definition of $F^*$.
	By symmetry, $A_4$ has no neighbor that is not in $\{L, A_1,A_2,A_3\}$.
	If $A_2$ has some neighbor $B$ that is not in $\{L, A_1,A_3,A_4\}$, then there must be a path in $S$ from $B$ to $A_3$ without passing through $A_2$, which implies a cycle of length at least $6$ in $S$ through $L$.
	We conclude that $S$ only contains $L,A_1,A_2,A_3$, and $A_4$.

	For each $i \in \{1,2,3,4\}$, let the vertices of $C_{A_i}$ be labeled as $v_1^i-v_2^i-v_3^i-v_4^i-v_5^i$.
	By previous arguments, we can assume that $v_1^iL \in E(G)$.
	For each $i$, if $v^i_2$ or $v^i_5$ has an incident edge to other nodes $A_j$ for $j \neq i$, then we can let $H':= H \setminus \{v^i_1v^i_2\} \cup \{LA_i, LA_j, v^i_2A_j\} $.
	We have $|H'| = |H| + 3 -1 = |H|+2$, $\credit(H') \leq \credit(H) - 2 - (5+5)\cre + 2 \leq \credit(H) - 2$, and $\cost(H') \leq \cost(H)$.
	In the following, we assume that $v^i_2$ and $v^i_5$ have no incident edges to other nodes $A_j$ for $j \neq i$.
	By~\cref{lem:local3matching}, there must be a $3$-matching between $C_{A_i}$ and $V(S) \setminus V(C_{A_i})$.
	Thus, there must be a $2$-matching between $\{v^i_3,v^i_4\}$ and $V(S) \setminus \{v^i_3,v^i_4\}$.
	If the $2$-matching is incident to different nodes $A_j, A_k$ for $i \neq j\neq k$, i.e., $v^i_3A_j, v^i_4A_k \in E(G)$, then we can set $H':= H \setminus \{v^i_3v^i_4\} \cup \{v^i_3A_j, v^i_4A_k, LA_j, LA_k\}$.
	We have $|H'| = |H| + 4 -1 = |H|+3$, $\credit(H') \leq \credit(H) - 2 - (5+5+5)\cre + 2 \leq \credit(H) - 3$, and $\cost(H') \leq \cost(H)$.
	In the remaining case, the $2$-matching is incident to the same node $A_j$.
	We call an incident edge of $v^i_3$ or $v^i_4$ an outgoing edge if the other vertex is in some other component $A_j$ for $j \neq i$.
	By previous arguments, we can assume without loss of generality that the outgoing edges of $v_3^1,v_4^1$ (resp. $v_3^2,v_4^2$) can only reach vertices in $\{v_1^2, v_3^2, v_4^2\}$ (resp. $\{v_1^1, v_3^1, v_4^1\}$).
	By~\cref{lem:local3matching}, there must be a $3$-matching between $C_{A_1}\cup C_{A_2}$ and the remaining vertices in $V(S)$.
	Since there are only $2$ edges between $C_L$ and $C_{A_1}\cup C_{A_2}$, incident to $v_1^1$ and $v_1^2$, there must be an edge in $G$ between $(C_{A_1}\cup C_{A_2}) \setminus \{v_1^1, v_1^2\}$ and $C_{A_3}\cup C_{A_4}$.
	This is a contradiction to the previous cases.
\end{proof}

\section{Algorithm for \MANY}\label{sec:manyC4}

In this section, we present an algorithm for the \MANY case. Intuitively, that is the case where
\emph{almost all} components of the initial canonical $2$-edge cover $H_0$ are a $\cFour$.
That is, this section is dedicated to proving \Cref{lem:main:MANY}, which we restate for convenience.

\lemmamainMANY*

We first recall and introduce some definitions.
Recall the definition of the component graph $\hG_H$ for a 2-edge cover $H$.
We say that a component $C$ of $H$ is a \emph{cut component} if $G \setminus V(C)$ is disconnected.
That is, if $c$ is the node in $\hG_H$ corresponding to $C$, then $\hG_H \setminus \{ c \}$ is disconnected. We then also say that $c$ is a \emph{cut node}. If $C$ is a $\cFour$, we call $C$ a $\cFour$-cut component and $c$ a $\cFour$-cut node.

Recall that we say that a component is \emph{huge} if it contains at least $32$ vertices and \emph{large} if it contains at least $9$ vertices and is 2EC (cf.~\Cref{def:large-huge}).
Note that a canonical $2$-edge cover always contains a huge component.
Further, note that a complex component of a canonical $2$-edge cover contains at least $11$ vertices (cf.~\Cref{def:canonicalD2})

In this section, we use the following credit scheme.

\definitionMANYcredit*

Recall that the \emph{cost} of a $2$-edge cover $H$ of $G$ is defined as $\cost(H) = |H| + \credit(H)$ where $\credit(H)$ denotes the sum total of all credits in $H$.
We can bound the cost of a canonical $2$-edge cover in terms of the cost as follows.

\lemmaMANYinitial*

\begin{proof}
	We first assign $\frac14 - \delta$ credit to each edge in a $\cFour$ and $4$ credits to each edge not contained in a $\cFour$.
	Then, we have $\cost(H) = |H| + \credit(H) = m_4 + m_r +  m_4 \cdot \cre + 4 \cdot m_r = (\frac 54 - \delta) \cdot m_4 + 5 \cdot m_r$, which is the bound in the lemma.

	Now we redistribute the credits to obtain a credit assignment as in~\cref{def:initial_credit_2}.
	For each edge in some $\mathcal{C}_i$ for $4 \leq i \leq 8$, we give the credit of the edges to the components.
	For each 2EC component that contains at least $9$ edges, we collect $2$ credits from the edges.
	For each complex component, we keep the credits in the bridge, and each block collects $1$ credit from its edges.
	Since a block only needs $1$ credit and an edge in the block has $4$ credits, we can assign $1$ credit to the complex component.
	In this way, we can redistribute the credits without increasing the cost.
\end{proof}

Our first goal in this section is to turn $H_0$ into a \emph{core-square configuration}, which we define next. See \cref{fig:core-square-configuration} for an example of a core-square configuration.

\definitionMANYcore*

	We have the following observation on the structure of a core-square configuration.

	\begin{lemma}\label{lem:pendent-segment-iscycle}
		Let $H$ be a core-square configuration and $C \neq L$ be a cut component of $H$.
		Let $S'$ be the unique segment of $\hG_H$ such that $C \in S'$ and $L \notin S'$.
		Then $S'$ consists of either $3$ or $4$ $\cFour$s and in $S'$ there is a cycle containing all components of $S'$.
	\end{lemma}

	\begin{proof}
		If $S'$ consists of only $2$ $\cFour$'s, then by~\cref{lem:local3matching}, there must be a $3$-matching between $C$ and the other $\cFour$ in $S'$, and we can always merge the two $\cFour$s into a $\cEight$, which contradicts that $H$ is canonical.
		Hence, $S'$ consists of either $3$ or $4$ $\cFour$'s.
		Since $S'$ is 2-node-connected, there is a cycle in $S'$ containing all the components of $S'$.
	\end{proof}

The main goal of this section is to compute a bridgeless core-square configuration, summarized in the following lemma.

\lemmaMANYcore*

Given a bridgeless core-square configuration, we then show how to prove \Cref{lem:main:MANY}.
In the next subsection we first prove \Cref{lem:main:MANY} assuming we are given a bridgeless core-square configuration.
In \Cref{sec:MANY:core}, we prove \Cref{lem:manyC4_core-configuration_main}, which is the main contribution of this section.

\subsection{Proof of \Cref{lem:main:MANY}}

We next prove \Cref{lem:main:MANY} assuming \Cref{lem:manyC4_core-configuration_main}. That is, we assume we are given a bridgeless core-square configuration $H$ and we want to turn it into a $2$-edge connected component that has at most $\left( \frac{5}{4} - \frac{1}{28} \right) \cdot \beta |H_0| + 5 (1 - \beta)|H_0| + \frac{1}{35} \opt$ many edges, where $\beta \in [0,1]$ such that the number of edges of $\cFour$'s in $H_0$ is $\beta |H_0|$.
For convenience, we let $m_4$ and $m_r$ be the number of edges of $H_0$ contained in a $\cFour$ or not contained in a $\cFour$, respectively. That is, $m_4 = \beta |H_0|$ and $m_r = (1- \beta) |H_0|$. We then show that we can compute a 2ECSS $H$ of $G$ using at most $(\frac 54 - \frac{1}{28}) \cdot m_4 + 5 \cdot m_r  +\frac{1}{35} \opt$ many edges.

\begin{proof}[Proof of \Cref{lem:main:MANY}]

	If $H$ consists of only a huge component $L$, we are done since there is only one component $L$, which does not contain a bridge. This component has a credit of $2$, and hence $|H| \leq \cost(H)-2 = \cost(H_0) + 2 - 2  \leq (\frac 54 - \frac{1}{28}) \cdot m_4 + 5 \cdot m_r $ by \Cref{lem:manyC4_core-configuration_main} and \Cref{lem:manyC4_starting-solution-cost}.

	Otherwise, let $C_1, ..., C_k$ be the $\cFour$ (which are also segments, by the definition of a core-square configuration) that are adjacent to $L$ in $\hG_H$.
	Furthermore, let	 $S_i$ be the inclusion-wise maximal segment of $\hG_H \setminus L$ which contains $C_i$. Note that these segments $S_1, ..., S_k$ are all pairwise disjoint by the definition of a core-square configuration.
	Furthermore, recall that each $S_i$ contains either 1, 3, or 4 components that are all $\cFour$ and that each $S_i$ is connected to $L$ in $\hG_H$.
	Consider some segment $S_i$ for $1 \leq i \leq k$.
	We glue $V(S_i)$ optimally to $L$, i.e., we contract $V(G) \setminus V(S_i)$ to a single vertex and compute an optimum solution $\OPT_i$.
	That is $\OPT_i$ is an optimum solution to $G|(V(G) \setminus V(S_i))$.
	Note that $\OPT_i$ can be computed in polynomial time since $|V(S_i)| \leq 16$.
	Let $|S_i|$ be the number of components of $H$ in $S_i$ and let $F_i$ be the set of edges of $H$ in $S_i$.
	If $|\OPT_i| \leq 5 \cdot |S_i| -1$, then set $H' = (H \setminus F_i) \cup \OPT_i$ and observe that $H'$ is still a bridgeless core-square configuration.
	Furthermore, $|H'| = |H| + |\OPT_i| - 4|S_i|$ and $\credit(H') \leq \credit(H) - |S_i| \cdot (1-4\delta)$.
	Since $\OPT_i \leq 5 \cdot |S_i| -1$, we have
	\begin{align*}
	\cost(H') & = |H'| + \credit(H') \\
	& \leq  |H| + |\OPT_i| - 4|S_i| + \credit(H) - |S_i| \cdot (1-4\delta)\\
	& \leq |H| + \credit(H) + 5 \cdot |S_i| -1 - 4|S_i| - |S_i| \cdot (1-4\delta)\\
	& \leq |H| + \credit(H) +|S_i| \cdot 4 \delta -1 \\
	& \leq |H| + \credit(H) +16 \cdot \delta -1 \\
	& \leq |H| + \credit(H) \ .
	\end{align*}

	Hence, assume that we are now given a bridgeless core-square configuration $H'$ and let again $S_1, S_2, ..., S_p$ be the remaining segments of $\hG_{H'} \setminus L$, such that $\OPT_i \geq 5 \cdot |S_i|$.
	For each $i \in \{1, \dots, p\}$, let $\ell_i$ be the number of components of $H'$ in $S_i$.

	We first show a different bound on $\OPT$.
	Let $G_i$ be the graph of $G$ induced by $V(S_i)$. Let $E(G_i)$ be the edges of this graph and let $E(L,G_i)$ be the edges in $G$ between the vertices of node $L$ and the vertices of $G_i$. Note that $(E(G_i) \cup E(L,G_i)) \cap (E(G_j) \cup E(L,G_j)) = \emptyset$ for $i \neq j$, by the definition of a core-square configuration. We have
	\[
		\opt \geq \sum_{i=1}^p |\OPT(G) \cap (E(G_i) \cup E(L,G_i))| \geq  \sum_{i=1}^k 5 \ell_i \ .
	\]

	We now show that we can turn $H'$ into a single 2EC component $H''$ such that $\cost(H'') \leq \cost(H) + 4 \delta \sum_{i=1}^p \ell_i$.
	For each $i \in \{1, \dots, p\}$, let $C_i$ be the cut component separating $L$ from the remaining components of $S_i$.
	We now find edge sets $X_i \subseteq E \setminus E(H')$ and $Y_i \subseteq E(H')$ such that $H'' = \Big( H' \cup \bigcup_{1 \leq i \leq p} X_i  \Big) \setminus \Big( \bigcup_{1 \leq i \leq p} \{ Y_i \} \Big)$ is 2EC and $|X_i| - |Y_i| \leq |S_i|$, which implies that $\cost(H'') \leq \cost(H') + 4 \delta \sum_{i=1}^p \ell_i \leq \cost(H) + 4 \delta \sum_{i=1}^p \ell_i$.
	To find these edge sets, we iteratively appply \Cref{lem:cycle-through-cfour}.
	First, using \Cref{lem:cycle-through-cfour} we find 2 edges incident to $L$ and $C_i$ such that those edges are incident to distinct vertices $u_i, v_i \in V (C_i)$ such that there is a Hamiltonian $u_i, v_i$-path in $G[V(C_i)]$. The two edges between $L$ and $C_i$ are added to $X_i$, while the edge $u_i v_i$ is added to $Y_i$.
	Then, for some other $C_i'$ in $S_i$ distinct from $C_i$ we apply \Cref{lem:cycle-through-cfour} to find a cycle $K$ through $C_i'$ and $C_i$ to distinct vertices $u_i', v_i' \in V (C_i')$ such that there is a Hamiltonian $u_i', v_i'$-path in $G[V(C_i')]$. The edges of $K$ are added to $X_i$, while the edge $u_i' v_i'$ is added to $Y_i$.
	If $K$ involves all nodes of $S_i$, we are done.
	Otherwise, we repeat this process where we temporarily shrink all nodes involved with $K$ to a single node and find a cycle $K'$ through this shrunk node and a not yet visited component $C_i''$ from $S_i$, where $K'$ is incident to distinct vertices $u_i'', v_i'' \in V (C_i'')$ such that there is a Hamiltonian $u_i'', v_i''$-path in $G[V(C_i'')]$, by \Cref{lem:cycle-through-cfour}. Again, the edges of $K''$ are added to $X_i$ and the edge $u_i'' v_i''$ is added to $Y_i$.
	This process stops after at most $|S_i|$ many steps and we have that $|X_i| - |Y_i| \leq |S_i|$.
	Furthermore, observe that $H'' = \Big( H' \cup \bigcup_{1 \leq i \leq p} X_i  \Big) \setminus \Big( \bigcup_{1 \leq i \leq p} \{ Y_i \} \Big)$ is 2EC. So $\cost(H'')-\cost(H') = |H''|-|H'| + \credit(H'')-\credit(H') \leq \sum_{i=1}^p \ell_i - \sum_{i=1}^p (1-4\delta)\cdot \ell_i  \leq 4 \delta \sum_{i=1}^p \ell_i $.

	Let $\sum_{i=1}^p \ell_i = \ell^*$.
	Combining \Cref{lem:manyC4_core-configuration_main}, we have
		$\cost(H'') \leq \cost(H_0) +2 + 4 \delta \cdot \ell^*$, where $H_0$ is the initial canonical $2$-edge cover.
		Since  $H''$ consists of a single 2EC component, which has a credit of 2, we have
		$|H''| \leq \cost(H_0) + 4 \delta \cdot \ell^* \leq (\frac 54 - \delta) \cdot m_4 + 5 \cdot m_r + 4 \delta \cdot \ell^*$ by \Cref{lem:manyC4_starting-solution-cost}.

		Using $\opt \geq \sum_{i=1}^k 5 \ell_i = 5 \ell^*$ and by setting $\delta = \frac{1}{28}$, we can bound
	\begin{align*}
	|H''| & \leq \Big( \frac 54 - \delta \Big) \cdot m_4 + 5 \cdot m_r + 4 \delta \cdot \ell^* \\
	& \leq \Big( \frac 54 - \delta \Big) \cdot m_4 + 5 \cdot m_r + \frac{4}{5} \delta \opt \\
	& = \Big( \frac 54 - \frac{1}{28} \Big) \cdot m_4 + 5 \cdot m_r + \frac{4}{5} \cdot \frac{1}{28} \opt \\
	& = \Big( \frac 54 - \frac{1}{28} \Big) \cdot m_4 + 5 \cdot m_r  +\frac{1}{35} \opt  \ ,
	\end{align*}
	which proves the statement.
\end{proof}

\subsection{Computing a Bridgeless Core-Square Configuration (Proof of \Cref{lem:manyC4_core-configuration_main})}
\label{sec:MANY:core}

To prove \Cref{lem:manyC4_core-configuration_main}, we start with a canonical $2$-edge cover $H$ and step-by-step decrease the number of components while always maintaining that the new canonical $2$-edge cover $H'$ satisfies $\cost(H') \leq \cost(H)$, except for at most one step in which we might increase the cost by at most $2$.
From now on we always assume that we are given some canonical $2$-edge cover $H$.

Throughout this section we focus on components of $H$ belonging to the same segment $S$ of $\hG_H$, i.e., the same $2$-node connected part $S$ of the component graph $\hG_H$.
For such components $C$ of $S$, we will usually apply the $3$-Matching \Cref{lem:3-matching}.
If $C$ is not a cut node, then it is clear that all edges implied by the $3$-Matching \Cref{lem:3-matching} are going to $S$.
If $C$ is a cut note, this is not immediate. However, we can divide $V$ into two sets $V_1$ and $V_2$ such that $V(S) \setminus V(C) \subseteq V_1$, $V(C) \subseteq V_2$ and the only edges between $V_1$ and $V_2$ are incident to $C$. Then, applying the $3$-Matching \Cref{lem:3-matching} to $(V_1, V_2)$ implies that these edges are incident to $V(C)$ and $S$.
We will always make use of this fact implicitly without explicitly giving this partition.

\paragraph*{Useful definitions and lemmas.}

Before we proceed, we first define certain situations in which we can turn $H$ into a canonical $2$-edge cover $H'$ such that $H'$ contains fewer components than $H$ and $\cost(H') \leq \cost(H)$.

\begin{definition}[covering a bridge]
	\label{def:coering-bridge}
	Let $H$ be a canonical $2$-edge cover of a structured graph $G$ that contains a complex
	component $C$. Let $X \subseteq E(G) \setminus E(H)$ be such that $X$ is a cycle in the
	component graph $\hG_H$ incident to $C$ at vertices $u_1, u_2 \in V(C)$ such that the
	path from $u_1$ to $u_2$ in $C$ contains at least one bridge $e \in E(C)$.
	We say that $X$ \emph{covers} the bridge $e$.

	Note that $X$ can also be a cycle containing only one edge if $e$ connects different vertices of $V(C)$.
\end{definition}

\begin{definition}[shortcut]
	\label{def:shortcut}
	Let $X \subseteq E(G) \setminus E(H)$ such that $X= X_1 \cup X_2$ with $X_1 \cap X_2 = \emptyset$,
	where $X_1$ is a simple cycle in the component graph $\hG_H$. $X_2$ is either a possibly empty simple path  in the
	component graph $\hG_H$ such that the endpoints of $X_2$ intersect with $X_1$ and all interior
	nodes of $X_2$ are distinct from $X_1$, or $X_2$ is a simple cycle that intersects $X_1$ at precisely one node.
	Let $C$ be a component incident to $X$ that is a $\cFour$.
	We say that $C$ is \emph{shortcut} w.r.t.\ $X$ if there are two edges $e_1, e_2 \in X$ with $e_1 = u_1 v_1$ and $e_2 = u_2 v_2$, $u_1, u_2 \in V(C)$ and $v_1, v_2 \in V(G) \setminus V(C)$ such that $u_1 u_2 \in E(C) \cap H$, i.e., $u_1$ and $u_2$ are adjacent in $H$. We also say that $e = u_1 u_2 \in E(H)$ is \emph{shortcut}.

	For each $\cFour$ component $C$ incident to $X$ that is shortcut, we let $e(X, C)$ be the edge that is shortcut.
	There might be more than one edge of $C$ that is shortcut w.r.t.\ $X$. In this case we pick an arbitrary such edge as $e(X, C)$.
	We define $F(X) \coloneqq \{ e(X, C) \mid C \text{ is shortcut w.r.t.\ $X$} \}$.
\end{definition}

\begin{figure}
	\centering
    \includegraphics[width=0.4\textwidth]{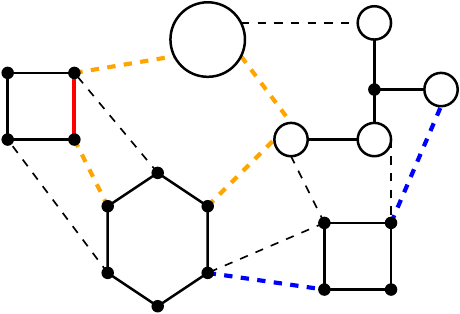}
    \caption{An example of a double $(5,1)$-merge. Edges in $X_1$ are orange and edges in $X_2$ are blue. The red edge is shortcut.}
\label{fig:double-merge}
\end{figure}

Note that in the above definition $X_2$ can also be empty, in which case $X$ is just a simple cycle in the component graph.

Using this definition of shortcut, we now define a single $(a, b)$-merge and a double $(a, b)$-merge.

\begin{definition}[single $(a, b)$-merge]
	\label{def:single-a-b-merge}
	Let $H$ be a canonical $2$-edge cover of a structured graph $G$.
	For $a, b \geq 0$, a single $(a, b)$-merge w.r.t.\ $H$ is a set of edges $X \subseteq E(G) \setminus E(H)$ such that
	\begin{itemize}[nosep]
		\item $X$ is a cycle in the component graph $\hG_H$,
		\item if $X$ is incident to a complex component $C'$, then $X$ is either incident to a vertex contained in a block of $C'$ or incident to at least 2 distinct vertices of $C'$,
		\item $a$ is the number of distinct components $C$ in $X$, and
		\item $b$ is the number of distinct $\cFour$ components $C$ in $X$ such that $X$ does not shortcut $C$.
	\end{itemize}
\end{definition}

\begin{definition}[double $(a, b)$-merge]
	\label{def:double-a-b-merge}
	Let $H$ be a canonical $2$-edge cover of a structured graph $G$.
	For $a, b \geq 0$, a double $(a, b)$-merge w.r.t.\ $H$ is a set of edges $X \subseteq E(G) \setminus E(H)$ such that
	\begin{itemize}[nosep]
		\item  $X \subseteq E(G) \setminus E(H)$ such that $X= X_1 \cup X_2$ with $X_1 \cap X_2 = \emptyset$, where $X_1$ is a simple cycle in the component graph $\hG_H$ and $X_2$ is a simple path in the component graph $\hG_H$ such that the endpoints of $X_2$ intersect with $X_1$ and all interior nodes of $X_2$ are distinct from $X_1$,
		\item if $X$ is incident to a complex component $C'$, then $X$ is either incident to a vertex contained in a block of $C'$ or incident to at least 2 distinct vertices of $C'$,
		\item $a$ is the number of distinct components $C$ in $X$, and
		\item $b$ is the number of distinct $\cFour$ components $C$ in $X$ such that $X$ does not shortcut $C$.
	\end{itemize}
\end{definition}

An example of a double $(5, 1)$-merge is given in \Cref{fig:double-merge}.
We further define \emph{short heavy cycles}.

\begin{definition}[short heavy cycle]
	\label{def:short-heavy-cycle}
	Let $H$ be a canonical $2$-edge cover of a structured graph $G$.
	A set of edges $X \subseteq E(G) \setminus E(H)$ is a short heavy cycle if
	\begin{itemize}[nosep]
		\item $X$ is a simple cycle in the component graph of length at most $6$ in $\hG_H$, and
		\item $X$ is incident to at least one component $C$ that is a $\mathcal{C}_i$, $5 \leq i \leq 8$ or complex. If $C$ is complex then $X$ covers at least one bridge of $C$.
	\end{itemize}
\end{definition}

The following lemma will be useful throughout this section.

\begin{lemma}
	\label{lem:abc-merge}
	Let $H$ be a canonical $2$-edge cover of a structured graph $G$ and let $X \subseteq E(G) \setminus E(H)$ such that one of the following conditions is satisfied:
	\begin{itemize}[nosep]
		\item $X$ is a single $(k, j)$-merge for some $k \geq 3$ and $j \leq 2$ such that $k-j \geq 3$,
		\item $X$ is a double $(k, j)$-merge for some $k \geq 5$ and $j \leq 3$ such that $k-j \geq 4$, or
		\item $X$ is a short heavy cycle.
	\end{itemize}
	Then $(H \cup X) \setminus F(X)$ is a canonical $2$-edge cover with fewer components than $H$ and $\cost((H \cup X) \setminus F(X)) \leq \cost(H)$.
\end{lemma}

\begin{proof}
	Let $H' = (H \cup X) \setminus F(X)$.
	We first prove that $H'$ is a canonical $2$-edge cover.
	This is clearly true for short heavy cycles, since we only add a cycle in the component graph and do not remove any edges, i.e., $F(X) = \emptyset$ in this case.
	If $X$ is a single or double merge, we have the same if we do not remove an edge, i.e., is $F(X) = \emptyset$.
	If we remove an edge, we only remove it from a $\cFour$ and we remove at most one edge from each component.
	Let $C$ be a component where we remove the edge $e = u v$. Note that there is still a Hamiltonian Path from $u$ to $v$ within $C$ in $H'$. Further, there is a connection from $u$ to $v$ in $H'$ using no edges within $C$.
	Since we have this for every component for which we removed an edge, this shows that the newly created block or 2EC component in $H'$ is indeed 2EC.
	Further, observe that we have not created any new bridges.

	This shows that $H'$ is a canonical $2$-edge cover.
	It remains to prove that $\cost(H') \leq \cost(H)$.
	Note that $|H'| \leq |H| + |X| - |F(X)|$.
	Let $C_1, ..., C_{p_c}$ be the set of complex components incident to $X$ and assume that $X$ contains $p_\ell$ 2EC components that are large or a $\mathcal{C}_i$ for some $5 \leq i \leq 8$ and $p_4$ components that are a $\cFour$. Further, for $1 \leq i \leq p_c$ let $bl_i$ be the number of blocks incident to $X$ and $br_i$ be the number of bridges in $C_i$ covered by $X$.
	We have that
	$$\credit(H') \leq \credit(H) + 2 - \Big( \sum_{i = 1}^{p_c} (bl_i + 4 \cdot br_i +1) + 2 p_\ell + (1 - 4 \delta) \cdot p_4   \Big) \ .$$
	The $+2$ comes from the fact that in $H'$, we either create a new large component (which receives a credit of $2$) or a new block within a new complex component (where the block and the component each receive a credit of $1$).
	The number of credits that are removed in $H'$ compared to $H$ are all credits inside components, blocks and bridges in $H$ that are now in the newly formed block or 2EC component in $H'$. These include all blocks incident to $X$ or bridges covered by $X$, and all 2EC components incident to $X$. Using the credit assignment scheme, we observe that the above inequality on $\credit(H')$ holds.

	If $X$ is a short heavy cycle, then $|X| \leq 6$, and either $X$ is incident to a component that is a $\mathcal{C}_i$ for some $5 \leq i \leq 8$ which has at least $20$ credits, or there is some complex component incident to $X$, say $C_1$, and $br_1 \geq 1$ and $bl_1 + 4 \cdot br_1 +1 \geq 5$.
	In either case, we have
	\begin{align*}
	\credit(H')
	& \leq \credit(H) + 2 - \Big( 5 + (|X|-1) \cdot (1-4\delta) \Big)\\
	& \leq \credit(H) - |X| \ ,
	\end{align*}
	since $|X| \leq 6$ and the worst-case is if there is exactly one complex component and $5$ components that are $\cFour$ and $X$ covers only one bridge of the complex component.
	Hence, we have $\cost(H') \leq \cost(H)$.

	Next, assume $X$ is a single $(k, j)$-merge for some $k \geq 3$ and $j \leq 2$ such that $k-j \geq 3$ or a double $(k, j)$-merge for some $k \geq 5$ and $j \leq 2$ such that $k-j \geq 4$.
	Then $|F(X)| \geq p_4 - j \geq 0$.
	Note that for each $i \in \{1, ..., p_c \}$ we have $bl_i + 4 \cdot br_i +1 \geq 2$ in this case since by the definition of a single or double merge, $X$ must be incident to either a block or it covers at least one bridge of each complex component.
	Further recall that $k = p_c + p_\ell + p_4$ and that $|X| \in \{k, k+1\}$.
	We have
	\begin{align*}
	\cost(H') & = |H'| + \credit(H') \\
	& \leq |H| + |X| - |F(X)| +\credit(H) + 2 - \Big( \sum_{i = 1}^{p_c} (bl_i + 4 \cdot br_i +1) + 2 p_\ell + (1 - 4 \delta) \cdot p_4   \Big)\\
	& \leq |H| + \credit(H) + |X|  + 2 - \Big( 2 \cdot (p_c + p_\ell) + (1 - 4 \delta) \cdot p_4 + |F(X)|   \Big) \\
	& \leq |H| + \credit(H) + |X|  + 2 - \Big( 2 \cdot (p_c + p_\ell) + (1 - 4 \delta) \cdot p_4 + p_4-j \Big) \\
	& \leq |H| + \credit(H) + |X|  + 2 - \Big( 2 \cdot (p_c + p_\ell) + (2 - 4 \delta) \cdot p_4 - j \Big) \\
	& \leq |H| + \credit(H) + |X|  + 2 - \Big( (2- 4\delta) \cdot (p_c + p_\ell + p_4) - j \Big) \\
	& \leq |H| + \credit(H) + |X|  + 2 + j - (2- 4\delta) \cdot k  \ ,
	\end{align*}
	where we used all inequalities from above.

	Consider first the case that $X$ is a single merge.
	Then we have $|X| = k \geq 3$, $j \leq 2$ and $k-j \geq 3$.
	Therefore, we obtain
	\begin{align*}
	\cost(H') & \leq |H| + \credit(H) + |X|  + 2 + j - (2- 4\delta) \cdot k  \\
	& = |H| + \credit(H) + k  + 2 + j - (2- 4\delta) \cdot k  \\
	& = |H| + \credit(H) + 2 + j - (1- 4\delta) \cdot k \\
	& \leq  |H| + \credit(H) + 2 + 2 - (1- 4\delta) \cdot 5 \\
	& \leq |H| + \credit(H) \\
	& = \cost(H) \ .
	\end{align*}

	If $X$ is a double merge, then we have $|X| = k + 1 \geq 6$, $j \leq 2$ and $k-j \geq 4$.
	Therefore, we obtain
	\begin{align*}
	\cost(H') & \leq |H| + \credit(H) + |X|  + 2 + j - (2- 4\delta) \cdot k  \\
	& = |H| + \credit(H) + k +1 + 2 + j - (2- 4\delta) \cdot k  \\
	& = |H| + \credit(H) + 3 + j - (1- 4\delta) \cdot k \\
	& \leq  |H| + \credit(H) + 3 + 3 - (1- 4\delta) \cdot 7 \\
	& \leq |H| + \credit(H)\\
	&  = \cost(H) \ .
	\end{align*}

	Note that all inequalities hold for $\delta \leq \frac{1}{28}$.
	This finishes the proof of the lemma.
\end{proof}

Furthermore, we will sometimes use the following two lemmas.

\begin{lemma}
	\label{lem:31-merge}
	Let $H$ be a canonical $2$-edge cover of a structured graph $G$ and let $C$ be a component of $H$ contained in a segment $S$ that consists of exactly three or four nodes that are all large or a $\cFour$.
	Then there is a single $(k, j)$-merge containing $C$ in $S$ not shortcutting $C$ such that $k-j \geq 2$ and $k \in \{3, 4 \}$.
\end{lemma}

\begin{proof}
We first consider the case that $S$ contains exactly $3$ nodes, i.e., $S$ contains the three components $C$, $C_1$, and $C_2$.
We assume that $C_1$ and $C_2$ are $\cFour$.
It is easy to see how to adapt to the case that one (or both) of $C_1$ or $C_2$ are large.
Let $v_1 - v_2 - v_3 - v_4 - v_1 $ be the cycle of $C_1$ and $u_1 - u_2 - u_3 - u_4 - u_1 $ be the cycle of $C_2$.
We show that there must exist a single $(3, 1)$-merge in $S$ that does not shortcut $C$.

By the $3$-Matching \Cref{lem:3-matching} and since $S$ is $2$-node connected, there must be a pair of edges $e_1$ and $e_2$ incident to adjacent vertices of $C_1$. W.l.o.g.\ assume that $e_1$ is incident to $v_1$ and $C$ and $e_2$ is incident to $v_2$ and $C_2$. Further, assume that $e_2$ is incident to $u_1$ at $C_2$.
$C_2$ must also be incident to an edge $f$ incident to $C$, as otherwise $S$ is not $2$-node connected.
If $f$ is incident to either $u_2$ or $u_4$, then $e_1, e_2$ and $f$ together form the desired single $(3, 1)$-merge that does not shortcut $C$.
Hence, assume that w.l.o.g.\ $f$ is incident to $u_3$.

Again, by the $3$-Matching \Cref{lem:3-matching}, either $u_2$ or $u_4$ is incident to an edge $g$ to either $C_1$ or $C$.
If $g$ is incident to $C$, we are clearly done.
Hence, assume $g$ is incident to $C_1$.
If $g$ is incident to either $v_1$ or $v_3$, then we can merge $C_1$ and $C_2$ into a $\cEight$, a contradiction.
If $g$ is incident to either $v_2$ or $v_4$, then $e_1, g$, and $f$ together form the desired $(3, 1)$-merge that does not shortcut $C$. This finishes the case that $S$ contains exactly $3$ nodes.

Let us now assume that $S$ contains exactly $4$ nodes, i.e., the components $C_1, C_2, C_3$, and $C$.
We show that there must exist a single $(k, 1)$-merge or a single $(4, 2)$ merge in $S$ that does not shortcut $C$, for some $k \in \{3, 4\}$.
This then implies the lemma statement.

We assume that $C_1, C_2$, and $C_3$ are $\cFour$.
It is easy to see how to adapt to the case that a subset of them is large.
We take the same labeling for $C_1$ and $C_2$ and let $C_3$ be the cycle $w_1 - w_2 - w_3 - w_4 - w_1$.
W.l.o.g.\ $C_1$ is adjacent to $C$.
By the $3$-Matching \Cref{lem:3-matching} and since $S$ is $2$-node connected, there must be a pair of edges $e_1$ and $e_2$ incident to adjacent vertices of $C_1$.
W.l.o.g.\ assume that $e_1$ is incident to $v_1$ and $C$ and $e_2$ is incident to $v_2$ and $C_2$. Further, assume that $e_2$ is incident to $u_1$ at $C_2$.

Assume first that $C_2$ is adjacent to $C$ by an edge $f$.
If $f$ is incident to either $u_2$ or $u_4$, we are done.
Hence, assume w.l.o.g.\ that $f$ is incident to $u_3$.
By the $3$-Matching \Cref{lem:3-matching}, there must be another edge $g$ incident to either $u_2$ or $u_4$.
If $g$ is incident to $C_1$, then either we can merge $C_1$ and $C_2$ into a $\cEight$ (if $g$ is incident to either $v_1$ or $v_3$) or we have a single $(3, 1)$-merge using $f, g$ and $e_1$ (if $g$ is incident to either $v_2$ or $v_4$).
If $g$ is incident to $C$, then $e_1, e_2$ and $g$ also form a single $(3, 1)$-merge.
Hence, $g$ is incident to $C_3$ at vertex $w_1$.
By the $3$-Matching \Cref{lem:3-matching}, there must be an edge $h$ incident to either $w_2$ or $w_4$.
If $h$ is incident to $C$, then $h, g$ and $f$ form a single $(3, 1)$-merge. If $h$ is incident to $C_1$, then $f, g, h$ and $e_1$ form a single $(4, 2)$-merge.
Hence, $h$ is incident to $C_2$.
If $h$ is incident to $u_1$ or $u_3$, then we can merge $C_2$ and $C_3$ to a single $\cEight$, a contradiction.
Hence, $h$ is incident to $u_2$ or $u_4$ (w.l.o.g. $u_2$).
Now there has to be an edge $i$ incident to $C_3$ and incident to some other node than $C_2$.
If $i$ is incident to $C$, then $f, i$ and either $g$ or $h$ is a single $(3, 1)$-merge.
If $i$ is incident to $C_1$, then $f, i, e_1$ and either $g$ or $h$ is a single $(4, 2)$-merge.

Next, assume that $C_2$ is not adjacent to $C$.
$C_3$ then must be adjacent to $C$ by an edge $f$, otherwise $S$ is not $2$-node connected.
Furthermore, $C_1$ must be adjacent to $C_3$ by an edge $g$.
If $g$ is incident to either $u_2$ or $u_3$, then $e_1, e_2, f$ and $g$ form a single $(4, 2)$-merge.
Hence, assume that $g$ is incident to, w.l.o.g.\ $u_1$. We assume that neither $u_2$ nor $u_4$ is adjacent to $C_3$.
However, by the $3$-Matching \Cref{lem:3-matching}, either $u_2$ or $u_4$ must have an edge, say $h$, to another component, which must be $C_1$.
If $h$ is incident to $v_1$ or $v_3$, then we can merge $C_1$ and $C_2$ into one $\cEight$, a contradiction.
Hence, $h$ is incident to, say $v_2$.
Now there is a single $(4, 2)$-merge using $e_1, h, g$ and $f$.
This finishes the case and the proof.
\end{proof}

\begin{lemma}
	\label{lem:cycle-useful-for-double-merge}
	Let $H$ be a canonical $2$-edge cover of a structured graph $G$ and let $C$ be a component of $H$ contained in two distinct segments $S, S'$ such that $S$ consists of exactly three or four nodes that are all large or a $\cFour$.
	Furthermore, there is a single $(k, j)$-merge $X'$ containing $C$ in $S'$ such that $k-j \geq 2$ and $k \in \{3, 4 \}$.
	Then there is also a double $(k, j)$-merge containing $C$ for some $k \geq 5$ and $k-j \geq 4$.
\end{lemma}

\begin{proof}
By \Cref{lem:31-merge}, there is a single $(k, j)$-merge $X$ containing $C$ in $S$ not shortcutting $C$ such that $k-j \geq 2$ and $k \in \{3, 4 \}$.
If the single merge $X'$ in $S'$ does not shortcut $C$ (or if $C$ is large), then $C$ is not shortcut in both single merges $X$ and $X'$ and hence $X \cup X'$ is a double $(k, j)$-merge containing $C$ for some $k \geq 5$ and $k-j \geq 4$, as desired.
If $X'$ does shortcut $C$, then in $X \cup X'$ there are also at least $k-j$ components that are shortcut or large.
Hence $X \cup X'$ is a double $(k, j)$-merge containing $C$ for some $k \geq 5$ and $k-j \geq 4$, as desired.
\end{proof}

\paragraph*{Main Lemmas.}

The main technical contribution is the following lemma. We postpone its proof, which is given in \Cref{sec:nice-path}.

\lemmaMANYgluingpath*

For the remainder we assume that we have applied the above lemma exhaustively.
Hence, from now on we assume that every segment $S$ with at least $5$ nodes does not contain a huge component.

Next, we show that if there is a segment that contains a $\mathcal{C}_i$ for $5 \leq i \leq 8$ and at most $i$ components, then we can turn $H$ into a canonical $2$-edge cover $H'$ with fewer components than $H$ and $\cost(H') \leq \cost(H)$.

\begin{lemma}
	\label{lem:core:small-segment-with-Ci}
	Given a canonical $2$-edge cover $H$ of a structured graph $G$ such that the component graph $\hG_H$ has a segment that contains a $\mathcal{C}_i$ for $5 \leq i \leq 8$ and at most $i$ components.
	Then we can compute a canonical $2$-edge cover $H'$ with fewer components than $H$ and $\cost(H') \leq \cost(H)$ in polynomial time.
\end{lemma}

\begin{proof}
	Since a $\mathcal{C}_i$ has $4i$ credits for $5 \leq i \leq 8$, we can redistribute the credits of the $\mathcal{C}_i$ such that each component of the segment has $2$ credits.
	Then we iteratively perform the following steps until the segment has only one component.
	We find any cycle $F$ of length at least $3$ in the segment and let $H'\coloneq H \cup F$.
	We have $|H'| = |H| + |F|$ and $\credit(H') \leq \credit(H) - 2|F| + 2$.
	Since $|F| \geq 2$, we have $\cost(H') \leq \cost(H) + 2 - |F| \leq \cost(H)$. This can be done in polynomial time.
\end{proof}

Hence, from now on we assume that we have exhaustively applied the above lemma.

Next, we show that if there is a segment $S$ with at most $4$ nodes containing two components that are large or complex, then we can turn $H$ into a canonical $2$-edge cover $H'$ with fewer components than $H$ and $\cost(H') \leq \cost(H)$.

\begin{lemma}
	\label{lem:core:2large-same-segment}
	Given a canonical $2$-edge cover $H$ of a structured graph $G$ such that the component graph $\hG_H$ contains a segment $S$ with at most $4$ components and $S$ contains two components that are large or a complex. Then we can compute a canonical $2$-edge cover $H'$ with fewer components than $H$ and $\cost(H') \leq \cost(H)$ in polynomial time.
\end{lemma}

\begin{proof}
	We first apply~\cref{lem:core:small-segment-with-Ci} exhaustively, so that we can assume that $S$ does not contain a $\mathcal{C}_i$ for $5 \leq i \leq 8$.
	Now we can assume that $S$ contains only large, complex, and $\cFour$ components.

	We will iteratively find a cycle $F$ in $S$ and merge the components together until only one component remains.
	We first argue that complex components can be treated equivalently as large components as follows.
	If the cycle $F$ covers some bridge of a complex component, then we gain $4$ credits from the bridge, which is much more than we gain from a large component (only $2$ credits).
	If the cycle $F$ does not cover any bridge of a complex component, then the two edges in $F$ incident to the complex component are incident to the same block of the complex component.
	Hence, if $F$ is incident to a block node $B$ of some complex component, the block has a credit of one and after  adding $F$ to $H$, the components incident to $F$ are part of $B$ and again need a credit of one.

	Therefore, we can first check if there is a cycle $F$ that covers some bridge of a complex component and add  if there is such a cycle.
	If there is no such cycle, all edges incident to a complex component are only incident to one block of it and hence we can effectively treat the complex component as a large component.
	In the following, we assume that $S$ consists only of large components and $\cFour$ components.

	If there is no $\cFour$, then each component has at least $2$ credits and we can iteratively find a cycle in $S$ and merge the components together until only one component remains, as we do in the proof of~\cref{lem:core:small-segment-with-Ci}.

	Let the components of $S$ be $C_1, C_2, C_3, C_4$.
	Let $C_1$ be a $\cFour$, with cycle $v_1 - v_2 - v_3 - v_4 - v_1$.
	By the $3$-Matching \Cref{lem:3-matching} and since $S$ is $2$-node connected, there must be two edges $e_1$ and $e_2$ such that, w.l.o.g., $e_1$ is incident to $v_1$ and $C_2$ and $e_2$ is incident to $v_2$ and $C_3$.

Assume there exists a path $P$ in $S$
 (of length 1 or 2) starting in $C_2$, ending in $C_3$, not containing $C_2$ but containing $C_4$ (if $S$ has 4 components).
	Then setting $H' = (H \cup P \cup \{ e_1, e_1 \}) \setminus \{ v_1 v_2 \}$ is a canonical $2$-edge cover where the components all components of $S$ are now in a single 2EC component of $H'$.
	Hence, $H'$ has fewer components than $H$.
	Furthermore, $\cost(H') \leq \cost(H')$ since $|H'| \leq |H| + |S| - 1$ and $\credit(H') \leq \credit(H) -  4 - (|S|-2) ( 1 - 4 \delta) +2 \leq \credit(H) - |S| +1$, since $|S| \in \{ 3, 4\}$ and $S$ contains at least 2 large components.

	Note that such a path $P$ as above exists if either $C_1$ is the only $\cFour$ in $S$ or there exist a $\cFour$ in $S$ that is only adjacent to two components in $S$. Hence, in those cases we are done.
	Therefore, assume there is no such path for any $\cFour$ in $S$.
	This implies that $S$ contains exactly two $\cFour$, say $C_1$ and $C_2$ and two large components, say $C_3$ and $C_4$. Let the cycle on $C_2$ be $u_1 - u_2 - u_3 - u_4 - u_1$.
	Furthermore, in $S$ the two $\cFour$'s $C_1$ and $C_2$ are adjacent to all components of $S$ and $C_3$ and $C_4$ are only adjacent to the components $C_1$ and $C_2$ in $S$. Otherwise, we have a path $P$ as above.
	Again, by the $3$-Matching \Cref{lem:3-matching} and since $S$ is $2$-node connected, there must be two edges $e_1$ and $e_2$ such that, w.l.o.g., $e_1$ is incident to $v_1$ and $C_2$ and $e_2$ is incident to $v_2$ and $C_3$. Note that $e_1$ can not be incident to $C_4$, as then again we can find the desired path $P$.
	Assume that $e_1$ is incident to $u_1$ on $C_2$.
	There must be an edge $f$ incident to $u_2$ or $u_4$ and some other component of $S$, by the 3-Matching \Cref{lem:3-matching}. Assume w.l.o.g.\ $f$ is incident to $u_2$. We make a case distinction on which component $f$ is incident to.
	If $f$ is incident to $C_3$, then also we can find a good cycle:
	Then $H' = (H \cup \{ e_1, e_2, f \} ) \setminus \{ v_1 v_2 , u_1 u_2 \}$ is a canonical 2-edge cover that has fewer components than $H$ and it can be checked similarly to the first case that $\cost(H') \leq \cost(H)$.
	Assume $f$ is incident to $C_4$.
	There must be an edge $g$ from $C_3$ to $C_2$ and $h$ from $C_4$ to $C_1$, since $S$ is $2$-node connected and $C_3$ and $C_4$ are not connected to each other in $S$, by assumption.
	Now set $H' = (H \cup \{ e_1, e_2, f, g, h \} ) \setminus \{ v_1 v_2 , u_1 u_2 \}$.
	Observe that $H'$ is a canonical $2$-edge cover with fewer components than $H$.
	Further, $\cost(H') \leq \cost(H')$ since $|H'| \leq |H| + 3$ and $\credit(H') \leq \credit(H) -  4 - 2 - 2 ( 1 - 4 \delta) +2 \leq \credit(H) - 3$.
	Hence, we assume that $f$ is incident to $C_1$. Note that $f$ cannot be incident to $v_2$ or
	$v_4$ as then $C_1$ and $C_2$ can be merged into a single $\cEight$, a contradiction to $H$ being
	canonical.
	Hence, assume that $f$ is incident to $v_1$ (the case that $f$ is incident to $v_3$ is the same).
	Again, there must be an edge from $C_3$ to $C_2$, say $g$.
	Now observe that either $H' = (H \cup \{ e_1, e_2, g \} ) \setminus \{ v_1 v_2 , u_1 u_2 \}$ or  $H'' = (H \cup \{ e_1, f, g \} ) \setminus \{ v_1 v_2 , u_1 u_2 \}$ is a canonical $2$-edge cover with fewer components than $H$.
	Further, $\cost(H') \leq \cost(H')$ since $|H'| \leq |H| + 1$ and $\credit(H') \leq \credit(H) -  2 - 2 - 2 ( 1 - 4 \delta) +2 \leq \credit(H) - 1$.
	This finishes the case distinction and the proof of the lemma.
\end{proof}

We assume that we have applied all previous lemmas exhaustively.
In particular, the last lemma implies that every segment with at most four components contains at most one component that is either large or complex.

Next, we show that if there is a segment with at least $5$ nodes that does not contain a huge component and does contain a $\mathcal{C}_i$ for some $4 \leq i \leq 8$,  then we can turn $H$ into a canonical $2$-edge cover $H'$ with fewer components than $H$ and $\cost(H') \leq \cost(H)$.

\begin{lemma}
	\label{lem:core:many-components-no-huge}
	Given a canonical $2$-edge cover $H$ of a structured graph $G$ such that the component graph $\hG_H$ contains a segment $S$ with at least $5$ components but $S$ does not contain a huge component.
	Then we can compute a canonical $2$-edge cover $H'$ with fewer components than $H$ and $\cost(H') \leq \cost(H)$ in polynomial time.
\end{lemma}

\begin{proof}
	Let $S'$ be another segment of $\hG_H$ that contains a huge component $L$.
	In $S$ there must be a cut node $z$ that separates $S \setminus \{z\}$ from $S'$ and therefore also from $L$.
	If $z$ is a $\mathcal{C}_i$ for $5 \leq i \leq 8$, then we can assume that $S\setminus \{z\}$ has at least $i$ components since we can apply~\cref{lem:core:small-segment-with-Ci} exhaustively.

	Suppose $z$ is a $\mathcal{C}_i$ for $4 \leq i \leq 8$, then $S \setminus \{z\}$ contains at least $i$ components and $4i$ vertices.
	This implies that $z$ is a large $\mathcal{C}_i$ cut since $L$ contains at least $32$ vertices, contradicting that $G$ is structured.
	Hence, $z$ corresponds to either a large or complex component.

	If $z$ and $L$ are in the same segment, then we can either apply~\cref{lem:core:many-components} if the segment has at least $5$ components or apply~\cref{lem:core:2large-same-segment} if the segment has at most $4$ components.
	Otherwise, consider the segment $S'' \neq S$ containing $z$ and another cut node $z'$ separating $L$ from $z$.
	If $S''$ has at least $5$ components, we continue with the same arguments as above, i.e., replace $S$ with $S''$.
	If $S''$ has at most $4$ components, then we can argue as above that $z'$ is also a large or complex component.
	Then we can apply~\cref{lem:core:2large-same-segment} in $S''$.
\end{proof}

We assume that we have applied the above lemmas exhaustively.
This implies the following: $H$ does not contain any $\mathcal{C}_i$ for $5 \leq i \leq 8$, and each segment has at most $4$ nodes and contains at most one large or complex component.

Next, we show that if there are two components that are large or complex, then we can turn $H$ into a canonical $2$-edge cover $H'$ with fewer components than $H$ and $\cost(H') \leq \cost(H)$.

\begin{lemma}
	\label{lem:core:2large}
	Given a canonical $2$-edge cover $H$ of a structured graph $G$ such that the component graph $\hG_H$ contains a huge component $L$ and another component $C \neq L$ that is either large or complex.
	Then we can compute a canonical $2$-edge cover $H'$ with fewer components than $H$ and $\cost(H') \leq \cost(H)$ in polynomial time.
\end{lemma}

\begin{proof}
We assume that we have applied all of the above lemmas exhaustively and hence no condition of the above lemmas is met.
	Let $S$ be the segment containing $L$ and $S'\neq S$ the segment containing $C$, such that there is a $\cFour$-cut node $z$ in $S$ separating $C$ and $L$.
	Since $H$ is canonical and we cannot merge two $\cFour$s into one $\cEight$, $z$ cannot be in a segment with only $2$ $\cFour$'s.
	Since $z$ cannot correspond to a large $\cFour$ cut (we cannot partition $G \setminus V(z)$ into two parts with at least $16$ vertices) as $G$ is structured and $|V(C)| \geq 9$,
	$z$ must be in exactly $2$ segments: $S$ and $S'$.
	Furthermore, $S'$ either consists of $z$ and $C$, or $z$, $C$ and another component $z'$ that is a $\cFour$.
	Let the vertices of $z$ be labeled as $v_1-v_2-v_3-v_4$.
	We apply~\cref{lem:cycle-through-cfour} to find a cycle $F$ in $S$ containing $L$ and $z$ such that $F$
	is incident to $v_1$ and $v_2$.

	If $S'$ consists of $z$ and $C$,
	by~\cref{lem:local3matching}, there must be a 3-matching between $V(z)$ and $V(C)$, which implies that there is a pair $v_i$ and $v_{i+1}$ for some $i \in \{2,3\} $, such that $v_iC,v_{i+1}C \in E(G)$.
	We let $H'\coloneq H\setminus \{v_1v_2,v_iv_{i+1}\}\cup F \cup \{v_iC,v_{i+1}C\}$.
	In this way we save the credits of the components incident to $F$ and $C$.
	Therefore, $\credit(H') \leq \credit(H) - (|F|-1)\cdot 4 \cdot \cre - 2-2 + 2 = \credit(H)-|F|+4\delta(|F|-1)-1 \leq  \credit(H)-|F|$.
	The last inequality holds since $S$ contains at most $4$ components and $|F| \geq 4$.
	We have $|H'| = |H| - 2+|F|+2 = |H|+|F|$, which implies that $\cost(H') \leq \cost(H)$.

	Assume next that $S'$ consists of $z$, $C$, and another component $z'$ that is a $\cFour$.
	Let the vertices of $z'$ be labeled as $v_1'-v_2'-v_3'-v_4'$.
	Since $S'$ is 2-node-connected, there is a cycle $F'$ in $S'$ through $z,z'$ and $C$.
	We show that there is a cycle $F''$ in $S'$ through $z,z'$ and $C$, and incident to neighboring vertices of $z'$, say, $v_1'$ and $v_2'$.
	Suppose $F'$ does not satisfy this condition, say, $F'$ is incident to $v_1'$ and $v_3'$.
	We apply~\cref{lem:local3matching} to find a 3-matching between $V(z')$ and $V(C) \cup V(z)$.
	We can assume that there is an edge in the 3-matching that is incident to $v_2'$.
	If the edge is incident to $C$, we let $F''\coloneq F \setminus \{z'C\} \cup {v_2C}$, otherwise we let $F''\coloneq F \setminus \{z'z\} \cup {v_2'z}$.
	It can be easily checked that $F''$ is a cycle in $S'$ through $z,z'$ and $C$, and incident to neighboring vertices of $z'$, say, $v_1'$ and $v_2'$ (or $v_2'$ and $v_3'$).
	We let $H'\coloneq H\setminus \{v_1v_2,v_1'v_2'\}\cup F'' \cup F$.
	Then we have $|H'| = |H| - 2+|F''|+|F| = |H|+|F|+1$ and $\credit(H') \leq \credit(H) - (|F|-1+1)\cdot 4 \cdot \cre - 2-2 + 2 = \credit(H)-|F|-2+4\delta|F| \leq \credit(H)-|F|-1 $.
	The last inequality holds since $S$ contains at most $4$ components and $|F| \geq 4$.
	Therefore, $\cost(H') \leq \cost(H)$.
\end{proof}

We assume that we have applied the above lemmas exhaustively.
This implies the following:
$H$ contains one huge component $L$, which may be complex, and all other components are $\cFour$'s.
Now we prove \Cref{lem:manyC4_core-configuration_main}.

\begin{proof}[Proof of \Cref{lem:manyC4_core-configuration_main}]

Recall that $L$ can be 2EC or complex, i.e., containing some bridges.
Assume $L$ is complex.
Before we prove the properties of a core-square configuration, we turn the component $L$ into a 2EC component without increasing $\cost(H)$.

	Consider a pendant block $B$ of $L$ that is adjacent to a bridge $e$.
	Since $G$ is 2EC, there must be a bridge-covering path $P$ in $G$ starting from $B$ that covers $e$.
	Since we have exhaustively applied the above lemmas, we know that $L$ is only contained in segments of $\hG_H$ of size at most $4$.
	Therefore, observe that there is a bridge-covering path $P$ of length at most $4$ starting in $B$ and covering $e$.
	Hence, $P$ is a short heavy cycle and according to \Cref{lem:abc-merge}, $H' = H \cup P$ is a canonical $2$-edge cover with $\cost(H) \leq \cost(H')$.

	If after the application of this bridge-covering step one of the above lemmas applies, then we exhaustively apply that lemma.
	Hence, we can cover all bridges of $L$ and turn it into a 2EC component.
	Therefore, we now assume that $L$ is bridgeless and we have exhaustively applied the above lemmas.
	These directly imply the Properties~(i) and~(ii) of the core-square configuration.

	We next show Property~(iv).
	That is, we show that each cut component which is a $\cFour$ is contained in at most $2$ segments.
	Suppose that there is a cut node $z$ which corresponds to a $\cFour$ and is contained in at least $3$ segments, say, $S_1, S_2$ and $S_3$.
	We can assume that $S_1$ and $S_2$ do not contain the huge component $L$.
	Then each of $S_1$ and $S_2$ must contain at least $3$ components, as otherwise we can merge two $\cFour$s into a $\cEight$.
	This implies that $z$ is a large $\cFour$ cut since $(S_1 \cup S_2) \setminus z$ already contains at least $16$ vertices, which contradicts that $G$ is structured.

	For each cut node $z$ corresponding to a $\cFour$,
	let $S_1$ and $S_2$ be the segments containing $z$.
	We show that either of $S_1$ or $S_2$ contains $L$.
	Otherwise, there must be another cut component $z'$ which is a $\cFour$ and separates $\{S_1 \cup S_2\} \setminus\{ z'\}$ from $L$.
	This implies that $z'$ is a large $\cFour$ cut since  $\{S_1 \cup S_2\} \setminus\{ z'\}$  contains at least $4 \cdot 4=16$ vertices and $L$ contains at least $32$ vertices, which contradicts that $G$ is structured.
	Without loss of generality, we can assume that $S_1$ contains $L$.
	Then $S_2$ must contain at most $4$ components, otherwise $z$ is a large $\cFour$ cut, which contradicts that $G$ is structured.
	Hence, we ensure the Property~(iv) of the core-square configuration.

	We next prove Property~(iii), i.e., we show that we can turn $H$ into a canonical $2$-edge cover such that $L$ is only contained in segments of size $2$.

	Assume first that $L$ is contained in at least two segments $S, S'$ of size at least 3 (and at most 4, since we can not apply the above lemmas).
	Then, by \Cref{lem:31-merge}, in $S$ there exists  a single $(k, j)$-merge containing $L$ such that $k-j \geq 2$ and $k \in \{3, 4 \}$.
	Then, by \Cref{lem:cycle-useful-for-double-merge}, there is a double $(k, j)$-merge containing $L$ for some $k \geq 5$ and $k-j \geq 4$.
	Hence, by \Cref{lem:abc-merge} we can turn $H$ into a canonical $2$-edge cover $H'$ with $\cost(H') \leq \cost(H)$.
	If after the application of this step one of the above lemmas applies, then we exhaustively apply that lemma.
	Observe that again the Properties~(i), (ii), and~(iv) are satisfied.

	Now assume there is only one segment $S$ containing $L$ and having at least 3 nodes (and at most 4).
	If one of the other components $C$ in $S$ (which is a $\cFour$), is a cut node, then $C$ is contained in another segment $S'$ containing either 3 or 4 components.
	Then we do the same as in the previous step:
By \Cref{lem:31-merge}, in $S$ there exists  a single $(k, j)$-merge containing $C$ such that $k-j \geq 2$ and $k \in \{3, 4 \}$.
	Furthermore, by \Cref{lem:cycle-useful-for-double-merge}, there is a double $(k, j)$-merge containing $C$ for some $k \geq 5$ and $k-j \geq 4$.
	Hence, by \Cref{lem:abc-merge} we can turn $H$ into a canonical $2$-edge cover $H'$ with $\cost(H') \leq \cost(H)$.

	If after the application of this step one of the above lemmas applies, then we exhaustively apply that lemma.
	Observe that again the Properties (i), (ii), and (iv) are satisfied.
	Hence, in the remaining case we have that there is only one segment $S$ containing $L$ and having at least 3 nodes (and at most 4).
	Furthermore, none of the nodes in the segment is a cut node.
	Observe that there exists a simple cycle $X$ in $S$ containing all nodes of $S$. Set $H' = H \cup X$.
	Then, $|H'| = |H| + |X|$ and
	$$\credit(H') \leq \credit(H) + 2 - 2 - (|X|-1) \cdot (1- 4\delta) \leq \credit(H) - |X| + 1 + (|X|-1) \cdot 4 \delta \leq \credit(H) - |X| + 2    \ ,$$
	since the newly created large 2EC component needs 2 credits, the huge 2EC component $L$ has a credit of 2 and every other component is a $\cFour$.
	Then we have $\cost(H') \leq \cost(H) +2$, as desired.
	Furthermore, observe that after this step Property~(iii) is satisfied, since now $L$ is only contained in segments of size 2.
This finishes the proof.
\end{proof}

\subsubsection{Proof of \Cref{lem:core:many-components}}
\label{sec:nice-path}
To prove \Cref{lem:core:many-components}, we use the definition of a \emph{gluing path}. A similar concept in a weaker form has been used in~\cite{GargGA23improved}.

\definitionMANYgluing*

An example of a gluing path is given in \Cref{fig:gluing-path-example}.
The benefit of a gluing path is the following: Assume we are given a gluing path $P$ of length $k \geq 4$.
If there is some edge $e$ from a node $v$ of $P$ to some node $u$ of $P$ such that their distance on $P$ is at least $4$, then consider the cycle $C$ formed by the subpath of $P$ from $u$ to $v$ and the edge $e$.
Clearly, adding the $k+1$ edges of $C$ to $H$ turns the $k+1$ components of $C$ into one single $2$-edge connected component $C'$. Let $H'$ be the new $2$-edge cover.
Furthermore, due to the definition of a gluing path, each component of $C$ that is a $\cFour$ (except the components corresponding to $u$ and $v$) can be shortcut. That is, we can remove one edge for each such component such that $C'$ remains $2$-edge connected.
Therefore, it can be easily checked that $H'$ is also a canonical $2$-edge cover such that $\cost(H') \leq \cost(H)$.

Of course, such edges may not always exist.
Hence, our goal is the following:
Given some canonical $2$-edge cover and some gluing path $P$ of length $k$, either we want to make progress by turning $H$ into some canonical $2$-edge cover $H'$ such that $H'$ contains fewer components than $H$ and $\cost(H') \leq \cost(H)$ or we can extend the gluing path $P$ to a gluing path $P'$ such that the length of $P$ is greater than the length of $P$.
Since the length of a gluing path can easily be bounded by the number of vertices of $G$, after polynomially many extension we must find the desired $H'$.
This motivates the definition of \emph{progress}.

\begin{definition}[progress]
	\label{def:many_C4_progress}
	Given a $2$-edge cover $H$ with a gluing path $P$, we say that we make \emph{progress} if in polynomial time either
	\begin{itemize}[nosep]
		\item[(i)] we can turn $H$ into a canonical $2$-edge cover $H'$ such that $\cost(H') \leq \cost(H)$ and either $H'$ contains fewer components than $H$ or $H'$ has the same number of components as $H$ but $H'$ has fewer bridges than $H$, or
		\item[(ii)] we can extend the gluing path $P$ to a gluing path $P'$ such that the length of $P'$ is greater than the length of $P$.
	\end{itemize}
\end{definition}

We will then show the following lemma.

\begin{lemma}
	\label{lem:many_C4_progress}
	Given a canonical $2$-edge cover $H$ of some structured graph $G$ that is not in a core-square configuration and a gluing path $P$, then we can make progress.
\end{lemma}

Thus, \Cref{lem:many_C4_progress} implies that eventually we end up with a core-square configuration.
Hence, proving \Cref{lem:many_C4_progress} implies \Cref{lem:core:many-components}.
Therefore, for now we want to prove \Cref{lem:many_C4_progress}.

Throughout this subsection, we assume that we are given some $2$-node connected segment $S$ of $\hG_H$ of $H$ that contains at least $5$ nodes (see prerequisites of \Cref{lem:core:many-components}).
Consider some gluing path $P = z_k - z_{k-1} - ... - z_1$ of length $k-1$ on the nodes $z_1, ..., z_k$ (each corresponding to some component of $H$).
We say that $z_k$ is the start of the gluing path and $z_1$ is the end.
If $z_i$ is a $\cFour$, we label the vertices of $z_i$ by $v_{1}^i, v_{2}^i, v_{3}^i, v_{4}^i$, which also defines the cycle of the $\cFour$.
If $z_i$ is not a $\cFour$, $V(z_i)$ contains at least $4$ vertices and we say that $v_{1}^i, v_{2}^i, v_{3}^i, v_{4}^i$ are some arbitrary distinct vertices of $V(z_i)$.
For $i \leq j$, the sub-path of $P$ from node $i$ to node $j$, i.e., $z_i - z_{i+1} - ... - z_{j-1} - z_j$ is denoted by $P_{ij}$.
For a node $z_i$, we say that the in-vertex is the vertex incident to the edge $z_{i+1} z_i$ and the out-vertex is the vertex incident to the edge $z_i z_{i-1}$.
Without loss of generality we let the in-vertex of node $z_i$ be $v_{1}^i$ and the out-vertex be $v_{2}^i$ (recall that according to the definition of a gluing path, the in-vertex and the out-vertex must be adjacent in $H$ for a $\cFour$; for complex components, the in-vertex and the out-vertex must be distinct if they do not belong to the same block).

\begin{figure}
	\centering
    \includegraphics[width=0.6\textwidth]{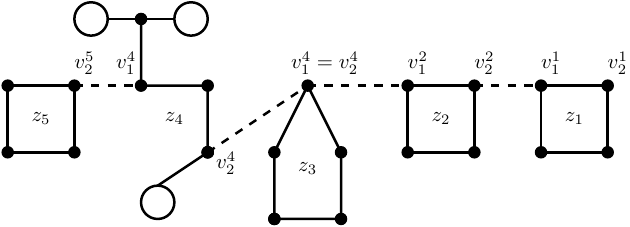}
    \caption{Example of a gluing path of length $5$: $z_5$, $z_2$, and $z_1$ are $\cFour$'s, $z_4$ is complex, and $z_3$ is a $\cFive$.}
    \label{fig:gluing-path-example}
\end{figure}

We first prove that we can compute a gluing path of length at least $3$ in polynomial time.

\begin{lemma}
	\label{lem:nicepath_start1}
	Given a $2$-edge cover $H$ of some structured graph $G$ such that $H$ contains a $2$-node connected segment $S$ of $\hG_H$ of $H$ that contains at least $3$ nodes, we can compute a gluing path of length at least $2$ in polynomial time.
\end{lemma}

\begin{proof}
	If $S$ contains a node $z_2$ that is not a $\cFour$, then there is clearly a gluing path of length $2$:
	Since $S$ is $2$-node connected, there must be two edges incident to $z_2$ to two distinct nodes $z_1, z_3$, that are not $z_2$. Say $e_1$ is incident to $z_1$ and $e_2$ is incident to $z_3$.
	Furthermore, there must be two such edges such that they are not incident to the same vertex of $z_2$, as otherwise $G$ is not $2$-vertex connected, a contradiction to $G$ being structured. Hence, we can use such two edges, which form a path $z_1 - z_2 - z_3$ that forms a gluing path.

	Hence, let us assume all nodes in $S$ are a $\cFour$.
	Let $z_2$ be one such component.
	Since $S$ is $2$-node connected, there must be two edges $e_1, e_2$ incident to $z_2$ to two distinct nodes $z_1, z_3$, that are not $z_2$.
	If $e_1, e_2$ are incident to vertices that are adjacent within $z_2$ w.r.t.\ the edges of $H$, then one can easily check that $z_1 - z_2 - z_3$ is a gluing path using the edges $e_1$ and $e_2$.
	Hence, let us assume this is not the case.
	Let $v_1^2 - v_2^2 - v_3^2 - v_4^2 - v_1^2$ be the $4$-cycle when expanding $z_2$.
	Let us assume $e_1, e_2$ are incident to the same vertex, say $v_1^2$. The case that one edge is incident to $v_1^2$ and the other to $v_3^2$ is analogous.
	By the $3$-Matching \Cref{lem:3-matching}, there must be an edge $e'$ incident to either $v_2^2$ or $v_4^2$ that is going to some component that is not $z_2$. Both cases are symmetric, so assume $e'$ is incident to $v_2^2$.
	If $e'$ is incident to some component $z' \neq z_1$, one can easily check that there is a gluing path $z_1 - z_2 - z'$ using the edges $e_1$ and $e'$.
	Otherwise, if $e'$ is incident to $z_1$, one can easily check that there is a gluing path $z_1 - z_2 - z_3$ using the edges $e'$ and $e_2$.
\end{proof}

Hence, from now on we assume that we are always given some gluing path $P = z_3 - z_2 - z_1$ of length at least $2$.

Next, we show the following useful lemma. An illustration of the lemma is given in \Cref{fig:1-3-edge}.

\begin{lemma}
	\label{lem:manyC4_1-3-edge}
	Let $P = z_k - z_{k-1} - ... - z_1$ be a gluing path of length $k-1 \geq 2$. Then either $v_2^1$ is adjacent to one of $\{v_2^3, v_4^3 \}$ or we can make progress.
	Similarly, either $v_1^k$ is adjacent to one of $\{v_1^{k-2}, v_3^{k-2} \}$ or we can make progress.
\end{lemma}

\begin{figure}
	\centering
    \includegraphics[width=0.7\textwidth]{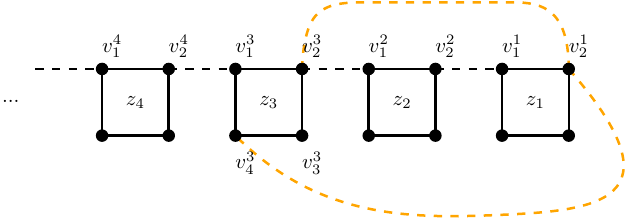}
    \caption{\Cref{lem:manyC4_1-3-edge}: Either at least one of the dashed orange edges exists, or we can make progress.}
    \label{fig:1-3-edge}
\end{figure}

\begin{proof}
	We prove the statement for $v_2^1$. The statement for $v_1^k$ is analogous.
	We can apply the $3$-Matching \Cref{lem:3-matching} to $V(z_1)$ and hence there must be an edge $e$ outgoing to some other node in $S$ from either $v_2^1$ or $v_4^1$.
	Without loss of generality, we assume that $e$ is incident to $v_2^1$, i.e., $e= v_2^1 u$.
	We assume that $u \notin \{v_2^3, v_4^3 \}$ and show that we can make progress.

	If $u$ is part of some component not on $P$, then note that we can add $e$
	to the gluing path $P$ and obtain a longer gluing path $P'$ and hence make progress.
	So now assume that $u \in V(P)$.
	This is independent of what type of component $z_1$ is.

	If $u \in V(z_i)$, $4 \leq i \leq k$, then we can make progress as follows:
	Note that $K = P_{1i} \cup e$ is a cycle in the component graph.
	Let $H' = H \cup K$ and note that it is a canonical $2$-edge cover.
	Further, note that for each node $z_j$, $1 < j < i$, that is a $\cFour$, the edge $v_1^j v_2^j$ can be removed from $H'$ such that $H'$ is still a canonical $2$-edge cover.
	Therefore, $C$ defines a single $(i, 1)$-merge, and hence, by \Cref{lem:abc-merge}, $\cost(H') \leq \cost(H)$.
	Thus, we can assume that $u \in V(z_2) \cup V(z_3)$.

	If $u \in \{ v_3^3, v_1^3\}$, then similar to the previous case we can make progress as follows: $P_{13} \cup e$ defines a single $(3, 0)$-merge, and hence by \Cref{lem:abc-merge}, we can make progress.

	Therefore, let us assume that $u \in V(z_2)$.
	If either $z_1$ or $z_2$ is a $\cFive$, $\cSix$, $\cSeven$, or $\cEight$, then there is a short heavy cycle between $z_1$ and $z_2$ and we can make progress according to \Cref{lem:abc-merge}.
	Similarly, if neither $z_1$ nor $z_2$ is a $\cFour$, i.e., if $z_1, z_2$ are either large or complex, then we can make progress according to \Cref{lem:abc-merge}.

	Hence, in the remaining case we have that at least one of $z_1, z_2$ is a $\cFour$ and at most one of $\{z_1, z_2\}$ is large or complex.
	We now show that either we can make progress, or we can switch the edge $z_2 z_1$ from $P$ with the edge $e$ such that we still have a gluing path $P'$ of length $k-1$.
	In particular, the in-vertex of $z_1$ on $P$ is $v_1^1$ and the in-vertex of $z_1$ on $P'$ is $v_2^1$.
	For the sake of contradiction, assume the statement is not true.
	Then $z_2$ can not be large, as otherwise the above statement is always true.
	If $z_2$ is complex, then $u$ can not be equal to $v_1^2$, the out-vertex of $z_2$, since then we can also find the desired gluing path $P'$.
	Furthermore, the out-vertex of $z_2$ can also not be in the same block as $u$ on $z_2$.
	Hence, the unique path from $v_2^2$ to $u$ contains at least one bridge.
	But then $e$ and the edge from $z_2$ to $z_1$ on $P$ defines a short heavy cycle and hence \Cref{lem:abc-merge} implies that we can make progress.
	Thus, $z_2$ is a $\cFour$.
	We now make a case distinction on whether $z_1$ is complex, large or a $\cFour$.
	If $z_1$ is a $\cFour$, then $u$ can neither be $v_3^2$ nor $v_1^2$:
	Otherwise this is a contradiction to $H$ being a canonical $2$-edge cover, as then $e$ and the edge on $P$ connecting $z_2$ and $z_1$ can be used to merge $z_1$ and $z_2$ into one $\cEight$.
	Therefore, $u \in \{ v_2^2, v_4^2 \}$.
	Now clearly there is such a desired gluing path $P'$ of length $k-1$.
	Hence, $z_1$ is either large or complex.
	There must be an edge $e'$ from $z_1$ to some other node distinct from $z_2$ (and $z_1$), as otherwise $z_1$ is not in a $2$-node connected component of $\hG_H$ containing at least $3$ nodes.
	If $z_1$ is large, then following the same proof as before but switching $e$ with $e'$, we can make progress.
	In particular, if $e'$ is incident to $z_i$ with $i \geq 4$, or $e'$ is incident to $v_3^3$ or $v_1^3$, then we can make progress since this defines a single $(i, 1)$-merge or a single $(3, 0)$-merge and the result follows by \Cref{lem:abc-merge}.
	Hence, $z_1$ is complex.
	If the in-vertex of $z_1$, i.e., $v_1^1$, is not part of some block of $z_1$, then the edge from $z_2$ to $z_1$ from $P$ together with $e$ defines a short heavy cycle, since there is at least one bridge on the unique path from $v_1^1$ to $v_2^1$. Therefore, we can make progress according to \Cref{lem:abc-merge}.
	Hence, assume that the in-vertex of $z_1$, i.e., $v_1^1$, is part of some block of $z_1$.
	But then the statement clearly follows because we can replace the edge $z_2 z_1$ on $P$ with $e$ to obtain a new gluing path $P'$ of length $k-1$ ending on $v_2^1$, as desired.

	Therefore, since we assume that we can not make progress, in the remaining case we assume that we can replace the edge $z_2 z_1$ of the gluing path by the edge $e$ and still have a gluing path of length $k-1$.
	Since there must be a $3$-matching outgoing from $z_1$ by \Cref{lem:3-matching}, there must be an additional edge $f$ incident to $z_1$.
	Since we can replace the gluing path $P$ with the gluing path $P'$, by symmetry, we can w.l.o.g.\ assume that $f$ is incident to $v_4^1$.
	Now for $f$ we must have the same conditions as for $e$ and hence, either $f$ is incident to one of $\{v_2^3, v_4^3\}$, to $z_2$, or we can make progress.
	The only case in which we do not show the statement of the lemma is that $f$ is incident to $z_2$.
	Now observe that this is true for any edge $f$ incident to $z_1$ and some other component that is not $z_1$.
	But then $z_1$ is not part of the $2$-node connected segment $S$, a contradiction.
	This finishes the proof.
\end{proof}

Hence, from now on we can always assume that $v_2^1$ is adjacent to one of $\{v_2^3, v_4^3 \}$ and $v_2^k$ is adjacent to one of $\{v_2^{k-2}, v_4^{k-2} \}$.

The next lemma states some useful facts about what concrete components $z_1$, $z_2$ and $z_3$ can be.

\begin{lemma}
	\label{lem:manyC4_typesz1-3}
	Let $P = z_k - z_{k-1} - ... - z_1$ be a gluing path of length $k-1 \geq 2$. We can make progress if
 one of $z_1, z_2, z_3$ (or one of $z_k, z_{k-1}, z_{k-2}$) is a $\cFive, \cSix, \cSeven$ or a $\cEight$.
\end{lemma}

\begin{proof}
	We show the lemma only for $z_2$ as the other statement is analogous.
	By \Cref{lem:manyC4_1-3-edge}, we know that either we can make progress or there is an edge $e$ from $v_1^2$ to one of $\{v_2^3, v_4^3\}$. In the former case we are done. Hence, assume we are in the latter case.
	Say without loss of generality $e$ is incident to $v_2^3$.
	Observe that there is a cycle $K$ on the nodes $z_1 - z_2 - z_3$ using $P_{13}$ and $e$.
	If one of $z_1, z_2, z_3$ is a $\cFive, \cSix, \cSeven$ or a $\cEight$, then this cycle is a short heavy cycle, and hence we can make progress.
	This proves the first statement.
\end{proof}

The next two lemmas show that we can always assume that we are given a gluing path of length at least $4$, i.e., the gluing path contains at least 5 nodes.

\begin{lemma}
	\label{lem:nicepath_start2}
	Given a $2$-edge cover $H$ of some structured graph $G$ such that $H$ contains a $2$-node connected segment $S$ of $\hG_H$ of $H$ that contains at least $4$ nodes, in polynomial time we can compute a gluing path of length at least $3$ or we can make progress.
\end{lemma}

\begin{proof}
	According to \Cref{lem:nicepath_start1} we can find a gluing path $P = z_3 - z_2 - z_1 $ of length $2$.
	Furthermore, by \Cref{lem:manyC4_typesz1-3}, we know that components on $P$ are either large, complex or a $\cFour$.
	Additionally, by \Cref{lem:manyC4_1-3-edge}, we know that $v_2^1$ is adjacent to one of $\{v_2^3, v_4^3 \}$ or we can make progress and $v_1^3$ is adjacent to one of $\{v_1^{1}, v_3^{1} \}$ or we can make progress.
	If we can make progress, we are done. Hence, we assume that those two edges, say $e_1, e_3$, exist. Let $e_1$ be the edge incident to $v_2^1$ and $e_3$ be the edge incident to~$v_1^3$.

	First, assume that both $z_1$ and $z_3$ are $\cFour$.
	But now it is easy to check that $e_1$ and $e_3$ can be used to turn $z_1$ and $z_3$ into a $\cEight$, a contradiction to our starting canonical $2$-edge cover.

	Hence, assume w.l.o.g.\ that $z_1$ is not a $\cFour$. Now it is easy to see that $P_{13} \cup \{ e_3 \}$ is a single $(3, 0)$-merge, and hence, by \Cref{lem:abc-merge}, we can make progress.
\end{proof}

\begin{lemma}
	\label{lem:nicepath_start3}
	Given a $2$-edge cover $H$ of some structured graph $G$ such that $H$ contains a $2$-node connected segment $S$ of $\hG_H$ of $H$ that contains at least $5$ nodes, in polynomial time we can compute a gluing path of length at least $4$ or we can make progress.
\end{lemma}

\begin{proof}
	According to \Cref{lem:nicepath_start2} we can find a gluing path $P = z_4 - z_3 - z_2 - z_1 $ of length $3$.
	Furthermore, by \Cref{lem:manyC4_typesz1-3}, we know that components on $P$ are either large, complex or a $\cFour$.
	Additionally, by \Cref{lem:manyC4_1-3-edge}, we know that $v_2^1$ is adjacent to one of $\{v_2^3, v_4^3 \}$ or we can make progress and $v_1^4$ is adjacent to one of $\{v_1^{2}, v_3^{2} \}$ or we can make progress.
	If we can make progress, we are done. Hence, we assume that those two edges, say $e_1, e_4$, exist. Let $e_1$ be the edge incident to $v_2^1$ and $e_4$ be the edge incident to $v_1^4$.
Now observe that
$$Z = \Big( P_{14} \cup E(z_1) \cup E(z_2) \cup E(z_3) \cup E(z_4) \cup \{e_1, e_2 \} \Big) \setminus \{v_1^1 v_2^1, v_1^2 v_2^2, v_1^3 v_2^3, v_1^4 v_2^4\} $$ makes $V(z_1) \cup V(z_2) \cup V(z_3) \cup V(z_4)$ $2$EC. Further, $H' = (H \setminus E(z_1) \cup E(z_2) \cup E(z_3 \cup E(z_4)) \cup Z)$ is a canonical $2$-edge cover with fewer components than $H$ and $\cost(H') \leq \cost(H)$. Hence, we can make progress.
\end{proof}

We next consider different types of \emph{branches}, i.e., we consider the case that we are given two gluing paths that overlap significantly in some way.
We show that if we have such branches, then we can make progress. This is summarized as follows.

\begin{figure}[tb]
	\centering
    \includegraphics[width=0.6\textwidth]{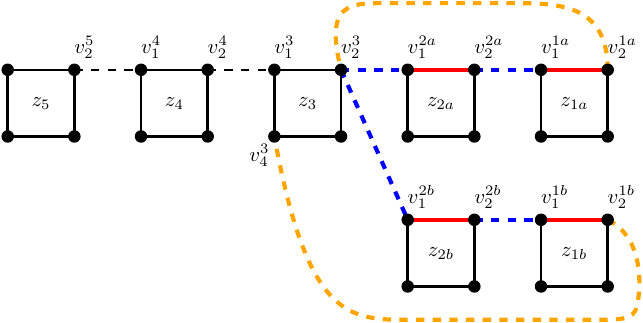}
    \caption{An example for a double $(5,1)$-merge in \Cref{lem:nicepath_branch-at-3-length-2}.}
    \label{fig:gluing-path-branch1}
\end{figure}

\begin{lemma}
	\label{lem:nicepath_branch-at-3-length-2}
	Let $P = z_k - z_{k-1} - ... - z_3 - z_{2a} - z_{1a}$ be a gluing path of length $k-1 \geq 2$ and let $P' = z_k - z_{k-1} - ... - z_3 - z_{2b} - z_{1b}$ be another gluing path of length $k-1 \geq 2$ such that the first $k-2$ nodes are the same in $P$ and $P'$ but the remaining $4$ nodes are all distinct, i.e., $z_{1a} \neq z_{2b} \neq z_{2a}$ and $z_{1a} \neq z_{1b} \neq z_{2a}$.
	Then, we can make progress.
\end{lemma}

\begin{proof}
	Assume we cannot make progress. Then, applying \Cref{lem:manyC4_1-3-edge} to $P$ and $P'$ on $z_{1a}$ and $z_{1b}$, we know there is an edge $e_a$ from $v_2^{1a}$ to one of $\{v_2^3, v_4^3\}$ and one edge $e_b$ from $v_2^{1b}$ to one of $\{v_2^3, v_4^3\}$.
	Now observe that $P_{1a3}$, $P'_{1b3}$ together with $e_a$ and $e_b$ form a double $(5, 1)$-merge, and hence by \Cref{lem:abc-merge} we can make progress. This situation is illustrated in \Cref{fig:gluing-path-branch1}.
\end{proof}

Before we proceed with the next type of branch, we first show that $z_3$ is a $\cFour$ or we can make progress.

\begin{lemma}
	\label{lem:manyC4_typesz1-3:2}
	Let $P = z_k - z_{k-1} - ... - z_1$ be a gluing path of length $k-1 \geq 2$. We can make progress if
		$z_3$ (or $z_{k-2}$) is not a $\cFour$.
\end{lemma}
\begin{proof}
	We show the lemma only for $z_3$ as the other statement is analogous.
	By \Cref{lem:manyC4_1-3-edge}, we know that either we can make progress or there is an edge $e$ from $v_1^2$ to one of $\{v_2^3, v_4^3\}$. In the former case we are done. Hence, assume we are in the latter case.
	Observe that there is a cycle $K$ on the nodes $z_1 - z_2 - z_3$ using $P_{13}$ and $e$.
	By \Cref{lem:manyC4_typesz1-3} $z_3$ is either large, complex or a $\cFour$.
	If $z_3$ is large, then $K$ is a single $(3, 0)$-merge, and hence we can make progress according to \Cref{lem:abc-merge}.
	Hence, in the remainder assume that $z_3$ is a complex component.
	If $e$ is not incident to $v_2^3$, then either $K$ is a short heavy cycle (if $K$ covers a bridge of $z_3$) or it is a single $(3, 0)$-merge (if $K$ is incident to a block of $z_3$). Hence, in either case we can make progress.
	Hence, for the remainder assume that $e$ is incident to $v_2^3$ and that this vertex is not contained in a block. Otherwise, as above, $K$ is a single $(3, 0)$-merge, and we are done.
	Furthermore, w.l.o.g.\ at most one vertex of $z_3$, say $v_2^3$, is adjacent to $z_1$.
	Otherwise, if there are two vertices of $z_3$ adjacent to $z_1$, by the previous arguments we know that adding both edges covers at least one bridge of $z_3$ and hence we make progress.

	Assume that one of $z_1$ or $z_2$ (w.l.o.g. $z_1$; the other case is analogous) is incident to some edge $e'$ which is incident to $z_i$ with $i \geq 4$.
	Then $P_{1i}$ together with $e'$ is a cycle, which is either a short heavy cycle (if $i \leq 6$) or a single $(i, j)$-merge for $j  \in \{0, 1\}$ (if $i \geq 7$), and hence in either case we can make progress.
	If $z_1$ or $z_2$ is incident to a vertex of $z_3$ distinct from $v_2^3$, then there is also a short heavy cycle.
	Hence, let us assume this is not the case. This implies that $z_1, z_2$ are only connected to each other, to $v_2^3$ and potentially to some other nodes distinct from $P$.

	We next prove that both $z_1, z_2$ are $\cFour$.
	Note that if $z_1$ or $z_2$ is complex, then $K$ does not cover a bridge of any of these components, since otherwise we have a heavy short cycle.
	Hence, if neither $z_1$ nor $z_2$ is $\cFour$,  it can be easily checked that adding $K$ to $H$ satisfies the credit invariant and we made progress.
	Hence, assume exactly one of $z_1, z_2$ is a $\cFour$. W.l.o.g.\ we can assume $z_2$ is a $\cFour$, since the other case is symmetric.
	If $z_1$ is complex, then there is only one block of $z_1$ which is connected to $S$, the segment of $P$.
	Otherwise, we can make progress since either we can extend the gluing path (if the edge is incident to some node of $S$ not on $P$), we have a short heavy cycle covering a bridge of $z_1$ (if the edge is incident to some node $z_i$ with $i \leq 6$), or there is a single $(i, j)$-merge for $j  \in \{0, 1\}$ (if the edge is incident to $z_i$ on $P$ with $i \geq 7$).
	Hence, since all edges incident to $z_1$ and $S$ are incident to a specific block $B_1$ of $z_1$, in fact $z_1$ acts like a large component locally for $S$.
	Since $z_1$ (and therefore $B_1$) is only incident to $v_2^3$
	 and $z_2$, $|V(z_3)| \geq 11$, and $|V(z_1) \cup V(z_2)| \geq 11$, by applying the local $4$-Matching Lemma to $V(z_1) \cup V(z_2)$ there must be a $3$-Matching from $V(z_2)$ to some nodes of $S$ not on $P$. This is true since $z_2$ is only incident to $z_1$,  $v_2^3$ and some other nodes of $S$ not on $P$.
	Let these outgoing edges from $z_2$ to other nodes of $S$ not on $P$ be $e', e'', e'''$.
	But then $P_{3k}$ together with $e$, the edge $z_1 z_2$ from $P$ and one of the edges of $\{e', e'', e'''\}$ is a longer gluing path since $\{e', e'', e'''\}$ is a $3$-matching incident to the $4$-cycle $z_2$. Hence, we can make progress.
	Thus, from now on, we assume that both $z_1, z_2$ are $\cFour$.

	We next prove that neither $z_2$ nor $z_1$ can be cut nodes, i.e., they are both only adjacent to nodes of $S$, the segment containing $P$.
	Assume the contrary and that $z_2$ is a cut node. The other case with $z_1$ is analogous.
	Hence, $z_2$ is part of some other segment $S'$.
	First note that $S$ contains at least $16$ vertices, as otherwise either $v_2^3$ is a cut vertex, which is a contradiction, or $S$ contains less than $6$ components and hence must contain a short heavy cycle.
	Hence, $S'$ contains at most $15$ vertices, as otherwise $z_2$ would be a large $\cFour$ cut, a contradiction.
	$S'$ only consists of components that are $\cFour$ or large, as otherwise there must be a short heavy cycle in $S'$ and we can make progress.
	Hence, $S'$ contains at most $4$ components, of which one is $z_2$.
	We can view $K$ as a single $(3, 1)$-merge since $z_3$ acts equivalently as a $C_4$ that is not shortcut.
	Hence, we can apply \Cref{lem:cycle-useful-for-double-merge} to $S$, $S'$ and $K$ and obtain that there is a good merge possible and we can make progress.
	Hence, neither $z_2$ nor $z_1$ is a cut node.

	If neither $z_2$ nor $z_1$ has an outgoing edge to a node of $S$ not on $P$, then $v_2^3$ is a cut node, a contradiction.
	Hence, assume without loss of generality that $z_1$ has an outgoing edge to some node $z$ of $S$ not on $P$.
	If one of $\{ v_1^2, v_3^2, v_1^1, v_3^1 \}$ has an outgoing edge to a node $z'$ of $S$ not on $P$, then we can extend the gluing path $P$ by one and we make progress.
	Hence, we assume this is not the case.
	If also $z_2$ has an outgoing edge to some node $z' \neq z$, then we can apply \Cref{lem:nicepath_branch-at-3-length-2} and we can make progress.
	In particular this case happens if $z_2$ has an outgoing edge and $z_1$ has outgoing edges to at least two distinct nodes of $S$ not on $P$, or if $z_2$ has outgoing edges to at least two distinct nodes of $S$ not on $P$.

	Hence, assume first that $z_2$ has an outgoing edge. If \Cref{lem:nicepath_branch-at-3-length-2} can not be applied, this means that $z_1$ and $z_2$ are only incident to $z$.
	If $z$ is not a $\cFour$, we can interchange $z_1$ with $z$ and obtain another gluing path of length $k-1$, which ends at a node that is not a $\cFour$. But then we can do the same as before where we assumed that $z_1$ is not a $\cFour$ and make progress. Hence, we know that $z$ is a $\cFour$.
	Let $v_1 - v_2 - v_3 - v_4 - v_1$ be the cycle of $z$ and let $e_1$ be the edge from $z_1$ to $z$ and let this edge be incident to $v_1$.
	Further, let $e_2$ be the edge from $z_2$ to $z$. Recall that $e_1$ is incident to either $v_1^1$ or $v_3^1$ and $e_1$ is incident to either $v_1^2$ or $v_3^2$.
	If $e_2$ is incident to either $v_2$ or $v_4$, we obtain a longer gluing path and hence make progress.
	Otherwise, $e_2$ is incident to either $v_1$ or $v_3$.
	Since $P_{2k}$ together with $e_2$ is also a gluing path of length $k-1$, we know that we can make progress, or w.l.o.g. $v_2$ is also incident to $v_2^3$, by applying \Cref{lem:manyC4_1-3-edge} and the previous argumentation in this proof for $z_1$.
	Since we can not apply \Cref{lem:nicepath_branch-at-3-length-2}, none of the nodes of $z_1, z_2, z$ is incident to some other component of $S$ not on $P$. But then $v_2^3$ is a cut vertex, a contradiction.

	Hence, $z_2$ does not have an outgoing edge to some node of $S$ not on $P$, which implies that $z_2$ is only incident to $v_2^3$ and $z_1$.
	Let $y_1, ..., y_j$ be the nodes of $S$ to which $z_1$ is adjacent to.
	Note to each $y_i$ there is a gluing path of length $k-1$ using $P_{3k}$, $e$ and the edge from $z_1$ to $y_i$.
	Hence, similar to before we can argue that each $y_i$ has to be a $\cFour$, as otherwise we can make progress.
	But now we can apply the same argumentation that we did for $z_2$ or $z_1$, respectively, to each $y_i$.
	In particular, each $y_i$ is only incident to $v_2^3$ and $z_1$, as otherwise we either have a longer gluing path, a good merge or we can apply \Cref{lem:nicepath_branch-at-3-length-2} and make progress.
	But then this implies that $v_2^3$ is a cut vertex, a contradiction.
	This finishes the proof.
\end{proof}

\begin{lemma}
	\label{lem:nicepath_branch-at-2-length-1}
	Let $P = z_k - z_{k-1} - ... - z_3 - z_{2} - z_{1a}$ be a gluing path of length $k-1 \geq 4$ and let $P' = z_k - z_{k-1} - ... - z_3 - z_{2} - z_{1b}$ be another gluing path of length $k-1 \geq 4$ such that the first $k-1$ nodes are the same in $P$ and $P'$ but the remaining $2$ nodes are distinct, i.e., $z_{1a} \neq z_{1b}$.
	Then, we can make progress.
\end{lemma}

\begin{proof}
	Let $S$ be the segment containing $P$ and $P'$.
	Throughout the proof we assume we can not make progress and obtain some properties, which we highlight in a bullet point list. We first apply \Cref{lem:manyC4_1-3-edge} to both $z_{1a}$ and $z_{1b}$ to obtain the edges $e_a = v_2^{1b} u_a$ and $e_b = v_2^{1b} u_b$, where $u_a \in \{v_2^3, v_4^3\}$ and $u_b \in \{v_2^3, v_4^3\}$.
	An illustration is given in \Cref{fig:gluing-path-2-2}.

	By \Cref{lem:manyC4_typesz1-3:2} we know that $z_3$ is a $\cFour$, as otherwise we can make progress.
	We next claim that is a $\cFour$.
	Next, assume that $z_2$ is not a $\cFour$.
	By \Cref{lem:manyC4_typesz1-3}, $z_2$ is large or complex. If $z_2$ is a large component, then there is a cycle $z_3 - z_{1a} - z_2 - z_{1b} - z_3$ that is a single $(4, 1)$-merge, a contradiction. Hence, assume that $z_2$ is complex.
	Then either $z_3 - z_2 - z_{1b}$ is a short heavy cycle (if $v_1^2$ or $v_2^2$ do not belong to the same block) or $z_3 - z_{1a} - z_2 - z_{1b} - z_3$ is a single $(4, 1)$-merge (if $v_1^2$ and $v_2^2$ belong to the same block) and hence, in either case, we make progress.
	Furthermore, by \Cref{lem:manyC4_typesz1-3}, $z_{1a}$ and $z_{1b}$ are either a $\cFour$, large or complex. Therefore, we have the first properties:
	\begin{itemize}[nosep]
		\item $z_2$ and $z_3$ are $\cFour$.
		\item $z_{1a}$ and $z_{1b}$ are either a $\cFour$, large or complex.
	\end{itemize}

	\begin{figure}[tb]
		\centering
    \includegraphics[width=0.6\textwidth]{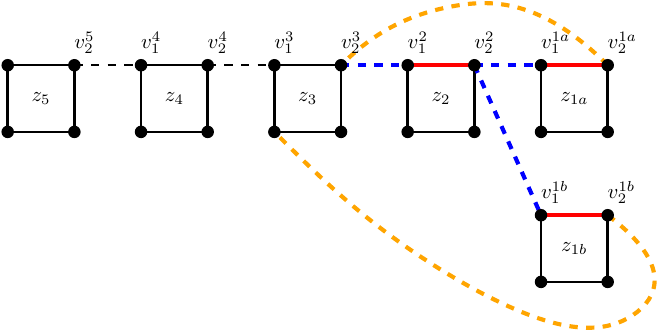}
    \caption{Example for situation in \Cref{lem:nicepath_branch-at-2-length-1}.}
    \label{fig:gluing-path-2-2}
\end{figure}

	Next, we show that neither $z_3$ nor $z_2$ can be a cut node in $\hG_H$. Otherwise, assume w.l.o.g.\ that $z_2$ is a cut node and let $S_2 \neq S$ be another segment of which $z_2$ is part of.
	Note that since $G$ is structured and $z_2$ is a $\cFour$ and a cut node, the number of components in $S_2$ is at most $4$. Otherwise $z_2$ is a large cut $\cFour$ since $S$ contains a huge component by assumption, a contradiction.
	Furthermore, there is no other segment besides $S_2$ containing any of the nodes of $S_2$. Otherwise, either there is a segment containing only two $\cFour$'s (which implies a merge of two $\cFour$'s into one $\cEight$, by the $3$-Matching \Cref{lem:3-matching}), or again $z_2$ is a large cut $\cFour$).
	If $S_2$ contains a complex component, then there exists a cycle of length at most $4$ in $S_2$ that covers at least one bridge of the complex component. Hence, by \Cref{lem:abc-merge}, we can make progress.
	Similarly we can make progress if $S_2$ contains a component that is a $\mathcal{C}_i$ for $5 \leq i \leq 8$.
	Hence, $S_2$ only consists of components that are large or a $\cFour$.
	Note that $P_{13} \cup e_a$ is a single $(3, 1)$-merge containing $z_2$.
	But then by \Cref{lem:cycle-useful-for-double-merge}, we can find a double $(k, j)$-merge containing $z_2$ for some $k \geq 5$ and $k-j \geq 4$.
	Hence, we can make progress according to \Cref{lem:abc-merge}.
	Therefore, we have the following additional property.
	\begin{itemize}[nosep]
		\item Neither $z_3$ nor $z_2$ can be a cut node in $\hG_H$.
	\end{itemize}

	Let $S_1$ be the subset of nodes of $S$ not on $P \setminus \{z_{1a}\}$ that are adjacent to $\{ v_2^2, v_4^2 \}$. 	Note that $z_{1a}, z_{1b} \in S_1$.
	We next show that the set $\{ v_2^2, v_4^2, v_2^3, v_4^3 \}$ separates $S_1$ from $P \setminus S_1$.
	We make a case distinction on where $z_{1a}$ can be incident to. The statement for all other nodes in $S_1$ is analogous.

	First, assume $z_{1a}$ has an outgoing edge $e$ to some node $z \in S$ with $z \notin P \cup S_1$.
	If $e$ is incident to $v_2^{1a}$ or $v_4^{1a}$, we can clearly extend $P$ to a longer gluing path $P''$ using $e$. Otherwise, if $e$ is incident to $v_1^{1a}$ or $v_3^{1a}$, then there are two gluing paths $P_1 = z_k - z_{k-1} - ... - z_3 - z_{1b} - z_{2}$ and $P_2 = z_k - z_{k-1} - ... - z_3 - z_{1a} - z$ of length $k-1$ that satisfy the requirements of \Cref{lem:nicepath_branch-at-3-length-2}, and hence we can make progress.
	Hence, both $z_{1a}$ and $z_{1b}$ are only incident in $S$ to nodes of $P$.

	If $z_{1a}$ is incident to a node $z_i \in P$ with $i \geq 5$, then this edge together with $P_{1ai}$ implies a single $(5, j)$-merge for some $j \leq 2$, and hence by \Cref{lem:abc-merge}, we can make progress.
	Let $C_b$ be the cycle $P'_{1b3} \cup \{ e_b \}$.
	If $z_{1a}$ is incident to $z_4$ via some edge $e$, we consider two cases:
	If $e$ is incident to one of $\{ v_1^{1a}, v_3^{1a} \}$ then define $C_a =  P_{34} \cup \{ e, e_a \}$ and observe that $X = C_a \cup C_b$ defines a double $(5, 1)$-merge. Hence, we can make progress according to \Cref{lem:abc-merge}.
	Otherwise, if $e$ is incident to $\{ v_2^{1a}, v_4^{1a} \}$ then define $C_a =  P_{34} \cup \{ e, P_{12} \}$ and observe that $X = C_a \cup C_b$ defines a double $(5, 1)$-merge. Hence, we can make progress according to \Cref{lem:abc-merge}.

	If $z_{1a}$ is incident to either $v_1^{3}$ or $v_3^{3}$ via some edge $e$, we show that we can make progress:
	If $z_{1a}$ is large then $P_{1a3} \cup e$ is a single $(3, 0)$-merge.
	Assume $z_{1a}$ is complex. $P_{1a3} \cup e_a$ is not a short heavy cycle if the edge of $z_2 z_{1a}$ and the edge $e_a$ are incident to the same block of $z_{1a}$. But then $P_{1a3} \cup e$ is either a short heavy cycle or a single $(3, 0)$-merge.
	Hence, in either case, \Cref{lem:abc-merge} implies that we can make progress.
	Otherwise, $z_{1a}$ is a $\cFour$.
	If $e$ is incident to one of $\{ v_1^{1a}, v_3^{1a} \}$, then either the two edges $e$ and $e_a$ can be used to turn $z_3$ and $z_{1a}$ into one $\cEight$ (if $e$ is incident to $v_1^{1a}$ or $v_3^{1a}$), a contradiction, or there is a single $(3,0)$ merge (if $e$ is incident to $v_2^{1a}$ or $v_4^{1a}$).

	Finally, if $z_{1a}$ is incident to either $v_1^{2}$ or $v_3^{2}$ via some edge $e$, we show that we can make progress and find a longer gluing path $z_k - ... - z_3 - z_{1b} - z_2 - z_{1a}$ using the edges $e_b$ and $e$.

	Hence, we have the following additional property:
	\begin{itemize}[nosep]
		\item No node of $S_1$ has an edge to a node of $S$ not in $P \cup S_1$. Furthermore, the set $\{ v_2^2, v_4^2, v_2^3, v_4^3 \}$ separates $S_1$ from $P \setminus S_1$.
	\end{itemize}

	We next show that $N(v_1^2), N(v_3^2) \subseteq \{ v_2^2, v_4^2, v_2^3, v_4^3 \}$, where $N(v)$ is the set of neighbors in $G$ of some vertex $v \in V(G)$.

	Assume there is an edge $e$ incident to one of $\{ v_3^2, v_1^2 \}$ which is incident to a vertex $u \notin \{ v_2^2, v_4^2, v_2^3, v_4^3 \}$.
	Recall that $z_2$ is not a cut node, hence $u$ is a vertex of $S$.
	If $u$ is part of a node $z$ of $S$ not on $P$ or $P'$, then we have a longer gluing path $ z_k - ... - z_3 - z_{1a} - z_2 - z$ or $ z_k - ... - z_3 - z_{1b} - z_2 - z$. Hence, this is not the case.
	If $u \in \{v_1^{3} , v_3^{3} \}$, then there is a single $(3, 0)$-merge using $P_{1a2}$, $e_a$ and $e$.
	If $u \in V(z_i)$ for some $i \geq 4$, there is a single $(i, 1)$-merge.
	Hence, in any case we can make progress according to \Cref{lem:abc-merge}.

	Finally, assume $e = v_1^2 v_3^2 \in E(G)$.
	W.l.o.g.\ we assume that the edge $e_a^1 = z_2 z_{1a}$ of $P$ is incident to $v_2^2$.
	If the edge $e_b^1 = z_2 z_{1b}$ of $P'$ is incident to $v_2^4$, then there is a single $(4, 1)$-merge using the edges $e_a, e_a^1, e_b^1, e_b$ and we can make progress according to \Cref{lem:abc-merge} (we can replace $z_2$ by $v_2^{2}-v_1^{2}-v_3^{2}-v_4^{2}$).
	Hence, $e_b^1$ has to be incident to $v_2^2$.
	There has to be an edge $e_2^4$ incident to $v_2^4$ and incident to either $z_{1a}$ or $z_{1b}$.
	Otherwise $\{v_2^2, v_2^3, v_4^3\}$ is a large $3$-vertex cut, a contradiction to $G$ being structured.
	W.l.o.g.\ assume that $e_2^4$ is incident to $z_{1b}$.
	If $e_2^4$ is incident to one of $\{ v_2^{1b}, v_4^{1b} \}$ then we can merge the two $\cFour$ $z_2$ and $z_{1b}$ into a single $\cEight$, a contradiction.
	Hence, $e_2^4$ is incident to one of $\{ v_1^{1b}, v_3^{1b} \}$.
	But then there is a single $(4, 1)$-merge using the edges $e_2^4, e_b, e_a, e_a^1$ and using the Hamiltonian Path $v_2^2 - v_3^2 - v_1^2 - v_4^2$ inside $z_2$. Hence, we can make progress according to \Cref{lem:abc-merge}.

	Hence, we have the following additional property:
	\begin{itemize}[nosep]
		\item We have $N(v_1^2), N(v_3^2) \subseteq \{ v_2^2, v_4^2, v_2^3, v_4^3 \}$.
	\end{itemize}

	Let $Y = \{ v_2^2, v_4^2, v_2^3, v_4^3\}$.
	By the above properties, we have that $Y$ is a separator separating $V_1$ from $V_2 \coloneqq V(G) \setminus (V_1 \cup Y)$, where $V_1$ contains all vertices of $V(G) \setminus Y$ that are in some connected component of $G \setminus Y$ containing at least one vertex of $V(S_1)\cup \{ v_1^2, v_3^2 \}$.
	Furthermore, let $V_1' \coloneqq V_1 \setminus  \{ v_1^2, v_3^2 \}$.
	From the condition of this section, we know that $S$ contains at least $5$ nodes and at least one component $L$ that is huge.
	We first assume that this $L$ is contained in $V_2$. We later show how the case that $L$ is contained in $V_1$ is essentially the same.

	Since $Y$ is a $4$-vertex cut and $G$ does not contain large $4$-vertex cuts, we know that $|V_1| \leq 23$, since $|V_2| \geq 32$ as it contains $L$.
	Hence, in $V_1'$ there are at most $5$ components if every component is a $\cFour$ or at most 4 components if one of them is complex or large.
	If $V_1$ contains a complex component, then it is easy to see that there must be a short heavy cycle covering a bridge of that complex component: To see this, observe first that $V_1'$ contains at most $4$ components and $V_1'$ is only connected to components of $V_1'$ and $z_2$ and $z_3$.
	Further, $\{ z_2, z_3 \}$ is a $2$-node cut in $S$ and there is an edge between $z_2$ and $z_3$ in $S$.
	Hence, any simple cycle covering a bridge of the complex component has length at most 6 (and hence is short). Since there is at least one such cycle within  the component graph induced by the components of $V_1 \cup V(z_2) \cup V(z_3)$ covering a bridge, we can make progress by \Cref{lem:abc-merge}.
	Hence, $V_1$ only contains components that are large or $\cFour$.

	We now show that $Y$ is a large $4$-vertex cut (cf.~\Cref{def:large_k_vertex_cut}), a contradiction to $G$ being structured.
	We make a case distinction on whether $V_1'$ contains only $2$ components that are $\cFour$ (in which case $z_{1a}$ and $z_{1b}$ are $\cFour$), or not.
	Note that in the latter case, in which $V_1'$ either contains $3$ components or $2$ components of which at least one is large, we have $|V_1'| \geq 12$.

	\textbf{Case 1:} First consider the latter case in which $|V_1'| \geq 12$.
	Let $\OPT_1'$ be an optimum solution for $G[V_1' \cup Y] | Y$ and let $\OPT_1$ be an optimum solution for $G[V_1 \cup Y] | Y$.
	Note that any optimum solution must pick at least $|\OPT_1'|+4 \geq 17$ edges from within $G[V_1 \cup Y] | Y$ and $|\OPT_1| \geq 17$,
	since $G[V_1]$ has at least three connected components: $V_1'$, $v_1^2$, and $v_3^2$.
	Since $|V_1'| \geq 12$, any feasible solution must pick at least $13$ edges from within $G[V_1' \cup Y]\mid Y$, and two edges incident to $v_1^2$, and $v_3^2$, respectively.
	These are disjoint, and hence this sums up to $|\OPT_1'|+4 \geq 17$ edges, since $|V_1'| \geq 12$.
	Recall that $\{ v_2^2, v_4^2 \} \subseteq N(v_3^2)$ and $\{ v_2^2, v_4^2, v_2^3 \} \subseteq N(v_1^2)$.
	We now let $Z_2 = \{ v_3^2 v_2^2, v_3^2 v_4^2, v_1^2 v_4^2, v_1^2 v_2^3 \}$.
	Note that $Z = OPT_1' \cup Z_2$ contains at most $2$ connected components on the vertex set $V_Z = V_1 \cup Y$, since $\{ v_2^2, v_4^2, v_2^3\}$ is connected via $v_3^2$ and $v_1^2$ and $V_1'$ is connected using the edges of $\OPT_1'$ and also connected to at least one vertex of $Y$.
	Let $C_1$ be the connected component containing $\{ v_2^2, v_4^2, v_2^3\}$ and $C_2$ be the connected component containing $v_4^3$.
	If $C_2 = \{ v_4^3 \}$, there must be an edge from $C_1$ to $C_2$ in $G$ as otherwise $Y \setminus \{ v_4^3 \}$ is a large $3$-vertex cut, a contradiction. If $C_2 \neq \{v_4^3 \}$ there must be an edge from $C_2$ to $C_1$ as otherwise $v_4^3$ is a cut vertex, a contradiction.
	In either case, let $e$ be such an edge.
	Now observe that adding $e$ to $Z$, makes $V_1$ and $Y$ connected.
	Now let $X_2$ be any solution for $(G \setminus V_1) | Y$. Note that in the same way as before (and as in the proof of \Cref{lemma:k-cuts-feasible1}), there is an edge set $Z'$ with $|Z'| \leq 3$ such that $X_2 \cup Z \cup Z' \cup \{e\}$ is 2EC.
	However, since $|Z' \cup \{e\}| \leq 4$, $|Z| \geq 17$ and $|\OPT_1| \geq |Z|+4$, we have $|Z \cup Z' \cup \{e\}| \leq \alpha \OPT_1$, which implies that $Y$ is a large $4$-vertex cut (cf.~second case of \Cref{def:large_k_vertex_cut}), a contradiction.

	\textbf{Case 2:} Now consider the other case, i.e., $z_{1a}$ and $z_{1b}$ are $\cFour$.
	We use the same definitions we used in the previous case.
	The main difference is that we only have $|\OPT_1'| \geq 9$ and hence any optimum solution must pick at least $|\OPT_1'|+4 \geq 13$ edges from within $G[V_1 \cup Y]$.
	We do mostly the same as in the previous case, but need to argue that we can find an edge set $Z$ with $|Z| \leq 13 + 3 = 16$ edges such that $X_2 \cup Z$ is 2EC.
	This then implies that $Y$ is a large $4$-vertex cut (cf.~second case of \Cref{def:large_k_vertex_cut}),, a contradiction, as desired.
	Recall that $e_a^1$ is incident to $v_1^{1a}$ and $z_2$, $e_a$ is incident to $v_2^{1a}$ and $z_3$, $e_b^1$ is incident to $v_1^{1b}$ and $z_2$ and $e_b$ is incident to $v_2^{1b}$ and $z_3$.
	W.l.o.g.\ we assume that $e_a^1$ is incident to $v_2^2$. We make a case distinction on where $e_a, e_b$ and $e_b^1$ are incident to.

	\textbf{Case 2.1:} If $e_a$ and $e_b$ are both incident to $v_2^3$, then set
	$$Z_1 = \Big( E(z_{1a}) \setminus  \{ v_1^{1a} v_2^{1a} \} \Big) \cup \Big( E(z_{1b}) \setminus  \{ v_1^{1b} v_2^{1b} \} \Big) \cup  \{ e_a, e_b, e_a^1, e_b^1, v_3^2 v_2^2, v_3^2 v_4^2, v_1^2 v_4^2, v_1^2 v_2^3 \}$$ and observe that $|Z_1| = 14$ and that all vertices of $(V_1 \cup Y) \setminus v_4^3$ are in a single 2EC component of $Z_1$.
	Similar to before, there must be an edge $e$ incident to $v_4^3$ and $V_1$.
	Since $(V_1 \cup Y) \setminus v_4^3$ is already 2EC, one can easily check that there is an edge set $Z_2$ with $|Z_2| \leq 1$ such that $X_2 \cup Z_1 \cup \{ e \} \cup Z_2$ is 2EC.
	Further, observe that $|Z_1 \cup \{ e \} \cup Z_2 | \leq 16$ as desired.

	\textbf{Case 2.2:} If $e_a$ and $e_b$ are both incident to $v_4^3$, then set
	$$Z_1 = \Big( E(z_{1a}) \setminus  \{ v_1^{1a} v_2^{1a} \} \Big) \cup \Big( E(z_{1b}) \setminus  \{ v_1^{1b} v_2^{1b} \} \Big) \cup  \{ e_a, e_b, e_a^1, e_b^1 \} \cup E(z_2) $$ and observe that $|Z_1| = 14$ and that all vertices of $(V_1 \cup Y) \setminus v_2^3$ are in a single 2EC component of $Z_1$.
	Now following the analogous steps as in Case 2.1, we obtain the result.

	\textbf{Case 2.3:} If neither Case 2.1 nor Case 2.2 applies, set
	$$Z_1 = \Big( E(z_{1a}) \setminus  \{ v_1^{1a} v_2^{1a} \} \Big) \cup\Big( E(z_{1b}) \setminus  \{ v_1^{1b} v_2^{1b} \} \Big) \cup  \{ e_a, e_b, e_a^1, e_b^1, v_3^2 v_2^2, v_3^2 v_4^2, v_1^2 v_4^2, v_1^2 v_2^3 \}$$ and observe that $|Z_1| = 14$ and that all vertices of $Y \setminus \{ v_4^3 \}$ are in a single 2EC component of $Z_1$, and that all vertices of $V_1 \cup Y$ are connected.
	Hence, similar to above, there must be an edge $e$ such that $X_2 \cup Z_1 \cup \{ e \}$ is 2EC.
	Further, observe that $|Z_1 \cup \{ e \}| \leq 15$ as desired.

	It remains to consider the case that $L$ is part of $V_1$. Assume w.l.o.g.\ that $z_{1a}$ is huge.
	Observe that no node of $S_1$ has an edge incident to $z_{1a}$, as if such an edge $e$ exists, we can extend the gluing path $P$ using this edge, unless $L$ is complex and $e_a^1$ and $e$ are incident to the same vertex $v$ of the complex component $L$, and $v$ is not part of any block.
	But then the three edges $P_{23}, e_a^1$ and $e_a$ are a short heavy cycle, since $e_a$ and $e_a^1$ are not incident to the same vertex of $L$ and hence this cycle covers a bridge of~$L$.

	Hence, assume that there is no edge from $L$ to any other node of $S_1$.
	Now observe that $Y$ also separates $L$ from $Y$ and $S' \coloneqq S \setminus (Y \cup \{ L \})$.
	Note that $S'$ contains at least the three nodes $z_{1b}$, $z_4$, and $z_5$, which contain in total at least $12$ vertices and the two vertices $v_1^2$ and $v_3^2$.
	Now we have the same situation as in Case 1 in which $L$ is in $V_2$ and $|V_1| \geq 12$. Hence, we can make progress.
	This finishes the proof of the lemma.
\end{proof}

Next, we show that $z_1$ and $z_2$ have no outgoing edges in $S$, i.e., edges to nodes of $S \setminus P$, where $P$ is the gluing path.

\begin{lemma}
	\label{lem:nicepath_z_1z_2_no-outside}
	Let $P = z_k - z_{k-1} - ... - z_3 - z_{2} - z_{1}$ be a gluing path of length $k-1 \geq 4$. If $z_1$ or $z_2$ has an outgoing edge, i.e., an edge to some node of $S$ not on $P$, then we can make progress.
\end{lemma}

\begin{proof}
Clearly, if there is an edge from $v_2^1$ or $v_4^1$ to some node not on $P$, we can make progress.
By \Cref{lem:manyC4_1-3-edge}, we know that the edge $e_{13}$ has to exist from $v_2^1$ to one of $\{v_2^3, v_4^3\}$.
Hence, if there is an edge from $v_3^1$ or $v_1^1$ to some node $z$ not on $P$, we have two gluing paths $P_1 = z_k - ... - z_3 - z_1 - z$ of length $k-1$ and $P_2 = z_k - ... - z_3 - z_1 - z_2$ of length $k-1$ and hence \Cref{lem:nicepath_branch-at-2-length-1} implies that we can make progress.
Therefore, $z_1$ has no outgoing edge to some node not on $P$.

Similarly, there can not be an edge from $v_1^2$ or $v_3^2$ to some node $z$ not on $P$ since we then have the longer gluing path $z_k - ... - z_3 - z_1 - z_2 - z$.
If there is an edge from $v_2^2$ or $v_4^2$ to some node $z$ not on $P$, then we again have two gluing paths of length $k-1$ that satisfy the conditions of \Cref{lem:nicepath_branch-at-2-length-1} and hence we can make progress.
\end{proof}

\begin{figure}
	\centering
    \includegraphics[width=0.6\textwidth]{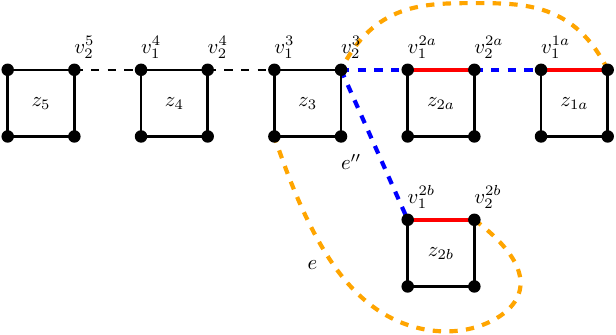}
    \caption{Example for situation in \Cref{lem:nicepath_branch-at-3-length-1}.}
    \label{fig:gluing-path-3-1}
\end{figure}

\begin{lemma}
	\label{lem:nicepath_branch-at-3-length-1}
	Let $P = z_k - z_{k-1} - ... - z_3 - z_{2a} - z_{1a}$ be a gluing path of length $k-1 \geq 4$ and let $P' = z_k - z_{k-1} - ... - z_3 - z_{2b}$ be another gluing path of length $k-2 \geq 3$ such that the first $k-2$ nodes are the same in $P$ and $P'$ but the remaining $3$ nodes are all distinct, i.e., $z_{1a} \neq z_{2b} \neq z_{2a}$.
	Then, we can make progress.
\end{lemma}

\begin{proof}
By \Cref{lem:manyC4_1-3-edge}, we know that the edge $e_{13}$ has to exist from $v_2^{1a}$ to one of $\{v_2^3, v_4^3\}$.
Either $v_2^{2b}$ or $v_4^{2b}$ must have an edge $e$ to a node $z \in S$ with $z \neq z_{2b}$, by the 3-Matching \Cref{lem:3-matching}. W.l.o.g.\ assume this edge is incident to $v_2^{2b}$.
If $z \notin P$, we have two gluing paths of length $k-1$ satisfying the conditions of  \Cref{lem:nicepath_branch-at-3-length-2} and we can make progress.
If $z = z_i$ for some $i \geq 4$, we can make progress as follows:
We take $P_{1a3}$ and $e_{13}$ together with $e$ and $P'_{2ai}$ and observe that this forms a double $(i+1, j)$-merge for some $j \leq 1$ and hence by \Cref{lem:abc-merge}, we can make progress.

Hence, $e$ is incident to $z_3$, since neither $z_1$ nor $z_2$ have an outgoing edge by \Cref{lem:nicepath_z_1z_2_no-outside} (or we can make progress).
If $e$ is incident to either $v_3^3$ or $v_1^3$, then there is a double $(4, 0)$-merge using $P_{1a3}$ and $e_{13}$ together with $P'_{2b3}$ and $e$, and hence, by \Cref{lem:abc-merge}, we can make progress.
Hence, $e$ is incident to either $v_2^3$ or $v_4^3$.
An illustration of this case is given in \Cref{fig:gluing-path-3-1}.

Since $z_{2b}$ can not be only connected to $z_3$ in $S$ (as then $z_{2b}$ is not in the segment $S$ we currently consider), either $v_1^{2b}$ or $v_3^{2b}$ must have an edge $e'$ to some node $z'$.
Again, note that $z_{1a} \neq z' \neq z_{2a}$ by \Cref{lem:nicepath_z_1z_2_no-outside}.

Let $e''$ be the edge from $z_3$ to $z_{2b}$ on $P'$. Observe that we can swap $e$ with $e''$ and we have the same conditions as in this lemma statement but $P'$ ends at $v_2^{2b}$ instead of $v_1^{2b}$.
Let this other path be $P''$.
There must be an edge $f$ to some node of $S$ that is not $z_3$, otherwise $z_{2b}$ is not in $S$.
Since we can interchange $P'$ and $P''$, no matter where $f$ is incident to at $z_{2b}$, we can do the exact same case distinction as for $e$ and observe that in any case we can make progress.
This finishes the proof.
\end{proof}

Next, we argue that a certain edge from $z_1$ to $z_4$ on the gluing path has to exist.

\begin{lemma}
	\label{lem:manyC4_1-4-edge}
	Let $P = z_k - z_{k-1} - ... - z_1$ be a gluing path of length $k$. Then either $v_3^1$ or $v_1^1$ is adjacent to one of $\{v_2^4, v_4^4 \}$ or we can make progress.
\end{lemma}

\begin{proof}
By the previous lemmas, we know that the edge $e_{13}$ from $v_2^1$ to either $v_2^3$ or $v_4^3$ exists.
Further, neither $z_1$ nor $z_2$ has an edge to some node in $S$ but not on $P$.

We first consider $z_1$.
If $z_1$ is incident to an edge $e$ that is incident to $z_i$ on $P$ for $i \geq 5$, then $P_{1i}$ together with $e$ forms a single $(i, j)$-merge for some $j \leq 2$, and hence we can make progress according to \Cref{lem:abc-merge}.
Assume $z_1$ has an edge to $z_4$. Clearly, if $v_3^1$ or $v_1^1$ is adjacent to one of $\{v_2^4, v_4^4 \}$, we are done.
If $v_3^1$ or $v_1^1$ is adjacent to one of $\{v_1^4, v_3^4 \}$ via some edge $e$, then $P_{14}$ together with $e$ forms a single $(4, 1)$-merge and we can make progress according to \Cref{lem:abc-merge}.
If $v_2^1$ or $v_4^1$ is adjacent to $z_4$ via some edge $e$, then $P_{14}$ together with $e$ forms a single $(4, 1)$-merge, and we can make progress according to \Cref{lem:abc-merge}.

Next, consider $z_2$.
If $z_2$ is incident to an edge $e$ that is incident to $z_i$ on $P$ for $i \geq 5$, then $P_{1i}$ together with $e$ and $e_{13}$ forms a double $(i, j)$-merge for some $j \leq 2$, and hence we can make progress according to \Cref{lem:abc-merge}.

Consider the case that $z_2$ has an edge $e$ to $z_4$.
Assume $e$ is incident to $v_2^2$ or $v_2^4$.
If $e$ is incident to one of $\{v_2^4, v_4^4 \}$, then we can rearrange $P$ to obtain the desired edge:
Consider $P' = z_k - ... - z_4 - z_3 - z_1 - z_2$ using the edge $e_{13}$ from $z_3$ to $z_1$.
Observe that $P'$ and $e$ satisfy the lemma statement.
If $e$ is incident to one of $\{v_2^3, v_4^1 \}$, then $P_{24}$ together with $e$ form a single $(3, 0)$-merge, and we can make progress according to \Cref{lem:abc-merge}.
Next, assume $e$ is incident to $v_1^2$ or $v_3^2$.
If $e$ is incident to one of $\{v_1^4, v_3^4 \}$ via some edge $e$, then $P_{34}$, $e_{13}$, $P_{12}$ and $e$ together form a single $(4, 1)$-merge and we can make progress according to \Cref{lem:abc-merge}.

Hence, if neither of the above conditions is met, then $z_3$ is a cut node, a contradiction to $P$ being a gluing path on the $2$-node-connected segment $S$ containing $z_1$ and $z_2$.
\end{proof}

Hence, from now on we assume that $v_3^1$ is adjacent to one of $\{v_2^4, v_4^4 \}$.

\begin{lemma}
	\label{lem:nicepath_no-outside-from-1-2-3}
	Let $P = z_k - z_{k-1} - ... - z_3 - z_{2} - z_{1}$ be a gluing path of length $k-1 \geq 4$.
	Then either we can make progress or in $G$ no vertex of the nodes $z_1, z_2, z_3$ has an edge to a node $z$ not on~$P$.
\end{lemma}

\begin{proof}
Let $e_{13}$ and $e_{14}$ be the edges implied by  \Cref{lem:manyC4_1-3-edge} and \Cref{lem:manyC4_1-4-edge}.
By \Cref{lem:nicepath_z_1z_2_no-outside}, we know that $z_1$ and $z_2$ are only adjacent to nodes of $P$ in $S$.
Furthermore, neither $v_2^3$ nor $v_4^3$ has an edge to a node not on $P$, as then we can make progress according to \Cref{lem:nicepath_branch-at-3-length-1}.
Hence, in the remaining case, assume $v_1^3$ or $v_3^3$ has an edge $e_{3b}$ to a node $z_b$.
Let $v_1^b - v_2^b - v_3^b - v_4^b - v_1^b$ be the $4$-cycle if $z_b$ is a $\cFour$, and otherwise those 4 vertices are simply $4$ distinct vertices of $z_b$.
Assume w.l.o.g.\ that $e_{3b}$ is incident to $v_1^b$.

By the $3$-Matching Lemma, either $v_2^b$ or $v_4^b$ has an edge $e$ to a node $z' \neq z_b$ of $S$.
If $z'$ is not on $P$, then we can extend the gluing path and obtain $P' = z_k - ... - z_4 - z_1 - z_3 - z_b - z$ using the edge $e_{14}$ and $e_{13}$.
If $z' \in \{z_1, z_2 \}$ then $P_{13}$ and $e_{13}$ together with $e_{3b}$ and $e$ form a double $(4, 0)$-merge and hence, by \Cref{lem:abc-merge}, we can make progress.
If $z'$ is a node of $P$ such that $i \geq 5$, then $P_{13}$ and $e_{13}$ together with $e_{3b}$, $e$ and $P_{3i}$ forms a double $(i+1, j)$-merge for some $j \leq 2$. Hence, by \Cref{lem:abc-merge}, we can make progress.
If $z'$ is $z_4$, consider the following two cases:
Either $e$ is incident to $v_2^4$ or $v_4^4$, then we can extend the gluing path to $P' = z_k - ... - z_4 - z_b - z_3 - z_2 - z_1$ using the edge $e$ and $e_{3b}$ incident to $z_b$.
If $e$ is incident to $v_3^4$ or $v_1^4$, then $P_{13}$ and $e_{13}$ together with $P_{34}$, $e_{3b}$ and $e$ forms a double $(5, 1)$-merge. Hence, by \Cref{lem:abc-merge}, we can make progress.

Therefore, $z' = z_3$, i.e., $e$ is incident to $z_3$.
Observe that if $e$ is incident to either $v_2^3$ or $v_4^3$, we can make progress according to \Cref{lem:nicepath_branch-at-3-length-1}.
Hence, $e$ is incident to either $v_1^3$ or $v_3^3$.
Moreover, $z_b$ must have an edge $f$ to some other node of $S$ which is not $z_3$, otherwise $z_b$ is not in $S$.
Now observe that we can interchange the roles of $e$ and $e_{3b}$.
Hence, we can apply the same case distinction for $f$ as for $e$ and observe that in any case we make progress.
This finishes the proof of the lemma.
\end{proof}

Hence, from now on we assume that in $G$ the vertices of nodes $z_1, z_2, z_3$ all cannot have an edge to a node $z$ that is not on $P$.

\begin{lemma}
	\label{lem:nicepath_remaining-case}
	Let $P = z_k - z_{k-1} - ... - z_3 - z_{2} - z_{1}$ be a gluing path of length $k \geq 4$ such that in $G$ the vertices of nodes $z_1, z_2, z_3$ are all only adjacent to vertices of nodes of $P$.
	Then we can make progress.
\end{lemma}

\begin{proof}
Let $e_{13}$ and $e_{14}$ be the edges implied by  \Cref{lem:manyC4_1-3-edge} and \Cref{lem:manyC4_1-4-edge}.
If $z_1$ is incident to an edge $e$ incident to some node $z_i$ for $i \geq 5$, we can make progress as $P_{1i}$ together with $e$ forms a single $(i, j)$-merge for some $j \leq 2$.
If $z_2$ is incident to an edge $e$ incident to some node $z_i$ for $i \geq 5$, we can make progress as $P_{1i}$ together with $e$ and $e_{13}$ forms a double $(i, j)$-merge for some $j \leq 1$.
Hence, in both cases we can make progress according to \Cref{lem:abc-merge}.
Thus, $z_1$ and $z_2$ are only adjacent to $z_3$ and $z_4$ in $S$.

Next, assume $z_3$ is incident to an edge $e$ incident to some node $z_i$ for $i \geq 5$.
If $e$ is incident to either $v_2^3$ or $v_4^3$, then $P_{13}$ and $e_{13}$ together with $P_{4i}$, $e$ and $e_{14}$ form a double $(i, j)$-merge for some $j \leq 1$.
If $e$ is incident to either $v_1^3$ or $v_3^3$, then $P_{4i}, e, e_{13}, e_{14}$ forms a single $(i-1, j)$-merge for some $j \leq 1$.
Hence, in both cases we can make progress according to \Cref{lem:abc-merge}.

Now observe that we can make progress if one of $z_1, z_2, z_3$ is connected to some $z_i, i \geq 5$.
Hence, assume that this is not the case.
Since $z_1, z_2, z_3$ have no edge to a node of $S$ not on $P$, this implies that $z_4$ is a cut node.
But then $z_5$ and $z_3, z_2, z_1$ are not in the same segment $S$, a contradiction.
This finishes the proof.
\end{proof}

Finally, using the above lemmas we can easily prove \Cref{lem:many_C4_progress}, which in turn proves \Cref{lem:core:many-components}.

\begin{proof}[Proof of \Cref{lem:many_C4_progress}]
By \Cref{lem:nicepath_start3} we can compute in polynomial time a gluing path $P = z_k - z_{k-1} - ... - z_2 - z_1$ of length $k-1 \geq 4$, unless we can make progress.
\Cref{lem:nicepath_no-outside-from-1-2-3} implies that the nodes $z_1$, $z_2$, and $z_3$ all have no edge to a node of $S$ that is not on $P$, unless we can make progress.
Now \Cref{lem:nicepath_remaining-case} implies that we can make progress, which proves the lemma.
\end{proof}

\newpage

\section*{Appendix}

\appendix

\section{Reduction to Structured Graphs}\label{app:prep}

This section is dedicated to the proof of \Cref{lem:reduction-to-structured}, which we first restate.

\lemmaReduction*

We assume throughout this section that $\alpha \in [\nicefrac{6}{5},\nicefrac{5}{4}]$ and $\varepsilon \in (0, \nicefrac{1}{100}]$.
Let $\ALG$ be a polynomial-time algorithm that, given an $(\alpha,\varepsilon)$-structured graph $G' = (V',E')$, returns an $\alpha$-approximate 2EC spanning subgraph $\ALG(G') \subseteq E'$ of $G'$.
In the following, we define a reduction algorithm which recursively modifies and partitions the input graph $G$ until we obtain $(\alpha,\varepsilon)$-structured graphs, then solves these with $\ALG$, and finally puts the resulting solutions together to an overall solution for $G$.
In particular, our reduction checks for structures that are forbidden in $(\alpha,\varepsilon)$-structured graphs. These are graphs of constant size (w.r.t.\ $\varepsilon$),
parallel edges and self loops,
1-vertex cuts, small $\alpha$-contractible subgraphs (w.r.t.\ $\varepsilon$), irrelevant edges, non-isolating 2-vertex cuts, large 3-vertex cuts, large $\mathcal C_k$ cuts for $k \in \{4,\ldots,8\}$, and large $4$-vertex cuts (cf.\ \Cref{def:structured}). We first restate definitions of some of these structures from~\cite{GargGA23improved}, and then formally define $\mathcal C_k$ cuts and large $4$-vertex cuts.

\begin{definition}[$\alpha$-contractible subgraph~\cite{GargGA23improved}]
	Let $\alpha \geq 1$ be a fixed constant. A 2EC subgraph $C$ of a 2EC graph $G$ is $\alpha$-\emph{contractible} if every 2ECSS of $G$ contains at least $\frac{1}{\alpha}\abs{E(C)}$ edges with both endpoints in $V(C)$.
\end{definition}

\begin{definition}[irrelevant edge~\cite{GargGA23improved}]
	Given a graph $G$ and an edge $e = uv \in E(G)$, we say that $e$ is \emph{irrelevant} if $\{ u, v\}$ is a 2-vertex cut of $G$.
\end{definition}

\begin{definition}[non-isolating 2-vertex cut~\cite{GargGA23improved}]
	Given a graph $G$, and a 2-vertex cut $\{u, v\}$ of $G$, we say that $\{u, v\}$ is \emph{isolating} if $G \setminus \{u, v\}$ has exactly two connected components, one of which is of size 1.
	Otherwise, it is \emph{non-isolating}.
\end{definition}

\begin{definition}[large 3-vertex cuts~\cite{GHL24}]\label{def:large_three_cut}
	Given a graph $G$ and a 3-vertex cut $\{u, v, w\}$ of $G$, we say that $\{u, v, w\}$ is \emph{large} if there is a partition $(V_1, V_2)$ of $V \setminus \{u, v, w\}$ with $E[V_1,V_2] = \emptyset$
	such that $|V_1|,|V_2| \geq \ceil{\frac{2}{\alpha - 1}} - 1$.
	Otherwise, we call it \emph{small}.
\end{definition}

Let $\mathcal C_k$ be a cycle on $k$ vertices $\{v_1,\ldots,v_k\}$ in a graph $G$.
We say that $\{v_1,\ldots,v_k\}$ is a \emph{$\mathcal C_k$ cut} if $G \setminus \{v_1,\ldots,v_k\}$ has strictly more connected components than $G$.

\begin{definition}[large $\mathcal{C}_k$ cuts]\label{def:large_cycle_cut}
	Let $k \in \{4,5,6,7,8\}$.
	Let $G = (V,E)$ be a graph and $\{v_1,\ldots,v_k\}$ be a $\mathcal C_k$ cut of $G$. We call $\{v_1,\ldots,v_k\}$ \emph{large} if there is a partition $(V_1,V_2)$ of $V \setminus \{v_1,\ldots,v_k\}$ with $E[V_1,V_2] = \emptyset$ and $|V_1|,|V_2| \geq \ceil{\frac{k}{\alpha - 1}} - 1$.
	Otherwise, we call it \emph{small}.
\end{definition}

\begin{definition}[large $4$-vertex cuts]\label{def:large_k_vertex_cut}
	Let $G = (V,E)$ be a graph and $\{v_1,v_2,v_3,v_4\}$ be a $4$-vertex cut of $G$.
	We call $\{v_1,v_2,v_3,v_4\}$ \emph{large} if
	\begin{enumerate}
		\item there is a partition $(V_1,V_2)$ of $V \setminus \{v_1,v_2,v_3,v_4\}$ with $E[V_1,V_2] = \emptyset$ and $\ceil{\frac{6}{\alpha - 1}} - 1 \leq |V_1|, |V_2|$, or
		\item there is a partition $(V_1,V_2)$ of $V \setminus \{v_1,v_2,v_3,v_4\}$ with $E[V_1,V_2] = \emptyset$, $G_i \coloneq G[V_i \cup \{v_1,v_2,v_3,v_4\}]$ and $G'_i := G_i | \{v_1,v_2,v_3,v_4\}$ for $i \in \{1,2\}$, such that
			  for any 2ECSS $H_2$ for $G'_2$
		      there exists
		      a set of edges $Z \subseteq E$
		      such that (i) $|Z| \leq \alpha \cdot \opt(G'_1)$ and (ii) $H_2 \cup Z$ is a 2ECSS of $G$.
	\end{enumerate}
	Otherwise, we call it \emph{small}.
\end{definition}

For all structures except large 3-vertex cuts, large $\mathcal C_k$ cuts, and large $k$-vertex cuts, we can use the reduction presented in~\cite{GargGA23improved}. The main contribution of this section is to show how to reduce to graphs without large 3-vertex cuts, large $\mathcal C_k$ cuts and large $k$-vertex cuts, given that all other substructures have already been removed. For large 3-vertex cuts, we mainly follows the approach of \cite{GHL24} but parameterize the required size of the cut by $\alpha$ instead of assuming $\alpha \geq \nicefrac54$ as in \cite{GHL24}.

This section is structured as follows. In \cref{sec:reduce-defs} we first introduce definitions to classify solutions regarding 2-vertex cuts and 3-vertex cuts, and then give the formal description of the reduction algorithm, which is split into four algorithm environments, where \cref{alg:reduce} is the main algorithm, which calls \cref{alg:2-vertex-cuts,alg:3-vertex-cuts,alg:cycle-cuts,alg:k-cuts} as subroutines.
In \cref{sec:reduce-auxiliary}, we give auxiliary results, mainly regarding the removal of 2-vertex cuts, 3-vertex cuts, large $\mathcal C_k$ cuts, and large $4$-vertex cuts. Finally, in \cref{sec:reduce-proofs} we combine these auxiliary results together and give the main proof of \cref{lem:reduction-to-structured}.

\subsection{Definitions and the Algorithm}\label{sec:reduce-defs}

We will heavily use the following two facts on 2EC graphs with contracted components. We will not refer to them explicitly.

\begin{fact}\label{fact:contract-2ec}
	Let $G$ be a 2EC graph and $W \subseteq V(G)$. Then $G | W$ is 2EC.
\end{fact}

\begin{fact}\label{fact:decontract}
	Let $H$ be a 2EC subgraph of a 2EC graph $G$, and $S$ and $S'$ be 2EC spanning subgraphs of $H$ and $G|H$, respectively. Then $S \cup S'$ is a 2EC spanning subgraph of $G$, where we interpret $S'$ as edges of $G$ after decontracting $V(H)$.
\end{fact}

\subsubsection{Definitions for Non-Isolating 2-Vertex Cuts}

Let $\{u,v\}$ be a $2$-vertex cut in a graph $G$.
Note that that we can find a partition $(V_1,V_2)$ of $V \setminus \{u,v\}$ such that $V_1 \neq \emptyset \neq V_2$ and there are no edges between $V_1$ and $V_2$. Once we delete $\{u,v\}$ from $G$, the graph breaks into various components such that each is either part of $V_1$ or $V_2$. Similarly, $\OPT(G)$ breaks into various components once we remove $\{u,v\}$ such that each component is either part of $V_1$ or $V_2$. In the following definition, we consider different types of graphs that match $H \coloneq \OPT(G)[V_1 \cup \{u,v\}]$ (or analogously $\OPT(G)[V_2 \cup \{u,v\}]$).

\begin{definition}[2-vertex cut solution types]\label{def:2vc-types}
	Let $H$ be a graph and $u,v \in V(H)$. Let $H'$ be obtained by contracting each 2EC component $C$ of $H$ into a single super-node $C$, and let $C(x)$ be the corresponding super-node of the 2EC component of $H$ that contains $x \in V(H)$. We define the following types of $H'$ with respect to $\{u,v\}$:
	\begin{description}
		\item[Type $\typA$:] $H'$ is composed of a single super-node $H' = C(u) = C(v)$.
		\item[Type $\typB$:] $H'$ is a $C(u) - C(v)$ path of length at least $1$.
		\item[Type $\typC$:] $H'$ consists of two isolated super-nodes $C(u)$ and $C(v)$.
	\end{description}
	Let $\typesTwoVC \coloneq \{\typA,\typB,\typC \}$ be the set of 2-vertex cut optimal solution types.
\end{definition}

\begin{definition}[Order on types]\label{def:2vc-ties}
	$\typA \succ \typB \succ \typC$.
\end{definition}

\subsubsection{Definitions for Large 3-Vertex Cuts}

Similarly to 2-vertex cuts and \Cref{def:2vc-types}, we give in the following definition types that match the subgraph of an optimal solution that is induced by the partition corresponding to a large $3$-vertex cut $\{u,v,w\}$ in $G$. Naturally, here are more patterns that can occur compared to 2-vertex cuts.

\begin{definition}[3-vertex cut solution types]
	Let $H$ be a graph and $u,v,w \in V(H)$. Let $H'$ be obtained by contracting each 2EC component $C$ of $H$ into a single super-node $C$, and let $C(x)$ be the corresponding super-node of the 2EC component of $H$ that contains $x \in V(H)$. We define the following types of $H'$ with respect to $\{u,v,w\}$:
	\begin{description}
		\item[Type $\typA$:] $H'$ is composed of a single super-node $H' = C(u) = C(v) = C(w)$.
		\item[Type $\typBi$:] $H'$ is a $C_1 - C_k$ path of length at least $1$ such that $u,v,w \in C_1 \cup C_k$ and $\abs{\{u,v,w\} \cap C_i} \leq 2$ for $i \in \{1,k\}$. We assume w.l.o.g.\ that $C_1 = C(u) = C(v)$.
		\item[Type $\typBii$:] $H'$ is composed of two isolated super-nodes $C_1$ and $C_2$, where $\abs{\{u,v,w\} \cap C_i} \leq 2$ for $i \in \{1,2\}$. We assume w.l.o.g.\ that $C_1 = C(u) = C(v)$.
		\item[Type $\typCi$:] $H'$ is a tree and $u$, $v$, and $w$ are in distinct super-nodes. %
		\item[Type $\typCii$:] $H'$ is composed of a $C_1-C_k$ path of length at least $1$ and an isolated super-node $C_\ell$ such that $\abs{C_i \cap \{u,v,w\}} = 1$ for $i \in \{1,k,\ell\}$.
		\item[Type $\typCiii$:] $H'$ is composed of the three isolated super-nodes $C(u), C(v)$ and $C(w)$.
	\end{description}
	Let $\typesThreeVC \coloneq \{\typA,\typBi,\typBii,\typCi,\typCii,\typCiii\}$ be the set of 3-vertex cut optimal solution types.
\end{definition}

\begin{definition}[Order on types]\label{def:3vc-ties}
	$\typA \succ \typBi \succ \typBii \succ \typCi \succ \typCii \succ \typCiii$.
\end{definition}

The following proposition lists compatible solution types of both sides of a 3-vertex cut $\{u,v,w\}$. The proof is straightforward and, thus, omitted.
For example, if $\OPT_1$ (using the notation introduced below) is of type $\typCiii$, it is composed of three distinct 2EC components $C(u)$, $C(v)$, and $C(w)$. Therefore, $u$, $v$, and $w$ must be in the same 2EC connected component in $\OPT_2$ as otherwise $\OPT$ cannot be a 2ECSS, and thus, $\OPT_2$ must be of type $\typA$.

\begin{proposition}[feasible type combinations]\label{lemma:3vc-type-combinations}
	Let $\{u,v,w\}$ be a large 3-vertex cut in a graph $G$, and $(V_1,V_2)$ be a partition of $V \setminus \{u,v,w\}$ such that $\ceil{\frac{2}{\alpha - 1}} - 1 \leq \abs{V_1} \leq \abs{V_2}$ and there are no edges between $V_1$ and $V_2$ in $G$.
	Let $G_1 = G[V_1 \cup \{u,v,w\}]$, $G_2 = G[V_2 \cup \{u,v,w\}] \setminus \{ uv, vw, uw \}$ and, for a fixed optimal solution $\OPT(G)$, let $\OPT_i = \OPT(G) \cap E(G_i)$ for $i \in \{1,2\}$. Then, with respect to $\{u,v,w\}$, if
	\begin{enumerate}[label=(\alph*)]
		\item $\OPT_1$ is of type $\typA$, $\OPT_2$ must be of type $t \in \{\typA,\typBi,\typBii,\typCi,\typCii,\typCiii\} = \typesThreeVC$.
		\item $\OPT_1$ is of type $\typBi$, $\OPT_2$ must be of type $t \in \{\typA,\typBi,\typBii,\typCi,\typCii\}$.
		\item $\OPT_1$ is of type $\typBii$, $\OPT_2$ must be of type $t \in \{\typA,\typBi,\typBii\}$.
		\item $\OPT_1$ is of type $\typCi$, $\OPT_2$ must be of type $t \in \{\typA,\typBi,\typCi\}$.
		\item $\OPT_1$ is of type $\typCii$, $\OPT_2$ must be of type $t \in \{\typA,\typBi\}$.
		\item $\OPT_1$ is of type $\typCiii$, $\OPT_2$ must be of type $\typA$.
	\end{enumerate}
\end{proposition}

Observe that according to the above definition we have $\abs{\OPT(G)} = \abs{\OPT_1} + \abs{\OPT_2}$.

\subsubsection{The Full Algorithm}

We are now ready to state our main algorithm $\Reduce$ (\Cref{alg:reduce}).
Let $\reduce(G) \coloneq |\Reduce(G)|$ denote the size of the solution computed by $\Reduce$ for a given 2ECSS instance $G$.
The reductions applied in \Cref{reduce:1vc,reduce:loop,reduce:contractible,reduce:irrelevant,reduce:2vc} (here and in the following, we do not distinguish between the line where the actual reduction is computed and the line in which the condition for the reduction is checked) are identical to the reductions given in \cite{GargGA23improved}.
Specifically, \Cref{alg:2-vertex-cuts}, which is called in \cref{reduce:2vc}, reduces  non-isolating 2-vertex cuts and was first given in \cite{GargGA23improved}. We emphasize that we present it slightly different, but without changing its actual behavior.
When \Cref{alg:reduce} reaches \Cref{reduce:3vc} it calls \Cref{alg:3-vertex-cuts} to handle a large 3-vertex cut $\{u,v,w\}$~\cite{GHL24}, when it reaches \Cref{reduce:c4cut} it calls \Cref{alg:cycle-cuts} to handle a large $\mathcal C_k$ cut $\{v_1,\ldots,v_k\} $ for $k \in \{4,5,6,7,8\}$ (that is, it checks for $k=4$, then for $k=5$, and so on), and when it reaches \Cref{reduce:k-cut} it calls \Cref{alg:k-cuts} to handle a large $4$-vertex cut $\{v_1,v_2,v_3,v_4\}$.

Whenever the reduction returns edges $K$ of contracted graphs $G|H$, i.e., $K=\Reduce(G|H)$, we interpret these as the corresponding edges after decontracting $H$, e.g., in \Cref{alg:reduce} in \Cref{reduce:contractible} or in \Cref{alg:3-vertex-cuts} in \Cref{reduce-3vc:both-large:recurse}.

Hence, from now on let us focus on \Cref{alg:3-vertex-cuts}.
In \Cref{app:prep:helper-3vc} we prove that the claimed sets of edges $F$ in \Cref{alg:3-vertex-cuts} exist and are of constant size.
The computations done in \Cref{reduce-3vc:compute-opt-types} (and in \cref{reduce-2vc:compute-opt-types} in \cref{alg:2-vertex-cuts}) can be done as explained in \cite{GargGA23improved}.
That is, we first compute via enumeration minimum-size subgraphs $\overline{\OPT}_1^t$ of $G_1$ of every type $t \in \typesThreeVC$ if such a solution exists, because $G_1$ is of constant size in this case.
If $\overline{\OPT}_1^t$ does not exist, $\OPT_1^t$ also cannot exist.
Otherwise, that is, $\overline{\OPT}_1^t$ exists, we check whether $G_2$ is a graph that admits a solution that is compatible with $\overline{\OPT}_1^t$ according to \Cref{lemma:3vc-type-combinations}, and, if that is the case, we set $\OPT_1^t \coloneq \overline{\OPT}_1^t$.

Finally, note that if the algorithm reaches \Cref{reduce:call-alg} in \Cref{alg:reduce}, then $G$ must be $(\alpha,\varepsilon)$-structured.
To show this, we will verify the following invariant via induction (\Cref{lem:reduction-to-structured-apx}) for all 2ECSS instances~$G$:
\begin{equation}\label{eq:reduce-invariant}
	\reduce(G) \leq \begin{cases}
		\opt(G)                     & \quad \text{if } \abs{V(G)} \leq \frac{8}{\varepsilon} \ , \text{ and} \\
		\alpha \cdot \opt(G) + f(G) & \quad \text{if } \abs{V(G)} > \frac{8}{\varepsilon} \ ,
	\end{cases}
\end{equation}
where $f(G) \coloneq 4\varepsilon \cdot \abs{V(G)} - 16$.

\begin{algorithm}[tb]
	\small
	\DontPrintSemicolon
	\caption{ $\Reduce(G)$: reduction to $(\alpha,\varepsilon)$-structured graphs}\label{alg:reduce}
	\KwIn{2EC graph $G=(V,E)$.}
	\If{$|V(G)| \leq \frac{8}{\varepsilon}$}{
		compute $\OPT(G)$ via enumeration and \Return $\OPT(G)$\label{reduce:bruteforce}
	}
	\If{$G$ has a 1-vertex cut $\{v\}$\label{reduce:1vc}}{
		let $(V_1,V_2)$, be a partition of $V \setminus \{v\}$ such that $V_1 \neq \emptyset \neq V_2$ and there are no edges between $V_1$ and~$V_2$. \;
		\Return $\Reduce(G[V_1 \cup \{v\}]) \cup \Reduce(G[V_2 \cup \{v\}])$. \;
	}
	\If{$G$ contains a self loop or a parallel edge $e$\label{reduce:loop}}{
		\Return $\Reduce(G \setminus \{e\})$. \;
	}
	\If{$G$ contains an $\alpha$-contractible subgraph $H$ with $|V(H)| \leq \frac{8}{\varepsilon}$\label{reduce:contractible}}{
		\Return $H \cup \Reduce(G \mid H)$.
	}
	\If{$G$ contains an irrelevant edge $e$\label{reduce:irrelevant}}{
		\Return $\Reduce(G \setminus \{e\})$. \;
	}
	\If{$G$ contains a non-isolating 2-vertex cut $\{u,v\}$\label{reduce:2vc}}{
		Execute \Cref{alg:2-vertex-cuts} for $\{u,v\}$. \;
	}
	\If{$G$ contains a large 3-vertex cut $\{u,v,w\}$\label{reduce:3vc}}{
		Execute \Cref{alg:3-vertex-cuts} for $\{u,v,w\}$.
	}
	\If{$G$ contains a large $\mathcal C_k$ cut $\{v_1,\ldots,v_k\}$ for $k \in \{4,5,6,7,8\}$ \label{reduce:c4cut}}{
		Execute \Cref{alg:cycle-cuts} for $\{v_1,\ldots,v_k\}$.
	}
	\If{$G$ contains a large $k$-vertex cut $\{v_1,\ldots,v_k\}$ for $k = 4$ \label{reduce:k-cut}}{
		Execute \Cref{alg:k-cuts} for $\{v_1,\ldots,v_k\}$ and $k=4$.
	}
	\Return $\ALG(G)$.\label{reduce:call-alg}
\end{algorithm}

\begin{algorithm}
	\small
	\DontPrintSemicolon
	\caption{Remove a non-isolating 2-vertex cut \cite{GargGA23improved}}\label{alg:2-vertex-cuts}
	\KwIn{A 2EC graph $G=(V,E)$ without cut vertices and without $\alpha$-contractible subgraphs with at most $\frac{8}{\varepsilon}$ vertices. A non-isolating 2-vertex cut $\{u,v\}$ in $G$.}
	Let $(V_1,V_2)$ be a partition of $V \setminus \{u,v\}$ such that $2 \leq \abs{V_1} \leq \abs{V_2}$ and there are no edges between $V_1$ and $V_2$ in $G$. \;
	Let $G_1 \coloneq G[V_1 \cup \{u,v\}]$ and $G_2 \coloneq G[V_2 \cup \{u,v\}] \setminus \{uv\}$.\;
	Let $G'_i \coloneq G_i | \{u,v\}$ for $i \in \{1,2\}$.\;
	\If{$\abs{V_1} > \frac{8}{\varepsilon}$}{
		\label{reduce-2vc:both-large}
		Let $H'_i \coloneq \Reduce(G'_i)$ for $i \in \{1,2\}$.\label{reduce-2vc:both-large:recurse} \;
		Let $F' \subseteq E$ be a minimum-size edge set such that $H' \coloneq H'_1 \cup H'_2 \cup F'$ is 2EC.
		\Return $H'$\;
	}
	\Else{\label{reduce-2vc:one-small}
	Let $\OPT_1^t$ be the minimum-size subgraphs of $G_1$ for every type $t \in \typesTwoVC$ (w.r.t.\ $\{u,v\}$ and $G_1$) that belong to some 2EC spanning subgraph of $G$, if exists. Let $\OPT_1^{\min}$ be the existing subgraph of minimum size among all types in $\typesTwoVC$, where ties are broken according to \Cref{def:2vc-ties} (that is, $t' \in \typesTwoVC$ is preferred over $t \in \typesTwoVC$ if $t' \succ t$), and $t^{\min} \in \typesTwoVC$ the corresponding type. Let $\opt_1^t \coloneq |\OPT_1^t|$ for every $t \in \typesTwoVC$ and $\opt_1^{\min} \coloneq |\OPT_1^{\min}|$. \label{reduce-2vc:compute-opt-types} \;
	\If{$t^{\min} = \typB$}{\label{reduce-2vc:b}
	Let $G_2^{\typB} \coloneq (V(G_2) \cup \{w\},E(G_2) \cup \{uw,vw\})$ and $H_2^{\typB} \coloneq \Reduce(G_2^{\typB}) \setminus \{uw,vw\}$, where $w$ is a dummy vertex and $uw, vw$ are dummy edges.  \;
	$H^{\typB} \coloneq \OPT_1^{\typB} \cup H_2^{\typB}$
	\Return $H^{\typB}$. \;
	}
	\ElseIf{$t^{\min} = \typC$}{\label{reduce-2vc:c}
	Let $G_2^{\typC} \coloneq (V(G_2),E(G_2) \cup \{uv\})$ and $H_2^{\typC} \coloneq \Reduce(G_2^{\typC}) \setminus \{uv\}$.  \;
	Let $F^{\typC} \subseteq E
	$ be a min-size edge set s.t.\ $H^{\typC} \coloneq \OPT_1^{\typC} \cup H_2^{\typC} \cup F^{\typC}$ is 2EC.
	\Return $H^{\typC}$. \;
	}
	}
\end{algorithm}

\begin{algorithm}
	\small
	\DontPrintSemicolon
	\KwIn{A 2EC graph $G=(V,E)$ without cut vertices, non-isolating 2-vertex cuts, or $\alpha$-contractible subgraphs with at most $\frac{8}{\varepsilon}$ vertices. A large 3-vertex cut $\{u,v,w\}$ in $G$.}
	Let $(V_1,V_2)$ be a partition of $V \setminus \{u,v,w\}$ such that $\ceil{\frac{2}{\alpha - 1}} - 1 \leq \abs{V_1} \leq \abs{V_2}$ and there are no edges between $V_1$ and $V_2$ in $G$. \;
	Let $G_1 \coloneq G[V_1 \cup \{u,v,w\}]$ and $G_2 \coloneq G[V_2 \cup \{u,v,w\}] \setminus \{uv, vw, uw \}$.\;
	Let $G'_i \coloneq G_i | \{u,v,w\}$ for $i \in \{1,2\}$.\;
	\If{$\abs{V_1} > \frac{8}{\varepsilon}$}{
		\label{reduce-3vc:both-large}
		Let $H'_i \coloneq \Reduce(G'_i)$ for $i \in \{1,2\}$. \label{reduce-3vc:both-large:recurse} \;
		Let $F' \subseteq E$ be a minimum-size edge set such that $H' \coloneq H'_1 \cup H'_2 \cup F'$ is 2EC.
		\Return $H'$\;
	}
	\Else{
	\label{reduce-3vc:one-small}
	Let $\OPT_1^t$ be the minimum-size subgraphs of $G_1$ for every type $t \in \typesThreeVC$ (w.r.t.\ $\{u,v,w\}$ and $G_1$) that belong to some 2EC spanning subgraph of $G$, if exists. Let $\OPT_1^{\min}$ be the existing subgraph of minimum size among all types in $\typesThreeVC$, where ties are broken according to \Cref{def:3vc-ties} (that is, $t' \in \typesThreeVC$ is preferred over $t \in \typesThreeVC$ if $t' \succ t$), and $t^{\min} \in \typesThreeVC$ the corresponding type. Let $\opt_1^t \coloneq |\OPT_1^t|$ for every $t \in \typesThreeVC$ and $\opt_1^{\min} \coloneq |\OPT_1^{\min}|$. \label{reduce-3vc:compute-opt-types} \;
	\If(\tcp*[h]{also if $t^{\min} = \typBi$}){$\OPT_1^{\typBi}$ exists and $\opt_1^{\typBi} \leq \opt_1^{\min} + 1$}{
	\label{reduce-3vc:b1}
	Let $G_2^{\typBi} \coloneq G_2'$ and $H_2^{\typBi} \coloneq \Reduce(G_2^{\typBi})$.  \;
	Let $F^{\typBi} \subseteq E
	$ be a min-size edge set s.t.\ $H^{\typBi} \coloneq \OPT_1^{\typBi} \cup H_2^{\typBi} \cup F^{\typBi}$ is 2EC.
	\Return $H^{\typBi}$. \;
	}
	\ElseIf{$t^{\min} = \typBii$}{
	\label{reduce-3vc:b2}
	Let $G_2^{\typBii} \coloneq G_2'$ and $H_2^{\typBii} \coloneq \Reduce(G_2^{\typBii})$. \;
	Let $F^{\typBii} \subseteq E
	$ be a min-size edge set s.t.\ $H^{\typBii} \coloneq \OPT_1^{\typBii} \cup H_2^{\typBii} \cup F^{\typBii}$ is 2EC.
	\Return $H^{\typBii}$. \;
	}
	\ElseIf{$t^{\min} = \typCi$}{
	\label{reduce-3vc:c1}
	Let $G_2^{\typCi} \coloneq G_2'$ and $H_2^{\typCi} \coloneq \Reduce(G_2^{\typCi})$. \;
	Let $F^{\typCi} \subseteq E
	$ be a min-size edge set s.t.\ $H^{\typCi} \coloneq \OPT_1^{\typCi} \cup H_2^{\typCi} \cup F^{\typCi}$ is 2EC.
	\Return $H^{\typCi}$. \;
	}
	\ElseIf(\tcp*[h]{from now on $y, z$ are dummy vertices with incident dummy edges}){$t^{\min} = \typCii$}{
	\label{reduce-3vc:c2-general}
	\If{Every $\OPT_1^{\typCii}$ solution contains a $C(u)-C(v)$ path.\label{reduce-3vc:c2-i}}{
	Let $G_2^{\typCii} \coloneq (V(G_2) \cup \{y\}, E(G_2) \cup \{uy,vy,vw\})$. \;
	Let $H_2^{\typCii} \coloneq \Reduce(G_2^{\typCii}) \setminus \{uy,vy,vw\}$. \;
	Let $F^{\typCii} \subseteq E
	$ be a min-size edge set s.t.\ $H^{\typCii} \coloneq \OPT_1^{\typCii} \cup H_2^{\typCii} \cup F^{\typCii}$ is 2EC.
	\Return $H^{\typCii}$. \;
	}
	\ElseIf{Every $\OPT_1^{\typCii}$ solution contains either a $C(u)-C(v)$ or a $C(v)-C(w)$ path.\label{reduce-3vc:c2-ii}}{
	Let $G_2^{\typCii} \coloneq (V(G_2) \cup \{y,z\}, E(G_2) \cup \{uy,vz,zy,wy\})$. \;
	Let $H_2^{\typCii} \coloneq \Reduce(G_2^{\typCii}) \setminus \{uy,vz,zy,wy\}$. \;
	Let $F^{\typCii} \subseteq E
	$ be a min-size edge set and $\OPT_1^{\typCii}$ be s.t.\ $H^{\typCii} \coloneq \OPT_1^{\typCii} \cup H_2^{\typCii} \cup F^{\typCii}$ is 2EC. \label{reduce-3vc:c2-ii-sol-constr}
	\Return $H^{\typCii}$. \;
	}
	\Else(\tcp*[h]{all paths are possible between $C(u)$, $C(v)$, and $C(w)$}){\label{reduce-3vc:c2-iii}
	Let $G_2^{\typCii} \coloneq (V(G_2) \cup \{y\}, E(G_2) \cup \{uy,vy,wy\})$. \;
	Let $H_2^{\typCii} \coloneq \Reduce(G_2^{\typCii}) \setminus \{uy,vy,wy\}$. \;
	Let $F^{\typCii} \subseteq E
	$ be a min-size edge set and $\OPT_1^{\typCii}$ be s.t.\ $H^{\typCii} \coloneq \OPT_1^{\typCii} \cup H_2^{\typCii} \cup F^{\typCii}$ is 2EC. \label{reduce-3vc:c2-iii-sol-constr}
	\Return $H^{\typCii}$. \;
	}
	}
	\ElseIf{$t^{\min} = \typCiii$}{
	\label{reduce-3vc:c3}
	Let $\OPT_1^{\typCiii}$ be a \typCiii solution such that there is an edge $e^\typCiii_{uv}$ in $G_1$ between $C(u)$ and $C(v)$ and there is an edge $e^\typCiii_{vw}$ in $G_1$ between $C(v)$ and $C(w)$. \label{reduce-3vc:c3-select-solution} \;
	Let $G_2^{\typCiii} \coloneq (V(G_2), E(G_2) \cup \{uv,uv,vw,vw\})$ and $H_2^{\typCiii} \coloneq \Reduce(G_2^{\typCiii}) \setminus \{uv,uv,vw,vw\}$. \;
	Let $F^{\typCiii} \subseteq E
	$ be a min-size edge set s.t.\ $H^{\typCiii} \coloneq \OPT_1^{\typCiii} \cup H_2^{\typCiii} \cup F^{\typCiii}$ is 2EC.  \label{reduce-3vc:c3-sol-constr}
	\Return $H^{\typCiii}$. \;
	}
	}
	\caption{Remove a large 3-vertex cut~\cite{GHL24}}\label{alg:3-vertex-cuts}
\end{algorithm}

\begin{algorithm}
	\small
	\DontPrintSemicolon
	\caption{Remove a large $\mathcal C_k$ cut}\label{alg:cycle-cuts}
	\KwIn{A 2EC graph $G=(V,E)$ with at least $\frac{8}{\eps}$ vertices.
		A large $\mathcal C_k$ cut $v_1 - v_2 - \ldots - v_k - v_1$ in $G$.}
	Let $(V_1,V_2)$ be the partition of $V \setminus \{v_1,\ldots,v_k\}$ such that $\ceil{\frac{k}{\alpha - 1}} - 1 \leq \abs{V_1} \leq \abs{V_2}$ and there are no edges between $V_1$ and $V_2$ in $G$. \;
	Let $G_i \coloneq G[V_i \cup \{v_1,\ldots,v_k\}]$ and $G'_i \coloneq G_i | \{v_1,\ldots,v_k\}$ for $i \in \{1,2\}$. \;
	Let $H_i \coloneq \Reduce(G'_i)$ for $i \in \{1,2\}$. \;
	\Return $H' \coloneq H_1 \cup H_2 \cup \{v_1v_2,v_2v_3,\ldots,v_{k-1}v_k,v_kv_1\}$. \;
\end{algorithm}

\begin{algorithm}
	\small
	\DontPrintSemicolon
	\caption{Remove a large $k$-vertex cut}\label{alg:k-cuts}
	\KwIn{A 2EC graph $G=(V,E)$ with at least $\frac{8}{\eps}$ vertices.
		A large $k$-vertex cut $\{v_1,\ldots,v_k\}$ in $G$.}
	\If{there is a partition $(V_1,V_2)$ of $V \setminus \{v_1,\ldots,v_k\}$ with $\abs{V_2} \geq \abs{V_1} \geq \ceil{\frac{2k-2}{\alpha - 1}} - 1$ and there are no edges between $V_1$ and $V_2$ in $G$. \label{alg-k-cuts:case1}}{
		Let $G_i \coloneq G[V_i \cup \{v_1,\ldots,v_k\}]$ and $G'_i \coloneq G_i | \{v_1,\ldots,v_k\}$ for $i \in \{1,2\}$. \;
		Let $H_i \coloneq \Reduce(G'_i)$ for $i \in \{1,2\}$. \;
		Let $F \subseteq E$ be a min-size set s.t.\
		$H' \coloneq H_1 \cup H_2 \cup F$ is a 2ECSS for $G$. \Return $H'$\;
	}
	\Else{\label{alg-k-cuts:case2}
		Let $(V_1,V_2)$ be a partition of $V \setminus \{v_1,\ldots,v_k\}$ with $\abs{V_1} \leq \ceil{\frac{2k-2}{\alpha - 1}} - 2$ and there are no edges between $V_1$ and $V_2$ in $G$ promised by Condition 2 in \Cref{def:large_k_vertex_cut}. \;
		Let $G_i \coloneq G[V_i \cup \{v_1,\ldots,v_k\}]$ and $G'_i \coloneq G_i | \{v_1,\ldots,v_k\}$ for $i \in \{1,2\}$. \;
		Let $H_2 \coloneq \Reduce(G'_2)$ \;
		Let $Z \subseteq E$ with $|Z| \leq \alpha \cdot \opt(G'_1)$ such that $H' \coloneq H_2 \cup Z$ is a 2ECSS for $G$.
		\Return $H'$. \;
	}
\end{algorithm}

\subsection{Auxiliary Lemmas}\label{sec:reduce-auxiliary}

The following two lemmas are standard results, and can be found in, e.g., \cite{HVV19,GargGA23improved}.

\begin{lemma}\label{lemma:cut-vertex-egal}
	Let $G$ be a 2EC graph, $v$ be a cut vertex of $G$, and $(V_1,V_2), V_1 \neq \emptyset \neq V_2$, be a partition of $V \setminus \{v\}$ such that there are no edges between $V_1$ and $V_2$. Then, $\opt(G) = \opt(G[V_1 \cup \{v\}]) + \opt(G[V_2 \cup \{v\}])$.
\end{lemma}

\begin{lemma}\label{lemma:self-loops-parallel-edges}
	Let $G$ be a 2VC multigraph with $|V(G)| \geq 3$. Then, there exists a minimum size 2EC spanning subgraph of $G$ that is simple, that is, it contains no self loops or parallel edges.
\end{lemma}

Next, we replicate a lemma of \cite{GargGA23improved} that
shows that irrelevant edges are irrelevant for obtaining an optimal 2ECSS.

\begin{lemma}[Lemma 2.1 in \cite{GargGA23improved}]
	\label{lemma:irrelevant-edge}
	Let $e = uv$ be an irrelevant edge of a 2VC simple graph $G = (V, E)$. Then there exists a minimum-size 2EC spanning subgraph of $G$ not containing $e$.
\end{lemma}

\begin{proof}
	Assume by contradiction that all optimal 2EC spanning subgraphs of $G$ contain $e$, and let $\OPT$ be one such solution. Define $\OPT' \coloneq \OPT \setminus \{e\}$.
	Clearly, $\OPT'$ cannot be 2EC as it would contradict the optimality of $\OPT$. Let $(V_1,V_2)$ be any partition of $V \setminus \{u,v\}$, $V_1 \neq \emptyset \neq V_2$, such that there are no edges between $V_1$ and $V_2$, which must exist since $\{u, v\}$ is a 2-vertex cut.
	Notice that at least one of $\OPT'_1 \coloneq \OPT'[V_1 \cup \{u,v\}]$ and $\OPT'_2 \coloneq \OPT'[V_2 \cup \{u,v\}]$, say $\OPT'_2$, needs to be connected, as otherwise $\OPT$ would not be 2EC.
	We also have that $\OPT'_1$ is disconnected, as otherwise $\OPT'$ would be 2EC.
	More precisely, $\OPT'_1$ consists of exactly two connected components $\OPT'_1(u)$ and $\OPT'_1(v)$ containing $u$ and $v$, respectively. Assume w.l.o.g.\ that $|V(\OPT'_1(u))| \geq 2$.
	Observe that there must exist an edge $f$ between $V(\OPT'_1(u))$ and $V(\OPT'_1(v))$. Otherwise, $u$ would be a 1-vertex cut separating $V(\OPT'_1(u)) \setminus \{u\}$ from $V \setminus (V(\OPT'_1(u)) \cup \{u\})$.
	Thus, $\OPT'' \coloneq \OPT' \cup \{f\}$ is an optimal 2EC spanning subgraph of $G$ not containing $e$, a contradiction.
\end{proof}

The next lemma shows that for both, 3-vertex cuts and 2-vertex cuts, type $\typA$ solutions are more expensive than the cheapest solution.

\begin{lemma}\label{lemma:type-A-high-cost}
	If \Cref{alg:2-vertex-cuts} (\Cref{alg:3-vertex-cuts}) reaches \Cref{reduce-2vc:one-small} (\Cref{reduce-3vc:one-small}) and $\OPT_1^{\typA}$ exists, then $\opt_1^{\typA} \geq \alpha \cdot \opt_1^{\min} + 1$.
\end{lemma}

\begin{proof}
	If $\OPT_1^{\typA}$ exists and $\opt_1^{\typA} \leq \alpha \cdot \opt_1^{\min}$ then $\OPT_1^{\typA}$ is 2EC and $\alpha$-contractible.
	However, since we assume that we reached \Cref{reduce-2vc:one-small} (\Cref{reduce-3vc:one-small}), we have that $|V_1| \leq \frac{8}{\varepsilon}$. This is a contradiction, because we would have contracted $G_1$ in \Cref{reduce:contractible} of \Cref{alg:reduce} earlier. Thus, it must be that $\opt_1^{\typA} \geq \alpha \cdot \opt_1^{\min} + 1$. %
\end{proof}

\subsubsection{Auxiliary Lemmas for \Cref{alg:2-vertex-cuts}}\label{app:prep:helper-2vc}

\begin{proposition}\label{prop:properties-for-2vc-alg}
	If \Cref{alg:reduce} executes \Cref{alg:2-vertex-cuts}, then $G$ contains no self loops, parallel edges, irrelevant edges, or $\alpha$-contractible subgraphs with at most $\frac{8}{\varepsilon}$ vertices.
\end{proposition}

\begin{lemma}[Lemma A.4 in \cite{GargGA23improved}]
	\label{lemma:2vc-type-combinations}
	Let $G$ be the graph when \Cref{alg:reduce} executes \Cref{alg:2-vertex-cuts}, $\{u,v\}$ be a non-isolating 2-vertex cut of $G$, $(V_1,V_2)$ be a partition of $V \setminus \{u,v\}$ such that $V_1 \neq \emptyset \neq V_2$ and there are no edges between $V_1$ and $V_2$, and $H$ be a 2EC spanning subgraph of $G$. Let $G_i = G[V_i \cup \{u,v\}]$ and $H_i = E(G_i) \cap H$ for $i \in \{1,2\}$. The following statements are true.
	\begin{enumerate}
		\item Both $H_1$ and $H_2$ are of a type in $\typesTwoVC$ with respect to $\{u,v\}$. Furthermore, if one of the two is of type $\typC$, then the other must be of type $\typA$.
		\item If $H_i$, for an $i \in \{1,2\}$, is of type $\typC$, then there exists one edge $f \in E(G_i)$ such that $H_i' \coloneq H_i \cup \{f\}$ is of type $\typB$. As a consequence, there exists a 2EC spanning subgraph $H'$ where $H_i' \coloneq H' \cap E(G_i)$ is of type $\typA$ or $\typB$.
	\end{enumerate}
\end{lemma}

\begin{proof}
	First note that $uv \notin E(G)$ since there are no irrelevant edges by \Cref{prop:properties-for-2vc-alg}.
	We first prove the first claim. We prove it for $H_1$, the other case being symmetric.
	First notice that if $H_1$ contains a connected component not containing $u$ nor $v$, $H$ would be disconnected. Hence, $H_1$ consists of one connected component or two connected components $H_1(u)$ and $H_1(v)$ containing $u$ and $v$, respectively.

	Suppose first that $H_1$ consists of one connected component. Let us contract the 2EC components of $H_1$, hence obtaining a tree $T$. If $T$ consists of a single vertex, $H_1$ is of type $\typA$. Otherwise, consider the path $P$ (possibly of length $0$)
	between the super nodes resulting from the contraction of the 2EC components $C(u)$ and $C(v)$ containing $u$ and $v$, respectively. Assume by contradiction that $T$ contains an edge $e$ not in $P$. Then $e$ does not belong to any cycle of $H$, contradicting the fact that $H$ is 2EC. Thus, $H_1$ is of type $\typB$ (in particular $P$ has length at least 1 since $H_1$ is not 2EC).

	Assume now that $H_1$ consists of 2 connected components $H_1(u)$ and $H_1(v)$. Let $T_u$ and $T_v$ be the two trees obtained by contracting the 2EC components of $H_1(u)$ and $H_1(v)$, respectively. By the same argument as before, the 2-edge connectivity of $H$ implies that these two trees contain no edge. Hence, $H_1$ is of type $\typC$.

	The second part of the first claim follows easily since if $H_1$ is of type $\typC$ and $H_2$ is not of type $\typA$, then $H$ would either be disconnected or it would contain at least one bridge edge, because $uv \notin E(G)$.

	We now move to the second claim of the lemma. There must exist an edge $f$ in $G_i$ between the two connected components of $H_i$ since otherwise at least one of $u$ or $v$ would be a cut vertex.
	Clearly, $H'_i \coloneq H_i \cup \{f\}$ satisfies the claim. For the second part of the claim, consider any 2EC spanning
	subgraph $H'$ and let $H'_i \coloneq H' \cap E(G_i)$.
	The claim holds if there exists one such $H'$ where $H'_i$ is of type $\typA$ or $\typB$, hence assume that this is not the case. Hence, $H'_i$ is of type $\typC$ by the first part of this lemma. The same argument as above implies the existence of an edge $f$ in $G_i$ such that $H''_i \coloneq H'_i \cup \{f\}$ is of type $\typB$, implying the claim for $H'' \coloneq H' \cup \{f\}$.
\end{proof}

\begin{lemma}\label{lemma:2vc-opt-1-large}
	In \Cref{alg:2-vertex-cuts} it holds that $\opt_1^{\min} \geq 3$.
\end{lemma}

\begin{proof}
	Consider any feasible 2ECSS solution $H_1$ for $G_1$, and let $H'_1 = H_1 \setminus \{u,v\}$ be the corresponding solution for $G_1 | \{u,v\}$.
	Since $G_1 | \{u,v\}$ contains at least $3$ vertices and $H'_1$ is a feasible 2ECSS solution, $H'_1$ must contain at least $3$ edges.
	Thus, $|H_1| \geq |H'_1| \geq 3$.
\end{proof}

We have an immediate corollary from \Cref{lemma:type-A-high-cost} and $\alpha > 1$. %

\begin{corollary}\label{lemma:2vc-type-A-condition}
	If \Cref{alg:2-vertex-cuts} reaches \Cref{reduce-2vc:one-small} and $\OPT_1^{\typA}$ exists, then $\opt_1^{\typA} \geq \opt_1^{\min} + 2$. In particular, $t^{\min} \neq \typA$.
\end{corollary}

\begin{lemma}[Claim 1 in \cite{GargGA23improved}]\label{lemma:2vc-only-b-and-c}
	If \Cref{alg:2-vertex-cuts} is executed on instance $G$, there exists an optimal solution $\opt(G)$ such that $\OPT_1$ is of type $\typB$ or $\typC$.
\end{lemma}

\begin{proof}
	For the sake of contradiction, consider any optimal solution $\OPT(G)$ and assume that $\OPT_1$ is of type $\typA$.
	\cref{lemma:2vc-type-combinations} implies that every feasible solution must use at least $\opt_1^{\min}$ edges from $G_1$. Moreover, $\opt_1^{\typA} > \alpha \cdot \min\{\opt_1^\typB,\opt_1^\typC\}$ by \Cref{lemma:type-A-high-cost}.
	If $\opt_1^\typA \geq \opt_1^\typB + 1$, then we obtain an alternative optimum solution with the desired property by taking $\OPT_1^\typB \cup \OPT_2$, and, in case $\OPT_2$ is of type $\typC$, by adding one edge $f$ between the connected components of $\OPT_2$ whose existence is guaranteed by \cref{lemma:2vc-type-combinations}. Otherwise, that is, $\opt_1^\typA \leq \opt_1^\typB$, we must have $\opt_1^\typA > \alpha \cdot \opt_1^\typC$. Specifically, $\opt_1^\typA \geq \opt_1^\typC + 1 \geq 4$ by \cref{lemma:2vc-opt-1-large}. \cref{lemma:2vc-type-combinations} also implies $\opt_1^\typB \leq \opt_1^\typC + 1$, which together gives $\opt_1^\typA = \opt_1^\typB = \opt_1^\typC + 1$. Then, $\opt_1^\typA > \alpha \cdot \opt_1^\typC = \alpha(\opt_1^\typA - 1)$ implies $\opt_1^\typA \leq 5$ since $\alpha \geq \frac{6}{5}$. We established that $\opt_1^\typA \in \{4,5\}$.

	If $\opt_1^\typA = 4$, $\OPT_1^\typA$ has to be a $4$-cycle $C =u - a -v -b- u$. Notice that the edge $ab$ cannot exist since otherwise $\opt_1^\typB = 3$ due to the path $u - a - b - v$. Then every feasible solution must include the 4 edges of $C$ to guarantee degree at least 2 on $a$ and $b$. But this makes $C$ an $\alpha$-contractible subgraph on $4 \leq \frac{4}{\varepsilon}$ vertices, which is a contradiction to \cref{prop:properties-for-2vc-alg}.

	If $\opt_1^\typA = 5$, the minimality of $\OPT_1^\typA$ implies that $\OPT_1^\typA$ is a $5$-cycle, say $C = u - a - v - b - c -u$. Observe that the edge $ab$ cannot exist, since otherwise $\{va,ab,bc,cu\}$ would be a type $\typB$ solution in $G_1$ of size $4$, contradicting $\opt_1^\typB = \opt_1^\typB = 5$. A symmetric construction shows that edge $ac$ cannot exist. Hence, every feasible solution restricted to $G_1$ must include the edges $\{au,av\}$ and furthermore at least $3$ more edge incident to $b$ and $c$, so at least $5$ edges altogether. This implies that $C$ is an $\alpha$-constractible subgraph of size $5 \leq \frac{8}{\varepsilon}$ vertices, which is a contradiction to \cref{prop:properties-for-2vc-alg}.
\end{proof}

\begin{lemma}[Part of the proof of Lemma A.6 in \cite{GargGA23improved}]\label{lemma:2vc-both-large-F}
	If the condition in \Cref{reduce-2vc:both-large} in \Cref{alg:2-vertex-cuts} applies and both, $H'_1$ and $H'_2$ are 2EC spanning subgraphs of $G'_1$ and $G'_2$, respectively, then there exist edges $F'$ with $|F'| \leq 2$ such that $H'$ is a 2EC spanning  subgraph of $G$.
\end{lemma}
\begin{proof}
	First note that $H'_i$, $i \in \{1,2\}$, must be of type $\typA$, $\typB$, or $\typC$. If one of them is of type $\typA$ or they are both of type $\typB$, then $H'_1 \cup H'_2$ is 2EC.
	If one of them is of type $\typB$, say $H'_1$, and the other of type $\typC$, say $H'_2$, then let $H'_2(u)$ and $H'_2(v)$ be the two 2EC components of $H'_2$ containing $u$ and $v$, respectively. By \cref{lemma:2vc-type-combinations}, there must exist an edge $f \in E(G_2)$ between $H'_2(u)$ and $H'_2(v)$. Hence, $H'_1(u) \cup H'_2(v) \cup \{f\}$ is 2EC.
	The remaining case is that $H'_1$ and $H'_2$ are both of type $\typC$. In this case, $H'_1 \cup H'_2$ consists of precisely two 2EC components $H'(u)$ and $H'(v)$ containing $u$ and $v$, respectively.
	Since $G$ is 2EC, there must exist two edges $f$ and $g$ between $H'(u)$ and $H'(v)$ such that $H'_1 \cup H'_2 \cup \{f,g\}$ is 2EC. Thus, in all cases a desired set $F^\typC$ of size at most 2 exists such that $H'$ is a 2ECSS of $G$.
\end{proof}

\begin{lemma}[Part of the proof of Lemma A.6 in \cite{GargGA23improved}]\label{lemma:2vc-b-F}
	If the condition in \Cref{reduce-2vc:b} in \Cref{alg:2-vertex-cuts} applies and $\Reduce(G_2^{\typB})$ is a 2EC spanning subgraph of $G_2^{\typB}$, then  $H^\typB$ is a 2EC spanning subgraph of~$G$.
\end{lemma}
\begin{proof}
	Note that $\Reduce(G_2^{\typB})$ must contain the two dummy edges $wv$ and $wu$, which induce a type $\typB$ graph. Hence, by \cref{lemma:2vc-type-combinations}, $H^{\typB}_2$ is of type $\typA$ or $\typB$.
	Thus, $H^{\typB} = \OPT_1^\typB \cup H^\typB$ is a 2ECSS of $G$.
\end{proof}

\begin{lemma}[Part of the proof of Lemma A.6 in \cite{GargGA23improved}]\label{lemma:2vc-c-F}
	If the condition in \Cref{reduce-2vc:c} in \Cref{alg:2-vertex-cuts} applies and $\Reduce(G_2^{\typC})$ is a 2EC spanning subgraph of $G_2^{\typC}$, then there exists an edge set $F^{\typC}$ with $|F^{\typC}| \leq 1$ such that $H^\typC$ is a 2EC spanning subgraph of $G$.
\end{lemma}

\begin{proof}
	Note that $H_2^\typC$ must be of type $\typA$ or $\typB$.
	\cref{lemma:2vc-type-combinations} guarantees the existence of an edge $f \in E(G_2)$ such that $\OPT_1^{\typC} \cup \{f\}$ is of type $\typB$ for $G_1$. Thus, $\OPT_1^{\typC} \cup H_2^\typC \cup \{f\}$ is a 2ECSS of $G$.
\end{proof}

Finally, we prove that the reductions of \cref{alg:2-vertex-cuts} preserve the approximation factor.

\begin{lemma}[Part of the proof of Lemma A.7 in \cite{GargGA23improved}]\label{lemma:2vc-approx}
	Let $G$ be a 2ECSS instance. If every recursive call to $\Reduce(G')$ in~\Cref{alg:2-vertex-cuts} on input $G$ satisfies \eqref{eq:reduce-invariant}, then  \eqref{eq:reduce-invariant} holds for input $G$.
\end{lemma}

\begin{proof}
	Let $\OPT_i \coloneq \OPT(G) \cap E(G_i)$ and $\opt_i \coloneq \abs{\OPT_i}$ for $i \in \{1,2\}$.
	Note that $\opt(G) = \opt_1 + \opt_2$.
	For the remainder of this proof, all line references (if not specified differently) refer to \Cref{alg:2-vertex-cuts}. We use in every case that $\opt_1 \geq 2$ (cf. \Cref{lemma:2vc-opt-1-large}), and $f(G) \geq 0$ by \Cref{prop:properties-for-2vc-alg}.
	We distinguish in the following which reduction the algorithm uses.

	\begin{description}
		\item[(\Cref{reduce-2vc:both-large}: $|V_1| > \frac{8}{\varepsilon}$)]
			Note that for $i \in \{1,2\}$ we have that $\OPT_i$ is a feasible solution for $G_i'$ and therefore $\opt_i \geq \opt(G_i')$.
			Using the assumption of the lemma, \Cref{lemma:2vc-both-large-F} and $\abs{V(G_1')} + \abs{V(G_2')} = \abs{V(G)}$, we conclude that
			\begin{align*}
				\reduce(G) = \abs{H'_1} + \abs{H'_2} + \abs{F'} & \leq \alpha \cdot (\opt_1 + \opt_2) + 4\varepsilon(\abs{V(G_1')} + \abs{V(G_2')}) -32 + 2          \\
				                                                & \leq \alpha \cdot \opt(G) + 4\varepsilon \cdot (\abs{V(G)}) - 16 = \alpha \cdot \opt(G) + f(G) \ .
			\end{align*}

		\item[(\Cref{reduce-2vc:b}: $t^{\min} = \typB$)]
			Note that $\opt_2(G_2^{\typB}) \leq \opt_2 + 2$, because we can turn $\OPT_2$, which we can assume to be of type $\typA$ or $\typB$ by \Cref{lemma:2vc-only-b-and-c} and \Cref{lemma:2vc-type-combinations}, using both dummy edges $\{uw,vw\}$ into a solution for $G_2^{\typB}$. Thus,
			\begin{align*}
				\reduce(G) = \abs{H^{\typB}} & \leq \opt_1^{\typB} + \abs{H^{\typB}}
				\leq \opt_1 + (\alpha \cdot \opt_2(G_2^{\typB}) + f(G_2^{\typB})) - 2             \\
				                             & \leq \opt_1 + (\alpha \cdot (\opt_2+2) + f(G)) - 2 \\
				                             & \leq \alpha \opt(G) + f(G) +
				(\alpha - 1)(2 - \opt_1) \leq \alpha \opt(G) + f(G) \ .
			\end{align*}

		\item[(\Cref{reduce-2vc:c}: $t^{\min} = \typC$)]
			We distinguish the following cases depending on which type combination of \Cref{lemma:2vc-type-combinations} is present for $(\OPT_1,\OPT_2)$. We exclude the case that $\OPT_1$ is of type $\typA$ due to \Cref{lemma:2vc-only-b-and-c}. Moreover, \cref{lemma:2vc-c-F} gives that $|F^{\typC}| \leq 1$.

			\begin{enumerate}[label=(\alph*)]

				\item $(\typB, \{\typA,\typB\})$.
				      In this case, $\opt_1^{\typB}$ exists but $t^{\min} = \typC$. Thus, the definition of \Cref{alg:2-vertex-cuts} and \Cref{def:2vc-ties} imply that
				      $\opt_1^{\typC} \leq \opt^{\typB}_1 - 1 = \opt_1 - 1$. Moreover, $\opt_2(G_2^{\typC}) \leq \opt_2 + 1$, because we can turn $\OPT_2$ with the dummy edge $\{uv\}$ into a solution for $G_2^{\typC}$. Thus,
				      \begin{align*}
					      \reduce(G) = \abs{H^{\typC}} & \leq \opt_1^{\typC} + \abs{H^{\typC}} + \abs{F^\typC}          \\
					                                   & \leq (\opt_1 - 1) + (\alpha \cdot (\opt_2 + 1) + f(G)) - 1 + 1 \\
					                                   & \leq \alpha \opt(G) + f(G) +
					      (\alpha - 1)(1 - \opt_1) \leq \alpha \opt(G) + f(G) \ .
				      \end{align*}

				\item $(\typC, \typA)$.
				      In this case, $\opt_1 = \opt_1^{\typC}$ and $\OPT_2$ is a feasible solution for $G_2^{\typC}$. Thus,
				      \begin{align*}
					      \reduce(G) = \abs{H^{\typC}} & \leq \opt_1^{\typC} + \abs{H^{\typC}} + \abs{F^\typC}                             \\
					                                   & \leq \opt_1 + (\alpha \cdot \opt_2 + f(G)) - 1 + 1 \leq \alpha \opt(G) + f(G) \ .
				      \end{align*}
			\end{enumerate}
	\end{description}

	This completes the proof of the lemma.
\end{proof}

\subsubsection{Auxiliary Lemmas for \Cref{alg:3-vertex-cuts}}\label{app:prep:helper-3vc}

Most of the analysis of \Cref{alg:3-vertex-cuts} is based on \cite{GHL24}.

\begin{proposition}\label{prop:properties-for-3vc-alg}
	If \Cref{alg:reduce} executes \Cref{alg:3-vertex-cuts}, then $G$ contains no non-isolating 2-vertex cuts, self loops, parallel edges, irrelevant edges, or $\alpha$-contractible subgraphs with at most $\frac{8}{\varepsilon}$ vertices.
\end{proposition}

\begin{lemma}
	Let $G$ be the graph when \Cref{alg:reduce} executes \Cref{alg:3-vertex-cuts}, $\{u,v,w\}$ be a large 3-vertex cut of $G$, $(V_1,V_2)$ be partition of $V \setminus \{u,v\}$ such that $V_1 \neq \emptyset \neq V_2$ and there are no edges between $V_1$ and $V_2$, and $H$ be a 2EC spanning subgraph of $G$. Let $G_i = G[V_i \cup \{u,v,w\}]$ and $H_i = E(G_i) \cap H$ for $i \in \{1,2\}$. Then, $H_1$ and $H_2$ are of a type in $\typesThreeVC$ with respect to $\{u,v,w\}$ and $G_1$ and $G_2$, respectively.
\end{lemma}

\begin{proof}
	Let $i \in \{1,2\}$. First, observe that $H_i$ consist of either one, two or three connected components, and each component contains at least one of $\{u,v,w\}$, as otherwise $H$ would be disconnected. Let $H_i'$ be obtained by contracting each 2EC component $C$ of $H_i$ into a single super-node $C$, and let $C_i(x)$ be the corresponding super-node of the 2EC component of $H$ that contains $x \in V(H_i)$.

	First, assume that $H_i$ is a single connected component.
	In this case $H_i'$ is a tree. If $H_i'$ is composed of a single vertex, $H_i$ is of type $\typA$. Otherwise, let $F$ be the set of edges of $H_i'$ that are on any simple path between $C_i(u)$, $C_i(v)$, and $C_i(w)$. If $H_i'$ contains an edge $e \in E(H_i') \setminus F$, then $e$ cannot belong to any cycle of $H$, because it does not belong to a cycle in $H_i$ and is not on a simple path between any two vertices of the cut $\{u,v,w\}$, contradicting that $H$ is 2EC. Thus, $H_i$ is of type $\typCi$ if $C_i(u)$, $C_i(v)$, and $C_i(w)$  are distinct, and of type $\typBi$ otherwise.

	Second, assume that $H_i$ is composed of two connected components.
	Let $H_i'(v)$ and $H_i'(w)$ denote the two trees in $H_i'$ such that w.l.o.g.\ $u$ and $v$ are in $H_i'(v)$ and $w$ is in $H_i'(w)$.
	First note that $H_i'(w)$ must be a single vertex. Otherwise, there is an edge $e$ in $H_i'(w)$ that cannot be part of any cycle in $H$, contradicting that $H$ is 2EC.
	If $H_i'(v)$ is also a single vertex, then $H_i$ is of type $\typBii$.
	Otherwise, let $F$ be the set of edges of $H_i'(v)$ on the simple path between $C_i(u)$ and $C_i(v)$.
	If $H_i'(v)$ contains an edge $e \in E(H_i'(v)) \setminus F$, then $e$ cannot belong to any cycle of $H$, contradicting that $H$ is 2EC. Thus, $H_i$ is of type $\typCii$.

	Finally, assume that $H_i$ is composed of three connected components.
	Let $H_i'(u)$, $H_i'(v)$, and $H_i'(w)$ denote the trees of $H_i'$ that contain $u$, $v$, and $w$, respectively.
	Using the same argument as before, we can show that each tree is composed of a single vertex, and, thus, $H_i$ is of type $\typCiii$.
\end{proof}

\begin{lemma}\label{lemma:3vc-opt-1-large}
	In \Cref{alg:3-vertex-cuts} it holds that $\opt_1^{\min} \geq \ceil{\frac{2}{\alpha-1}} \geq 8$.
\end{lemma}

\begin{proof}
	Consider any feasible 2ECSS solution $H_1$ for $G_1$, and let $H'_1 = H_1 \setminus \{u,v,w\}$ be the corresponding solution for $G_1 | \{u,v,w\}$.
	Since $G_1 | \{u,v,w\}$ contains at least $\ceil{\frac{2}{\alpha-1}}$ vertices and $H'_1$ is a feasible 2ECSS solution, $H'_1$ must contain at least $\ceil{\frac{2}{\alpha-1}}$ edges.
	Thus, $|H_1| \geq |H'_1| \geq \ceil{\frac{2}{\alpha-1}}$.
\end{proof}

We have an immediate corollary from the above lemma, \cref{lemma:type-A-high-cost}, and $\alpha \geq \frac{6}{5}$.

\begin{corollary}\label{lemma:3vc-type-A-condition}
	If \Cref{alg:3-vertex-cuts} reaches \Cref{reduce-3vc:one-small} and $\OPT_1^{\typA}$ exists, then $\opt_1^{\typA}  \geq \opt_1^{\min} + 3$. In particular, $t^{\min} \neq \typA$.
\end{corollary}

\begin{proof}
	\[
		\opt_1^{\typA} \geq \frac{6}{5} \opt_1^{\min} + 1
		\geq \opt_1^{\min} +  \frac{1}{5} 8 + 1 \geq \opt_1^{\min} + \frac{13}{5}
	\]
	Thus, $\opt_1^{\typA} \geq \opt_1^{\min} + 3$.
\end{proof}

\begin{lemma}\label{lemma:3vc-spanning-tree}
	Let $G = (V,E)$ be a (multi-)graph that does not contain cut vertices or non-isolating 2-vertex cuts, and let $\{u,v,w\}$ be a large 3-vertex cut of $G$.
	Let $(V_1,V_2)$ be a partition of $V \setminus \{u,v,w\}$ such that $7 \leq \abs{V_1} \leq \abs{V_2}$ and there are no edges between $V_1$ and $V_2$ in $G$, and let $G_1 \coloneq G[V_i \cup \{u,v,w\}]$ and $G_2 \coloneq G[V_i \cup \{u,v,w\}] \setminus \{uv, vw, uw\}$.
	For $i \in \{1,2\}$, let $H_i \subseteq E(G_i)$, let $H_i(z)$ be the connected component of $H_i$ that contains $z \in V(H_i)$ and let $H_i$ be such that $V(H_i) = V(H_i(u)) \cup V(H_i(v)) \cup V(H_i(w))$, i.e., each vertex $y \in V(H_i)$ is connected to $u$, $v$, or $w$ in $H_i$. Then, the following holds.
	\begin{enumerate}
		\item For $x \in \{u,v,w\}$ and $i \in \{1,2\}$, if $H_i(x) \neq H_i$, then there exists an edge $e \in E(G_i) \setminus H_i$ between the vertices of $H_i(x)$ and the vertices of $H_i \setminus H_i(x)$.
		\item For $i \in \{1,2\}$, if $H_i(u)$, $H_i(v)$ and $H_i(w)$ are all pairwise vertex-disjoint, then there exist edges $F \subseteq E(G_i) \setminus H_i$ such that $|F| \geq 2$ and $F \subseteq \{e_{uv},e_{vw},e_{uw}\}$ where $e_{xy}$ is between the vertices of $H_i(x)$ and the vertices of $H_i(y)$.
	\end{enumerate}
\end{lemma}

\begin{proof}
	We prove the first claim, and note that the second claim follows from applying the first claim twice.
	Let $i \in \{1,2\}$, $x \in \{u,v,w\}$ and $H_i(x) \neq H_i$.
	If $H_i(x)$ is composed of the single vertex $x$ such that there is no edge between $x$ and $H_i \setminus H_i(x)$, then $\{u,v,w\} \setminus \{x\}$ must be a non-isolating two-vertex cut in $G$, a contradiction. Otherwise, that is, $H_i(x)$ consists of at least two vertices, and there is no edge between $H_i(x)$ and $H_i \setminus H_i(x)$ then $x$ is a cut vertex in $G$ (separating $V(H_i(x)) \setminus \{x\}$ from $V \setminus V(H_i(x))$), a contradiction.
\end{proof}

\begin{lemma}\label{lemma:3vc-make-2ec}
	Let $H \subseteq E$ be such that $H$ is connected and contains a spanning subgraph of $G$, and let $H_i \coloneq H \cap E(G_i)$ for $i \in \{1,2\}$.
	If for all $i \in \{1,2\}$ every edge $e \in H_i$ is part of a 2EC component of $H_i$ or lies on an $x-y$ path in $H_i$ where $x,y \in \{u,v,w\}$ such that $x \neq y$ and there exists an $x-y$ path in $H_{j}$ where $j=1$ if $i=2$ and $j=2$ if $i=1$, then $H$ is a 2EC spanning subgraph of $G$.
\end{lemma}

\begin{proof}
	We show that every edge $e \in H$ is part of a 2EC component of $H$, which proves the lemma because $H$ is connected and spanning.
	If for all $i \in \{1,2\}$ every edge $e \in H_i$ is part of a 2EC component of $H_i$, it is also part of a 2EC component of $H$. Otherwise, for all $i \in \{1,2\}$, if $e$ lies on an $x-y$ path in $H_i$ where $x,y \in \{u,v,w\}$ such that $x \neq y$ and there exists an $x-y$ path in $H_{j}$ where $j=1$ if $i=2$ and $j=2$ if $i=1$, then $e$ is on a cycle of $H$ as both paths are edge-disjoint.
\end{proof}

\begin{lemma}\label{lemma:3vc-both-large-F}
	If the condition in \Cref{reduce-3vc:both-large} in \Cref{alg:3-vertex-cuts} applies and both, $H'_1$ and $H'_2$ are 2EC spanning subgraphs of $G'_1$ and $G'_2$, respectively, then there exist edges $F'$ with $|F'| \leq 4$ such that $H'$ is a 2EC spanning subgraph of $G$.
\end{lemma}

\begin{proof}
	Let $i \in \{1,2\}$, and let $H'_i(u)$, $H'_i(v)$, and $H'_i(w)$ be the connected components of $H'_i$ after decontracting $\{u,v,w\}$ that contain $u$, $v$ and $w$, respectively.
	By \Cref{lemma:3vc-spanning-tree}, there exist edges $F_i \subseteq E(G_i)$ such that $|F_i| \leq 2$ and $H'_i \cup F_i$ is connected in $G_i$.
	Since $H'_i$ is spanning in $G_i$, setting $F \coloneq F_1 \cup F_2$ implies that $|F| \leq 4$ and that $H' = H'_1 \cup H'_2 \cup F$ is spanning and connected in $G$.
	Moreover, observe that for $i \in \{1,2\}$ each edge $e \in H'_i \cup F_i$ is either in a 2EC component of $H'_i \cup F_i$ or lies on an $x-y$ path in $H'_i \cup F_i$, where $x, y \in \{u, v, w\}$, $x \neq y$. Thus, \Cref{lemma:3vc-make-2ec} implies that $H'$ is a 2EC spanning subgraph of $G$.
\end{proof}

\begin{lemma}\label{lemma:3vc-b1-F}
	If the condition in \Cref{reduce-3vc:b1} in \Cref{alg:3-vertex-cuts} applies and $\Reduce(G_2^{\typBi})$ is a 2EC spanning subgraph of $G_2^{\typBi}$, then there exist edges $F^{\typBi}$ with $|F^{\typBi}| \leq 1$ such that $H^\typBi$ is a 2EC spanning subgraph of $G$.
\end{lemma}

\begin{proof}
	Assume w.l.o.g.\ that $u$ and $v$ belong to the same 2EC component of $\OPT_1^{\typBi}$.
	Let $H_1^{\typBi} \coloneq \OPT_1^{\typBi}$ and $C_i^{\typBi}(x)$ be the 2EC component of $H_i^{\typBi}$ that contains $x$, where $x \in \{u, v, w\}$ and $i \in \{1, 2\}$.
	Since $C_1^{\typBi}(u) = C_1^{\typBi}(v)$, we have that in $\OPT_1^{\typBi} \cup H_2^{\typBi}$, the vertices of $C_1^{\typBi}(u)$, $C_2^{\typBi}(u)$, and $C_2^{\typBi}(v)$ are in the same 2EC component.
	Furthermore, in $\OPT_1^{\typBi}$, there is a connection between $C_1^{\typBi}(u)$ and $C_1^{\typBi}(w)$.
	If $w$ and $v$ are already connected in $H_2^{\typBi}$ then we can set $F^{\typBi} \coloneq \emptyset$ and see that this satisfies the lemma, since $H^{\typBi}$ contains a spanning tree and every edge in $H^{\typBi}$ is either in the 2EC component containing $u$ and $v$, in a 2EC component containing $w$ or on a $w-v$ path in $G_i$, $i \in \{1, 2\}$.
	Otherwise, if $w$ and $v$ are not in the same connected component in $H_2^{\typBi}$, by \Cref{lemma:3vc-spanning-tree}, there must be an edge $e$ from $C_2^{\typBi}(w)$ to either $C_2^{\typBi}(v)$ or $C_2^{\typBi}(u)$ in $E(G_2) \setminus H^{\typBi}$.
	Now it can be easily verified that $F^{\typBi} \coloneq \{ e \}$ satisfies the lemma, since $H^{\typBi}$ contains a spanning tree and every edge in $H^{\typBi}$ is either in the 2EC component containing $u$ and $v$, in a 2EC component containing $w$ or on a $w-v$ path in $G_i$, $i \in \{1, 2\}$.
\end{proof}

\begin{lemma}\label{lemma:3vc-b2-F}
	If the condition in \Cref{reduce-3vc:b2} in \Cref{alg:3-vertex-cuts} applies and $\Reduce(G_2^{\typBii})$ is a 2EC spanning subgraph of $G_2^{\typBii}$, then there exist edges $F^{\typBii}$ with $|F^{\typBii}| \leq 2$ such that $H^\typBii$ is a 2EC spanning subgraph of $G$.
\end{lemma}

\begin{proof}
	The proof is similar to the previous one.
	Assume w.l.o.g.\ that $u$ and $v$ belong to the same 2EC component of $\OPT_1^{\typBii}$.
	Let $H_1^{\typBii} \coloneq \OPT_1^{\typBii}$ and $C_i^{\typBii}(x)$ be the 2EC component of $H_i^{\typBii}$ that contains $x$, where $x \in \{u, v, w\}$ and $i \in \{1, 2\}$.
	Since $C_1^{\typBii}(u) = C_1^{\typBii}(v)$, we have that in $\OPT_1^{\typBii} \cup H_2^{\typBii}$, the vertices of $C_1^{\typBii}(u)$, $C_2^{\typBii}(u)$ and $C_2^{\typBii}(v)$ are in the same 2EC component.
	Furthermore, in $\OPT_1^{\typBii}$, there is no connection between $C_1^{\typBi}(u)$ and $C_1^{\typBii}(w)$, since it is of type~$\typBii$.
	Hence, by \Cref{lemma:3vc-spanning-tree}, there must be an edge $e_1 \in E(G_1) \setminus \OPT_1^{\typBii}$.
	Note that $\OPT_1^{\typBii} \cup \{e_1\}$ is now a solution of type $\typBi$ for $G_1$, and hence we can do the same as in the proof of the previous lemma and observe that $\abs{F^{\typBii}} \leq 2$.
\end{proof}

\begin{lemma}\label{lemma:3vc-c1-F}
	If the condition in \Cref{reduce-3vc:c1} in \Cref{alg:3-vertex-cuts} applies and $\Reduce(G_2^{\typCi})$ is a 2EC spanning subgraph of $G_2^{\typCi}$, then there exist edges $F^{\typCi}$ with $|F^{\typCi}| \leq 2$ such that $H^\typCi$ is a 2EC spanning subgraph of $G$.
\end{lemma}

\begin{proof}
	Let $C_2^{\typCi}(x)$ be the 2EC component of $H_2^{\typCi}$ that contains $x$, for %
	$x \in \{u,v,w\}$.
	Note that $\OPT_1^{\typCi}$ is connected and spanning for $G_1$ and that each edge $e \in \OPT_1^{\typCi}$ is either in some 2EC component or it is on a $y-z$-path in $\OPT_1^{\typCi}$ for $y, z \in \{u,v,w\}$, $y \neq z$.
	We consider three cases.
	If $H_2^{\typCi}$ is connected, then we claim that $F^{\typCi} \coloneq \emptyset$ satisfies the lemma.
	In this case every edge of $H_2^{\typCi}$ is part of a 2EC component of $H_2^{\typCi}$, or it is on a $y-z$-path in $H_2^{\typCi}$ for $y, z \in \{u,v,w\}$, $y \neq z$.
	Hence, \Cref{lemma:3vc-make-2ec} implies that $H_2^{\typCi} \cup \OPT_1^{\typCi}$ is a 2EC spanning subgraph of $G$.

	Next, assume that $H_2^{\typCi}$ is composed of two connected components, say w.l.o.g.\ one containing $u$ and $v$ and the other containing $w$, and observe that each component is 2EC. Then \Cref{lemma:3vc-spanning-tree} guarantees the existence of an edge $e$ in $G_2$ between $C_2^{\typCi}(w)$ and $C_2^{\typCi}(u)$ such that $H_2^{\typCi} \cup \{e\}$ is connected.
	Thus, we can analogously to the previous case show that
	$\OPT_1^{\typCi} \cup H_2^{\typCi} \cup F^{\typCi}$ with $F^{\typCi} \coloneq \{e\}$ is a 2EC spanning subgraph of $G$, since in $ H_2^{\typCi} \cup F^{\typCi}$ each edge is either in some 2EC component or it is on a $y-z$-path in $ H_2^{\typCi} \cup F^{\typCi}$ for $y, z \in \{u,v,w\}$, $y \neq z$ (hence \Cref{lemma:3vc-make-2ec} holds).
	In the remaining case $H_2^{\typCi}$ is composed of three connected components, each containing exactly one vertex of $\{u, v, w\}$.
	Observe that each component is 2EC.
	Hence, \Cref{lemma:3vc-spanning-tree} guarantees the existence of at least two edges $e_1$ and $e_2$ in $G_2$ between two pairs of connected components of $H_2^{\typCi}$ such that $H_2^{\typCi} \cup \{e_1,e_2\}$ is connected.
	Thus, we can analogously to the previous cases show that $\OPT_1^{\typCi} \cup H_2^{\typCi} \cup F^{\typCi}$ with $F^{\typCi} \coloneq \{e_1,e_2\}$ is a 2EC spanning subgraph of $G$, since in $H_2^{\typCi} \cup F^{\typCi}$ each edge is either in some 2EC component or it is on a $y-z$-path in $ H_2^{\typCi} \cup F^{\typCi}$ for $y, z \in \{u,v,w\}$, $y \neq z$ (hence \Cref{lemma:3vc-make-2ec} holds).
\end{proof}

\begin{lemma}\label{lemma:3vc-c2-i-F}
	If the condition in \Cref{reduce-3vc:c2-i} in \Cref{alg:3-vertex-cuts} applies and $\Reduce(G_2^{\typCii})$ is a 2EC spanning subgraph of $G_2^{\typCii}$, then there exist edges $F^{\typCii}$ with $|F^{\typCii}| \leq 1$ such that $H^\typCii$ is a 2EC spanning subgraph of $G$ and $|F^{\typCii}| - |\Reduce(G_2^{\typCii})| + |H_2^{\typCii}| \leq - 2$.
\end{lemma}

\begin{proof}
	Fix an optimal \typCii solution $\OPT_1^{\typCii}$ for $G_1$. Note that both dummy edges $uy$ and $vy$ must be part of $\Reduce(G_2^{\typCii})$ as otherwise $\Reduce(G_2^{\typCii})$ cannot be feasible for $G_2^{\typCii}$. By the same reason, $H_2^{\typCii}$ must be connected and spanning for $G_2$.

	We first consider the case where $vw \in \Reduce(G_2^{\typCii})$.
	Let $C_1^{\typCii}(w)$ be the 2EC component of $\OPT_1^{\typCii}$ that contains $w$. %
	Since $C_1^{\typCii}(w)$ is isolated in $\OPT_1^{\typCii}$, \Cref{lemma:3vc-spanning-tree} implies that there exists an edge $e$ in $G_1$ between $C_1^{\typCii}(w)$ and $V(G_1) \setminus C_1^{\typCii}(w)$.
	Then, $\OPT_1^{\typCii} \cup \{e\}$ is connected and spanning for $G_1$, and every edge in $\OPT_1^{\typCii} \cup \{e\}$ is part of a 2EC component of $\OPT_1^{\typCii} \cup \{e\}$ or lies on a path in $\OPT_1^{\typCii} \cup \{e\}$ between any two distinct vertices of $\{u,v,w\}$ between which is a path in $H_2^{\typCii}$.
	Similarly, every edge in $H_2^{\typCii}$ is part of a 2EC component of $H_2^{\typCii}$ or lies on a path in $H_2^{\typCii}$ between any two distinct vertices of $\{u,v,w\}$ between which is a path in $\OPT_1^{\typCii} \cup \{e\}$.
	Thus, $\OPT_1^{\typCii} \cup H_2^{\typCii} \cup F^{\typCii}$ with $F^{\typCii} \coloneq \{e\}$ is a 2EC spanning subgraph of $G$ by \Cref{lemma:3vc-make-2ec}.  Moreover, $|F^{\typCii}| - |\Reduce(G_2^{\typCii})| + |H_2^{\typCii}| = 1 - 3 = -2$.

	If $vw \notin \Reduce(G_2^{\typCii})$, then every edge in $\OPT_1^{\typCii}$ is part of a 2EC component of $\OPT_1^{\typCii}$ or lies on a path in $\OPT_1^{\typCii}$ between $u$ and $v$ between which is also a path in $H_2^{\typCii}$.
	Further, every edge in $H_2^{\typCii}$ is part of a 2EC component of $H_2^{\typCii}$ or lies on a path in $H_2^{\typCii}$ between $u$ and $v$, and there is a path between $u$ and $v$ in $\OPT_1^{\typCii}$.
	Note that there cannot be an edge in $H_2^{\typCii}$ that is on a path between $w$ and $u$ or $v$ but not in a 2EC component, because then $\Reduce(G_2^{\typCii})$ cannot be feasible for $G_2^{\typCii}$ as $vw \notin \Reduce(G_2^{\typCii})$.
	Thus, $\OPT_1^{\typCii} \cup H_2^{\typCii}$ is a 2EC spanning subgraph of $G$ by \Cref{lemma:3vc-make-2ec}.
	We have that $|F^{\typCii}| - |\Reduce(G_2^{\typCii})| + |H_2^{\typCii}| = 0 - 2 = -2$.
\end{proof}

\begin{lemma}\label{lemma:3vc-c2-ii-F}
	If the condition in \Cref{reduce-3vc:c2-ii} in \Cref{alg:3-vertex-cuts} applies and $\Reduce(G_2^{\typCii})$ is a 2EC spanning subgraph of $G_2^{\typCii}$, then there exist edges $F^{\typCii}$ with $|F^{\typCii}| \leq 1$ and $\OPT_1^{\typCii}$ such that $H^\typCii$ is a 2EC spanning subgraph of $G$ and $|F^{\typCii}| - |\Reduce(G_2^{\typCii})| + |H_2^{\typCii}| \leq - 3$.
\end{lemma}

\begin{proof}
	Let $D = \{uy,vz,zy,wy\}$.
	First observe that $H_2^{\typCii}$ must be connected and spanning for $G_2$.
	Note that $vz$ and $zy$ must be part of $\Reduce(G_2^{\typCii})$ to make $z$ incident to 2 edges, and additionally $vy$ or $wy$ must be part of $\Reduce(G_2^{\typCii})$ to make $y$ incident to at least 2 edges.
	Thus, we can distinguish the following cases.

	If $\Reduce(G_2^{\typCii}) \cap D = \{uy,vz,zy\}$ (or by symmetry  $\Reduce(G_2^{\typCii}) \cap D = \{wy,vz,zy\}$), let $\OPT_1^{\typCii}$ be an optimal $\typCii$ solution for $G_1$ that contains a $C_1^{\typCii}(u)-C_1^{\typCii}(v)$ path, where $C_1^{\typCii}(x)$ denotes the 2EC component of $\OPT_1^{\typCii}$ that contains $x$, for $x \in \{u,v,w\}$.
	Note that every edge of $\OPT_1^{\typCii}$ is part of a 2EC component of $\OPT_1^{\typCii}$ or lies on a path in $\OPT_1^{\typCii}$ between $u$ and $v$ between which is also a path in $H_2^{\typCii}$.
	Further, every edge in $H_2^{\typCii}$ is part of a 2EC component of $H_2^{\typCii}$ or lies on a path in $H_2^{\typCii}$ between $u$ and $v$, and there is a path between $u$ and $v$ in $\OPT_1^{\typCii}$. Note that there cannot be an edge in $H_2^{\typCii}$ that is on a path between $w$ and $u$ or $v$ but not in a 2EC component, because then $\Reduce(G_2^{\typCii})$ cannot be feasible for $G_2^{\typCii}$ as $wy \notin \Reduce(G_2^{\typCii})$.
	Thus, $\OPT_1^{\typCii} \cup H_2^{\typCii}$ is a 2EC spanning subgraph of $G$ by \Cref{lemma:3vc-make-2ec}. We have that $|F^{\typCii}| - |\Reduce(G_2^{\typCii})| + |H_2^{\typCii}| = 0-3 = -3$.

	Otherwise, that is, $\Reduce(G_2^{\typCii}) \cap D = \{vy,vz,zy,wy\}$,
	let $\OPT_1^{\typCii}$ be an optimal \typCii solution for $G_1$ that contains a $C_1^{\typCii}(u)-C_1^{\typCii}(v)$ path of $\OPT_1^{\typCii}$, where $C_1^{\typCii}(x)$ denotes the 2EC component of $\OPT_1^{\typCii}$ that contains $x$, for $x \in \{u,v,w\}$.
	\Cref{lemma:3vc-spanning-tree} implies that there exists an edge $e$ in $G_1$ between $C_1^{\typCii}(w)$ and the connected component containing $V(G_1) \setminus V(C_1^{\typCii}(w))$. Therefore, $\OPT_1^{\typCii} \cup \{e\}$ is connected.
	Thus, every edge in $\OPT_1^{\typCii} \cup \{e\}$ (resp.\ $H_2^{\typCii}$) is part of a 2EC component of $\OPT_1^{\typCii} \cup \{e\}$ ($H_2^{\typCii}$) or lies on a path in $\OPT_1^{\typCii} \cup \{e\}$  ($H_2^{\typCii}$) between any two distinct vertices of $\{u,v,w\}$ between which is a path in $H_2^{\typCii}$ ($\OPT_1^{\typCii} \cup \{e\}$).
	We conclude using \Cref{lemma:3vc-make-2ec} that $H_2^{\typCii} \cup \OPT_1^{\typCii} \cup F^{\typCii}$ with $F^{\typCii} \coloneq \{e\}$ is a 2EC spanning subgraph of $G$. We have that $|F^{\typCii}| - |\Reduce(G_2^{\typCii})| + |H_2^{\typCii}| = 1-4=-3$.
\end{proof}

\begin{lemma}\label{lemma:3vc-c2-iii-F}
	If the condition in \Cref{reduce-3vc:c2-iii} in \Cref{alg:3-vertex-cuts} applies  and $\Reduce(G_2^{\typCii})$ is a 2EC spanning subgraph of $G_2^{\typCii}$, then there exist edges $F^{\typCii}$ with $|F^{\typCii}| \leq 1$ and $\OPT_1^{\typCii}$ such that $H^\typCii$ is a 2EC spanning subgraph of $G$ and $|F^{\typCii}| - |\Reduce(G_2^{\typCii})| + |H_2^{\typCii}| \leq - 2$.
\end{lemma}

\begin{proof}
	Let $D = \{uy,vy,wy\}$.
	First observe that $H_2^{\typCii}$ must be connected and spanning for $G_2$.
	Note that $|D \cap \Reduce(G_2^{\typCii})| \geq 2$ as $y$ needs to be 2EC in $\Reduce(G_2^{\typCii})$.

	If $|D \cap \Reduce(G_2^{\typCii})| = 2$, let w.l.o.g.\ by symmetry $D \cap \Reduce(G_2^{\typCii}) = \{uy,vy\}$ and $\OPT_1^{\typCii}$ be an optimal $\typCii$ solution for $G_1$ that contains a $C_1^{\typCii}(u)-C_1^{\typCii}(v)$ path, where $C_i^{\typCii}(x)$ denotes the 2EC component of $\OPT_1^{\typCii}$ that contains $x$, for $x \in \{u,v,w\}$.
	Note that every edge of $\OPT_1^{\typCii}$ is part of a 2EC component of $\OPT_1^{\typCii}$ or lies on a path in $\OPT_1^{\typCii}$ between $u$ and $v$ between which is also a path in $H_2^{\typCii}$.
	Further, every edge in $H_2^{\typCii}$ is part of a 2EC component of $H_2^{\typCii}$ or lies on a path in $H_2^{\typCii}$ between $u$ and $v$, and there is a path between $u$ and $v$ in $\OPT_1^{\typCii}$. Note that there cannot be an edge in $H_2^{\typCii}$ that is on a path between $w$ and $u$ or $v$ but not in a 2EC component, because then $\Reduce(G_2^{\typCii})$ cannot be feasible for $G_2^{\typCii}$ as $wy \notin \Reduce(G_2^{\typCii})$.
	Thus, $\OPT_1^{\typCii} \cup H_2^{\typCii}$ is a 2EC spanning subgraph of $G$ by \Cref{lemma:3vc-make-2ec}.
	We have that $|F^{\typCii}| - |\Reduce(G_2^{\typCii})| + |H_2^{\typCii}| = 0-2=-2$.

	Otherwise, that is, $\Reduce(G_2^{\typCii}) \cap D = \{uy,vy,wy\}$,
	let $\OPT_1^{\typCii}$ be an optimal \typCii solution for $G_1$ that contains a $C_1^{\typCii}(u)-C_1^{\typCii}(v)$ path, where $C_i^{\typCii}(x)$ denotes the 2EC component of $\OPT_1^{\typCii}$ that contains $x$, for $x \in \{u,v,w\}$.
	\Cref{lemma:3vc-spanning-tree} implies that there exists an edge $e$ in $G_1$ between $C_1^{\typCii}(w)$ and the connected component containing $V(G_1) \setminus V(C_1^{\typCii}(w))$.
	Therefore, $\OPT_1^{\typCii} \cup \{e\}$ is connected.
	Thus, every edge in $\OPT_1^{\typCii} \cup \{e\}$ (resp.\ $H_2^{\typCii}$) is part of a 2EC component of $\OPT_1^{\typCii} \cup \{e\}$ ($H_2^{\typCii}$) or lies on a path in $\OPT_1^{\typCii} \cup \{e\}$  ($H_2^{\typCii}$) between any two distinct vertices of $\{u,v,w\}$ between which is a path in $H_2^{\typCii}$ ($\OPT_1^{\typCii} \cup \{e\}$).
	We conclude using \Cref{lemma:3vc-make-2ec} that $H_2^{\typCii} \cup \OPT_1^{\typCii} \cup F^{\typCii}$ with $F^{\typCii} \coloneq \{e\}$ is a 2EC spanning subgraph of $G$. We have that $|F^{\typCii}| - |\Reduce(G_2^{\typCii})| + |H_2^{\typCii}| = 1-3=-2$.
\end{proof}

\begin{lemma}\label{lemma:3vc-c3-F}
	If the condition in \Cref{reduce-3vc:c3} in \Cref{alg:3-vertex-cuts} applies and $\Reduce(G_2^{\typCiii})$ is a 2EC spanning subgraph of $G_2^{\typCiii}$, then there exist edges $F^{\typCiii}$ with $|F^{\typCiii}| \leq 4$ such that $H^\typCiii$ is a 2EC spanning subgraph of $G$ and $|F^{\typCiii}| - |\Reduce(G_2^{\typCiii})| + |H_2^{\typCiii}| \leq 0$.
\end{lemma}

\begin{proof}
	First note that the claimed $\typCiii$ solution in \Cref{reduce-3vc:c3-select-solution} is guaranteed by \Cref{lemma:3vc-spanning-tree} (among renaming $u$, $v$ and $w$).
	Let $C_2^{\typCiii}(x)$ be the 2EC component of $H_2^{\typCiii}$ that contains $x$, for $x \in \{u,v,w\}$.
	We have 4 main cases.

	\textbf{Case 1:}
	$|\Reduce(G_2^{\typCiii}) \cap \{uv,vw\}| = \ell \leq 2$. First observe that $H_2^{\typCiii}$ is connected as otherwise $\Reduce(G_2^{\typCiii})$ cannot be feasible for $G_2^{\typCiii}$. Moreover, we can select $\ell$ edges $F^{\typCiii} \subseteq \{e^\typCiii_{uv},e^\typCiii_{vw}\}$ such that
	$\OPT_1^{\typCiii} \cup F^{\typCiii}$ is connected.
	Further, note that every edge in $\OPT_1^{\typCiii} \cup F^{\typCiii}$ (resp.\ $H_2^{\typCiii}$) is part of a 2EC component of $\OPT_1^{\typCiii} \cup F^{\typCiii}$ ($H_2^{\typCiii}$) or lies on a path in $\OPT_1^{\typCiii} \cup F^{\typCiii}$  ($H_2^{\typCiii}$) between any two distinct vertices of $\{u,v,w\}$ between which is a path in $H_2^{\typCiii}$ ($\OPT_1^{\typCiii} \cup F^{\typCiii}$).
	Thus, \Cref{lemma:3vc-make-2ec} implies that $H_2^{\typCiii} \cup \OPT_1^{\typCiii} \cup F^{\typCiii}$ is a 2EC spanning subgraph of $G$.
	We have that $|F^{\typCiii}| - |\Reduce(G_2^{\typCiii})| + |H_2^{\typCiii}| = \ell - \ell = 0$.

	\textbf{Case 2:}
	$|\Reduce(G_2^{\typCiii}) \cap \{uv,uv\}| = 2$ (the case $|\Reduce(G_2^{\typCiii}) \cap \{vw,vw\}| = 2$ is symmetric).
	We have that $C_2^{\typCiii}(v) = C_2^{\typCiii}(w)$, as otherwise $\Reduce(G_2^{\typCiii})$ cannot be feasible for $G_2^{\typCiii}$.
	If $H_2^{\typCiii}$ is 2EC, then clearly $H_2^{\typCiii} \cup \OPT_2^\typCiii$ is a 2EC spanning subgraph of $G$, and $|F^{\typCiii}| - |\Reduce(G_2^{\typCiii})| + |H_2^{\typCiii}| = 0 - 2 \leq 0$.

	If $C_2^{\typCiii}(u)$ is connected to $C_2^{\typCiii}(v)$ in $H_2^{\typCiii}$, then $H_2^\typCiii$ is connected.
	Further, $\OPT_1^\typCiii \cup F^\typCiii$ with $F^\typCiii \coloneq \{e^\typCiii_{uv}\}$ is connected.
	Every edge in $\OPT_1^{\typCiii} \cup F^{\typCiii}$ is part of a 2EC component of $\OPT_1^{\typCiii} \cup F^{\typCiii}$ or lies on a path in $\OPT_1^{\typCiii} \cup F^{\typCiii}$ between $u$ and $v$, which are also connected in $H_2^{\typCiii}$. Note that there cannot be an edge in $\OPT_1^{\typCiii} \cup F^{\typCiii}$ on a path between $w$ and $u$ or $v$ that is not in a 2EC component by the choice of the $\typCiii$ solution $\OPT_1^{\typCiii}$.
	Moreover, every edge in $H_2^{\typCiii}$ is part of a 2EC component of $H_2^{\typCiii}$ or lies on a path in $H_2^{\typCiii}$ between $u$ and $v$ or $w$, which are also connected in $\OPT_1^{\typCiii} \cup F^{\typCiii}$. Note that there cannot be an edge in $H_2^{\typCiii}$ on a path between $v$ and $w$ that is not in a 2EC component, because $C_2^{\typCiii}(v) = C_2^{\typCiii}(w)$ is 2EC.
	Thus, \Cref{lemma:3vc-make-2ec} implies that $H_2^{\typCiii} \cup \OPT_1^{\typCiii} \cup F^{\typCiii}$ is a 2EC spanning subgraph of $G$, and $|F^{\typCiii}| - |\Reduce(G_2^{\typCiii})| + |H_2^{\typCiii}| = 1 - 2 \leq 0$.

	Otherwise, that is, $C_2^{\typCiii}(u)$ is not connected to $C_2^{\typCiii}(v)$ in $H_2^{\typCiii}$,
	by \Cref{lemma:3vc-spanning-tree} there must exist an edge $e$ in $G_2$ between $C_2^{\typCiii}(u)$ and $C_2^{\typCiii}(v)$ such that $H_2^{\typCiii} \cup \{e\}$ is connected.
	Thus, we can analogously to the previous case show that
	$H_2^{\typCiii} \cup \OPT_1^{\typCiii} \cup F^{\typCiii}$ with $F^{\typCiii} \coloneq \{e, e^\typCiii_{uv}\}$ is a 2EC spanning subgraph of $G$.
	We have that $|F^{\typCiii}| - |\Reduce(G_2^{\typCiii})| + |H_2^{\typCiii}| = 2 - 2 = 0$.

	\textbf{Case 3:}
	$|\Reduce(G_2^{\typCiii}) \cap \{uv,uv,vw\}| = 3$ (the case $|\Reduce(G_2^{\typCiii}) \cap \{uv,vw,vw\}| = 3$ is symmetric).
	Note that in this case $C_2^{\typCiii}(u)$ is connected to $C_2^{\typCiii}(v)$ in $H_2^{\typCiii}$ as otherwise $u$ and $w$ cannot be 2EC in $\Reduce(G_2^{\typCiii})$.
	We have two cases.

	If $C_2^{\typCiii}(v)$ is connected to $C_2^{\typCiii}(w)$ in $H_2^{\typCiii}$, then $H_2^{\typCiii}$ is connected.
	Further, $\OPT_1^{\typCiii} \cup F^{\typCiii}$ with $F^{\typCiii} \coloneq \{e^{\typCiii}_{uv},e^{\typCiii}_{vw}\}$ is connected and spanning for $G_1$.
	Thus, every edge in $\OPT_1^{\typCiii} \cup F^{\typCiii}$ (resp.\ $H_2^{\typCiii}$) is part of a 2EC component of $\OPT_1^{\typCiii} \cup F^{\typCiii}$ ($H_2^{\typCiii}$) or lies on a path in $\OPT_1^{\typCiii} \cup F^{\typCiii}$  ($H_2^{\typCiii}$) between any two distinct vertices of $\{u,v,w\}$ between which is a path in $H_2^{\typCiii}$ ($\OPT_1^{\typCiii} \cup F^{\typCiii}$).
	We conclude via \Cref{lemma:3vc-make-2ec} that $H_2^{\typCiii} \cup \OPT_1^{\typCiii} \cup F^{\typCiii}$ is a 2EC spanning subgraph of $G$, and  $|F^{\typCiii}| - |\Reduce(G_2^{\typCiii})| + |H_2^{\typCiii}| = 2-3\leq 0$.

	Otherwise, that is, $C_2^{\typCiii}(v)$ is not connected to $C_2^{\typCiii}(w)$ in $H_2^{\typCiii}$, \Cref{lemma:3vc-spanning-tree} guarantees the existence of an edge $e$ in $G_2$ between the vertices of $V(C_2^{\typCiii}(w))$ and the vertices of $V(H_2^{\typCiii}) \setminus V(C_2^{\typCiii}(w))$.
	Therefore, $H_2^{\typCiii} \cup \{e\}$ is connected.
	Hence, we can analogously to the previous case show that $H_2^{\typCiii} \cup \OPT_1^{\typCiii} \cup F^{\typCiii}$ with $F^{\typCiii} \coloneq \{e,e^{\typCiii}_{uv},e^{\typCiii}_{vw}\}$ is a 2EC spanning subgraph of $G$. We have that $|F^{\typCiii}| - |\Reduce(G_2^{\typCiii})| + |H_2^{\typCiii}| = 3-3 = 0$.

	\textbf{Case 4:}
	$|\Reduce(G_2^{\typCiii}) \cap \{uv,uv,vw,vw\}| = 4$. We distinguish three cases.

	First, assume that $H_2^{\typCiii}$ is connected. We have that $\OPT_1^{\typCiii} \cup F^{\typCiii}$ with $F^{\typCiii} \coloneq \{e^{\typCiii}_{uv},e^{\typCiii}_{vw}\}$ is connected and spanning for $G_1$.
	Thus, every edge in $\OPT_1^{\typCiii} \cup F^{\typCiii}$ (resp.\ $H_2^{\typCiii}$) is part of a 2EC component of $\OPT_1^{\typCiii} \cup F^{\typCiii}$ ($H_2^{\typCiii}$) or lies on a path in $\OPT_1^{\typCiii} \cup F^{\typCiii}$  ($H_2^{\typCiii}$) between any two distinct vertices of $\{u,v,w\}$ between which is a path in $H_2^{\typCiii}$ ($\OPT_1^{\typCiii} \cup F^{\typCiii}$).
	We conclude via \Cref{lemma:3vc-make-2ec} that $H_2^{\typCiii} \cup \OPT_1^{\typCiii} \cup F^{\typCiii}$ is a 2EC spanning subgraph of $G$, and  $|F^{\typCiii}| - |\Reduce(G_2^{\typCiii})| + |H_2^{\typCiii}| = 2-4\leq 0$.

	If $H_2^{\typCiii}$ is composed of two connected components, assume w.l.o.g.\ that $C_2^{\typCiii}(u)$ is connected to $C_2^{\typCiii}(v)$ in $H_2^{\typCiii}$ and the other connected component in $H_2^{\typCiii}$ is $C_2^{\typCiii}(w)$.
	In this case, \Cref{lemma:3vc-spanning-tree} guarantees that there is an edge $e$ in $E(G_2) \setminus C_2^{\typCiii}(w)$ between the vertices of $C_2^{\typCiii}(w)$ and the vertices of $V(H_2^{\typCiii}) \setminus V(C_2^{\typCiii}(w))$.
	Therefore, $H_2^{\typCiii} \cup \{e\}$ is connected.
	Hence, we can analogously to the previous case show that $H_2^{\typCiii} \cup \OPT_1^{\typCiii} \cup F^{\typCiii}$ with $F^{\typCiii} \coloneq \{e,e^{\typCiii}_{uv},e^{\typCiii}_{vw}\}$ is a 2EC spanning subgraph of $G$, and $|F^{\typCiii}| - |\Reduce(G_2^{\typCiii})| + |H_2^{\typCiii}| = 3-4\leq 0$.

	Otherwise, that is, $C_2^{\typCiii}(u)$, $C_2^{\typCiii}(v)$ and $C_2^{\typCiii}(w)$ are all pairwise disjoint, \Cref{lemma:3vc-spanning-tree} guarantees the existence of two edges $F \subseteq E(G_2)$ such that $H_2^{\typCiii} \cup F$ is connected.
	Hence, we can analogously to the previous case(s) show that $H_2^{\typCiii} \cup \OPT_1^{\typCiii} \cup F^{\typCiii}$
	with $F^{\typCiii} \coloneq F \cup \{e^{\typCiii}_{uv},e^{\typCiii}_{vw}\}$ is a 2EC spanning subgraph of $G$. We have that $|F^{\typCiii}| - |\Reduce(G_2^{\typCiii})| + |H_2^{\typCiii}| = 4-4 = 0$.
\end{proof}

Finally, we prove that the reductions of \cref{alg:3-vertex-cuts} preserve the approximation factor.
While the proof is quite long, it is of the same spirit of the analogue for \cref{alg:2-vertex-cuts} (cf.\ \cref{lemma:2vc-approx}) and easy to verify: We consider different cases depending on the reduction the algorithm uses and, if \Cref{alg:3-vertex-cuts} is executed, we distinguish for some cases further which solution types the actual optimal solution uses for both sides of the cut. For each case, we can easily derive bounds for the different parts of the solution that our algorithm produces, and put them together to obtain the overall bound, which is stated below.

\begin{lemma}\label{lemma:3vc-approx}
	Let $G$ be a 2ECSS instance. If every recursive call to $\Reduce(G')$ in~\Cref{alg:3-vertex-cuts} on input $G$ satisfies \eqref{eq:reduce-invariant}, then  \eqref{eq:reduce-invariant} holds for input $G$.
\end{lemma}

\begin{proof}
	Let $\OPT_i \coloneq \OPT(G) \cap E(G_i)$ and $\opt_i \coloneq \abs{\OPT_i}$ for $i \in \{1,2\}$.
	Note that $\opt(G) = \opt_1 + \opt_2$.
	For the remainder of this proof, all line references (if not specified differently) refer to \Cref{alg:3-vertex-cuts}. Note that \Cref{lemma:3vc-opt-1-large} gives $\opt_1 \geq \opt_1^{\min} \geq  \frac{2}{\alpha - 1}$. Furthermore, throughout the proof we require $\alpha \geq \frac{6}{5}$. We also use that $f(G) \geq 0$, because $|V(G)| \geq \frac{8}{\eps}$ by \Cref{prop:properties-for-3vc-alg}. We distinguish in the following which reduction the algorithm uses.

	\begin{description}
		\item[(\Cref{reduce-3vc:both-large}: $|V_1| > \frac{8}{\varepsilon}$)]
			Note that for $i \in \{1,2\}$, we have that $\OPT_i$ is a feasible solution for $G_i'$ and therefore $\opt_i \geq \opt(G_i')$.
			Using the induction hypothesis, \Cref{lemma:3vc-both-large-F} and $\abs{V(G_1)} + \abs{V(G_2)} = \abs{V_1} + \abs{V_2} + 2 \leq \abs{V(G)}$, we conclude that
			\begin{align*}
				\reduce(G) & = \abs{H'_1} + \abs{H'_2} + \abs{F'}                                                               \\
				           & \leq \alpha \cdot (\opt_1 + \opt_2) + 4\varepsilon(\abs{V_1} + \abs{V_2} + 2) - 32 + 4             \\
				           & \leq \alpha \cdot \opt(G) + 4\varepsilon \cdot (\abs{V(G)}) - 16 = \alpha \cdot \opt(G) + f(G) \ .
			\end{align*}
		\item[(\Cref{reduce-3vc:b1}: $\OPT_1^{\typBi}$ exists and $\opt_1^{\typBi} \leq \opt_1^{\min} + 1$)]
			Note that $\opt(G_2^{\typBi}) \leq \opt_2$, because $u$, $v$, and $w$ are contracted into a single vertex in $G_2^{\typBi}$.
			Thus, $\opt(G) = \opt_2 + \opt_1 \geq \opt(G_2^{\typBi}) + \opt_1$.
			Moreover, $\opt_1^{\typBi} \leq \opt_1^{\min} + 1 \leq \opt_1 + 1$, and $\abs{F^{\typBi}} \leq 1$ by \Cref{lemma:3vc-b1-F}. Therefore, we conclude
			\begin{align*}
				\reduce(G) = \abs{H^{\typBi}} & \leq \opt_1^{\typBi} + \abs{H_2^{\typBi}} + \abs{F^{\typBi}}
				\leq (\opt_1 + 1) + (\alpha \cdot \opt_2(G_2^{\typBi})  + f(G)) + 1                                                            \\
				                              & \leq \alpha \cdot \opt(G) + f(G) - (\alpha - 1)\opt_1 + 2 \leq \alpha \cdot \opt(G) + f(G) \ ,
			\end{align*}
			where we use $\opt_1 \geq \frac{2}{\alpha - 1}$.

		\item[(\Cref{reduce-3vc:b2,reduce-3vc:c1}: $t^{\min} = \typBii$ or $t^{\min} = \typCi$)]
			Let $\tau \coloneq t^{\min}$.
			Note that $\opt(G_2^{\tau}) \leq \opt_2$, because $u$, $v$, and $w$ are contracted into a single vertex in $G_2^{\tau}$. Thus, $\opt(G) \geq \opt_2 + \opt_1(G_1^{\tau}) \geq \opt(G_2^{\tau}) + \opt_1(G_1^{\tau})$. Moreover, $\abs{F^{\tau}} \leq 2$ by \Cref{lemma:3vc-b2-F,lemma:3vc-c1-F}. Therefore, we conclude
			\begin{align*}
				\reduce(G) = \abs{H^{\tau}} & \leq \opt_1^{\tau} + \abs{H_2^{\tau}} + \abs{F^{\tau}}
				\leq \opt^{\tau}_1 + (\alpha \cdot \opt_2(G_2^{\tau})  + f(G)) + 2                                                                  \\
				                            & \leq \alpha \cdot \opt(G) + f(G) - (\alpha - 1)\opt^{\tau}_1 + 2 \leq \alpha \cdot \opt(G) + f(G) \ ,
			\end{align*}
			where we use $\opt_1 \geq \frac{2}{\alpha - 1}$.

		\item[(\Cref{reduce-3vc:c2-general}: $t^{\min} = \typCii$)]
			In \Cref{alg:3-vertex-cuts} we have the following $\typCii$ subtypes, which use different reductions~$G_2^{\typCii}$ for the recursion (\Cref{reduce-3vc:c2-i,reduce-3vc:c2-ii,reduce-3vc:c2-iii}).
			\begin{description}
				\item[Subtype $\typCiisubba$ (\Cref{reduce-3vc:c2-i}):] \emph{Every $\OPT_1^{\typCii}$ solution contains a $C(u)-C(v)$ path.} In this case
					the definition of \Cref{alg:3-vertex-cuts}, \Cref{lemma:3vc-c2-i-F}, and the induction hypothesis gives
					\begin{align}
						\reduce(G) = |H^{\typCii}| & \leq |\OPT_1^{\typCii}| + \abs{H_2^{\typCii}} + \abs{F^{\typCii}} \notag \leq \opt_1^{\typCii} + \reduce(G_2^{\typCii}) - 2 \\
						                           & \leq \opt_1^{\typCii} + \alpha \cdot \opt(G_2^{\typCii}) + f(G) - 2 \ . \label{eq:3vc-c2-i}
					\end{align}
				\item[Subtype $\typCiisubb$ (\Cref{reduce-3vc:c2-ii}):] \emph{Every $\OPT_1^{\typCii}$ solution contains either a $C(u)-C(v)$ path or a $C(v)-C(w)$ path.} In this case
					the definition of \Cref{alg:3-vertex-cuts}, \Cref{lemma:3vc-c2-ii-F}, and the induction hypothesis gives
					\begin{align}
						\reduce(G) = |H^{\typCii}| & \leq |\OPT_1^{\typCii}| + \abs{H_2^{\typCii}} + \abs{F^{\typCii}} \notag \leq \opt_1^{\typCii} + \reduce(G_2^{\typCii}) - 3 \\
						                           & \leq \opt_1^{\typCii} + \alpha \cdot \opt(G_2^{\typCii}) + f(G) - 3 \ . \label{eq:3vc-c2-ii}
					\end{align}
				\item[Subtype $\typCiisubc$ (\Cref{reduce-3vc:c2-iii}):] \emph{Every $\OPT_1^{\typCii}$ solution contains some path between $C(u)$, $C(v)$, and $C(w)$.} In this case
					the definition of \Cref{alg:3-vertex-cuts}, \Cref{lemma:3vc-c2-iii-F}, and the induction hypothesis gives
					\begin{align}
						\reduce(G) = |H^{\typCii}| & \leq |\OPT_1^{\typCii}| + \abs{H_2^{\typCii}} + \abs{F^{\typCii}} \notag \leq \opt_1^{\typCii} + \reduce(G_2^{\typCii}) - 2 \\
						                           & \leq \opt_1^{\typCii} + \alpha \cdot \opt(G_2^{\typCii}) + f(G) - 2 \ . \label{eq:3vc-c2-iii}
					\end{align}
			\end{description}
			We distinguish the following cases depending on which type combination of \Cref{lemma:3vc-type-combinations} is present for $(\OPT_1,\OPT_2)$.
			\begin{enumerate}[label=(\alph*)]
				\item $(\typA, \{\typA,\typBi,\typBii,\typCi,\typCii,\typCiii\})$. Since $\OPT_1^{\typA}$ exists, \Cref{lemma:3vc-type-A-condition} gives $\opt_1 = \opt_1^{\typA} \geq \opt_1^{\typCii} + 3$.
				      \begin{description}
					      \item[Subtype \typCiisubba:] We have $\opt(G_2^{\typCii}) \leq \opt_2 + 5$ because we can turn $\OPT_2$ with the dummy edges $\{uy,vy,vw\}$ and at most two more edges, which are guaranteed by \Cref{lemma:3vc-spanning-tree} if required, into a solution for $G_2^{\typCii}$. Thus, \eqref{eq:3vc-c2-i} gives
						      \begin{align*}
							      \reduce(G) & \leq (\opt_1 - 3) + (\alpha \cdot (\opt_2 + 5) + f(G) - 2                                         \\
							                 & \leq \alpha \cdot \opt(G) + f(G) + (\alpha - 1)(5 - \opt_1) \leq \alpha \cdot \opt(G) + f(G)  \ .
						      \end{align*}
					      \item[Subtype \typCiisubb:] We have $\opt(G_2^{\typCii}) \leq \opt_2 + 6$ because we can turn $\OPT_2$ with the dummy edges $\{uy,vz,zy,wy\}$ and at most two more edges, which are guaranteed by \Cref{lemma:3vc-spanning-tree} if required, into a solution for $G_2^{\typCii}$. Thus, \eqref{eq:3vc-c2-ii} gives
						      \begin{align*}
							      \reduce(G) & \leq (\opt_1 - 3) + (\alpha \cdot (\opt_2 + 6) + f(G) - 3                                         \\
							                 & \leq \alpha \cdot \opt(G) + f(G) + (\alpha - 1)(6 - \opt_1) \leq \alpha \cdot \opt(G) + f(G)  \ .
						      \end{align*}
					      \item[Subtype \typCiisubc:] We have $\opt(G_2^{\typCii}) \leq \opt_2 + 5$ because we can turn $\OPT_2$ with the dummy edges $\{uy,vy,wy\}$ and at most two more edges, which are guaranteed by \Cref{lemma:3vc-spanning-tree} if required, into a solution for $G_2^{\typCii}$. Thus, \eqref{eq:3vc-c2-iii} gives
						      \begin{align*}
							      \reduce(G) & \leq (\opt_1 - 3) + (\alpha \cdot (\opt_2 + 5) + f(G) - 2                                          \\
							                 & \leq \alpha \cdot \opt(G) + f(G) + (\alpha - 1)(5 - \opt_1) \leq  \alpha \cdot \opt(G) + f(G)  \ .
						      \end{align*}
				      \end{description}

				\item $(\typBi, \{\typA,\typBi,\typBii,\typCi,\typCii\})$.
				      Since $\OPT_1^{\typBi}$ exists but the condition in \Cref{reduce-3vc:b1} did not apply, we conclude that $\opt_1^{\typBi} \geq \opt_1^{\min} + 2 = \opt_1^{\typCii} + 2$.
				      \begin{description}
					      \item[Subtype \typCiisubba:]
						      We have $\opt(G_2^{\typCii}) \leq \opt_2 + 4$ because we can turn $\OPT_2$ with the dummy edges $\{uy,vy, vw\}$ and at most one more edge, which is guaranteed by \Cref{lemma:3vc-spanning-tree} if required, into a solution for $G_2^{\typCii}$.
						      Thus, continuing \eqref{eq:3vc-c2-i} gives
						      \begin{align*}
							      \reduce(G) & \leq \opt_1 - 2 + \alpha \cdot (\opt_2 + 4) + f(G) - 2                                           \\
							                 & \leq \alpha \cdot \opt(G) + f(G) + (\alpha - 1)(4 - \opt_1) \leq \alpha \cdot \opt(G) + f(G) \ .
						      \end{align*}
					      \item[Subtype \typCiisubb:] We have $\opt(G_2^{\typCii}) \leq \opt_2 + 5$ because we can turn $\OPT_2$ with the dummy edges $\{uy,vz,zy,wy\}$ and at most one more edge, which is guaranteed by \Cref{lemma:3vc-spanning-tree} if required, into a solution for $G_2^{\typCii}$.
						      Thus,  \eqref{eq:3vc-c2-ii} gives
						      \begin{align*}
							      \reduce(G) & \leq \opt_1 - 2 + \alpha \cdot (\opt_2 + 5) + f(G) - 3                                           \\
							                 & \leq \alpha \cdot \opt(G) + f(G) + (\alpha - 1)(5 - \opt_1) \leq \alpha \cdot \opt(G) + f(G) \ .
						      \end{align*}
					      \item[Subtype \typCiisubc:] We have $\opt(G_2^{\typCii}) \leq \opt_2 + 4$ because we can turn $\OPT_2$ with the dummy edges $\{uy,vy,wy\}$ and at most one more edge, which is guaranteed by \Cref{lemma:3vc-spanning-tree} if required, into a solution for $G_2^{\typCii}$.
						      Thus, \eqref{eq:3vc-c2-iii} gives
						      \begin{align*}
							      \reduce(G) & \leq \opt_1 - 2 + \alpha \cdot (\opt_2 + 4) + f(G) - 2                                           \\
							                 & \leq \alpha \cdot \opt(G) + f(G) + (\alpha - 1)(4 - \opt_1) \leq \alpha \cdot \opt(G) + f(G) \ .
						      \end{align*}
				      \end{description}

				\item $(\typBii, \{\typA,\typBi,\typBii\})$. In this case $\OPT_1^{\typBii}$ exists but $t^{\min} = \typCii$. Thus, the definition of \Cref{alg:3-vertex-cuts} and \Cref{def:3vc-ties} imply that $\opt_1^{\typCii} \leq \opt^{\typBii}_1 - 1 = \opt_1 - 1$.
				      Furthermore, in this case at least two vertices of $\{u, v, w\}$ are in a 2EC component in $\OPT_2$.
				      \begin{description}
					      \item[Subtype \typCiisubba:] We have $\opt(G_2^{\typCii}) \leq \opt_2 + 3$ because we can turn $\OPT_2$ with two of the dummy edges $\{uy,vy,vw\}$ and at most one more edge, which is guaranteed by \Cref{lemma:3vc-spanning-tree} if required, into a solution for $G_2^{\typCii}$.
						      Thus, \eqref{eq:3vc-c2-i} gives
						      \begin{align*}
							      \reduce(G) & \leq \opt_1 - 1 + \alpha \cdot (\opt_2 + 3) + f(G) - 2                                           \\
							                 & \leq \alpha \cdot \opt(G) + f(G) + (\alpha - 1)(3 - \opt_1) \leq \alpha \cdot \opt(G) + f(G) \ .
						      \end{align*}
					      \item[Subtype \typCiisubb:] We have $\opt(G_2^{\typCii}) \leq \opt_2 + 4$ because we we can turn $\OPT_2$ with the dummy edges $\{vz,zy \}$, one of the dummy edges $\{uy, wy\}$, and at most one more edge, which is guaranteed by \Cref{lemma:3vc-spanning-tree} if required, into a solution for $G_2^{\typCii}$.
						      Thus, \eqref{eq:3vc-c2-ii} gives
						      \begin{align*}
							      \reduce(G) & \leq \opt_1 - 1 + \alpha \cdot (\opt_2 + 4) + f(G) - 3                                           \\
							                 & \leq \alpha \cdot \opt(G) + f(G) + (\alpha - 1)(4 - \opt_1) \leq \alpha \cdot \opt(G) + f(G) \ .
						      \end{align*}
					      \item[Subtype \typCiisubc:] We have $\opt(G_2^{\typCii}) \leq \opt_2 + 3$ because we can turn $\OPT_2$ with two of the dummy edges $\{uy,vy,wy\}$ and at most one more edge, which is guaranteed by \Cref{lemma:3vc-spanning-tree} if required, into a solution for $G_2^{\typCii}$.
						      Thus, \eqref{eq:3vc-c2-iii} gives
						      \begin{align*}
							      \reduce(G) & \leq \opt_1 - 1 + \alpha \cdot (\opt_2 + 3) + f(G) - 2                                           \\
							                 & \leq \alpha \cdot \opt(G) + f(G) + (\alpha - 1)(3 - \opt_1) \leq \alpha \cdot \opt(G) + f(G) \ .
						      \end{align*}
				      \end{description}

				\item $(\typCi, \{\typA,\typBi,\typCi\})$.
				      In this case $\opt_1^{\typCi}$ exists but $t^{\min} = \typCii$. Thus, the definition of \Cref{alg:3-vertex-cuts} and \Cref{def:3vc-ties} imply that $\opt_1^{\typCii} \leq \opt^{\typCi}_1 - 1 = \opt_1 - 1$.
				      \begin{description}
					      \item[Subtype \typCiisubba:] We have $\opt(G_2^{\typCii}) \leq \opt_2 + 3$ because we can turn $\OPT_2$ with the dummy edges $\{uy,vy,vw\}$ into a solution for $G_2^{\typCii}$.
						      Thus, continuing \eqref{eq:3vc-c2-i} gives
						      \begin{align*}
							      \reduce(G) & \leq \opt_1 - 1 + \alpha \cdot (\opt_2 + 3) + f(G) - 2                                           \\
							                 & \leq \alpha \cdot \opt(G) + f(G) + (\alpha - 1)(3 - \opt_1) \leq \alpha \cdot \opt(G) + f(G) \ .
						      \end{align*}
					      \item[Subtype \typCiisubb:] We have $\opt(G_2^{\typCii}) \leq \opt_2 + 4$ because we can turn $\OPT_2$ with the dummy edges $\{uy,vz,zy,wy\}$ into a solution for $G_2^{\typCii}$.
						      Thus, continuing \eqref{eq:3vc-c2-ii} gives
						      \begin{align*}
							      \reduce(G) & \leq \opt_1 - 1 + \alpha \cdot (\opt_2 + 4) + f(G) - 3                                           \\
							                 & \leq \alpha \cdot \opt(G) + f(G) + (\alpha - 1)(4 - \opt_1) \leq \alpha \cdot \opt(G) + f(G) \ .
						      \end{align*}
					      \item[Subtype \typCiisubc:] We have $\opt(G_2^{\typCii}) \leq \opt_2 + 3$ because we can turn $\OPT_2$ with the dummy edges $\{uy,vy,wy\}$ into a solution for $G_2^{\typCii}$.
						      Thus, continuing \eqref{eq:3vc-c2-iii} gives
						      \begin{align*}
							      \reduce(G) & \leq \opt_1 - 1 + \alpha \cdot (\opt_2 + 3) + f(G) - 2                                           \\
							                 & \leq \alpha \cdot \opt(G) + f(G) + (\alpha - 1)(3 - \opt_1) \leq \alpha \cdot \opt(G) + f(G) \ .
						      \end{align*}
				      \end{description}

				\item $(\typCii, \{\typA,\typBi\})$. We have $\opt_1^{\typCii} = \opt_1$.
				      \begin{description}
					      \item[Subtype \typCiisubba:] We have $\opt(G_2^{\typCii}) \leq \opt_2 + 3$ because we can turn $\OPT_2$ with the dummy  edges $\{uy,vy,vw\}$ into a solution for $G_2^{\typCii}$.
						      Thus, continuing \eqref{eq:3vc-c2-i} gives
						      \begin{align*}
							      \reduce(G) & \leq \opt_1 + \alpha \cdot (\opt_2 + 3) + f(G) - 2                   \\
							                 & \leq \alpha \cdot \opt(G) + f(G) + (\alpha - 1)(2 - \opt_1) + \alpha \\
							                 & \leq \alpha \cdot \opt(G) + f(G) \ .
						      \end{align*}
					      \item[Subtype \typCiisubb:] We have $\opt(G_2^{\typCii}) \leq \opt_2 + 4$ because we can turn $\OPT_2$ with the dummy edges $\{uy,vz,zy,wy\}$ into a solution for $G_2^{\typCii}$.
						      Thus, continuing \eqref{eq:3vc-c2-ii} gives
						      \begin{align*}
							      \reduce(G) & \leq \opt_1 + \alpha \cdot (\opt_2 + 4) + f(G) - 3                   \\
							                 & \leq \alpha \cdot \opt(G) + f(G) + (\alpha - 1)(3 - \opt_1) + \alpha \\
							                 & \leq \alpha \cdot \opt(G) + f(G) \ ,
						      \end{align*}
						      where we use $\opt_1 \geq \frac{2}{1-\alpha} \geq \frac{4\alpha - 3}{\alpha - 1}$ for $\alpha \in [\frac{6}{5},\frac{5}{4}]$.
					      \item[Subtype \typCiisubc:] We have $\opt(G_2^{\typCii}) \leq \opt_2 + 3$ because we can turn $\OPT_2$ with the dummy edges $\{uy,vy,wy\}$ into a solution for $G_2^{\typCii}$.
						      Thus, continuing \eqref{eq:3vc-c2-iii} gives
						      \begin{align*}
							      \reduce(G) & \leq \opt_1 + \alpha \cdot (\opt_2 + 3) + f(G) - 2                   \\
							                 & \leq \alpha \cdot \opt(G) + f(G) + (\alpha - 1)(2 - \opt_1) + \alpha \\
							                 & \leq \alpha \cdot \opt(G) + f(G) \ .
						      \end{align*}
				      \end{description}

				\item $(\typCiii, \typA)$. We have $\opt_1^{\typCii} = \opt_1^{\min} \leq \opt_1^{\typCiii} = \opt_1$.
				      \begin{description}
					      \item[Subtype \typCiisubba:] We have $\opt(G_2^{\typCii}) \leq \opt_2 + 2$ because we can turn $\OPT_2$ with the dummy edges $\{uy,vy\}$ into a solution for $G_2^{\typCii}$.
						      Thus, continuing \eqref{eq:3vc-c2-i} gives
						      \begin{align*}
							      \reduce(G) & \leq \opt_1 + \alpha \cdot (\opt_2 + 2) + f(G) - 2          \\
							                 & \leq \alpha \cdot \opt(G) + f(G) + (\alpha - 1)(2 - \opt_1) \\
							                 & \leq \alpha \cdot \opt(G) + f(G) \ .
						      \end{align*}
					      \item[Subtype \typCiisubb:] We have $\opt(G_2^{\typCii}) \leq \opt_2 + 3$ because we can turn $\OPT_2$ with the dummy edges $\{uy,vz,zy\}$ into a solution for $G_2^{\typCii}$.
						      Thus, continuing \eqref{eq:3vc-c2-ii} gives
						      \begin{align*}
							      \reduce(G) & \leq \opt_1 + \alpha \cdot (\opt_2 + 3) + f(G) - 3          \\
							                 & \leq \alpha \cdot \opt(G) + f(G) + (\alpha - 1)(3 - \opt_1) \\
							                 & \leq \alpha \cdot \opt(G) + f(G) \ .
						      \end{align*}
					      \item[Subtype \typCiisubc:] We have $\opt(G_2^{\typCii}) \leq \opt_2 + 2$ because we can turn $\OPT_2$ with the dummy edges $\{uy,vy\}$ into a solution for $G_2^{\typCii}$.
						      Thus, continuing \eqref{eq:3vc-c2-iii} gives
						      \begin{align*}
							      \reduce(G) & \leq \opt_1 + \alpha \cdot (\opt_2 + 2) + f(G) - 2          \\
							                 & \leq \alpha \cdot \opt(G) + f(G) + (\alpha - 1)(2 - \opt_1) \\
							                 & \leq \alpha \cdot \opt(G) + f(G) \ .
						      \end{align*}
				      \end{description}
			\end{enumerate}

		\item[(\Cref{reduce-3vc:c3}: $t^{\min} = \typCiii$)]
			Let $D^{\typCiii} \coloneq \{uv,uv,vw,vw\}$ denote the set of dummy edges used for the construction of $G_2^\typCiii$.
			First observe that the definition of \Cref{alg:3-vertex-cuts}, \Cref{lemma:3vc-c3-F}, and the induction hypothesis gives
			\begin{align}
				\reduce(G) = |H^{\typCiii}| & \leq |\OPT_1^{\typCiii}| + \abs{H_2^{\typCiii}} + \abs{F^{\typCiii}} \leq \opt_1^{\typCiii} + \reduce(G_2^{\typCiii}) \notag \\
				                            & \leq \opt_1^{\typCiii} + \alpha \cdot \opt(G_2^{\typCiii}) + f(G) \ . \label{eq:3vc-c3}
			\end{align}

			We distinguish the following cases depending on which type combination of \Cref{lemma:3vc-type-combinations} is present for $(\OPT_1,\OPT_2)$.
			\begin{enumerate}[label=(\alph*)]
				\item $(\typA, \{\typA,\typBi,\typBii,\typCi,\typCii,\typCiii\})$. Since $\OPT_1^{\typA}$ exists, \Cref{lemma:3vc-type-A-condition} gives $\opt_1 = \opt_1^{\typA} \geq \opt_1^{\typCiii} + 3$. Moreover, $\opt(G_2^{\typCiii}) \leq \opt_2 + 4$ because we can turn $\OPT_2$ with at most four dummy edges of $D^{\typCiii}$ into a solution for $G_2^{\typCiii}$. Thus, \eqref{eq:3vc-c3} gives
				      \begin{align*}
					      \reduce(G) & \leq \opt_1 - 3 + \alpha \cdot (\opt_2 + 4) + f(G)                                                        \\
					                 & \leq \alpha \cdot \opt(G) + f(G) + (\alpha - 1)(3 - \opt_1) + \alpha \leq \alpha \cdot \opt(G) + f(G) \ ,
				      \end{align*}
				      where we use $\opt_1 \geq \frac{2}{1-\alpha} \geq \frac{4\alpha - 3}{\alpha - 1}$ for $\alpha \in [\frac{6}{5},\frac{5}{4}]$.

				\item $(\typBi, \{\typA,\typBi,\typBii,\typCi,\typCii\})$.
				      Since $\OPT_1^{\typBi}$ exists but the condition in \Cref{reduce-3vc:b1} did not apply, we conclude that $\opt_1 = \opt_1^{\typBi} \geq \opt_1^{\min} + 2 = \opt_1^{\typCiii} + 2$.
				      Moreover, $\opt(G_2^{\typCiii}) \leq \opt_2 + 3$ because we can turn $\OPT_2$ with at most three dummy edges of $D^{\typCiii}$ into a solution for $G_2^{\typCiii}$.  Thus, \eqref{eq:3vc-c3} gives
				      \begin{align*}
					      \reduce(G) & \leq \opt_1 - 2 + \alpha \cdot (\opt_2 + 3) + f(G)                                                        \\
					                 & \leq \alpha \cdot \opt(G) + f(G) + (\alpha - 1)(2 - \opt_1) + \alpha \leq \alpha \cdot \opt(G) + f(G) \ .
				      \end{align*}

				\item $(\typBii, \{\typA,\typBi,\typBii\})$.
				      In this case $\OPT_1^{\typBii}$ exists but $t^{\min} = \typCiii$. Thus, the definition of \Cref{alg:3-vertex-cuts} and \Cref{def:3vc-ties} imply that $\opt_1^{\typCiii} \leq \opt^{\typBii}_1 - 1 = \opt_1 - 1$.
				      Moreover, $\opt(G_2^{\typCiii}) \leq \opt_2 + 2$ because we can turn $\OPT_2$ with at most two dummy edges of $D^{\typCiii}$ into a solution for $G_2^{\typCiii}$.  Thus, \eqref{eq:3vc-c3} gives
				      \begin{align*}
					      \reduce(G) & \leq \opt_1 - 1 + \alpha \cdot (\opt_2 + 2) + f(G)                                                        \\
					                 & \leq \alpha \cdot \opt(G) + f(G) + (\alpha - 1)(1 - \opt_1) + \alpha \leq \alpha \cdot \opt(G) + f(G) \ .
				      \end{align*}

				\item $(\typCi, \{\typA,\typBi,\typCi\})$.
				      In this case $\OPT_1^{\typCi}$ exists but $t^{\min} = \typCiii$. Thus, the definition of \Cref{alg:3-vertex-cuts} and \Cref{def:3vc-ties} imply that $\opt_1^{\typCiii} \leq \opt^{\typCi}_1 - 1 = \opt_1 - 1$.
				      Moreover, $\opt(G_2^{\typCiii}) \leq \opt_2 + 2$ because we can turn $\OPT_2$ with at most two dummy edges of $D^{\typCiii}$ into a solution for $G_2^{\typCiii}$. Thus, \eqref{eq:3vc-c3} gives
				      \begin{align*}
					      \reduce(G) & \leq \opt_1 - 1 + \alpha \cdot (\opt_2 + 2) + f(G)                                                        \\
					                 & \leq \alpha \cdot \opt(G) + f(G) + (\alpha - 1)(1 - \opt_1) + \alpha \leq \alpha \cdot \opt(G) + f(G) \ .
				      \end{align*}

				\item $(\typCii, \{\typA,\typBi\})$.
				      In this case $\OPT_1^{\typCii}$ exists but $t^{\min} = \typCiii$. Thus, the definition of \Cref{alg:3-vertex-cuts} and \Cref{def:3vc-ties} imply that $\opt_1^{\typCiii} \leq \opt^{\typCii}_1 - 1 = \opt_1 - 1$.
				      Moreover, $\opt(G_2^{\typCiii}) \leq \opt_2 + 1$ because we can turn $\OPT_2$ with at most one dummy edge of $D^{\typCiii}$ into a solution for $G_2^{\typCiii}$.  Thus, \eqref{eq:3vc-c3} gives
				      \begin{align*}
					      \reduce(G) & \leq \opt_1 - 1 + \alpha \cdot (\opt_2 + 1) + f(G)                                               \\
					                 & \leq \alpha \cdot \opt(G) + f(G) + (\alpha - 1)(1 - \opt_1) \leq \alpha \cdot \opt(G) + f(G) \ .
				      \end{align*}

				\item $(\typCiii, \typA)$. Using $\opt_1^{\typCiii} = \opt_1$ and $\opt(G_2^{\typCiii}) \leq \opt_2$ as $\OPT_2$ is already feasible for $G_2^{\typCiii}$, \eqref{eq:3vc-c3} implies
				      \begin{align*}
					      \reduce(G) \leq \opt_1 + \alpha \cdot \opt_2 + f(G)  \leq  \alpha \cdot \opt(G) + f(G) \ .
				      \end{align*}
			\end{enumerate}
	\end{description}

	This completes the proof of the lemma.
\end{proof}

\subsubsection{Auxiliary Lemmas for \Cref{alg:cycle-cuts}}

\begin{proposition}\label{prop:properties-for-cycle-cut-alg}
	If \Cref{alg:reduce} executes \Cref{alg:cycle-cuts}, then $|V(G)| > \frac{8}{\eps}$.
\end{proposition}

\begin{lemma}\label{lemma:cycle-cuts-opt-large}
	In \Cref{alg:cycle-cuts}, it holds that $\opt(G'_1) \geq \ceil{\frac{k}{\alpha-1}}$.
\end{lemma}

\begin{proof}
	Since $G'_1 = G_1 | \{v_1,\ldots,v_k\}$ contains at least $|V_1| + 1 \geq (\ceil{\frac{k}{\alpha-1}} - 1) + 1$ vertices and $\opt(G'_1)$ is a feasible 2ECSS solution for $G'_1$, $\opt(G'_1)$ must contain at least $\ceil{\frac{k}{\alpha-1}}$ edges.
\end{proof}

\begin{lemma}\label{lemma:cycle-cuts-feasible}
	In \Cref{alg:cycle-cuts}, if both $H_1 = \Reduce(G'_1)$ and $H_2 = \Reduce(G'_2)$ are 2EC spanning subgraphs of $G'_1$ and $G'_2$, respectively, then $H$ is a 2EC spanning subgraph of $G$.
\end{lemma}

\begin{proof}
	Note that both $G_1$ and $G_2$ and hence both $G'_1$ and $G'_2$ must be 2EC since $G$ is 2EC.
	By assumption $H_1$ and $H_2$ are 2EC spanning subgraphs of $G'_1$ and $G'_2$, respectively. Thus, their union is a 2EC spanning subgraph of $G|\{v_1,\ldots,v_k\}$. Since $H$ additionally to $H_1 \cup H_2$ contains a 2EC spanning subgraph of $G[\{v_1,\ldots,v_k\}]$, we conclude that $H$ is a 2EC spanning subgraph of $G$.
\end{proof}

\begin{lemma}\label{lemma:cycle-cuts-approx}
	Let $G$ be a 2ECSS instance. If every recursive call to $\Reduce(G')$ in \Cref{alg:cycle-cuts} on input $G$ satisfies \eqref{eq:reduce-invariant}, then  \eqref{eq:reduce-invariant} holds for input $G$.
\end{lemma}

\begin{proof}
	Assume $|V_1| \leq |V_2|$. Note that $\opt(G'_1) + \opt(G'_2) \leq \opt(G)$ and $|V(G'_1)| + |V(G'_2)| \leq |V(G)|$.
	By \Cref{prop:properties-for-cycle-cut-alg}, we have $|V(G)| \geq \frac{8}{\eps}$. Thus, we need to show that
	\(
	\reduce(G) \leq \alpha \cdot \opt(G) + f(G) = \alpha \cdot \opt(G) + 4\varepsilon \cdot |V(G)| - 16 \ .
	\)
	We distinguish three cases.
	\begin{description}
		\item[\textbf{($|V_2| \leq \frac{8}{\eps}$ and $|V_1| \leq \frac{8}{\eps}$).}] By \eqref{eq:reduce-invariant}, we have $H_i = \opt(G'_i)$ for all $i \in \{1,2\}$. Then
			\[
				\reduce(G) = |H_1| + |H_2| + k = \opt(G'_1) + \opt(G'_2) + k \leq \alpha \opt(G) + 4\varepsilon \cdot |V(G)| - 16 = \alpha \opt(G) + f(G) \ ,
			\]
			because $\alpha \geq 1$ and $|V(G)| \geq \frac{8}{\eps}$ gives $k \leq 16 \leq 4\varepsilon \cdot |V(G)| - 16$.
		\item[\textbf{($|V_2| > \frac{8}{\eps}$ and $|V_1| \leq \frac{8}{\eps}$).}] By \eqref{eq:reduce-invariant}, we have $H_1 = \opt(G'_1)$ and $|H_2| \leq \alpha \cdot \opt(G'_2) + f(G'_2)$. Then,
			\[
				\reduce(G) = |H_1| + |H_2| + k
				\leq \opt(G'_1) + \alpha \cdot \opt(G'_2) + f(G'_2) + k \leq \alpha \cdot \opt(G) + f(G) \ ,
			\]
			where we use that $f(G'_2) \leq f(G)$ and $\opt(G'_1) \geq \frac{k}{\alpha - 1}$ (\Cref{lemma:cycle-cuts-opt-large}).
		\item[\textbf{($|V_2| > \frac{8}{\eps}$ and $|V_1| > \frac{8}{\eps}$).}] By \eqref{eq:reduce-invariant}, we have
			\begin{align*}
				\reduce(G) & = |H_1| + |H_2| + k                                                                                                                \\
				           & \leq \alpha \cdot \opt(G'_1) + 4\varepsilon \cdot |V(G'_1)| - 16 + \alpha \cdot \opt(G'_2) + 4\varepsilon \cdot |V(G'_2)| - 16 + k \\
				           & \leq \alpha \cdot \opt(G) + 4\varepsilon \cdot |V(G)| - 16 = \alpha \cdot \opt(G) + f(G) \ ,
			\end{align*}
			because $k \leq 8$.
	\end{description}
	This completes the proof of the lemma.
\end{proof}

\subsubsection{Auxiliary Lemmas for \Cref{alg:k-cuts}}\label{app:prep:helper-k-cuts}

\begin{proposition}\label{prop:properties-for-k-cut-alg}
	If \Cref{alg:reduce} executes \Cref{alg:k-cuts} for $k=4$, then $|V(G)| > \frac{8}{\eps}$, and $G$ contains no cut vertices, non-isolating 2-vertex cuts, or large $3$-vertex cuts.
\end{proposition}

\begin{lemma}\label{lemma:k-cuts-opt-large}
	If the condition in \Cref{alg-k-cuts:case1} in \Cref{alg:k-cuts} holds, we have $\opt(G'_1) \geq \ceil{\frac{2k-2}{\alpha-1}}$.
\end{lemma}

\begin{proof}
	Since $G'_1 = G_1 | \{v_1,\ldots,v_k\}$ contains at least $|V_1| + 1 \geq (\ceil{\frac{2k-2}{\alpha-1}} - 1) + 1$ vertices and $\opt(G'_1)$ is a feasible 2ECSS solution for $G'_1$, $\opt(G'_1)$ must contain at least $\ceil{\frac{2k-2}{\alpha-1}}$ edges.
\end{proof}

\begin{lemma}\label{lemma:k-cuts-feasible1}
    Let $k=4$.
	If the condition in \Cref{alg-k-cuts:case1} in \Cref{alg:k-cuts} holds and if both $H_1 = \Reduce(G'_1)$ and $H_2 = \Reduce(G'_2)$ are 2EC spanning subgraphs of $G'_1$ and $G'_2$, respectively,
	then there exists $F \subseteq E$ with $|F| \leq 2k-2$ such that $H = H_1 \cup H_2 \cup F$ is a 2EC spanning subgraph of~$G$.
\end{lemma}

\begin{proof}
	We show the existence of the desired edges $F = F_1 \cup F_2$ inductively.

	We treat $H_1$ as a subgraph of $G_1$ and assume that $V(G_1) = V(H_1)$.
	Start with $F_1 = \emptyset$. Note that $H_1 \cup F_1$ contains at most $k$ connected components, because $H_1$ is 2EC for $G'_1$.
	For each $v_i \in \{v_1,\ldots,v_k\}$, let $H_1(v_i)$ denote the connected component of $H_1 \cup F_1$ that contains $v_i$. Note that $H_1(v_i)$ can be composed of the single vertex $v_i$.

	If $H_1 \cup F_1$ is not connected, let $v_i \in \{v_1,\ldots,v_k\}$ such that $H_1(v_i) \neq (H_1 \cup F_1)$, and suppose that there is no edge between $H_1(v_i)$ and $(H_1 \cup F_1) \setminus H_1(v_i)$ in $G$.
	If $H_1(v_i)$ is composed of the single vertex $v_i$, then $\{v_1,\ldots,v_k\} \setminus \{v_i\}$ is a large $3$-vertex cut in $G$, a contradiction to \Cref{prop:properties-for-k-cut-alg}. Otherwise, that is, $H_1(v_i)$ consists of at least two vertices, then $v_i$ is a cut vertex in $G$ that separates $V(H_1(v_i)) \setminus \{v_i\}$ from $V \setminus V((H_1 \cup F_1) \setminus H_1(v_i))$, a contradiction to \Cref{prop:properties-for-k-cut-alg}.
	Thus, there is an edge $e$ between $H_1(v_i)$ and $(H_1 \cup F_1) \setminus H_1(v_i)$ in $G$.
	We now add $e$ to $F_1$, and observe that this reduces the number of connected component of $H_1 \cup F_1$ by one.

	Thus, after repeating this argument by at most $k-1$ times, we showed the existence of a set of edges $F_1$ with $|F_1| \leq k-1$ such that $H_1 \cup F_1$ is connected and spans $V(G_1)$.

	By symmetry, we can prove the existence of edges $F_2$ with $|F_2| \leq k-1$ such that $H_2 \cup F_2$ is connected and spans $V(G_2)$.
	Finally, we set $F = F_1 \cup F_2$, and have that $H_1 \cup H_2 \cup F$ is 2EC and spans $V(G)$.
\end{proof}

\begin{lemma}\label{lemma:k-cuts-feasible2}
	If the condition in \Cref{alg-k-cuts:case2} in \Cref{alg:k-cuts} for $k=4$ holds and if $H_2 = \Reduce(G'_2)$ is a 2EC spanning subgraph of $G'_2$,
	then $H' = H_2 \cup Z$ is a 2EC spanning subgraph of~$G$.
\end{lemma}

\begin{proof}
	Since for the $k$-vertex cut, Condition 2 of \Cref{def:large_k_vertex_cut} applies and $H_2$ is a 2ECSS of $G'_2$, we conclude that $H'$ is a 2ECSS of $G$.
\end{proof}

\begin{lemma}\label{lemma:k-cuts-approx}
	Let $G$ be a 2ECSS instance. If every recursive call to $\Reduce(G')$ in \Cref{alg:k-cuts} for $k=4$ on input $G$ satisfies \eqref{eq:reduce-invariant}, then  \eqref{eq:reduce-invariant} holds for input $G$.
\end{lemma}

\begin{proof}
	By \Cref{prop:properties-for-k-cut-alg}, we have $|V(G)| \geq \frac{8}{\eps}$. Thus, we need to show that
	\(
	\reduce(G) \leq \alpha \cdot \opt(G) + f(G) = \alpha \cdot \opt(G) + 4\varepsilon \cdot |V(G)| - 16 \ .
	\)

	We first consider the case where the condition in \Cref{alg-k-cuts:case1} holds, that is,
	there exists a partition $(V_1,V_2)$ of $V \setminus \{v_1,\ldots,v_k\}$ with $E[V_1,V_2] = \emptyset$ and
	$\ceil{\frac{2k-2}{\alpha-1}} - 1 \leq |V_1| \leq |V_2|$.
	Note that $\opt(G'_1) + \opt(G'_2) \leq \opt(G)$ and $|V(G'_1)| + |V(G'_2)| \leq |V(G)|$.
	We distinguish three cases.
	\begin{description}
		\item[\textbf{($|V_2| \leq \frac{8}{\eps}$ and $|V_1| \leq \frac{8}{\eps}$).}]
			By \eqref{eq:reduce-invariant}, we have $H_i = \opt(G'_i)$ for all $i \in \{1,2\}$. Thus
			\begin{align*}
				\reduce(G) & = |H_1| + |H_2| + |F|                                                            \\
				           & \leq \opt(G'_1) + \opt(G'_2) + (2k-2)                                            \\
				           & \leq \alpha \opt(G) + 4\varepsilon \cdot |V(G)| - 16 = \alpha \opt(G) + f(G) \ ,
			\end{align*}
			because $\alpha \geq 1$ and $|V(G)| \geq \frac{8}{\eps}$ gives $2k - 2 \leq 16 \leq 4\varepsilon \cdot |V(G)| - 16$.
		\item[\textbf{($|V_2| > \frac{8}{\eps}$ and $|V_1| \leq \frac{8}{\eps}$).}]
			By \eqref{eq:reduce-invariant}, we have $H_1 = \opt(G'_1)$ and $|H_2| \leq \alpha \cdot \opt(G'_2) + f(G'_2)$. Then,
			\[
				\reduce(G) = |H_1| + |H_2| + |F|
				\leq \opt(G'_1) + \alpha \cdot \opt(G'_2) + f(G'_2) + (2k-2) \leq \alpha \cdot \opt(G) + f(G) \ ,
			\]
			where we use that $f(G'_2) \leq f(G)$ and $\opt(G'_1) \geq \frac{2k-2}{\alpha - 1}$ (\Cref{lemma:k-cuts-opt-large}).
		\item[\textbf{($|V_2| > \frac{8}{\eps}$ and $|V_1| > \frac{8}{\eps}$).}]
			By \eqref{eq:reduce-invariant}, we have
			\begin{align*}
				\reduce(G) & = |H_1| + |H_2| + |F|                                                                                                                   \\
				           & \leq \alpha \cdot \opt(G'_1) + 4\varepsilon \cdot |V(G'_1)| - 16 + \alpha \cdot \opt(G'_2) + 4\varepsilon \cdot |V(G'_2)| - 16 + (2k-2) \\
				           & \leq \alpha \cdot \opt(G) + 4\varepsilon \cdot |V(G)| - 16 = \alpha \cdot \opt(G) + f(G) \ ,
			\end{align*}
			because $2k - 2 \leq 16$.
	\end{description}

	Next, we consider the case where the condition in \Cref{alg-k-cuts:case2} holds.
	Note that $\opt(G'_1) + \opt(G'_2) \leq \opt(G)$ and $|V(G'_1)| + |V(G'_2)| \leq |V(G)|$.
	Since Condition 1 in \Cref{def:large_k_vertex_cut} does not hold, $|V_1| \leq \ceil{\frac{2k-2}{\alpha-1}} - 2 \leq \frac{8}{\eps}$ as $\eps \leq \frac{1}{100}$. Thus, we distinguish two cases.
	\begin{description}
		\item[\textbf{($|V_2| \leq \frac{8}{\eps}$ and $|V_1| \leq \frac{8}{\eps}$).}]
			We have $|Z| \leq \alpha \cdot \opt(G'_1)$, and by \eqref{eq:reduce-invariant}, we have $H_2 = \opt(G'_2)$. Thus,
			\begin{align*}
				\reduce(G) & = |Z| + |H_2| \leq \alpha \opt(G'_1) + \opt(G'_2)                                \\
				           & \leq \alpha \opt(G) + 4\varepsilon \cdot |V(G)| - 16 = \alpha \opt(G) + f(G) \ ,
			\end{align*}
			because $\alpha \geq 1$ and $|V(G)| \geq \frac{8}{\eps}$ gives $0 \leq 4\varepsilon \cdot |V(G)| - 16$.
		\item[\textbf{($|V_2| > \frac{8}{\eps}$ and $|V_1| \leq \frac{8}{\eps}$).}]
			We have $|Z| \leq \alpha \cdot \opt(G'_1)$, and by \eqref{eq:reduce-invariant}, we have  $|H_2| \leq \alpha \cdot \opt(G'_2) + f(G'_2)$. Thus,
			\begin{align*}
				\reduce(G) & = |Z| + |H_2| \leq \alpha \cdot \opt(G'_1) + \alpha \cdot \opt(G'_2) + f(G'_2) \\
				           & \leq \alpha \cdot \opt(G) + f(G) \ ,
			\end{align*}
			where we use that $f(G'_2) \leq f(G)$.
	\end{description}

	This completes the proof of the lemma.
\end{proof}

\subsection{Proof of \Cref{lem:reduction-to-structured}}\label{sec:reduce-proofs}

In this final subsection, we prove \Cref{lem:reduction-to-structured}. The statement is implied by the following three lemmas.
We first show that $\Reduce(G)$ runs in polynomial time. The proof is similar to the corresponding proof in \cite{GargGA23improved}.

\begin{lemma}[running time]\label{lem:reduction-to-structured-runtime}
	For any 2ECSS instance $G$, $\Reduce(G)$ runs in polynomial time in $|V(G)|$ if $\ALG(G)$ runs in polynomial time in $|V(G)|$.
\end{lemma}

\begin{proof}
	Let $n = |V(G)|$ and $m = |V(G)|$, and $\sigma = n^2 + m^2$ be the size of the problem. We first argue that every non-recursive step of \Cref{alg:reduce,alg:2-vertex-cuts,alg:3-vertex-cuts,alg:cycle-cuts,alg:k-cuts} can be performed in time polynomial in $\sigma$.
	Clearly, all computations performed in \Cref{alg:reduce} can be performed in polynomial time in $\sigma$. In particular, it is possible to check in polynomial-time whether $G$ contains an $\alpha$-contractible subgraph with at most $\frac{8}{\varepsilon}$ vertices: Enumerate over all subsets $W$ of vertices of the desired size; for each such $W$, one can compute in polynomial time a minimum-size 2ECSS $\OPT'$ of $G[W]$, if any such subgraph exists. Furthermore, one can compute in polynomial time a maximum cardinality subset of edges $F$ with endpoints in $W$ such that $G \setminus F$ is 2EC. Then, it is sufficient to compare the size of $\OPT'$, if exists, with $|E(G[W]) \setminus F|$.

	We now consider \Cref{alg:3-vertex-cuts}.
	We already argued that the computation in \Cref{reduce-3vc:compute-opt-types} can be executed in polynomial time in $\sigma$. Moreover, determining the $\typCii$ subcases in \Cref{reduce-3vc:c2-i,reduce-3vc:c2-ii,reduce-3vc:c2-iii} and selecting a desired $\OPT_1^{\typCii}$ solution in \Cref{reduce-3vc:c2-ii-sol-constr,reduce-3vc:c2-iii-sol-constr} and $\OPT_1^{\typCiii}$ solution in \Cref{reduce-3vc:c3-select-solution} can be done in constant time, as $G_1$ is of constant size in this case.
	Finding the corresponding sets of edges $F$ in \Cref{alg:3-vertex-cuts} can also be done in polynomial time in $\sigma$ by enumeration, because \Cref{lemma:3vc-both-large-F,lemma:3vc-b1-F,lemma:3vc-b2-F,lemma:3vc-c1-F,lemma:3vc-c2-i-F,lemma:3vc-c2-ii-F,lemma:3vc-c2-iii-F,lemma:3vc-c3-F} guarantee that there exists a desired set of edges $F$ of constant size.
	The same arguments can also be used to argue that  \cref{alg:2-vertex-cuts} can be implemented in polynomial time.
	Finally, \Cref{alg:cycle-cuts,alg:k-cuts} can clearly also be implemented in polynomial time. In particular, checking whether $G$ contains a large $4$-vertex cut by \Cref{def:large_k_vertex_cut}(2) can be done in polynomial time by enumerating all possible connectivity configurations of $H_2$, similar to enumerating the different solution types of 3-vertex cuts explained above.

	To bound the total running time including recursive steps, let $T(\sigma)$ be the running time of \Cref{alg:reduce} as a function of $\sigma$. In each step, starting with a problem of size $\sigma$, we generate either a single subproblem of size $\sigma' < \sigma$ (to easily see this in \Cref{alg:2-vertex-cuts,alg:3-vertex-cuts,alg:cycle-cuts,alg:k-cuts}, as we assume that the cuts are sufficiently large) or two subproblems of size $\sigma'$ and $\sigma''$ such that $\sigma',\sigma'' < \sigma$ and $\sigma' + \sigma'' \leq \sigma$. Thus, $T(\sigma) \leq \max\{T(\sigma'), T(\sigma') + T(\sigma'') \} + \mathrm{poly}(\sigma)$ with $\sigma',\sigma'' < \sigma$ and $\sigma' + \sigma'' \leq \sigma$. Therefore, $T(\sigma) \leq \sigma \cdot \mathrm{poly}(\sigma)$.
\end{proof}

Next, we show that $\Reduce(G)$ returns a feasible solution to 2ECSS on $G$.

\begin{lemma}[correctness]\label{lem:reduction-to-structured-correctness}
	For any 2ECSS instance $G$, $\Reduce(G)$ returns a 2EC spanning subgraph of~$G$.
\end{lemma}

\begin{proof}
	The proof is by induction over the tuple $(\abs{V(G)}, \abs{E(G)})$. The base cases are given by \Cref{reduce:bruteforce,reduce:call-alg} of \Cref{alg:reduce}.
	If we enter \Cref{reduce:bruteforce}, the claim trivially holds. If we call $\ALG$ in \Cref{reduce:call-alg}, the claim follows by the correctness of $\ALG$ for $(\alpha,\varepsilon)$-structured instances.

	If the condition in \Cref{reduce:1vc}
	applies, let $G_i = G[V_i \cup \{v\}]$. Note that both $G_1$ and $G_2$, which must be 2EC since $G$ is 2EC, each contain at most $|V(G)| - 1$ vertices, hence by induction hypothesis $\Reduce(G_1)$ and $\Reduce(G_2)$ are 2EC spanning subgraphs of $V_1 \cup \{v\}$ and $V_2 \cup \{v\}$, respectively. Thus, their union is a 2EC spanning subgraph of $G$.
	If the conditions in \Cref{reduce:loop} or \Cref{reduce:irrelevant} apply, $G' \coloneq G \setminus \{e\}$ must be 2EC by \cref{lemma:self-loops-parallel-edges,lemma:irrelevant-edge}, respectively, and $|V(G')| = |V(G)|$ and $|E(G')| = |E(G)| - 1$. Hence, by induction hypothesis $\Reduce(G')$ is a 2EC spanning subgraph of $G'$, and thus also for $G$.
	If the condition in \Cref{reduce:contractible} applies, the claim follows from \Cref{fact:decontract}.

	It remains to consider the cases when \cref{alg:2-vertex-cuts,alg:3-vertex-cuts,alg:cycle-cuts,alg:k-cuts} are executed.
	Observe that in these cases, $G$ is 2VC since the condition of \cref{reduce:1vc} is checked before. Moreover, note that all recursive calls to $\Reduce$ made in \cref{alg:2-vertex-cuts,alg:3-vertex-cuts,alg:cycle-cuts} consider graphs with strictly fewer vertices than $G$, hence the induction hypothesis applies.
	If \cref{alg:reduce} executes \cref{alg:2-vertex-cuts},
	the correctness follows from the induction hypothesis, \cref{lemma:2vc-type-A-condition} and \cref{lemma:2vc-both-large-F,lemma:2vc-b-F,lemma:2vc-c-F}.
	If \cref{alg:reduce} executes \Cref{alg:3-vertex-cuts},
	the correctness follows from the induction hypothesis, \Cref{lemma:2vc-type-A-condition} and \Cref{lemma:3vc-b1-F,lemma:3vc-b2-F,lemma:3vc-c1-F,lemma:3vc-c2-i-F,lemma:3vc-c2-ii-F,lemma:3vc-c2-iii-F,lemma:3vc-c3-F}.
	If \cref{alg:reduce} executes \Cref{alg:cycle-cuts}, the correctness follows from the induction hypothesis and \Cref{lemma:cycle-cuts-feasible}, and if \cref{alg:reduce} executes \Cref{alg:k-cuts}, the correctness follows from the induction hypothesis and \Cref{lemma:k-cuts-feasible1,lemma:k-cuts-feasible2}.
\end{proof}

The following lemma implies $\reduce(G) \leq (\alpha + 4\varepsilon) \opt(G)$ using the trivial bound $\opt(G) \geq |V(G)|$.

\begin{lemma}[approximation guarantee]\label{lem:reduction-to-structured-apx}
	For any 2ECSS instance $G$, \eqref{eq:reduce-invariant} is satisfied. That is,
	\[
		\reduce(G) \leq \begin{cases}
			\opt(G)                                                   & \quad \text{if } \abs{V(G)} \leq \frac{8}{\varepsilon} \ , \text{ and} \\
			\alpha \cdot \opt(G) + 4\varepsilon \cdot \abs{V(G)} - 16 & \quad \text{if } \abs{V(G)} > \frac{8}{\varepsilon} \ .
		\end{cases}
	\]
\end{lemma}
\begin{proof}
	The proof is by induction over the tuple $(\abs{V(G)}, \abs{E(G)})$. The base cases are given by \Cref{reduce:bruteforce,reduce:call-alg} of \Cref{alg:reduce}.
	If we enter \Cref{reduce:bruteforce}, the stated bound trivially holds.
	If the condition in \cref{reduce:loop} or \cref{reduce:irrelevant} applies, we consider a graph $G' \coloneq G \setminus \{e\}$ with the same number of vertices as $G$ and the same optimal size of a 2ECSS by \Cref{lemma:self-loops-parallel-edges} and \cref{lemma:irrelevant-edge}, respectively. Thus, the claim follows from the induction hypothesis. If we call $\ALG$ in \Cref{reduce:call-alg}, the claim follows because we assume that $\ALG(G)$ returns an $\alpha$-approximate solution to 2ECSS on $G$ and $G$ is in this case $(\alpha,\varepsilon)$-structured by the correctness of $\Reduce$ (cf.\ \Cref{lem:reduction-to-structured-correctness}).

	If the condition in \cref{reduce:1vc} applies, recall that by \cref{lemma:cut-vertex-egal} we have $\opt(G) = \opt(G_1) + \opt(G_2)$, where $G_i \coloneq G[V_i \cup \{v\}]$ for $i \in \{1,2\}$. We distinguish a few cases depending on the sizes of $V_1$ and $V_2$. Assume w.l.o.g.\ $|V_1| \leq |V_2|$.
	If $|V_2| \leq \frac{8}{\varepsilon}$, then $\reduce(G) = \opt(G_1) + \opt(G_2) = \opt(G) \leq \alpha \cdot \opt(G) + 4\varepsilon |V(G)| - 16$ since $\alpha \geq 1$ and $|V(G)| \geq \frac{8}{\varepsilon}$.
	Otherwise, $|V_2| \geq \frac{8}{\varepsilon}$.
	If $|V_1| \leq \frac{8}{\varepsilon}$, then by the induction hypothesis we have $\reduce(G) \leq \opt(G_1) + \alpha \opt(G_2) + f(G_2) \leq \alpha \opt(G) + f(G)$ using $\alpha \geq 1$ and $f(G_2) \leq f(G)$.
	Otherwise, $|V_2| \geq |V_1| \geq \frac{8}{\varepsilon}$, and the induction hypothesis gives
	$\reduce(G) \leq \alpha \opt(G_1) + \alpha \opt(G_2) + 4\varepsilon (|V(G_1)| + |V(G_2)|) - 16 \leq \alpha \opt(G) + 4\varepsilon |V(G)| + 4\varepsilon - 16 \leq \alpha \opt(G) + f(G)$, because $4\varepsilon \leq 16$.

	If the condition in \cref{reduce:contractible} applies,
	we have $\opt(G) = |\OPT(G) \cap {V(H)}^2| + |\OPT(G) \setminus {V(H)}^2| \geq \frac{1}{\alpha}|H| + \opt(G|H)$ by the definition of a contractible subgraph $H$ and observing that $\OPT(G) \setminus {V(H)}^2$ induces a feasible 2EC spanning subgraph of $G|H$.
	If $|V(G|H)| \leq \frac{8}{\varepsilon}$ one has $\reduce(G) = |H| + \opt(G|H) \leq \alpha \opt(G) \leq \alpha \opt(G) + 4\varepsilon|V(G)| - 16$ since $|V(G)| \geq \frac{8}{\varepsilon}$ whenever \cref{alg:reduce} reaches \cref{reduce:contractible}.
	Otherwise, $|V(G|H)| > \frac{8}{\varepsilon}$, we conclude that $\reduce(G) \leq |H| + \alpha \opt(G|H) + 4\varepsilon |V(G|H)| - 4 \leq \alpha \opt(G) + 4\varepsilon |V(G)| - 16$ since $|V(G)| \geq |V(G|H)|$.

	Finally, if \cref{alg:reduce} executes \cref{alg:2-vertex-cuts,alg:3-vertex-cuts,alg:cycle-cuts,alg:k-cuts}, we have that $|V(G)| > \frac{8}{\varepsilon}$. Thus, in these cases the statement follows from the induction hypothesis and \cref{lemma:2vc-approx}, \cref{lemma:3vc-approx}, \cref{lemma:cycle-cuts-approx}, or \cref{lemma:k-cuts-approx}, respectively.
\end{proof}

\section{Canonical 2-Edge Cover}
\label{sec:canonical}

\lemmaCanonicalMain*

\begin{proof}
	We first compute a triangle-free $2$-edge cover $H'$, which can be done in polynomial time~\cite{kobayashi2023approximation,hartvigsen2024finding}, and show how to modify it to obtain a canonical $2$-edge cover with the same number of edges.
	As in~\cite{BGGHJL25}, we can assume that in $H'$, there is a component which has at least $32$ vertices.
	This can be achieved by guessing a subset of $31$ edges in an optimal 2ECSS, which induce a connected graph.
	We can fix these $31$ edges to be in $H'$ by subdividing each edge into a path of length $2$ so that any $2$-edge cover will include these paths.

	Since $H'$ is a triangle-free $2$-edge cover, each component is either a 2EC component that contains at least $4$ edges, or complex.

	If $H'$ is canonical, we are done. Otherwise, we apply the following algorithm to $H'$:
	While there exist sets of edges $F_A \subseteq E \setminus H'$ and $F_R \subseteq H'$ with $|F_A| \leq |F_R| \leq 2$ such that $H'' \coloneq (H' \setminus F_R) \cup F_A$ is a triangle-free $2$-edge cover and either $|H''| < |H'|$, or $|H''| = |H'|$ and $H''$ contains strictly less connected components than $H'$, or $H''$ contains strictly less bridges than $H'$, or $H''$ contains less cut vertices in 2EC components than $H'$, we update $H' \coloneq H''$.

	Clearly, each iteration of this algorithm can be executed in polynomial time. Moreover, there can only be polynomially many iterations.
	Thus, let $H$ be equal to $H'$ when the algorithm terminates.
	We prove that $H$ is a canonical 2-edge cover of $G$. This proves the lemma because the algorithm does not increase the number of edge compared to a minimum triangle-free $2$-edge cover of $G$. To this end, we consider the following cases.

	\textbf{Case 1:} $H$ contains a 2EC component $C$ such that $4 \leq |E(C)| \leq 8$ and $C$ is not a cycle on $|E(C)|$ vertices. We distinguish three exhaustive cases, and prove in any case that Case~1 cannot occur.

	\textbf{Case 1.1:}
	$C$ has a cut vertex $v$. Observe that $v$ must have degree at least $4$ in $C$.
	Furthermore, each connected component in $C \setminus \{v\}$ contains at least $2$ vertices, as otherwise $C$ contains parallel edges, a contradiction to $G$ being structured.
	Assume that there is a component $K_1$ in $C \setminus \{v\}$ with $2$ vertices and let $u_1$ and $u_2$ be these vertices.
	Note that the three edges $\{v u_1, u_1 u_2, u_2 v\}$ must be in $C$ as otherwise $G$ is not structured since it must contain parallel edges.
	If there exists an edge $u_1 w \in E \setminus H$ (or $u_2 w \in E\setminus H$ by symmetry) for $w \in V \setminus V(C)$, then we can remove $v u_1$ from $H$ and add $u_1 w$ to obtain a triangle-free $2$-edge cover $H''$ with strictly less components and $|H| = |H''|$, a contradiction to the algorithm.
	If there exists an edge $u_1 w \in E \setminus H$ (or $u_2 w \in E$ by symmetry) for $w \in V(C) \setminus \{v, u_1, u_2 \}$, then we can remove $v u_1$ from $H$ and add $u_1 w$ to obtain a triangle-free $2$-edge cover $H''$ with $|H''| = |H|$ in which $v$ is not a cut vertex anymore.
	If neither of these two cases apply, we conclude $\deg_G(u_i) = 2$ for $i=1,2$. Thus, $v - u_1 - u_2 - v$ is a contractible cycle, a contradiction to $G$ being structured.
	Now we assume each connected component in $C \setminus \{v\}$ has at least $3$ vertices (and at least $2$ edges).
	Since the degree of $v$ is at least $4$ in $C$ and $|E(C)| \leq 8$, the only possible case is that $C$ consists of two $4$-cycles intersecting only at one vertex $v$, say, $v-x_1-x_2-x_3-v$ and $v-y_1-y_2-y_3-v$.
	Note $x_1,y_1,x_3,y_3$ are symmetric.
	We focus on $x_1$.
	If there exists some edge $x_1w$ for some $w \in V \setminus V(C)$,
	then we can remove $x_1v$ from $H$ and add $x_1w$ to obtain a triangle-free $2$-edge cover $H''$ with $|H''| = |H|$ with fewer components.
	if $x_1w \in E(G)$ for some $w \in \{y_1,y_2,y_3\}$, then we can remove $x_1v$ from $H$ and add $x_1w$ to obtain a triangle-free $2$-edge cover $H''$ with $|H''| = |H|$ in which $v$ is not a cut vertex anymore.
	Hence either $x_1$ (by symmetry also $x_3$) has degree $2$ in $G$ or $x_1x_3 \in E(G)$.
	In the former case, $x_1-x_2-x_3-v$ is contractible.
	The case for $y_1,y_3$ is symmetric.
	So the only remaining case is that the degrees of $x_1,x_3,y_1,y_3$ in $G$ are exactly $3$ in $G$ and $x_1x_3, y_1y_3 \in E(G)$.
	Now we consider $x_2$.
		There must be some edge $x_2w \in E$ with $w \in V\setminus \{x_1,x_3,v\}$, as otherwise $v$ is a cut vertex of $G$.
		We add $x_2w$ and $x_1x_3$ to $H$, and remove $x_1v,x_2x_3$ to obtain a triangle-free $2$-edge cover $H''$ with $|H''|=|H|$.
		If $w \notin C$, $H''$ has fewer connected components.
		If $w \in C$, $H'$ has the same number of components but fewer cut vertices.

	\textbf{Case 1.2:}
	$C$ contains no 2-vertex cut (and, thus, $C$ is 3-vertex-connected ). This implies that every vertex has degree at least 3 in $C$, as otherwise the neighboring vertices of a vertex of $C$ with degree $2$ form a $2$-vertex cut.
	But then it must hold that $4 \leq |V(C)| \leq 5 $, since $|V(C)| \leq 3$ is not possible with the above degree constraint and if $|V(C)| \geq 6$, then the number of edges in $C$ is at least $9$, a contradiction.
	But if $|V(C)| = 4$ and $C$ is $3$-vertex-connected, $C$ must be a complete graph on $4$ vertices. Therefore, there is an edge $e \in E(C)$ such that $H \setminus \{e\}$ is a triangle-free 2-edge cover, a contradiction to the algorithm.
	Similarly, if $|V(C)|=5$, there must be exactly $8$ edges and $4$ vertices with degree $3$ and $1$ vertex with degree $4$.
	Since $C$ is $3$-vertex-connected, removing the only vertex with degree $4$ from $C$ results in a connected graph with every vertex having degree $2$, which must be a $\cFour$.
	Hence $C$ is $\cFour$ plus one vertex connected to each vertex of the $\cFour$.
	Therefore, there is an edge $e \in E(C)$ such that $H \setminus \{e\}$ is a triangle-free 2-edge cover, a contradiction to the algorithm.

	\textbf{Case 1.3:} $C$ contains a 2-vertex cut $\{ v_1, v_2 \}$.
	Note that each connected component of $C \setminus \{v_1, v_2 \}$ must have an edge to both $v_1$ and $v_2$, as otherwise $C$ contains a cut vertex and we are in Case 1.1.
	Hence, there are at most $4$ connected components in $C \setminus \{v_1, v_2 \}$, as otherwise the number of edges in $C$ is at least 10.

	First, assume that $C \setminus \{v_1, v_2 \}$ contains exactly $4$ connected components.
	Then each of the connected components in  $C \setminus \{v_1, v_2 \}$ must be a single vertex and we label them as $u_1, u_2, u_3, u_4$.
	If there is some edge between $u_iw$ for some $i$ and $w \in V\setminus V(C)$, then we can remove $v_1 u_i$ from $H$ and add $u_i w$ to obtain a triangle-free $2$-edge cover $H''$ with strictly fewer components and $|H| = |H''|$, a contradiction to the algorithm.
	If there is an edge between $u_i$ and $u_j$ in $G$ for some $1 \leq i < j \leq 4$, then we can remove $u_iv_2, u_jv_1$ from $H$ and add $u_iu_j$ to obtain a triangle-free 2-edge cover with fewer edges, which contradicts the algorithm.
	Hence $u_i$ has degree $2$ in $G$ for each $i$.
	This implies that $C$ is contractible since any optimal solution includes the incident edges of $u_i$ and in total at least $8$ edges in $C$.

	Second, assume that $C \setminus \{v_1, v_2 \}$ contains exactly $3$ connected components.
	Observe that each connected component of $C \setminus \{v_1, v_2 \}$
	has at most $3$ vertices.
	Otherwise $|E(C)| > 8$.
	If there is some connected component of $C \setminus \{v_1, v_2 \}$ with $3$ vertices (a path $u_1-u_2-u_3$),
	then the other two connected components of $C \setminus \{v_1, v_2 \}$ must consist of a single vertex, say, $u_4$ and $u_5$.
	In this case,
	both $v_1$ and $v_2$ have some incident edge with the other endpoint in $\{u_1,u_2,u_3\}$.
	Otherwise, either $v_1$ or $v_2$ is a cut vertex in $C$ and we are in Case 1.1.
	Further, since both $u_1$ and $u_3$ have at least $2$ incident edges in $C$, both $u_1$ and $u_3$ have some incident edge with the other endpoint in ${v_1,v_2}$.
	Hence we can assume that $u_1v_1$ and $u_3v_2$ are in $E(C)$ and $E(C) = \{u_1u_2,u_2u_3,u_1v_1,u_3v_2, v_1u_4, v_1u_5, v_2u_4, v_2u_5\}$.
	Similar to the previous arguments in the last paragraph, there is no edge $u_iw$ in $G$ for $i \neq 2$ and $w \in V \setminus V(C)$, no edge between $\{u_1, u_3\}$ and $\{u_4,u_5\}$, and $u_4u_5 \notin E(G)$.
	So the optimal solution has to include at least $2$ incident edges of $u_1,u_3, u_4, u_5$ and hence at least $7$ edges in $G[C]$, which implies $C$ is contractible.
	Now consider the case where each connected component of $C \setminus \{v_1, v_2 \}$ has at most $2$ vertices.
	By the same arguments in the previous case, there is no edge between these connected component of $C \setminus \{v_1, v_2 \}$ in $G$ and there is no edge between these components and $V\setminus V(C)$.
	Hence we can argue that each vertex in $C \ \{v_1, v_2\}$ has degree $2$ in $G$ and hence $C$ is contractible.

	Hence, assume that $C \setminus \{v_1, v_2 \}$ contains exactly 2 connected components $K_1$ and $K_2$ and without loss of generality assume $|E(K_1)| \geq |E(K_2)|$.
	Recall that each connected component of $C \setminus \{v_1, v_2 \}$ must have an edge to both $v_1$ and $v_2$.
	Hence, $\{v_1, v_2 \}$ is incident to at least 4 edges in $C$. Thus, $|E(K_2)| \leq |E(K_1)| \leq 3$.
	If $|E(K_1)| = 3$, $K_1$ is a path on 4 vertices, or a triangle on 3 vertices, or a star on 4 vertices (one center and 3 leaves).
	If $K_1$ is a triangle, then either there is an edge $e$ in $C$ such that $H \setminus \{e\}$ is a triangle-free $2$-edge cover (a contradiction to the algorithm) or one vertex in triangle is a cut vertex of $C$ (we are in Case 1.1).
	If $K_1$ is a star on 4 vertices, then there must be 3 edges from the leaves of $K_1$ to $\{v_1, v_2 \}$, as otherwise $H$ is not a 2-edge cover.
	Then $C$ contains exactly 8 edges.
	Let $u_1$ be the center of the star $K_1$, and let $u_2,u_3,u_4$ be the leaves of the star.
	Let $u_5$ be the only vertex in $K_2$.
	Then by symmetry we can assume that
	$E(C)=\{u_1u_2,u_1u_3,u_1u_4,v_1u_2,v_1u_3,v_2u_4, v_1u_5,v_2u_5\}$.
	By similar arguments in the previous cases, there is no edge in $G$ between $\{u_2,u_3,u_4,u_5\}$ and $V \setminus V(C)$.
	Then we claim that $G[\{u_2,u_3,u_4,u_5\}]$ either has no edge or has only one edge $u_4u_5$, which implies that any optimal solution has to include at least $7$ edges in $G[C]$ and $C$ is contractible.
	To prove the claim, if $u_2u_3 \in G$, then we can remove $v_1u_2, u_1u_3$ from $H$ and add $u_2u_3$
	to obtain a triangle-free $2$-edge cover with fewer edges.
	If $u_2u_5 \in E(G)$,
	then we can remove $u_1u_2, v_1u_5$ from $H$ and add $u_2u_5$
	to obtain a triangle-free $2$-edge cover with fewer edges.
	The case for $u_3u_5$ is symmetric.
	If $u_3u_4 \in E(G)$, then we can remove $v_1u_3, u_1u_3$ from $H$ and add $u_3u_4$
	to obtain a triangle-free $2$-edge cover with fewer edges.
	The case for $u_2u_4$ is symmetric.
	We conclude with the claim and that $K_1$ cannot be a star.
	Finally consider the case where $K_1$ is a path on 4 vertices.
	But then one leaf of $K_1$ must be adjacent to $v_1$ in $C$ and the other leaf of $K_1$ must be adjacent to $v_2$ in $C$. Since $K_2$ is a single vertex in this case and must be incident to both $v_1$ and $v_2$, $C$ contains a \cSeven as a spanning subgraph, a contradiction.
	Hence, we have that $|E(K_2)| \leq |E(K_1)| \leq 2$.
	Since $G$ does not contain parallel edges, $K_1$ must be a simple path on $i$ vertices for some $i \in \{1, 2, 3\}$. The same is true for $K_2$.
	By similar reasoning as before we can observe that the leaves of $K_1$ and $K_2$, respectively, must be incident to $v_1$ and $v_2$ (or, if $K_1$ or $K_2$ is a single vertex, then it must be incident to both $v_1$ and $v_2$) and therefore we observe that $C$ contains a spanning cycle as a subgraph, a contradiction.

	\textbf{Case 2:}
	There is some complex component $C$ of $H$ containing a pendant block $B$ with less than 6 edges.
	Assume there is a pendant block $B$ with exactly 5 edges. The cases in which $B$ has 3 or 4 edges is similar.
	Note that $B$ must contain either 4 or 5 vertices.
	If $B$ contains exactly $4$ vertices %
	then $B$ must contain a simple cycle on 4 vertices with an additional edge $e$ being a chord of this cycle, which is redundant, and, thus, a contradiction to the algorithm.
	Hence, assume that $B$ contains exactly $5$ vertices, i.e., $B$ is a cycle $b_1 - b_2 - b_3 - b_4 - b_5 - b_1$, and assume that $b_1$ is incident to the unique bridge $f$ in $C$.
	If either $b_2$ (or $b_5$ by symmetry) is incident to an edge $e$ with neighbor in $V \setminus V(B)$ we can remove $b_1 b_2$ from $H$ and add $e$ to obtain a triangle-free $2$-edge cover $H''$. If $e$ is incident to some vertex in $V \setminus V(C)$, the number of connected components in $H''$ is strictly less compared to $H$, and otherwise, that is, $e$ is incident to some vertex in $V(C) \setminus V(B)$, or $H''$ has strictly less bridges compared to $H$. This is a contradiction to the algorithm. Thus, in the following we can assume that $b_2$ and $b_5$ are only adjacent to vertices of $V(B)$.

		By \Cref{lem:3-matching},
		$b_3$ and $b_4$ both must be incident to some edge $e_3$ and $e_4$ going to some vertex in $V \setminus V(B)$, respectively.

	If $b_2b_5 \in E(G)$,
	then we can remove $b_2 b_3$ and $b_1 b_5$ from $H$ and add $b_2 b_5$ and $e_3$ to obtain a triangle-free $2$-edge cover $H''$.
	If $e_3$ is incident to some vertex $V \setminus V(C)$, the number of connected components in $H''$ is strictly less compared to $H$, and otherwise, that is, $e_3$ is incident to some vertex $V(C) \setminus V(B)$, $H''$ has strictly less bridges compared to $H$. This is a contradiction to the algorithm.

		If $b_3b_5 \in E(G)$,
		then we can remove $b_1 b_5$ and $b_3 b_4$ from $H$ from $H$ and add $b_3 b_5$ and $e_4$ to obtain a triangle-free $2$-edge cover $H''$.
		If $e_4$ is incident to some vertex $V \setminus V(C)$, the number of connected components in $H''$ is strictly less compared to $H$, and otherwise, that is, $e_4$ is incident to some vertex $V(C) \setminus V(B)$, $H''$ has strictly less bridges compared to $H$.
		By symmetry, we also have $b_2b_4 \notin E(G)$.
		Hence $C$ is a $\cFive$ with $b_2, b_5$ having degree $2$ in $G$.

	\textbf{Case 3:}
	There is some complex component $C$ of $H$ containing a non-pendant block $B$ that contains less than $5$ edges.
	If $B$ has $3$ edges, i.e., $B$ is a triangle.
	In this case, we observe that in any case there must be an edge $e$ in $B$ such that $H \setminus \{e\}$ is a triangle-free $2$-edge cover, a contradiction to the algorithm.
	If $B$ has $4$ edges, i.e., $B$ is a $\cFour$.
	We number the vertices as $b_1,b_2,b_3,b_4$.
	If there are two adjacent vertices $b_i$ and $b_{i+1}$ ($b_1$ and $b_4$ are also adjacent), then $H \setminus \{b_ib_{i+1}\}$ is a triangle-free $2$-edge cover, a contradiction to the algorithm.
	Since $B$ is a non-pendant block, w.l.o.g.\ we assume that $b_1$ and $b_3$ are incident to some bridge in $C$,
	and $b_2$ and $b_4$ are vertices of degree $2$ in $C$.
	If $b_2b_4 \in E$, we can add the edge $b_2b_4$, and remove $b_1b_4$ and $b_2b_3$ to obtain a triangle-free $2$-edge cover with fewer edges, a contradiction to the algorithm.
	If $b_2w \in E$ for some $w \notin V(B)$, then we can add $b_2w$ and remove $b_2b_3$ to obtain a triangle-free $2$-edge cover $H''$ with either fewer bridges (if $w \in V(C)$) or fewer connected components (if $w \in V \setminus V(C)$).
	It holds the same for $b_4$.
	Hence we can assume that $b_2, b_4$ are vertices of degree $2$ in $G$, which implies that $B$ is contractible, a contraction.

		\textbf{Case 4:}
		There are two non-complex components of $H$ that are both $\cFour$s and can be merged into a single $\cEight$ by removing two edges of $H$ and adding two edges.
		We can enumerate all pairs of $\cFour$s in $H$ that are adjacent in $G$ and check whether they can be merged into a $\cEight$.
		Then we obtain $H''$ with fewer components and $|H''| = |H|$.
\end{proof}

\printbibliography

\end{document}